\pdfoutput=1
\RequirePackage{ifpdf}
\ifpdf 
\documentclass[pdftex]{sigma}
\else
\documentclass{sigma}
\fi

\numberwithin{equation}{section}

\newtheorem{thm}{Theorem}[section]
\newtheorem*{Theorem*}{Theorem}

\newtheorem{lem}[thm]{Lemma}
\newtheorem{prop}[thm]{Proposition}
 { \theoremstyle{definition}
\newtheorem{Def}[thm]{Definition}

\newtheorem{ex}[thm]{Example}
\newtheorem{rem}[thm]{Remark} }

\newcommand{\e}{\mathrm{e}}

\def\IZ{\mathbb {Z}}
\def\IQ{\mathbb {Q}}
\def\IR{\mathbb {R}}
\def\IC{\mathbb {C}}
\def\IP{\mathbb {P}}

\def\cC{\mathcal{C}}
\def\cE{\mathcal{E}}
\def\cM{\mathcal{M}}
\def\cW{\mathcal{W}}

\def\mfm{\mathfrak{m}}

\def\sff{\mathsf{f}}

\def\sfu{\mathsf{u}}
\def\sfv{\mathsf{v}}
\def\sfx{\mathsf{x}}
\def\sfy{\mathsf{y}}
\def\sfA{\mathsf{A}}
\def\sfB{\mathsf{B}}
\def\sfC{\mathsf{C}}
\def\sfD{\mathsf{D}}

\def\VO{V}
\def\sfm{\mathsf{f}^{\mathrm{MV}}}

\begin{document}
\allowdisplaybreaks

\newcommand{\arXivNumber}{2303.14154}

\renewcommand{\PaperNumber}{043}

\FirstPageHeading

\ShortArticleName{Some Generalizations of Mirzakhani's Recursion and Masur--Veech Volumes}

\ArticleName{Some Generalizations of Mirzakhani's Recursion\\ and Masur--Veech Volumes via Topological Recursions}

\Author{Hiroyuki FUJI~$^{\rm a}$ and Masahide MANABE~$^{\rm bc}$}

\AuthorNameForHeading{H.~Fuji and M.~Manabe}

\Address{$^{\rm a)}$~Center for Mathematical and Data Sciences and Department of Mathematics, Kobe University,\\
\hphantom{$^{\rm a)}$}~Rokko, Kobe 657-8501, Japan}
\EmailD{\href{mailto:hfuji@math.kobe-u.ac.jp}{hfuji@math.kobe-u.ac.jp}}

\Address{$^{\rm b)}$~Osaka Central Advanced Mathematical Institute, Osaka Metropolitan University,\\
\hphantom{$^{\rm b)}$}~3-3-138 Sugimoto, Sumiyoshi-ku, Osaka, 558-8585, Japan}
\EmailD{\href{mailto:masahidemanabe@gmail.com}{masahidemanabe@gmail.com}}

\Address{$^{\rm c)}$~Department of Mathematics, Graduate School of Science, Osaka University, \\
\hphantom{$^{\rm c)}$}~Toyonaka, Osaka 560-0043, Japan}

\ArticleDates{Received April 04, 2023, in final form May 09, 2024; Published online May 27, 2024}

\Abstract{Via Andersen--Borot--Orantin's geometric recursion, a twist of the topological recursion was proposed, and a recursion for the Masur--Veech polynomials was uncovered. The purpose of this article is to explore generalizations of Mirzakhani's recursion based on physical two-dimensional gravity models related to the Jackiw--Teitelboim gravity and to provide an introduction to various realizations of topological recursion. For generalized Mirzakhani's recursions involving a Masur--Veech type twist, we derive Virasoro constraints and cut-and-join equations, and also show some computations of generalized volumes for the physical two-dimensional gravity models.}


\Classification{81T45; 14D21; 14N10}

\section{Introduction}\label{sect:Intro}

A remarkable identity \cite{McShane} concerning the lengths of simple closed geodesics on a once-punctured torus
with a complete finite-area hyperbolic structure was discovered by McShane in his Ph.D.\ Thesis \cite{McShane_PhD}.
After his remarkable discovery, McShane's identity was generalized
to bordered hyperbolic Riemann surfaces of higher genus in a series of Mirzakhani's papers \cite{Mirz3,Mirz1}.
The generalized identity leads to a striking recursion relation for the Weil--Petersson volumes of moduli spaces of
bordered Riemann surfaces, referred to as \emph{Mirzakhani's recursion}.

\begin{thm}[Mirzakhani's recursion \cite{Mirz3,Mirz1}]
Let $V_{g,n}^{\mathrm{WP}}(L_1,\dots,L_n)$ be
the Weil--Petersson volume for the moduli space of bordered connected Riemann surfaces of genus $g$
with $n$ ordered boundary components of lengths $L_1,\dots,L_n$.
The Weil--Petersson volumes for $2g-2+n>0$ obey Mirzakhani's recursion
\begin{align}
L_1V_{g,n}^{\mathrm{WP}}(L_1,\dots,L_n)
={}&\frac{1}{2}\int_{\mathbb{R}_+^2}
D^{\mathrm{WP}}(L_1,\ell,\ell') P^{\mathrm{WP}}_{g,n}(\ell,\ell',L_K)
\ell\ell'{\rm d}\ell {\rm d}\ell'\nonumber
\\
&
+\sum_{m=2}^n\int_{\mathbb{R}_+}
R^{\mathrm{WP}}(L_1,L_m,\ell) V^{\mathrm{WP}}_{g,n-1}(\ell,L_{K\setminus\{m\}})\ell {\rm d}\ell,\label{eq:Mirzakhani_rec}
\end{align}
where $L_K=\{L_2,\dots,L_n\}$ and $L_J=\{L_{i_1},\dots,L_{i_{|J|}}\}$
for $J=\{i_1,\dots,i_{|J|}\}\subset K=\{2,\dots,n\}$.
Here $P^{\mathrm{WP}}_{g,n}$ is
\begin{align}
&
P^{\mathrm{WP}}_{g,n}(\ell,\ell',L_K)
=V_{g-1,n+1}^{\mathrm{WP}}(\ell,\ell',L_K)+
\mathop{\sum_{h+h'=g}}
\limits_{J \cup J'=K}^{\mathrm{stable}}
V^{\mathrm{WP}}_{h,1+|J|}(\ell,L_J) V^{\mathrm{WP}}_{h',1+|J'|}(\ell',L_{J'}),
\label{eq:P_WP}
\end{align}
where stable in the sum means that $h$, $h'$, $J$, $J'$ obey $2h-1+|J|>0$
and $2h'-1+|J'|>0$, and~$D^{\mathrm{WP}}$ and $R^{\mathrm{WP}}$ are
\begin{align*}
D^{\mathrm{WP}}(x,y,z)={}&\int_{0}^{x}H^{\mathrm{WP}}(y+z,x') {\rm d}x',
\\
R^{\mathrm{WP}}(x,y,z)={}&\frac{1}{2}\int_{0}^{x}
\left(H^{\mathrm{WP}}(z,x'+y)+H^{\mathrm{WP}}(z,x'-y)\right){\rm d}x',
\end{align*}
where the recursion kernel $H^{\mathrm{WP}}(x,y)$ is
\begin{align*}
H^{\mathrm{WP}}(x,y)=\frac{1}{1+\mathrm{e}^{\frac{x+y}{2}}}+\frac{1}{1+\mathrm{e}^{\frac{x-y}{2}}}.
\end{align*}
Especially for $(g,n)=(0,3)$ and $(1,1)$,\footnote{We employ a different normalization for $V^{\mathrm{WP}}_{1,1}(L_1)$
compared to the one used in \cite{Mirz3,Mirz1} by a factor of $2$.}
\begin{align*}
&V^{\mathrm{WP}}_{0,3}(L_1,L_2,L_3)=1,
\qquad
L_1V^{\mathrm{WP}}_{1,1}(L_1)=
\frac12 \int_{\mathbb{R}_+}
D^{\mathrm{WP}}(L_1,\ell,\ell) \ell {\rm d}\ell.
\end{align*}
\end{thm}

\begin{table}[t]\renewcommand{\arraystretch}{1.2}
 \centering
 \begin{tabular}{|ll|}
 \hline
 Bosonic model & $\mathsf{y}$-coordinate function \\
 \hline
KdV (Kontsevich's matrix model) &$\mathsf{y}^{\mathrm{KdV}}(z)=z+\sum_{a\ge 2}\mathsf{u}_a z^a$\\
 Weil--Petersson (JT gravity) & $\mathsf{y}^{\mathrm{WP}}(z)=\frac{1}{2\pi}\sin(2\pi z)$ \\
 Airy (topological gravity) & $\mathsf{y}^{\mathrm{A}}(z)=z$ \\
 FZZT ($(2,p)$ minimal string) & $\mathsf{y}^{\mathrm{M}(p)}(z)=\frac{(-1)^{\frac{p-1}{2}}}{2\pi} T_p\big(\frac{2\pi}{p}z\big)$
 \\
 \hline
 Supersymmetric model & $\mathsf{y}$-coordinate function \\
 \hline
 BGW (BGW matrix model) & $\mathsf{y}^{\mathrm{BGW}}(z)=\frac{1}{z}+\sum_{a\ge 1}\mathsf{v}_a z^a$ \\
 Super Weil--Petersson (JT supergravity) & $\mathsf{y}^{\mathrm{SWP}}(z)=\frac{1}{z}\cos(2\pi z)$ \\
 Bessel (analogue of topological gravity) & $\mathsf{y}^{\mathrm{B}}(z)=\frac{1}{z}$ \\
 Brane ($(2,2p-2)$ minimal superstring) & $\mathsf{y}^{\mathrm{SM}(p)}(z)=\frac{(-1)^{\frac{p-1}{2}}}{z} U_{p-1}\big(\frac{2\pi}{p}z\big)$
\\
\hline
 \end{tabular}
 \caption{$\mathsf{y}$-coordinate functions of physical 2D gravity models.
 $T_p(z)$ and $U_p(z)$ denote the Chebyshev polynomials of the first and second kind, respectively.}\label{tab:spectral_curve}
\end{table}

In \cite{EO2007,Eynard:2007fi},
it was shown that a Laplace transform of Mirzakhani's recursion for the Weil--Petersson volumes
obeys the \emph{Chekhov--Eynard--Orantin $($CEO$)$ topological recursion} \cite{EO2007}.
The CEO topological recursion was originally found in the asymptotic analysis of correlation functions of Hermitian matrix models \cite{Alexandrov:2003pj,Eynard:2004mh},
and the basic data of the recursion relation is extracted from algebro-geometric data of a \emph{spectral curve}.
The spectral curve $\mathcal{C}$ consists of basic data~$(\Sigma;\mathsf{x},\mathsf{y},B)$: a compact Riemann surface $\Sigma$, coordinate functions $(\mathsf{x},\mathsf{y})$ on $\Sigma^{\otimes 2}$, and
a bidifferential $B$ on $\Sigma^{\otimes 2}$.
In this paper, we focus on the following class of spectral curves with the basic data:
\begin{align}
\label{eq:spectral_curve_data_1}
\Sigma=\IP^1,\qquad
\mathsf{x}(z)=\frac{1}{2}z^2,\qquad B(z,w)=\frac{{\rm d}z\otimes {\rm d}w}{(z-w)^2},
\end{align}
where the remaining $\mathsf{y}$-coordinate function is specified depending on the models.
In particular, the $\mathsf{y}$-coordinate function for
the Laplace dual of Mirzakhani's recursion
is $\sfy^{\mathrm{WP}}(z)$ in Table \ref{tab:spectral_curve}.
This spectral curve resides in the class of the KdV spectral curve in Table \ref{tab:spectral_curve} which involves time variables $\mathsf{u}_a$ and leads to the asymptotic expansion of the tau-function of the KdV hierarchy
given by Kontsevich's matrix integral \cite{Kontsevich}
via the CEO topological recursion \cite{EO2007}.

In recent years, several fascinating developments and extensions of the CEO topological recursion for the Weil--Petersson volumes
have been reported in theoretical physics and geometry.
In theoretical physics, the non-perturbative studies of the \emph{Jackiw--Teitelboim $($JT$)$ gravity}
motivated by gauge/gravity correspondence uncovered a novel aspect of the Weil--Petersson volumes and their recursions.
In Saad--Shenker--Stanford's work \cite{Saad:2019lba},
the Weil--Petersson volume of the moduli space of hyperbolic bordered Riemann surfaces arises
in the computation of the path integral of the partition function in the JT gravity,
and the physical interpretation of the JT gravity partition function as a matrix integral was pointed out.

In terms of the JT gravity interpretation, the coordinate function of the spectral curve for the Weil--Petersson volumes can be found from the disk partition function of the $(2,p)$ minimal string theory in the background of \emph{Fateev--Zamolodchikov--Zamolodchikov--Teschner $($FZZT$)$ brane}
\mbox{\cite{Fateev:2000ik,Teschner:2000md}} in the $p\to\infty$ limit \cite{Saad:2019lba}.%
\footnote{
The spectral curve for the $(2,p)$ minimal string was also considered in \cite{BergereEynard09}.}
The $(2,p)$ minimal string theory for $p=1$ is, in particular, equivalent to the topological gravity, which is also known as Kontsevich--Witten's intersection theory on the moduli space of stable curves.
The coordinate function of the spectral curve
for the topological gravity
is the KdV spectral curve with all time variables set to zero, referred to as the Airy spectral curve.

Saad--Shenker--Stanford's analysis was further extended to the JT supergravity by Stanford--Witten's work \cite{Stanford:2019vob}.
The path integral for the partition function of the JT supergravity is performed over the moduli space of super Riemann surfaces which are constructed as Riemann surfaces equipped with a spin structure \cite{Crane,Felder:2019iqj,Ip:2016ojn,Ip:2017msi,LeBrun,Penner:2015xla,Rosly,Witten:2012ga}.
In \cite{Stanford:2019vob}, a supersymmetric extension of Mirzakhani's recursion for hyperbolic (Neveu--Schwarz) bordered super Riemann surfaces was derived, and the spectral curve of the CEO topological recursion for the supersymmetric extension of the Weil--Petersson volumes was unveiled.

\begin{thm}[Stanford--Witten's recursion \cite{Norbury:2020vyi,Stanford:2019vob}]
Let $V_{g,n}^{\mathrm{SWP}}(L_1,\dots,L_n)$ be
the supersymmetric analogue of the Weil--Petersson volume, referred to as the super Weil--Petersson volume, for the moduli space of bordered connected super Riemann surfaces of genus $g$
with $n$ ordered NS boundary components of lengths $L_1,\dots,L_n$.%
\footnote{
There are several choices of the orientation- and time-reversal symmetries to define the JT supergravity.
Depending on the choice of these symmetries, the sign and power of $2$ factors must be implemented to find the partition function and the supersymmetric volume $V_{g,n}^{\mathrm{SWP}}$.
In this article, we adopt the normalization of $V_{g,n}^{\mathrm{SWP}}$ to agree with $V^{\Theta}_{g,n}$ defined in \cite{Norbury:2020vyi} (see equation~\eqref{sw_volume}).
}
The super Weil--Petersson volumes for~${2g-2+n>0}$ obey
the same recursion relation as Mirzakhani's recursion \eqref{eq:Mirzakhani_rec} with replacements:
\begin{align*}
V^{\mathrm{WP}}\to V^{\mathrm{SWP}},\qquad
D^{\mathrm{WP}}\to D^{\mathrm{SWP}},\qquad
R^{\mathrm{WP}}\to R^{\mathrm{SWP}},
\end{align*}
where $D^{\mathrm{SWP}}$ and $R^{\mathrm{SWP}}$ are
\begin{align*}
D^{\mathrm{SWP}}(x,y,z)={}&H^{\mathrm{SWP}}(y+z,x),
\\
R^{\mathrm{SWP}}(x,y,z)={}&\frac{1}{2}
\left(H^{\mathrm{SWP}}(z,x+y)+H^{\mathrm{SWP}}(z,x-y)\right),
\end{align*}
and the kernel function $H^{\mathrm{SWP}}(x,y)$ is
\begin{align*}
H^{\mathrm{SWP}}(x,y)=\frac{1}{4\pi}\left(
\frac{1}{\cosh\frac{x-y}{4}}-\frac{1}{\cosh\frac{x+y}{4}}
\right).
\end{align*}
\end{thm}

The spectral curve \cite{Norbury:2020vyi,Stanford:2019vob} for the super Weil--Petersson volumes was found as a specialization of the BGW spectral curve for the tau-function of the Br\'ezin--Gross--Witten (BGW) model \cite{Brezin:1980rk,Gross:1980he}.
The BGW spectral curve involves time variables $\mathsf{v}_a$, and we find the well-known spectral curve referred to as the Bessel spectral curve \cite{Do:2016odu}
by setting all the time variables $\mathsf{v}_a$ to zero.
By comparison with the above bosonic models in the KdV hierarchy, one can naturally consider a one-parameter family
of spectral curves, which interpolates Stanford--Witten's curve for the super Weil--Petersson volumes and the Bessel spectral curve,
a supersymmetric analogue of the spectral curve for the $(2,p)$ minimal string.%
\footnote{
The $\mathsf{y}$-coordinate functions for the bosonic model $\mathsf{y}^{\mathrm{boson}}$ and the supersymmetric model $\mathsf{y}^{\mathrm{super}}$ are related by \smash{$\mathsf{y}^{\mathrm{super}}=\frac{\partial \mathsf{y}^{\mathrm{boson}}}{\partial \mathsf{x}}$},
where $\mathsf{x}(z)=z^2/2$ for the physical 2D gravity models in Table \ref{tab:spectral_curve}.
}
From some physical observations, the basic data of the spectral curve for this supersymmetric analogue are
expected to be found from some brane partition functions of type 0A $(2,2p-2)$ minimal superstring.\footnote{$p$ is an odd positive integer in the $(2,p)$ minimal string and the $(2,2p-2)$ minimal superstring.}

In geometry, the framework of Mirzakhani's recursion was generalized
on the basis of the Teichm\"uller theory by Andersen--Borot--Orantin's work \cite{Andersen_GR}, which is named as the \emph{geometric recursion}.
The basic data of the geometric recursion consists of measurable functions on the Teichm\"uller space of a bordered Riemann surface,
and McShane-Mirzakhani's identity is represented in the framework of the geometric recursion.
In this article, we call the generalized Mirzakhani's recursion as \emph{Andersen--Borot--Orantin $($ABO$)$ topological recursion},
which arises from the geometric recursion.
And for the above physical 2D gravity models, the ABO topological recursion is obtained as a Laplace dual of the CEO topological recursion.

In another work of Mirzakhani's \cite{Mirz_MV},
the enumerative problem of simple closed geodesics in hyperbolic bordered Riemann surfaces was extended and an elegant combinatorial approach to the computation of the \emph{Masur--Veech volume} of the moduli space of quadratic differentials
on Riemann surfaces with marked points was formulated explicitly.
In this approach,
the combinatorial data of the distribution of simple closed geodesics is described by
stable graphs, and the Masur--Veech volumes are computed by combinations of Weil--Petersson volumes.
In recent years, Mirzakhani's combinatorial approach to compute the Masur--Veech volumes was established further in a series of works by Delecroix, Goujard, Zograf, and Zorich \cite{DGZZ1,DGZZ3,DGZZ2,DGZZ_latest,DGZZ4}. (See also \cite{CMSZBGL,GM18} for related works.)

The Masur--Veech volume $\mathrm{Vol}\mathcal{Q}_{g,n}$ for the moduli space of quadratic differentials $q\in \mathcal{Q}_{g,n}$
is labeled by the order of zeros and poles of $q$.
In particular for the \emph{principal stratum} of the moduli space $\mathcal{Q}_{g,n}$ of quadratic differentials,
a novel connection between Delecroix--Goujard--Zograf--Zorich's result \cite{DGZZ_latest} and the ABO topological recursion was proposed in \cite{Andersen_MV,Andersen_GR}.
The ABO topological recursion to compute the Masur--Veech volumes is the Laplace dual of the CEO topological recursion
for the Airy spectral curve accompanied with an action of \emph{twist}.
The twist action shifts
the basic data of the ABO topological recursion (i.e., functions $D$ and $R$ in Mirzakhani's recursion), which implements the combinatorial data of stable graphs.
It was shown in \cite{Andersen_MV} that
the constant term in the polynomial obtained from the twisted ABO topological recursion for the Airy spectral curve provides the Masur--Veech volume $\mathrm{Vol}\mathcal{Q}_{g,n}$.

In this article, we discuss the following points on the basis of the above developments for the physical 2D gravity models listed in Table~\ref{tab:spectral_curve}:
\begin{enumerate}\itemsep=0pt
\item[(1)] derivation of a Mirzakhani type ABO topological recursion as the Laplace dual of the CEO topological recursion (Sections \ref{sec:examples_gr}, \ref{sec:examples_gr_super}, and Appendix \ref{sec:Derivation_Mirz});
\item[(2)] generalizations of the Masur--Veech volume
(Sections \ref{sec:examples_gr_tw} and \ref{sect:MasurVeech}), and
a direct proof of the twisted CEO topological recursion
in Theorem \ref{thm_tw_CEO} (Theorem \ref{thm:Laplace_dual_GR}
of Section \ref{sec:twisting_tr});
\item[(3)] derivation of Virasoro constraints with or without Masur--Veech type twist and their solutions via cut-and-join equations (Section \ref{sec:virasoro_constraint});
\item[(4)] physical derivations of the basic data of the spectral curves and the Masur--Veech twist action of the ABO topological recursion (see Appendix \ref{sec:physical}).
\end{enumerate}

On the first point, we will derive kernel functions of the generalized Mirzakhani's recursions for the 2D gravity models
in Table \ref{tab:spectral_curve}.
The recursion kernel for each model is listed in Table \ref{tab:kernel}.
\begin{table}[t]\renewcommand{\arraystretch}{1.2}
 \centering
 \begin{tabular}{|ll|}
 \hline
 Bosonic model & Recursion kernel $H(x,y)$ \\
 \hline
 Weil Petersson (JT gravity) & $H^{\mathrm{WP}}(x,y)=\frac{1}{1+\mathrm{e}^{\frac{x+y}{2}}}+\frac{1}{1+\mathrm{e}^{\frac{x-y}{2}}}$ \\
 Airy (topological gravity) & $H^{\mathrm{A}}(x,y)=\theta(y-x)+\theta(-x-y)$ \\
 FZZT ($(2,p)$ minimal string) & $u_j:=\frac{p}{2\pi}\sin\big(j\pi/p\big)$\\
 \multicolumn{2}{|l|}{
 \hspace{1cm}$H^{\mathrm{M}(p)}(x,y)=-\sum_{j=1}^{(p-1)/2}(-1)^j\cos\big(\frac{\pi}{p}j\big)
\left(\mathrm{e}^{-u_j(x+y)} \theta(x+y)
+\mathrm{e}^{-u_j(x-y)} \theta(x-y)\right)
$}\\
\multicolumn{2}{|l|}{$
\hspace{3.45cm}{}+\sum_{j=0}^{(p-1)/2}(-1)^j\cos\big(\frac{\pi}{p}j\big)
\left(\mathrm{e}^{u_j(x+y)} \theta(-x-y)
+\mathrm{e}^{u_j(x-y)} \theta(y-x)\right)$}\\
 \hline
 Supersymmetric model & Recursion kernel $H(x,y)$ \\
 \hline
 Super Weil--Petersson (JT supergravity) & $H^{\mathrm{SWP}}(x,y)=\frac{1}{4\pi}\Big(
\frac{1}{\cosh\frac{x-y}{4}}-\frac{1}{\cosh\frac{x+y}{4}}
\Big)$ \\
 Bessel (topological gravity) & $H^{\mathrm{B}}(x,y)=\delta(x-y)-\delta(x+y)$ \\
 Brane ($(2,2p-2)$ minimal superstring)& $u'_j:=\frac{p}{2\pi}\sin\left((j-1/2)\pi/p\right)$\\
 \multicolumn{2}{|l|}{
 \hspace{1cm}$H^{\mathrm{SM}(p)}(x,y)=\frac{1}{2\pi}\sum_{j=1}^{(p-1)/2}(-1)^j
\cos^2\big(\frac{\pi}{p}\big(j-\frac{1}{2}\big)\big)\big(\mathrm{e}^{-u'_j(x+y)} \theta(x+y)$}\\
\multicolumn{2}{|l|}{
$\hspace{3.5cm}
{}-\mathrm{e}^{-u'_j(x-y)} \theta(x-y)+\mathrm{e}^{u'_j(x+y)} \theta(-x-y)
-\mathrm{e}^{u'_j(x-y)} \theta(y-x)\big)
$}\\
\multicolumn{2}{|l|}{$
\hspace{3.5cm}
{}+\delta_{p,1}\left(\delta(x-y)-\delta(x+y)\right)$}
\\
\hline
 \end{tabular}
 \caption{Recursion kernels of physical 2D gravity models ($\theta$ denotes the Heaviside step function).}\label{tab:kernel}
\end{table}

On the second point, we will discuss twisted volume polynomials $V_{g,n}\big[\mathsf{f}^{\mathrm{MV}}\big]$ with Masur--Veech type twist $\mathsf{f}^{\mathrm{MV}}$ for the 2D gravity models listed in Table \ref{tab:kernel}.
For these physical models, the twisted Mirzakhani type ABO topological recursion is
\begin{align*}
L_1V_{g,n}\big[\mathsf{f}^{\mathrm{MV}}\big](L_1,\dots,L_n)
={}&\frac{1}{2}\int_{\mathbb{R}_+^2}
D\big[\mathsf{f}^{\mathrm{MV}}\big](L_1,\ell,\ell') P_{g,n}\big[\mathsf{f}^{\mathrm{MV}}\big](\ell,\ell',L_K)
\ell\ell' {\rm d}\ell {\rm d}\ell'
\\
&
+\sum_{m=2}^n\int_{\mathbb{R}_+} R\big[\mathsf{f}^{\mathrm{MV}}\big](L_1,L_m,\ell) V_{g,n-1}\big[\mathsf{f}^{\mathrm{MV}}\big](\ell,L_{K\setminus\{m\}})
\ell {\rm d}\ell,
\end{align*}
where $P_{g,n}\big[\mathsf{f}^{\mathrm{MV}}\big]$ is given by equation~\eqref{eq:P_WP} for the twisted volume polynomials, and $D\big[\mathsf{f}^{\mathrm{MV}}\big]$ and~$R\big[\mathsf{f}^{\mathrm{MV}}\big]$ are
\begin{align*}
D\big[\mathsf{f}^{\mathrm{MV}}\big](L_1,L_2,L_3)={}&
D(L_1,L_2,L_3)+R(L_1,L_2,L_3) \mathsf{f}^{\mathrm{MV}}(L_2)
\\
&
+R(L_1,L_3,L_2) \mathsf{f}^{\mathrm{MV}}(L_3)
+L_1 \mathsf{f}^{\mathrm{MV}}(L_2) \mathsf{f}^{\mathrm{MV}}(L_3),
\\
R\big[\mathsf{f}^{\mathrm{MV}}\big](L_1,L_2,L_3)={}&
R(L_1,L_2,L_3)+L_1 \mathsf{f}^{\mathrm{MV}}(L_3).
\end{align*}
Here the Masur--Veech type twist function $\mathsf{f}^{\mathrm{MV}}$ is $\mathsf{f}^{\mathrm{MV}}(\ell)=\frac{1}{\mathrm{e}^{\ell}-1}$. In fact, the Masur--Veech volume $\mathrm{Vol}\mathcal{Q}_{g,n}$ for the moduli space of quadratic differentials on a~Riemann surface of genus $g$ with $n$ marked points is the constant term of the twisted volume~$V^{\mathrm{A}}_{g,n}\big[\mathsf{f}^{\mathrm{MV}}\big]$ for the symplectic volume $V^{\mathrm{A}}_{g,n}$ of the moduli space $\mathcal{M}_{g,n}$ of stable curves of genus $g$ with $n$ marked points in Kontsevich--Witten's theory.
In this article, we refer to the twist action of the topological recursion by the function $\mathsf{f}^{\mathrm{MV}}$
as the \emph{Masur--Veech type twist}.
We will compute an analogue of the Masur--Veech volume for each 2D gravity model
by a combinatorial method developed in~\mbox{\cite{DGZZ_latest,Mirz_MV}}.

The main claim of this part is a derivation of the CEO topological recursion for twisted multidifferentials $\omega_{g,n}\big[\mathsf{f}^{\mathrm{MV}}\big]$
as a Laplace dual of the twisted volume polynomials $V_{g,n}\big[\mathsf{f}^{\mathrm{MV}}\big]$.

\begin{thm}[twisted CEO topological recursion \cite{Andersen_MV}]\label{thm_tw_CEO}
Let $V_{g,n}\big[\mathsf{f}^{\mathrm{MV}}\big]$ be the twisted volume polynomials for the physical $2$D gravity models in Table {\rm \ref{tab:kernel}}, which are expanded as
\begin{align*}
V_{g,n}\big[\mathsf{f}^{\mathrm{MV}}\big](L_1,\dots,L_n)=\sum_{a_1,\dots,a_n\ge 0}F^{(g)}\big[\mathsf{f}^{\mathrm{MV}}\big]_{a_1,\dots,a_n}\prod_{i=1}^n\frac{L_i^{2a_i}}{(2a_i+1)!}.
\end{align*}
Then, for $2g-2+n>0$, the multidifferentials $\omega_{g,n}\big[\mathsf{f}^{\mathrm{MV}}\big]$
obtained from $V_{g,n}\big[\mathsf{f}^{\mathrm{MV}}\big]$,
\begin{align*}
\omega_{g,n}\big[\mathsf{f}^{\mathrm{MV}}\big](z_1,\dots,z_n)=
\sum_{a_1,\dots,a_n\ge 0}F^{(g)}\big[\mathsf{f}^{\mathrm{MV}}\big]_{a_1,\dots,a_n}
\otimes_{i=1}^n \zeta_{\mathrm{H}}(2a_i+2;z_i) {\rm d}z_i,
\end{align*}
where
\begin{align*}
\zeta_{\mathrm{H}}(2d;z)=
\frac{1}{z^{2d}} + \frac12 \sum_{\mfm \in \IZ^*} \frac{1}{(z+\mfm)^{2d}}
\end{align*}
is the Hurwitz zeta function,\footnote{$\mathbb{Z}^*$ implies $\mathbb{Z}\setminus \{0\}$.}
obey the CEO topological recursion twisted by $\mathsf{f}^{\mathrm{MV}}$ such that
\begin{align*}
&
\omega_{g,n}\big[\mathsf{f}^{\mathrm{MV}}\big](z_1,\dots,z_n)
=\mathop{\mathrm{Res}}\limits_{w=0}K\big[\sfm\big](z_1,w)
\mathcal{R}\omega_{g,n}\big[\sfm\big](w,z_K),
\end{align*}
where $z_K=\{z_2, \dots, z_n\}$,
\begin{align*}
&
K\big[\mathsf{f}^{\mathrm{MV}}\big](z,w)
=\frac{(-1) {\rm d}z}{\left(\mathsf{y}(w)-\mathsf{y}(-w)\right){\rm d}w}
\left(
\frac{1}{z^2-w^2}+\frac{1}{2}\sum_{\mathfrak{m}\in\mathbb{Z}^*}
\frac{1}{(z+\mfm)^2-w^2}
\right),
\\
&
\mathcal{R}\omega_{g,n}\big[\sfm\big](w,z_K)
=\omega_{g-1,n+1}\big[\sfm\big](w,-w,z_K)
\\
&\phantom{\mathcal{R}\omega_{g,n}\big[\sfm\big](w,z_K)
=}{}
+\mathop{\sum_{h+h'=g}}
\limits_{J \cup J'=K}^{\mathrm{no (0,1)}}
\omega_{h,1+|J|}\big[\sfm\big](w, z_J) \omega_{h',1+|J'|}\big[\sfm\big](-w, z_{J'}),
\end{align*}
and
\begin{align*}
\omega_{0,2}\big[\mathsf{f}^{\mathrm{MV}}\big](z_1,z_2)
={}&B\big[\mathsf{f}^{\mathrm{MV}}\big](z_1,z_2)
=\frac{{\rm d}z_1\otimes {\rm d}z_2}{(z_1-z_2)^2}+\frac{1}{2}\sum_{\mathfrak{m}\in\mathbb{Z}^*}\frac{{\rm d}z_1\otimes {\rm d}z_2}{(z_1-z_2+\mathfrak{m})^2}
\\
={}&\zeta_{\mathrm{H}}(2;z_1-z_2) {\rm d}z_1 \otimes {\rm d}z_2.
\end{align*}
\end{thm}

On the third point, we will focus on an algebraic aspect,
which is formulated as the \emph{quantum Airy structure}
\cite{Andersen:2017vyk,Kontsevich:2017vdc},
of the ABO topological recursion and the CEO topological recursion.
For the physical 2D gravity models, we see that the quantum Airy structures
are equivalent to the Virasoro constraints, where the quantum Airy structures admit the Masur--Veech type twist by a group action in \cite{Andersen:2017vyk}
and then the Virasoro constraints are twisted as well.
We explicitly obtain solutions of the Virasoro constraints with or without
Masur--Veech type twist by using the cut-and-join equations in \cite{Alexandrov:2010bn,Alexandrov:2016kjl}, which are derived from the Virasoro constraints, and the group action mentioned above.

On the fourth point, we will discuss a physical interpretation of the Masur--Veech type twist of the ABO topological recursion
in terms of the JT gravity.
Via the path integral computations, the bidifferential $B$ in the basic data \eqref{eq:spectral_curve_data_1} of the spectral curve $\cC$ for
the Weil--Petersson volumes is found from the JT gravity partition function
on a hyperbolic double trumpet \cite{Saad:2019lba}.
As mentioned above, the basic data $\big(\IP^1;\mathsf{x},\mathsf{y},B\big[\mathsf{f}^{\mathrm{MV}}\big]\big)$ of the spectral curve for the twisted CEO topological recursion differs from $\mathcal{C}$ only by a shift of the bidifferential.
In this article, we find that such a shift of the bidifferential is obtained from the partition function of
a massless scalar field coupled to the JT gravity fields \cite{Jafferis:2022wez}.
We also discuss a derivation of the basic data of spectral curves
for the other physical 2D gravity models in the parallel way as the JT gravity.

Here we highlight the consequences of this article.\footnote{To make this article a valuable resource for readers in the physical and mathematical community, some introductory aspects of the Masur--Veech volumes, topological recursions and two-dimensional gravities are provided with explicit computations.}
From our observation in Appendix \ref{sec:physical}, the Masur--Veech type twist function $\mathsf{f}^{\mathrm{MV}}$ is found in the partition function of the massless scalar field coupled to the metric field of the JT gravity \cite{Jafferis:2022wez}. This physical interpretation is quite novel and matches with Mirzakhani's enumeration of simple closed geodesics in hyperbolic bordered Riemann surfaces \cite{Mirz_MV}.
To apply our physical interpretation of the Masur--Veech type twist further,
we perform a reverse construction of the ABO topological recursion data $\mathsf{A}$, $\mathsf{B}$, $\mathsf{C}$, $\mathsf{D}$ for the $(2,p)$ minimal string endowed with the FZZT boundary condition and its supersymmetric analogue, and the generalizations of the combinatorial formula of the Masur--Veech volume in~\cite{DGZZ_latest, Mirz_MV}.
The geometry of moduli spaces of $(2,p)$ minimal strings is still veiled in secrecy,
and the symplectic volume of such moduli spaces is not studied well even in the physical context.
We hope that our computational results of the generalized symplectic volume and its Masur--Veech type twist may be helpful for further studies on the Liouville gravity.

This paper is organized as follows.
In Section \ref{sec:geometric_recursion},
we summarize the formulation of the ABO topological recursion and discuss physical 2D gravity examples.
In Section \ref{sect:MasurVeech},
we show the combinatorial computation of the Masur--Veech volume $\mathrm{Vol}\mathcal{Q}_{g,n}$
and its generalizations to the physical 2D gravity models.
In Section \ref{sec:topological_recursion},
we discuss the CEO topological recursion for the 2D gravity models, and
derive the twisted CEO topological recursion for generalized Masur--Veech polynomials as a Laplace transform of the twisted ABO topological recursion.
In Section \ref{sec:virasoro_constraint},
we derive the manifest form of Virasoro generators from the (twisted) ABO topological recursion on the basis of the quantum Airy structure,
and compute free energies by solving cut-and-join equations iteratively for the 2D gravity models.
In Appendix \ref{sec:physical},
we give a physical interpretation of the Masur--Veech type twist
of the topological recursions by an extra scalar field coupled to the JT gravity fields,
and discuss a derivation of the basic data of spectral curves for the 2D gravity models.
In Appendix \ref{sec:Derivation_Mirz},
we derive the functions $D$ and $R$ in the Mirzakhani type ABO topological recursions for
the FZZT brane in the $(2,p)$ minimal string and its supersymmetric analogue
from the CEO topological recursion in the similar way as the paper \cite{Eynard:2007fi} by Eynard and Orantin.
In Appendix \ref{sec:table_vol},
we give 
the (twisted) volume polynomials for the 2D gravity models.

\section{ABO topological recursion}\label{sec:geometric_recursion}

In this section, after recalling the ABO topological recursion \cite{Andersen_GR} which generalizes Mirzakhani's recursion \cite{Mirz3,Mirz1}, we apply it to the physical 2D gravity models in Table \ref{tab:kernel}. In particular, we provide the kernel functions \eqref{eq:H_minimal} and \eqref{eq:H_super_minimal} for the $(2,p)$ minimal string and the $(2,2p-2)$ minimal superstring.
In Sections \ref{sec:twisting_gr} and \ref{sec:examples_gr_tw}, we also recall the ABO topological recursion with a twist proposed in \cite{Andersen_MV,Andersen_GR}, which generalizes
the combinatorial formula of the Masur--Veech volume in
\cite{DGZZ_latest,Mirz_MV} (see Section \ref{sect:MasurVeech}),
and apply it to the physical 2D gravity models.

\subsection{Formulation}\label{sec:def_gr}

The ABO topological recursion is a framework of recursions for \emph{volume polynomials} defined on the moduli space of connected bordered Riemann surfaces, which is a generalization of Mirzakhani's recursion.%
\footnote{The physical meaning of the volume polynomials for the 2D gravity models will be discussed in Appendix \ref{sec:physical}.
At present, the volume polynomials for such models are not defined on the moduli space of connected bordered Riemann surfaces. In this article, we define these volume polynomials as solutions of the ABO topological recursion~\eqref{geom_rec} whose initial data are found from the inverse Laplace transforms of the CEO topological recursions for the 2D gravity models. (In this approach, the volume polynomial $V_{1,1}(L_1)=V\mathsf{D}(L_1)$ is not found as the integral of the initial data $\mathsf{D}$ on $\mathcal{M}_{1,1}(L_1)$.)
}

\begin{Def}[ABO topological recursion \cite{Andersen_GR}]\label{def:geom_rec}
Let $\VO_{g,n}(L_1,\dots,L_n)$ be a volume polynomial
labeled by $g \ge 0$, $n \ge 1$ satisfying $2g-2+n > 0$,
on the moduli space $\cM_{g,n}(L_1,\dots,L_n)$ of connected bordered Riemann surfaces of genus $g$ with $n$ ordered boundary components of lengths~${L_1, \dots, L_n}$,
which obeys the ABO topological recursion such that%
\footnote{
In \cite{Andersen_MV,Andersen_GR}, the volume polynomial $\VO_{g,n}(L_1,\dots,L_n)$ is denoted as
$V\Omega_{g,n}(L_1,\dots,L_n)$.}
\begin{align}
\VO_{g,n}(L_1,\dots,L_n)={}&
\sum_{m=2}^n \int_{\IR_+} \sfB(L_1,L_m,\ell)
\VO_{g,n-1}(\ell, L_{K\setminus \{m\}}) \ell {\rm d}\ell\nonumber
\\
&
+ \frac12 \int_{\IR_+^2} \sfC(L_1, \ell, \ell')
P_{g,n}(\ell, \ell', L_K)
\ell\ell'{\rm d}\ell {\rm d}\ell',
\label{geom_rec}
\end{align}
where $\IR_+ = [0, \infty)$, $K=\{2,\dots,n\}$.
The topological recursion requires our initial data $\sfB(L_1,L_2,\ell)$, $\sfC(L_1, \ell, \ell')$, and
\begin{align*}
\VO_{0,3}(L_1,L_2,L_3)=\sfA(L_1,L_2,L_3),\qquad
\VO_{1,1}(L_1)=V\sfD(L_1)
=\int_{\cM_{1,1}(L_1)} \sfD(\sigma) {\rm d}\mu_{\mathrm{WP}}(\sigma),
\end{align*}
where $\mu_{\mathrm{WP}}(\sigma)$ denotes the Weil--Petersson measure on the moduli space
$\cM_{1,1}(L_1)$ endowed with a hyperbolic metric $\sigma$ on a torus with one boundary, and $\sfD(\sigma)$ is a measurable function on~${\cM_{1,1}(L_1)}$.
The initial data satisfies some decaying constraints and symmetry properties called \textit{admissibility conditions} \cite{Andersen_GR}.
Here
\begin{align}
P_{g,n}(\ell, \ell', L_K)=
\VO_{g-1, n+1}(\ell, \ell', L_K)
+ \mathop{\sum_{h+h'=g}}
\limits_{J \cup J'=K}^{\mathrm{stable}}
\VO_{h, 1+|J|}(\ell, L_J) \VO_{h', 1+|J'|}(\ell', L_{J'}),
\label{geo_rec_p}
\end{align}
where $\mathrm{stable}$ in the sum means that
$h$, $h'$, $J$, $J'$ obey $2h-1+|J|>0$ and $2h'-1+|J'|>0$, and
$L_J=\{L_{i_1}, \dots, L_{i_{|J|}}\}$,
$L_{J'}=\{L_{i_{|J|+1}}, \dots, L_{i_{n-1}}\}$ for
$J=\{i_1, \dots, i_{|J|}\} \subseteq K$.
\end{Def}

Assume that the volume polynomials $\VO_{g,n}$ are expanded as
\begin{align}
\VO_{g,n}(L_1,\dots,L_n)=
\sum_{a_1,\dots,a_n \ge 0} F^{(g)}_{a_1,\dots,a_n}
\prod_{i=1}^n e_{a_i}(L_i),
\label{gr_exp}
\end{align}
where $F^{(g)}_{a_1,\dots,a_n}$ is referred to as the \emph{volume coefficient}, and
\begin{align}
e_{a}(L)=\frac{L^{2a}}{(2a+1)!}.
\label{e_base}
\end{align}
By
\begin{align}
&
\int_{\IR_+} \sfB(L_1,L_2,\ell) e_{a}(\ell) \ell {\rm d}\ell
=\sum_{a_1, a_2 \ge 0} \sfB^{a_1}_{a_2, a} e_{a_1}(L_1) e_{a_2}(L_2),\nonumber
\\
&
\int_{\IR_+^2} \sfC(L_1,\ell,\ell') e_{a}(\ell) e_{b}(\ell')
\ell \ell' {\rm d}\ell {\rm d}\ell'
=\sum_{a_1 \ge 0} \sfC^{a_1}_{a, b} e_{a_1}(L_1),
\label{gr_coeff_bc}
\end{align}
the ABO topological recursion \eqref{geom_rec} gives
a recursion for the volume coefficients:
\begin{align}
F^{(g)}_{a_1,\dots,a_n}={}&
\sum_{m=2}^n \sum_{b \ge 0} \sfB^{a_1}_{a_m, b}
F^{(g)}_{b, a_2, \dots, \widehat{a}_m, \dots,a_n}\nonumber
\\
&
+\frac12 \sum_{a,b\ge 0} \sfC^{a_1}_{a, b}
\Bigg(F^{(g-1)}_{a,b, a_2, \dots, a_n} +
\mathop{\sum_{h+h'=g}}
\limits_{J \cup J'=K}^{\mathrm{stable}}
F^{(h)}_{a,a_J} F^{(h')}_{b,a_{J'}} \Bigg),
\label{geom_rec_coeff}
\end{align}
where $a_J=\{a_{i_1},\dots,a_{i_{|J|}}\}$ and $a_{J'}=\{a_{i_{|J|+1}},\dots,a_{i_{n-1}}\}$ for $J=\{i_1,\dots,i_{|J|}\} \subseteq K$. The initial inputs are $\sfB^{a_1}_{a_2, a}$, $\sfC^{a_1}_{a, b}$, and
\begin{align}
F^{(0)}_{a_1,a_2,a_3}=\sfA^{a_1}_{a_2, a_3},\qquad
F^{(1)}_{a_1}=\sfD^{a_1},
\label{geom_rec_coeff_ini}
\end{align}
where note that $\sfB^{a}_{b, c}=\sfB^{a}_{c, b}$, $\sfC^{a}_{b, c}=\sfC^{a}_{c, b}$ and $\sfA^{a}_{b, c}=\sfA^{a}_{c, b}=\sfA^{b}_{a, c}$.

\begin{rem}[Mirzakhani type ABO topological recursion]\!\label{rem:mir_rec}
The ABO topological re\-cur\-sion~\eqref{geom_rec} is a generalization of
Mirzakhani's re\-cur\-sion \cite{Mirz3,Mirz1} for the Weil--Petersson volume
$V_{g,n}^{\mathrm{WP}}\!(L_1,\allowbreak\dots, L_n)$ of the moduli space
$\cM_{g,n}(L_1,\dots,L_n)$ of genus $g$ hyperbolic surfaces with $n$ geodesic boundaries of length $L_1, \dots, L_n$.
In this article, we call the following form of the ABO topological recursion
the \emph{Mirzakhani type ABO topological recursion}:
\begin{align}
L_1V_{g,n}(L_1,\dots,L_n)={}&
\sum_{m=2}^n \int_{\IR_+}
x R(L_1,L_m,x) V_{g,n-1}(x,L_{K \setminus\{m\}}) {\rm d}x\nonumber
\\
&
+\frac{1}{2}\int_{\IR_+^2}
xy D(L_1,x,y) P_{g,n}(x,y,L_K) {\rm d}x{\rm d}y,
\label{eq:Mirzakhani's}
\end{align}
where
\begin{align}
R(x,y,z)=x \sfB(x,y,z),
\qquad
D(x,y,z)=x \sfC(x, y, z).
\label{mirz_abo_data}
\end{align}
\end{rem}

\subsection{Bosonic models}\label{sec:examples_gr}

We refer to a class of physical 2D gravity models such as the JT gravity, the topological gravity,
and the $(2,p)$ minimal string (denoted by resp.\ WP, A, and M$(p)$) as \emph{bosonic models} (see Table \ref{tab:kernel} in Section \ref{sect:Intro}, and Appendix \ref{sec:physical} for physical arguments).
For the JT gravity, the Weil--Petersson volumes appear in a part of the path integral of the partition function \cite{Saad:2019lba}.
In the Mirzakhani type ABO topological recursion \eqref{eq:Mirzakhani's}
for each bosonic model, two functions $R(x,y,z)$ and $D(x,y,z)$
are given in terms of a kernel function $H(x,y)$ as
\begin{align}
R(x,y,z)=\frac{1}{2}\int_0^{x}\left(H(z,t+y)+H(z,t-y)\right){\rm d}t,
\qquad
D(x,y,z)=\int_0^x H(y+z,t) {\rm d}t.
\label{eq:DR_H_bosonic}
\end{align}
In the following, we will provide the kernel functions $H(x,y)$ for the bosonic models, and find their topological recursions for the volume polynomials $\VO_{g,n}$.

\subsubsection{Weil--Petersson volumes}\label{subsec:gr_mir}

The initial data of the ABO topological recursion \cite{Andersen_MV,Andersen_GR} for the Weil--Petersson volumes of moduli spaces of connected bordered Riemann surfaces are
\begin{align}
&\sfA^{\mathrm{WP}}(L_1,L_2,L_3)=1,
\qquad
\sfB^{\mathrm{WP}}(L_1,L_2,\ell)
=
1 - \frac{1}{L_1}\log
\left(
\frac{\cosh\bigl(\frac{L_2}{2}\bigr)+\cosh\bigl(\frac{L_1+\ell}{2}\bigr)}
{\cosh\bigl(\frac{L_2}{2}\bigr)+\cosh\bigl(\frac{L_1-\ell}{2}\bigr)}
\right),\nonumber
\\
&\sfC^{\mathrm{WP}}(L_1, \ell, \ell')=\frac{2}{L_1} \log
\left(
\frac{\e^{\frac{L_1}{2}}+\e^{\frac{\ell + \ell'}{2}}}
{\e^{-\frac{L_1}{2}}+\e^{\frac{\ell + \ell'}{2}}}
\right),
\qquad
V\sfD^{\mathrm{WP}}(L_1)=\frac{\pi^2}{12} + \frac{1}{48}L_1^2.
\label{gr_mir_abcd}
\end{align}
Here the initial data $\sfB^{\mathrm{WP}}$ and $\sfC^{\mathrm{WP}}$ are given
by the formulae \eqref{mirz_abo_data} and \eqref{eq:DR_H_bosonic}
with the kernel function
\begin{align}
H^{\mathrm{WP}}(x,y)=
\frac{1}{1+\mathrm{e}^{\frac{x-y}{2}}}+\frac{1}{1+\mathrm{e}^{\frac{x+y}{2}}}.
\label{eq:H_WP}
\end{align}

For this model, the volume polynomial gives
the Weil--Petersson volume of the moduli space~${\cM_{g,n}(L_1,\dots,L_n)}$ of
connected bordered Riemann surfaces \cite{Mirz1}:
\begin{align}
\VO_{g,n}^{\mathrm{WP}}(L_1,\dots,L_n)=
\int_{\cM_{g,n}(L_1,\dots,L_n)} \exp{\omega^{\mathrm{WP}}}=
\int_{\overline{\cM}_{g,n}}
\exp\left(2\pi^2 \kappa_1 + \sum_{i=1}^n \frac{L_i^2}{2}\psi_i\right).
\label{mir_volume}
\end{align}
Here $\omega^{\mathrm{WP}}$ denotes the Weil--Petersson symplectic form, and
in the last equality, \smash{$\VO_{g,n}^{\mathrm{WP}}$} is represented by
the integral of the $\psi$ and $\kappa_1$ classes on
the moduli space $\overline{\cM}_{g,n}$,
which is the Deligne--Mumford compactification of the moduli space of stable curves of genus $g$ with $n$ marked points.
(This equality is proved in Wolpert's work \cite{Wolpert}.)
The psi class $\psi_i$ is the first Chern class of the line bundle over
$\overline{\cM}_{g,n}$ with fiber over $(C, p_1, \dots , p_n)$ being the cotangent space $T^*_{p}C$.
The first Miller--Morita--Mumford class $\kappa_1$ is a tautological class
defined by considering the pushforward of $\psi_{n+1}^2$ with respect to the forgetful map $\pi\colon\overline{\mathcal{M}}_{g,n+1}\to \overline{\mathcal{M}}_{g,n}$.
By comparison of equation~\eqref{gr_exp} with equation~\eqref{mir_volume}, the Weil--Petersson volume coefficients are
\begin{align}
F^{\mathrm{WP}(g)}_{a_1,\dots,a_n}
=
\left(\prod_{i=1}^n(2a_i+1)!!\right)
\int_{\overline{\cM}_{g,n}}
\e^{2\pi^2 \kappa_1} \psi_1^{a_1} \cdots \psi_n^{a_n}.
\label{feg_mir}
\end{align}
Here the volume coefficient $F^{\mathrm{WP}(g)}_{a_1,\dots,a_n}$ does not vanish, if the condition below is satisfied: $
\sum_{i=1}^n a_i \le 3g-3+n$. Some explicit results of $\VO_{g,n}^{\mathrm{WP}}$ are listed in~\eqref{lst_wp}.

\begin{rem}\label{rem:fugacity_kappa}
One can introduce a deformation parameter $s$ for the Weil--Petersson vol\-ume \eqref{mir_volume} by replacing $\pi^2$ with $\pi^2 s$
in $V\sfD^{\mathrm{WP}}$ of equation~\eqref{gr_mir_abcd}.
\end{rem}

\subsubsection{Kontsevich--Witten symplectic volumes}\label{subsec:gr_kw}

In Witten's work \cite{Witten:1990hr},
a novel approach to the intersection theory of the moduli space $\overline{\cM}_{g,n}$ of stable curves
is proposed based on the two-dimensional topological gravity, and it is conjectured that
the generating function of integrals over $\overline{\cM}_{g,n}$ is given by
the tau function of the KdV hierarchy. Witten's conjecture is proved elegantly by Kontsevich \cite{Kontsevich}, and the cell decomposition of $\overline{\cM}_{g,n}$ on the basis of Strebel's quadratic differential is realized by metric ribbon graphs \cite{MP98} in his proof.
On the basis of Kontsevich's work, the symplectic volume of the moduli space of stable curves is defined on the space of metric ribbon graphs in \cite{BCSW}, and is referred to as the \emph{Kontsevich--Witten symplectic volume}.

The initial data of the ABO topological recursion \cite{Andersen_MV,Andersen_GR,BCSW} for the Kontsevich--Witten symplectic
volumes of moduli spaces of stable curves are
\begin{align}
&\sfA^{\mathrm{A}}(L_1,L_2,L_3)=1,\nonumber
\\
&\sfB^{\mathrm{A}}(L_1,L_2,\ell)=
\frac{1}{2L_1}\left(
\left[L_1-L_2-\ell\right]_+ - \left[-L_1+L_2-\ell\right]_+
+\left[L_1+L_2-\ell\right]_+ \right)\nonumber
\\
&\phantom{\sfB^{\mathrm{A}}(L_1,L_2,\ell)}{}=
\begin{cases}
1\ \
& \textrm{if}\ \ 0 \le \ell \le L_2 - L_1,
\\
1-\dfrac{\ell}{L_1}\ \
& \textrm{if}\ \ 0 \le \ell \le L_1 - L_2,
\vspace{1mm}\\
\dfrac{L_1+L_2-\ell}{2L_1}\ \
& \textrm{if}\ \ |L_1-L_2| \le \ell \le L_1+L_2,
\\
0\ \
& \textrm{if}\ \ L_1+L_2 \le \ell,
\end{cases}\nonumber
\\
&\sfC^{\mathrm{A}}(L_1, \ell, \ell')=\frac{1}{L_1} \left[ L_1 - \ell - \ell' \right]_+,
\qquad
V\sfD^{\mathrm{A}}(L_1)=\frac{1}{48}L_1^2,
\label{gr_kw_abcd}
\end{align}
where $[x]_+=x$ for $x>0$ and $[x]_+=0$ for $x \le 0$.
The kernel function $H^{\mathrm{A}}$ which provides the initial data $\sfB^{\mathrm{A}}$ and $\sfC^{\mathrm{A}}$
is found from $H^{\mathrm{WP}}$ in equation~\eqref{eq:H_WP} for the Weil--Petersson volume in the scaling limit such that
\begin{align*}
H^{\mathrm{A}}(x,y)=\lim_{\beta \to \infty} H^{\mathrm{WP}}(\beta x,\beta y)=
\theta(y-x)+\theta(-x-y),
\end{align*}
where $\theta(t)$ denotes the Heaviside step function,
\begin{align*}
\theta(t)=\left\{
\begin{array}{lc}
1 & \textrm{if}\ t>0, \\
0 & \textrm{if}\ t\le 0.
\end{array}
\right.
\end{align*}
A geometric interpretation of the kernels $\mathsf{B}^{\mathrm{A}}$ and $\mathsf{C}^{\mathrm{A}}$ is given in \cite{Andersen_GR}.

For the initial data \eqref{gr_kw_abcd}, the volume polynomial gives
the Kontsevich--Witten symplectic volume of the moduli space $\overline{\cM}_{g,n}$ of stable curves:
\begin{align}
\VO_{g,n}^{\mathrm{A}}(L_1,\dots,L_n)={}&
\int_{\overline{\cM}_{g,n}}
\exp\left(\sum_{i=1}^n \frac{L_i^2}{2}\psi_i\right)\nonumber
\\
={}&
\mathop{\sum_{a_1, \dots, a_n \ge 0}}\limits_{|\mathbf{a}|=3g-3+n}
\left(\prod_{i=1}^n(2a_i+1)!!\right)
\int_{\overline{\cM}_{g,n}} \psi_1^{a_1} \cdots \psi_n^{a_n}
\prod_{i=1}^n e_{a_i}(L_i),
\label{kw_volume}
\end{align}
where $e_{a}(L)$ is defined in equation~\eqref{e_base},
and note the homogeneity condition
\begin{align}
|\mathbf{a}|=\sum_{i=1}^n a_i = 3g-3+n.
\label{hom_kw}
\end{align}
By comparison of equation~\eqref{gr_exp} with equation~\eqref{kw_volume},
the volume coefficients are
\begin{align}
F^{\mathrm{A}(g)}_{a_1,\dots,a_n}
=
\left(\prod_{i=1}^n(2a_i+1)!!\right)
\int_{\overline{\cM}_{g,n}} \psi_1^{a_1} \cdots \psi_n^{a_n}.
\label{feg_kw}
\end{align}
Note that the volume polynomials $\VO_{g,n}^{\mathrm{A}}$ are obtained from
the Weil--Petersson volumes $\VO_{g,n}^{\mathrm{WP}}$ in equation~\eqref{mir_volume} by
\begin{align*}
\VO_{g,n}^{\mathrm{A}}(L_1,\dots,L_n)=
\lim_{\beta \to \infty} \frac{1}{\beta^{6g-6+2n}} \VO_{g,n}^{\mathrm{WP}}(\beta L_1,\dots,\beta L_n),
\end{align*}
and some computational results are listed in~\eqref{lst_wp}.

\begin{rem}
Essentially the recursion relation \eqref{geom_rec_coeff} for Kontsevich--Witten's symplectic volume coefficients
is equivalent to the Dijkgraaf--Verlinde--Verlinde formula \cite{Dijkgraaf:1990rs} for the intersection numbers on the moduli space
of stable curves.
\end{rem}

\subsubsection[(2,p) minimal string]{$\boldsymbol{(2,p)}$ minimal string}\label{subsec:gr_mg}

Let $p$ be an odd positive integer.
The $(2,p)$ minimal string reviewed in Appendix \ref{sec:minimal} resides in a class of two-dimensional gravity,
which yields the JT gravity for $(p=\infty)$ and the topological gravity for $(p=1)$.
Accordingly, the volume polynomial \smash{$\VO^{\mathrm{M}(p)}_{g,n}$} for the $(2,p)$ minimal string interpolates
the Weil--Petersson volume $\VO^{\mathrm{WP}}_{g,n}$ in Section \ref{subsec:gr_mir} and
the Kontsevich--Witten symplectic volume $\VO^{\mathrm{A}}_{g,n}$ in Section \ref{subsec:gr_kw}.
The kernel function given below for the $(2,p)$ minimal string is derived in
Appendix \ref{sec:derivation_minimal} from the CEO topological recursion for the spectral curve found from the physical amplitude for the disk topology ending on the FZZT brane \cite{Fateev:2000ik,Seiberg:2003nm,Teschner:2000md}.
Here we just define it by
\begin{align}
H^{\mathrm{M}(p)}(x,y)={}&
-\sum_{j=1}^{(p-1)/2}(-1)^j\cos\left(\frac{\pi}{p}j\right)
\big(\mathrm{e}^{-u_j(x+y)} \theta(x+y)
+\mathrm{e}^{-u_j(x-y)} \theta(x-y)\big)\nonumber
\\
&
+\sum_{j=0}^{(p-1)/2}(-1)^j\cos\left(\frac{\pi}{p}j\right)
\big(\mathrm{e}^{u_j(x+y)} \theta(-x-y)
+\mathrm{e}^{u_j(x-y)} \theta(y-x)\big),
\label{eq:H_minimal}
\end{align}
where
\begin{align}
u_j=&\frac{p}{2\pi}\sin\left(\frac{\pi}{p}j\right).
\label{eq:uj_bosonic}
\end{align}
The kernel function \eqref{eq:H_minimal} yields
$H^{\mathrm{M}(1)}(x,y)=H^{\mathrm{A}}(x,y)$ and $H^{\mathrm{M}(\infty)}(x,y)=H^{\mathrm{WP}}(x,y)$.
The formulae \eqref{mirz_abo_data} and \eqref{eq:DR_H_bosonic} provide
the initial data $\sfB^{\mathrm{M}(p)}$ and $\sfC^{\mathrm{M}(p)}$ of
the ABO topological recursion, and the remaining initial data are
\begin{align*}
\sfA^{\mathrm{M}(p)}(L_1,L_2,L_3)=1,
\qquad
V\sfD^{\mathrm{M}(p)}(L_1)=\frac{\pi^2}{12}\left(1-\frac{1}{p^2}\right)
+ \frac{1}{48}L_1^2.
\end{align*}
From these initial data, one finds the volume polynomials $\VO^{\mathrm{M}(p)}_{g,n}$ for the $(2,p)$ minimal string, and
some computational results are listed in~\eqref{lst_mst}.

The following theorem is proved from the formula \eqref{mst_airy} in Section \ref{subsec:tr_mg} by Theorem \ref{thm_tw_CEO}.

\begin{thm}\label{thm:mst_vol_poly}
The volume polynomials for the $(2,p)$ minimal string obey
\begin{align}
\VO^{\mathrm{M}(p)}_{g,n}(L_1,\dots,L_n)=
\mathop{\sum_{a_1,\dots,a_n \ge 0}}\limits_{|\mathbf{a}|\le 3g-3+n}
F^{\mathrm{M}(p) (g)}_{a_1,\dots,a_n}
\prod_{i=1}^n e_{a_i}(L_i),
\label{gr_exp_mg}
\end{align}
where the volume coefficients are
\begin{align}
F^{\mathrm{M}(p) (g)}_{a_1,\dots,a_n}
={}&
\left(\prod_{i=1}^n(2a_i+1)!!\right)\nonumber
\\
& \times
\sum_{m \ge 0} \frac{(-1)^m}{m!}
\mathop{\sum_{b_1,\dots,b_m \ge 1}}\limits_{|\mathbf{b}|=3g-3+n-|\mathbf{a}|}
\Bigg(\prod_{j=1}^m \frac{\bigl(-2\pi^2\bigr)^{b_j}}{b_j!}
\prod_{k=1}^{b_j}\left(1-\frac{(2k-1)^2}{p^2}\right)\Bigg)\nonumber
\\
& \times
\int_{\overline{\mathcal{M}}_{g,n+m}}
\psi_1^{a_1} \cdots \psi_n^{a_n}
\psi_{n+1}^{b_1+1} \cdots \psi_{n+m}^{b_m+1}\nonumber
\\
={}&
\left(\prod_{i=1}^n(2a_i+1)!!\right)\nonumber
\\
&\times
\Bigg<
\exp{\Biggl(-\sum_{b \ge 1} \frac{\bigl(-2\pi^2\bigr)^{b}}{b!}
\Bigg(\prod_{k=1}^{b}\Bigg(1-\frac{(2k-1)^2}{p^2}\Bigg)\Bigg)\tau_{b+1}\Biggr)}
\tau_{a_1} \cdots \tau_{a_n}\Bigg>_{g},
\label{feg_mg}
\end{align}
and the condition for non-zero
volume coefficients $F^{\mathrm{M}(p) (g)}_{a_1,\dots,a_n}$ is
$|\mathbf{a}|=\sum_{i=1}^n a_i \le 3g-3+n$.
Here we set $\tau_a=\tau_{a,i}=\psi_i^a$ and introduced the notation
\begin{align*}
\big<\e^{ \xi \tau_a}\big>_{g}=
\sum_{m \ge 0}\frac{\xi^m}{m!}
\left<\tau_a^m\right>_{g}=
\sum_{m \ge 0}\frac{\xi^m}{m!}
\int_{\overline{\mathcal{M}}_{g,m}}
\psi_{1}^{a} \cdots \psi_{m}^{a}.
\end{align*}
\end{thm}

\begin{rem}
As mentioned above, the volume polynomial \eqref{gr_exp_mg} interpolates
the Weil--Pe\-ters\-son volume \eqref{mir_volume} at $p=\infty$ and
the Kontsevich--Witten symplectic volume \eqref{kw_volume} at $p=1$:
\begin{align*}
\VO^{\mathrm{M}(1)}_{g,n}(L_1,\dots,L_n)=
\VO^{\mathrm{A}}_{g,n}(L_1,\dots,L_n),
\qquad
\VO^{\mathrm{M}(\infty)}_{g,n}(L_1,\dots,L_n)=
\VO^{\mathrm{WP}}_{g,n}(L_1,\dots,L_n).
\end{align*}
In particular, by comparison of the Weil--Petersson volume coefficients \eqref{feg_mir} with the formula~\eqref{feg_mg} for $p=\infty$, we obtain
a formula
\begin{align}
\int_{\overline{\cM}_{g,n}}
\e^{2\pi^2 \kappa_1} \psi_1^{a_1} \cdots \psi_n^{a_n}=
\Bigg<
\exp{\Bigg(-\sum_{b \ge 1} \frac{\bigl(-2\pi^2\bigr)^{b}}{b!} \tau_{b+1}\Bigg)}
\tau_{a_1} \cdots \tau_{a_n}\Bigg>_{g},
\label{feg_wp_formula}
\end{align}
which is found in the literature (see, e.g., of \cite[equation (2.24)]{Dijkgraaf:2018vnm})
and intended as
\begin{align*}
\int_{\overline{\cM}_{g}} \frac{\kappa_1^{3g-3}}{(3g-3)!}=
\sum_{m=0}^{3g-3} \frac{(-1)^{3g-3+m}}{m!}
\mathop{\sum_{b_1,\dots,b_m \ge 1}}\limits_{|\mathbf{b}|=3g-3}
\int_{\overline{\cM}_{g,m}}
\frac{\psi_1^{b_1+1}\cdots \psi_m^{b_m+1}}{b_1!\cdots b_m!},
\end{align*}
by the homogeneity condition \eqref{hom_kw}.
\end{rem}

\subsection{Supersymmetric models}\label{sec:examples_gr_super}

Supersymmetric generalizations of the bosonic models are considered
as we will discuss in Appendices \ref{sec:super_models} and \ref{sec:type0A}.
We refer to a class of models such as
the JT supergravity \cite{Stanford:2019vob}, the BGW model \cite{Brezin:1980rk,Gross:1980he,Norbury:2020vyi} in the limit of all time variables set to zero,
and the $(2,2p-2)$ minimal superstring (denoted by resp.\ SWP, B, and SM$(p)$) as \emph{supersymmetric models}
(see Table \ref{tab:kernel} in Section~\ref{sect:Intro}).
The super Weil--Petersson volumes \cite{Norbury:2020vyi,Stanford:2019vob} arise in a part of
the path integral of the partition function of the JT supergravity.
In the Mirzakhani type ABO topological recursion~\eqref{eq:Mirzakhani's}
for each supersymmetric model, the two functions $R(x,y,z)$ and $D(x,y,z)$ are given by a~kernel function $H(x,y)$ as
\begin{align}
R(x,y,z)=\frac{1}{2}\left(H(z,x+y)+H(z,x-y)\right),
\qquad
D(x,y,z)=H(y+z,x).
\label{eq:DR_H_super}
\end{align}
We will provide the kernel functions $H(x,y)$ for the supersymmetric models as well as their volume polynomials $\VO_{g,n}$ in the following.

\subsubsection{Super Weil--Petersson volumes}\label{subsec:gr_sw}

The initial data of the ABO topological recursion for the super Weil--Petersson volumes%
\footnote{
In this article, the normalization of the super Weil--Petersson volume $\VO_{g,n}^{\mathrm{SWP}}$ is chosen to be identified with~$\VO_{g,n}^{\Theta}$ defined in \cite{Norbury:2020vyi}, which is equivalent to a specialization of the BGW tau function of the KdV hierarchy. The normalization is checked as follows:
\begin{enumerate}\itemsep=0pt
\item[1)] The choice of the model in \cite{Norbury:2020vyi} corresponds to the supergravity for an odd spin structure in Stanford--Witten's work \cite{Stanford:2019vob}.
The functions $D$ and $R$ in \cite{Norbury:2020vyi} are related to the functions \smash{$\widetilde{\mathcal{D}}$} and \smash{$\widetilde{\mathcal{T}}$} of \cite[equations~(D.44) and~(D.47)]{Stanford:2019vob} by \smash{$D(x,y,z)=-2\widetilde{\mathcal{D}}(x,y,z)$} and \smash{$R(x,y,z)=2\widetilde{\mathcal{T}}(x,y,z)$}.
\item[2)] Let $\VO_{g,n}^{\mathrm{SW}}$ denote the volume function of the supergravity with the odd spin structure in \cite{Stanford:2019vob}. Two volumes~$\VO_{g,n}^{\mathrm{SW}}$ and~$V_{g,n}^{\Theta}$ are related by $V_{g,n}^{\mathrm{SW}}=(-1)^{n}2^{1-g}V_{g,n}^{\Theta}$.
\item[3)] Multiplying a factor $(-1)^{n}2^{g-1}$ to the supersymmetric recursion relation (D.30) in \cite{Stanford:2019vob}, we recover
the recursion relation (7) in \cite{Norbury:2020vyi}.
\end{enumerate}
}
are \cite{Norbury:2020vyi,Stanford:2019vob},
\begin{align*}
\sfA^{\mathrm{SWP}}(L_1,L_2,L_3)=0,
\qquad
V\sfD^{\mathrm{SWP}}(L_1)={}&\frac{1}{8},
\end{align*}
and the remaining ones $\sfB^{\mathrm{SWP}}$ and $\sfC^{\mathrm{SWP}}$ are found by the formulae \eqref{mirz_abo_data} and \eqref{eq:DR_H_super} from the kernel function
\begin{align*}
H^{\mathrm{SWP}}(x,y)=
\frac{1}{4\pi}\left(\frac{1}{\cosh\frac{x-y}{4}}-\frac{1}{\cosh\frac{x+y}{4}}\right).
\end{align*}

For the above initial data, the volume polynomial gives
the super Weil--Petersson volume of the moduli space of super Riemann surfaces
which is given by an integral over the moduli space of stable curves
\begin{align}
\begin{split}
\VO_{g,n}^{\mathrm{SWP}}(L_1,\dots,L_n)={}&
\int_{\overline{\cM}_{g,n}} \Theta_{g,n}
\exp\left(2\pi^2 \kappa_1 + \sum_{i=1}^n \frac{L_i^2}{2}\psi_i\right),
\label{sw_volume}
\end{split}
\end{align}
and the super Weil--Petersson volume coefficients are
\begin{align}
F^{\mathrm{SWP}(g)}_{a_1,\dots,a_n}
=
\left(\prod_{i=1}^n(2a_i+1)!!\right)
\int_{\overline{\cM}_{g,n}}
\e^{2\pi^2 \kappa_1} \Theta_{g,n} \psi_1^{a_1} \cdots \psi_n^{a_n},
\label{feg_sw}
\end{align}
where the Norbury classes $\Theta_{g,n}\in H^{4g-4+2n}\big(\overline{\cM}_{g,n}, \IQ\big)$ are defined in \cite{Chidambaram:2022cqc,{Norbury:1712}}.
The super Weil--Petersson volume coefficients
\smash{$F^{\mathrm{SWP}(g)}_{a_1,\dots,a_n}$} do not vanish only if $\sum_{i=1}^n a_i \le g-1$, and this condition implies that
$\VO_{0,n}^{\mathrm{SWP}}=0$ and $\VO_{1,n}^{\mathrm{SWP}}$'s are constants which do not depend on $L_i$.
Some computational results of the volume polynomials \smash{$\VO_{g,n}^{\mathrm{SWP}}$} are listed in~\eqref{lst_swp}
(see also~\cite{Norbury:2020vyi}).

\subsubsection{Super symplectic volumes}\label{subsec:gr_be}

Using the scaling relation
\begin{align*}
\lim_{\beta \to\infty}\frac{\beta}{\pi}\frac{1}{\cosh(\beta x)}=\delta(x),
\end{align*}
one finds the kernel function
\begin{align*}
H^{\mathrm{B}}(x,y)
=\lim_{\beta \to \infty} \beta H^{\mathrm{SWP}}(\beta x,\beta y)
=\delta(x-y)-\delta(x+y),
\end{align*}
for a supersymmetric analogue of the Kontsevich--Witten symplectic volumes
referred to as the \emph{super symplectic volumes}.
The initial data of the ABO topological recursion is given by
\begin{align}
&\sfA^{\mathrm{B}}(L_1,L_2,L_3)=0,\nonumber
\\
&\sfB^{\mathrm{B}}(L_1,L_2,\ell)=
\frac{1}{2L_1}
\left(
\delta(L_1-L_2-\ell) - \delta(-L_1+L_2-\ell) + \delta(L_1+L_2-\ell)
\right),\nonumber
\\
&\sfC^{\mathrm{B}}(L_1, \ell, \ell')=\frac{1}{L_1} \delta(L_1 - \ell - \ell'),
\qquad
V\sfD^{\mathrm{B}}(L_1)=\frac{1}{8}.
\label{gr_be_abcd}
\end{align}

The volume polynomial for the initial data \eqref{gr_be_abcd}
gives the super symplectic volume
\cite[Proposition 6.2]{Norbury:2020vyi}:
\begin{align}
\VO_{g,n}^{\mathrm{B}}(L_1,\dots,L_n)={}&
\int_{\overline{\cM}_{g,n}} \Theta_{g,n}
\exp\left(\sum_{i=1}^n \frac{L_i^2}{2}\psi_i\right)\nonumber
\\
={}&
\mathop{\sum_{a_1, \dots, a_n \ge 0}}\limits_{|\mathbf{a}|=g-1}
\left(\prod_{i=1}^n(2a_i+1)!!\right)
\int_{\overline{\cM}_{g,n}}
\Theta_{g,n} \psi_1^{a_1} \cdots \psi_n^{a_n}
\prod_{i=1}^n e_{a_i}(L_i),
\label{be_volume}
\end{align}
where note the homogeneity condition
\begin{align}
|\mathbf{a}|=\sum_{i=1}^n a_i = (3g-3+n) - (2g-2+n) = g-1.
\label{hom_be}
\end{align}
By comparison of equation~\eqref{gr_exp} with equation~\eqref{be_volume},
the volume coefficients are
\begin{align}
F^{\mathrm{B}(g)}_{a_1,\dots,a_n}
=
\left(\prod_{i=1}^n(2a_i+1)!!\right)
\int_{\overline{\cM}_{g,n}}
\Theta_{g,n} \psi_1^{a_1} \cdots \psi_n^{a_n}.
\label{feg_be}
\end{align}
Note that some of the volume polynomials $\VO_{g,n}^{\mathrm{B}}$ are found from
the super Weil--Petersson volumes~$\VO_{g,n}^{\mathrm{SWP}}$ in~\eqref{lst_swp}
by
\begin{align*}
\VO_{g,n}^{\mathrm{B}}(L_1,\dots,L_n)=
\lim_{\beta \to \infty} \frac{1}{\beta^{2g-2}} \VO_{g,n}^{\mathrm{SWP}}(\beta L_1,\dots,\beta L_n),
\end{align*}
or from the volume polynomials $\VO^{\mathrm{SM}(p)}_{g,n}$ for
the $(2,2p-2)$ minimal superstring below in equation~\eqref{vol_minimal_superstring} by \smash{$\VO_{g,n}^{\mathrm{B}}=\VO^{\mathrm{SM}(1)}_{g,n}$}.

\subsubsection[(2,2p-2) minimal superstring]{$\boldsymbol{(2,2p-2)}$ minimal superstring}\label{subsec:gr_mg_super}
Consider a family of ABO recursions interpolating the ABO topological recursions of super Weil--Petersson volumes in Section \ref{subsec:gr_sw} and super symplectic volumes in Section~\ref{subsec:gr_be}.
Such a~model is provided by the type 0A $(2,2p-2)$ minimal superstring with any odd positive integers~$p$.
A spectral curve of the CEO topological recursion
in the $(2,2p-2)$ minimal superstring is heuristically obtained
in Appendix~\ref{sec:type0A},
and then the kernel function in the Mirzakhani type ABO topological recursion is derived in Appendix \ref{sec:derivation_super_minimal}:
\begin{align}
H^{\mathrm{SM}(p)}(x,y)={}&
\frac{1}{2\pi}\sum_{j=1}^{(p-1)/2}(-1)^j
\cos^2\left(\frac{\pi}{p}\left(j-\frac{1}{2}\right)\right)\bigl(\mathrm{e}^{-u'_j(x+y)} \theta(x+y)\nonumber\\
&-\mathrm{e}^{-u'_j(x-y)} \theta(x-y)
+\mathrm{e}^{u'_j(x+y)} \theta(-x-y)
-\mathrm{e}^{u'_j(x-y)} \theta(y-x)\bigr)\nonumber
\\
&+\left(\delta(x-y)-\delta(x+y)\right)\delta_{p,1},\label{eq:H_super_minimal}
\end{align}
where
\begin{align}
u'_j= \frac{p}{2\pi}\sin\left(\frac{\pi}{p}\left(j-\frac{1}{2}\right)\right).
\label{eq:uj_super}
\end{align}
This kernel function obeys
$H^{\mathrm{SM}(1)}(x,y)=H^{\mathrm{B}}(x,y)$ and $H^{\mathrm{SM}(\infty)}(x,y)=H^{\mathrm{SWP}}(x,y)$.
From the formulae \eqref{mirz_abo_data} and \eqref{eq:DR_H_super},
the initial data $\sfB^{\mathrm{SM}(p)}$ and $\sfC^{\mathrm{SM}(p)}$
of the ABO topological recursion are obtained, and the remaining ones are
\begin{align*}
\sfA^{\mathrm{SM}(p)}(L_1,L_2,L_3)=0,
\qquad
V\sfD^{\mathrm{SM}(p)}(L_1)=\frac18.
\end{align*}
Using the initial data,
one can compute the volume polynomials \smash{$\VO^{\mathrm{SM}(p)}_{g,n}$} for
the $(2,2p-2)$ minimal superstring iteratively.
In particular for the supersymmetric model, the recursion for the volume coefficients simplifies drastically,
and the general form of the volume polynomials for any odd positive integers $p$ is obtained for $g=0,1,2,3$ as follows:
\begin{gather}
\VO^{\mathrm{SM}(p)}_{0,n}(L_1,\dots,L_n)= 0,\qquad
\VO^{\mathrm{SM}(p)}_{1,n}(L_1,\dots,L_n)=\frac{(n-1)!}{8},\nonumber
\\
\VO^{\mathrm{SM}(p)}_{2,n}(L_1,\dots,L_n)= 
\frac{3(n+1)!}{128}
\left[(n+2)\left(1-\frac{1}{p^2}\right)\pi^2
+ \frac14 \sum_{i=1}^n L_i^2\right],\nonumber
\\
\VO^{\mathrm{SM}(p)}_{3,n}(L_1,\dots,L_n)= 
\frac{(n+3)!}{2^{16}\cdot 5}
\Biggl[16(n+4)\left(42n\left(1-\frac{1}{p^2}\right)+185+\frac{15}{p^2}\right)
\left(1-\frac{1}{p^2}\right)\pi^4 \nonumber
\\
\hphantom{\VO^{\mathrm{SM}(p)}_{3,n}(L_1,\dots,L_n)=}{}+336(n+4)\left(1-\frac{1}{p^2}\right)\pi^2 \sum_{i=1}^nL_i^2 \nonumber\\
\hphantom{\VO^{\mathrm{SM}(p)}_{3,n}(L_1,\dots,L_n)=}{}+84\sum_{i \ne j}^n L_i^2 L_j^2
+25\sum_{i=1}^n L_i^4\Biggr].\label{vol_minimal_superstring}
\end{gather}
More computational results are listed in~\eqref{lst_super_mst}.

Similar to Theorem \ref{thm:mst_vol_poly}, the following theorem is proved from the formula \eqref{msst_bessel} in Section~\ref{subsec:tr_mg_super}
by Theorem \ref{thm_tw_CEO}.

\begin{thm}\label{thm:mst_vol_poly_super}
The volume polynomials for the $(2,2p-2)$ minimal superstring obey
\begin{align}
\VO^{\mathrm{SM}(p)}_{g,n}(L_1,\dots,L_n)=
\mathop{\sum_{a_1,\dots,a_n \ge 0}}\limits_{|\mathbf{a}|\le g-1}
F^{\mathrm{SM}(p) (g)}_{a_1,\dots,a_n}
\prod_{i=1}^n e_{a_i}(L_i),
\label{gr_exp_smg}
\end{align}
where the volume coefficients are
\begin{align}
F^{\mathrm{SM}(p) (g)}_{a_1,\dots,a_n}
={}&
\left(\prod_{i=1}^n(2a_i+1)!!\right)\nonumber
\\
&
\times
\sum_{m \ge 0} \frac{(-1)^m}{m!}
\mathop{\sum_{b_1,\dots,b_m \ge 1}}\limits_{|\mathbf{b}|=g-1-|\mathbf{a}|}
\Bigg(\prod_{j=1}^m \frac{\bigl(-2\pi^2\bigr)^{b_j}}{b_j!}
\prod_{k=1}^{b_j}\left(1-\frac{(2k-1)^2}{p^2}\right)\Bigg)\nonumber
\\
&
\times
\int_{\overline{\mathcal{M}}_{g,n+m}}
\Theta_{g,n+m}
\psi_1^{a_1} \cdots \psi_n^{a_n}
\psi_{n+1}^{b_1} \cdots \psi_{n+m}^{b_m},
\label{feg_smg}
\end{align}
and do not vanish only if
$|\mathbf{a}|=\sum_{i=1}^n a_i \le g-1$.
\end{thm}

\begin{rem}
The volume polynomial \eqref{gr_exp_smg} interpolates
the super Weil--Petersson volume \eqref{sw_volume} at $p=\infty$ and
the super symplectic vol\-ume \eqref{be_volume} at $p=1$:
\begin{align*}
\VO^{\mathrm{SM}(1)}_{g,n}(L_1,\dots,L_n)=
\VO^{\mathrm{B}}_{g,n}(L_1,\dots,L_n),
\qquad
\VO^{\mathrm{SM}(\infty)}_{g,n}(L_1,\dots,L_n)=
\VO^{\mathrm{SWP}}_{g,n}(L_1,\dots,L_n).
\end{align*}
In particular, by comparing the super Weil--Petersson volume coefficients \eqref{feg_sw} with the formula~\eqref{feg_smg} for $p=\infty$, we obtain
a super analogue of the formula \eqref{feg_wp_formula}.
\end{rem}

\subsection{Twisting}\label{sec:twisting_gr}

Here we consider a twist action of the ABO topological recursion \cite{Andersen_MV,Andersen_GR}.
The twisting of the ABO topological recursion is originated from the study of the statistics of length of multicurves in a connected bordered Riemann surface, which leads to a combinatorial computation of the Masur--Veech volume \cite{Masur,Veech} for the moduli space of quadratic differentials on Riemann surfaces.
To discuss this aspect of the ABO topological recursion, we summarize a working definition of
stable graphs \cite{Mirz_MV} (see \cite[Appendix B]{DGZZ_latest} for a formal definition of stable graphs).

Let $\gamma_i$ ($i=1,\dots,k$) be simple closed geodesics on a bordered surface $S_{g,n}$ of genus $g$ with $n$ boundary components.
We assume that $\gamma_i$ and $\gamma_j$ ($i\ne j$) are pairwise non-isotopic with regard to the action of the mapping class group
$\mathrm{Mod}_{g,n}$ on $S_{g,n}$ and not intersecting each other
(see Figure \ref{fig:multi_curve} (left) for an example).
\begin{figure}[t]\centering
\includegraphics[width=100mm]{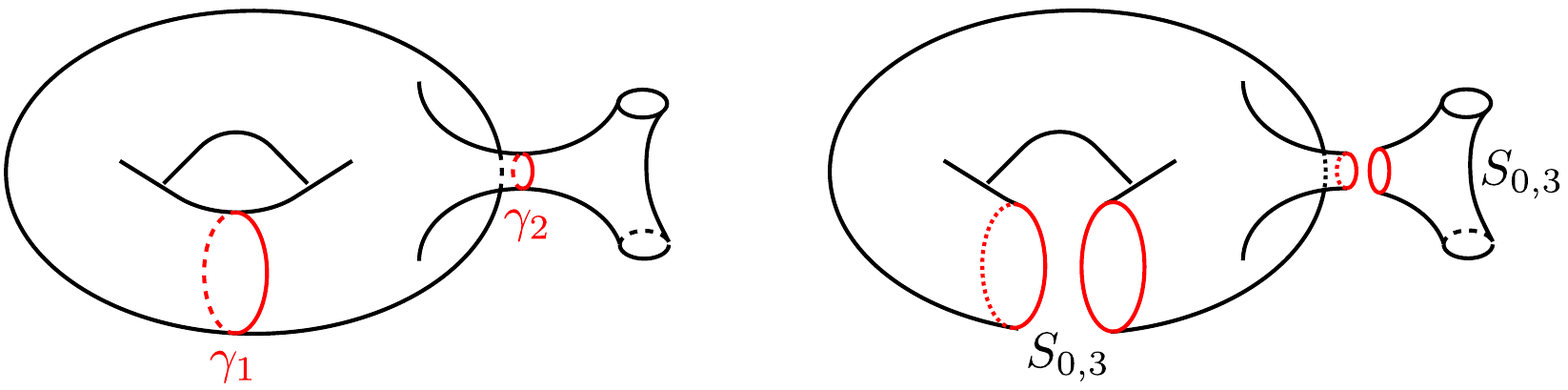}
\caption{Left: Two simple closed geodesics $\gamma_1$, $\gamma_2$ in $S_{1,2}$.
Right: Decomposition of $S_{1,2}$ cut by $\gamma_1$,~$\gamma_2$.}\label{fig:multi_curve}
\end{figure}
A multicurve $\gamma$ is given by
\begin{align*}
\gamma=\sum_{i=1}^kH_i\gamma_i,
\end{align*}
where $(H_1,\dots,H_k)$ is a set of positive integers and $\gamma_i$'s are disjoint, essential, non-peripheral
simple closed curves in $S_{g,n}$.\footnote{A curve $\gamma_i$ is said to be essential (resp.\ non-peripheral)
if $S_{g,n}\setminus\gamma_i$ does not have disk (resp.\ annulus) components.}
For the above multicurve $\gamma$, the reduced multicurve $\gamma_{\mathrm{red}}$ is $\gamma_{\mathrm{red}}=\sum_{i=1}^k\gamma_i$.

A \emph{stable graph} $\Gamma$ is associated with the pair
$(S_{g,n},\gamma_{\mathrm{red}})$ by cutting a bordered surface $S_{g,n}$
along a reduced multicurve $\gamma_{\mathrm{red}}$ such that
\begin{align}
\label{eq:cut_Riemann_surface}
S_{g,n}\setminus\gamma_{\mathrm{red}}=
\bigsqcup_{a=1}^{N} S_{g_a,n_a},
\end{align}
where $S_{g_a,n_a}$'s are connected stable bordered surfaces with $2g_a-2+n_a>0$.
For example, a~decomposition $S_{1,2}\setminus\{\gamma_1,\gamma_2\}$ is found in Figure \ref{fig:multi_curve} (right).
The associated stable graph $\Gamma$ is the dual decorated graph
made of decorated vertices $v_a$ and edges $e_i$ found from the decomposition~\eqref{eq:cut_Riemann_surface} as in Figure \ref{fig:stable_graph_components}. The basic data of the stable graph are given as follows:
\begin{figure}[t]\centering
\includegraphics[width=120mm]{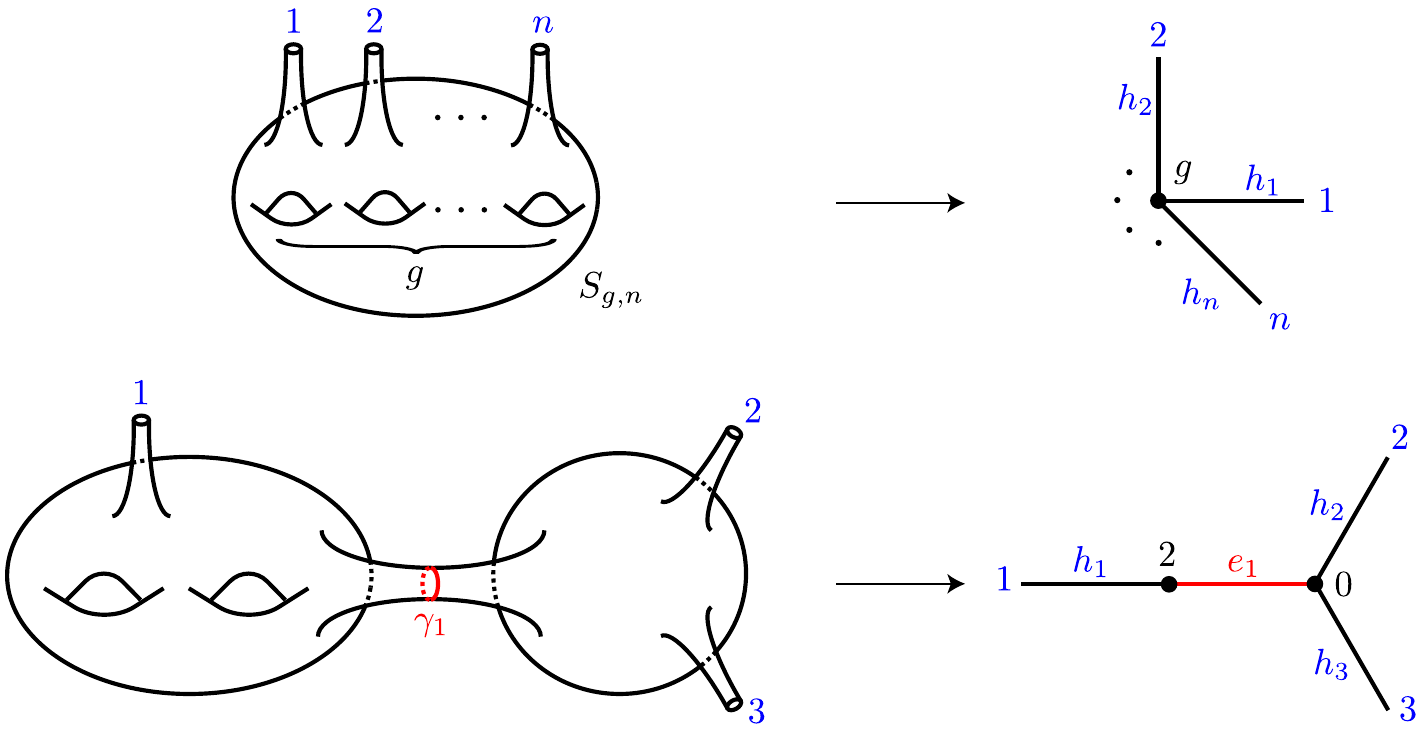}

\caption{Top: Identification of $S_{g,n}$ with
an $n$-valent vertex decorated by an integer $g$ of the stable graph.
Bottom: Identification of a decomposition of a bordered surface with
an edge and half-edges of a~stable graph.}\label{fig:stable_graph_components}
\end{figure}

\begin{itemize}\itemsep=0pt
\item
Vertex $v_a$ ($a=1,\dots,N$):
an $n_a$-valent vertex decorated by an integer $g_a$ associated with~$S_{g_a,n_a}$ ;
\item
Internal edge $e_i$ ($i=1,\dots,k$):
an edge associated with the simple closed curve $\gamma_i$;
\item
Incidence relation:
two vertices $v_a$, $v_b$ are connected by an edge $e_i$ if
$S_{g_a,n_a}$ and $S_{g_b,n_b}$ are adjacent to a simple closed curve $\gamma_i$ in $S_{g,n}$ s.t.\
$S_{g_a,n_a}\sqcup_{\gamma_i}S_{g_b,n_b}\subset S_{g,n}$,
where the edge $e_i$ forms a loop in the stable graph for the case of $a=b$;
\item
Half-edge $b_p$ ($p=1,\dots,n$): a half-edge associated with bordered boundaries in $S_{g,n}$.
\end{itemize}
We denote by $G_{g,n}$ the set of stable graphs associated with $S_{g,n}$,
and by $\mathsf{V}_{\Gamma}$ and $\mathsf{E}_{\Gamma}$
the sets of edges and vertices in $\Gamma \in G_{g,n}$, respectively.

\begin{Def}[twisted volume polynomials \cite{Andersen_GR}]\label{def:tw_abo_amp}
Let $\sff\colon {\IR}_+ \to {\IC}$ be an \emph{admissible test function},
i.e., a Riemann-integrable function on $\IR_+$
such that
\begin{align*}
\mathop{\sup}\limits_{\ell>0} (1+\ell)^s |\sff(\ell)|<\infty
\qquad
\textrm{for any } s>0.
\end{align*}
The \emph{twisted volume polynomials} $\VO_{g,n}[\sff]$ are defined as combinations of the volume polynomials~$\VO_{g,n}$ on the basis of the basic data of stable graphs by
(see \cite[Lemma 7.4]{Andersen_GR}),
\begin{align}
\VO_{g,n}[\sff](L_1,\dots,L_n)={}&
\sum_{\Gamma\in G_{g,n}} \frac{1}{|\mathrm{Aut}(\Gamma)|}
\int_{\IR_+^{|\mathsf{E}_{\Gamma}|}}
\prod_{v\in \mathsf{V}_{\Gamma}}
\VO_{h(v),k(v)}((\ell_e)_{e\in\mathsf{E}(v)},
(L_{\lambda})_{\lambda\in\Lambda(v)})\nonumber
\\
&\times
\prod_{e\in \mathsf{E}_{\Gamma}} \sff(\ell_e) \ell_e {\rm d}\ell_e,\label{eq:twist_volume}
\end{align}
where $\mathsf{E}(v)$ and $\Lambda(v)$ denote
the sets of edges and half-edges emanating from a vertex
$v\in \mathsf{V}_{\Gamma}$, respectively.
The expansion \eqref{gr_exp} for $\VO_{g,n}[\sff]$ defines
the \emph{twisted volume coefficients} $F^{(g)}[\sff]_{a_1,\dots,a_n}$:
\begin{align}
V_{g,n}[\sff](L_1,\dots,L_n)=
\sum_{a_1,\dots,a_n\ge 0} F^{(g)}[\sff]_{a_1,\dots,a_n}
\prod_{i=1}^n e_{a_i}(L_i),
\label{eq:expansion_volume}
\end{align}
where $e_{a}(L)=L^{2a}/(2a+1)!$.
\end{Def}

\begin{prop}[twisted initial data \cite{Andersen_GR}]\label{def:twist}
The twisted volume polynomials $\VO_{g,n}[\sff]$ are
obtained from the ABO topological recursion with the following
four twisted initial data $(\sfA[\sff], \sfB[\sff], \allowbreak\sfC[\sff], \sfD[\sff])$:
\begin{align}
&\sfA[\sff](L_1, L_2, L_3)=\sfA(L_1, L_2, L_3),
\qquad
\sfB[\sff](L_1, L_2, \ell)=\sfB(L_1, L_2, \ell) + \sfA(L_1, L_2, \ell) \sff(\ell),\nonumber
\\
&\sfC[\sff](L_1, \ell, \ell')=\sfC(L_1, \ell, \ell')
+ \sfB(L_1, \ell, \ell') \sff(\ell) + \sfB(L_1, \ell', \ell) \sff(\ell')
+ \sfA(L_1, \ell, \ell') \sff(\ell) \sff(\ell'),\nonumber
\\
&\sfD[\sff](\sigma)=\sfD(\sigma)
+\frac12 \sum_{\gamma \in S_T^{\circ}}
\sfA(\ell_{\sigma}(\partial T), \ell_{\sigma}(\gamma),
\ell_{\sigma}(\gamma)) \sff(\ell_{\sigma}(\gamma)),
\label{abcd_twist}
\end{align}
where $S_T^{\circ}$ is the set of simple closed curves in a torus with one boundary $T$.
$\ell_{\sigma}(\partial T)$ and $\ell_{\sigma}(\gamma)$ denote the lengths of the boundary $\partial T$
and the shortest geodesic in the homotopy class of $\gamma$ with respect to a hyperbolic metric $\sigma$ on $T$, respectively.
\end{prop}

From equations~\eqref{gr_coeff_bc} and \eqref{geom_rec_coeff_ini},
the twisted initial data \eqref{abcd_twist} imply
\begin{align}
&\sfA[\sff]^{a_1}_{a_2,a_3}=\sfA^{a_1}_{a_2,a_3},
\qquad
\sfB[\sff]^{a_1}_{a_2,a_3}=\sfB^{a_1}_{a_2,a_3}
+ \sum_{a \ge 0} \sfA^{a_1}_{a_2,a} u_{a, a_3},\nonumber
\\
&\sfC[\sff]^{a_1}_{a_2,a_3}=\sfC^{a_1}_{a_2,a_3}
+ \sum_{a \ge 0}
\left(\sfB^{a_1}_{a,a_2} u_{a, a_3} + \sfB^{a_1}_{a,a_3} u_{a, a_2}\right)
+ \sum_{a, b \ge 0} \sfA^{a_1}_{a,b} u_{a, a_2} u_{b, a_3},\nonumber
\\
&\sfD[\sff]^{a_1}=\sfD^{a_1}
+\frac12 \sum_{a, b \ge 0} \sfA^{a_1}_{a,b} u_{a,b},
\label{vir_abcd_twist}
\end{align}
where
\begin{align}
u_{a, b}=
u[\sff]_{a, b}=
\int_{\IR_+} \frac{\ell^{2a + 2b +1}}
{\left(2a+1\right)!\left(2b+1\right)!} \sff(\ell) {\rm d}\ell.
\label{twist_u}
\end{align}
In particular, we consider the Masur--Veech type twist,
\begin{align}
\sfm(\ell)=\frac{1}{\e^{\ell}-1},
\label{mv_twist}
\end{align}
as an admissible test function, and then the twist function \eqref{twist_u} is
\begin{align}
u_{a, b}^{\mathrm{MV}}=
u\big[\sfm\big]_{a, b}=
\frac{\left(2a+2b+1\right)!}{\left(2a+1\right)!\left(2b+1\right)!}
\zeta(2a+2b+2).\label{mv_twist_u}
\end{align}

\begin{rem}\label{rem:ge_mv_tw}
We can consider the following one-parameter generalization of the Masur--Veech type twist function \cite{Giacchetto:2021cbo}:
\begin{align}
\label{eq:gen_MV_twist}
\mathsf{f}^{\mathrm{g\mathchar`-MV}}(\ell;s):=\frac{1}{\mathrm{e}^{s\ell}-1}.
\end{align}
The volume polynomial \eqref{eq:twist_volume} twisted by $\mathsf{f}^{\mathrm{g\mathchar`-MV}}(\ell;s)$ depends on the parameter $s$ such that
\begin{align*}
V_{g,n}\big[\mathsf{f}^{\mathrm{g\mathchar`-MV}}\big](L_1,\dots,L_n;s).
\end{align*}
The factor of twist functions in equation~\eqref{eq:twist_volume} is expanded for $\mathsf{f}^{\mathrm{g\mathchar`-MV}}$:
\begin{align*}
\prod_{e\in\mathsf{E}_{\Gamma}}\mathsf{f}^{\mathrm{g\mathchar`-MV}}(\ell_e;s)
=\sum_{\substack{k_e\ge 1\\ e\in\mathsf{E}_{\Gamma}}}\mathrm{e}^{-s\sum_{e\in\mathsf{E}_{\Gamma}} k_e\ell_e}.
\end{align*}
Since the inverse Laplace transform of $\big(\mathrm{e}^{-s\sum_e k_e\ell_e}\big)/s$ with respect to the parameter $s$
gives a~Heaviside step function $\theta(\mathfrak{l}-\sum_e k_e\ell_e)$ with the dual parameter $\mathfrak{l}$,
the inverse Laplace transform of the twisted Weil--Petersson volumes
$V^{\mathrm{WP}}_{g,n}\big[\mathsf{f}^{\mathrm{g\mathchar`-MV}}\big]$
with respect to the parameter $s$
gives the average number of multicurves
whose geodesic lengths are
bounded by $\mathfrak{l}$ on the moduli space of bordered hyperbolic Riemann surfaces \cite{Mirz_MV}.
\end{rem}

\subsection{Masur--Veech type twist}\label{sec:examples_gr_tw}

Here we focus on the Masur--Veech type twist of
the Kontsevich--Witten symplectic volumes in Section \ref{subsec:gr_kw},
the volume polynomials for the $(2,p)$ minimal string in Section \ref{subsec:gr_mg},
the super symplectic volumes in Section \ref{subsec:gr_be} and
the volume polynomials for the $(2,2p-2)$ minimal superstring in Section \ref{subsec:gr_mg_super}.
In the following, we summarize the computational results of the twisted volume polynomials.

\subsubsection{Masur--Veech polynomials}\label{subsec:gr_mvp}

In \cite{Andersen_MV}, it is shown that the constant term of
the twisted Kontsevich--Witten symplectic volume
\begin{align}
V^{\mathrm{MV}}_{g,n}(L_1,\dots,L_n):=
V^{\mathrm{A}}_{g,n}\big[\sfm\big](L_1,\dots,L_n),
\label{mv_tw_kw}
\end{align}
which is referred to as the \emph{Masur--Veech polynomial},
gives the Masur--Veech volume
$\mathrm{Vol}\mathcal{Q}_{g,n}$ reviewed in Section \ref{sec_MV_stable}
(see equation~\eqref{eq:Masur--Veech_twisted}).
Some computational results of the Masur--Veech polynomials are
listed in~\eqref{lst_mv}. 

\subsubsection[Twisted volume polynomials for (2,p) minimal string]{Twisted volume polynomials for $\boldsymbol{(2,p)}$ minimal string}\label{subsec:gr_mg_twisted}

The twisted volume polynomial $\VO^{\mathrm{M}(p)}_{g,n}\big[\sfm\big]$ for the $(2,p)$ minimal string interpolates the Masur--Veech polynomial $V^{\mathrm{MV}}_{g,n}$ in equation~\eqref{mv_tw_kw} at $p=1$ and the twisted Weil--Petersson volume
$\VO^{\mathrm{WP}}_{g,n}\big[\sfm\big]$ at $p=\infty$:
\begin{gather*}
\VO^{\mathrm{M}(1)}_{g,n}\big[\sfm\big](L_1,\dots,L_n)= 
\VO^{\mathrm{A}}_{g,n}\big[\sfm\big](L_1,\dots,L_n)=
\VO^{\mathrm{MV}}_{g,n}(L_1,\dots,L_n),
\\
\VO^{\mathrm{M}(\infty)}_{g,n}\big[\sfm\big](L_1,\dots,L_n)= 
\VO^{\mathrm{WP}}_{g,n}\big[\sfm\big](L_1,\dots,L_n).
\end{gather*}
Some of the twisted Weil--Petersson volumes
$\VO^{\mathrm{WP}}_{g,n}\big[\sfm\big]$ and
the twisted volume polynomial \smash{$\VO^{\mathrm{M}(p)}_{g,n}\big[\sfm\big]$} are
listed in \eqref{lst_wp_tw_s} and~\eqref{lst_mst_twist_s}, 
respectively,
where a deformation parameter $s$ is introduced by replacing
$\pi$ with $\pi \sqrt{s}$, as in Remark~\ref{rem:fugacity_kappa},
before the twist.
A~combinatorial formula for the constant term of \smash{$\VO^{\mathrm{M}(p)}_{g,n}\big[\sfm\big]$} is provided in Proposition~\ref{prop:minimal_combi} of Section~\ref{sec:combi_minimal}.

\subsubsection{Super Masur--Veech polynomials}\label{subsec:gr_supermvp}

As a supersymmetric analogue of the Masur--Veech polynomial \eqref{mv_tw_kw},
we define the \emph{super Masur--Veech polynomial} by
\begin{align}
V^{\mathrm{SMV}}_{g,n}(L_1,\dots,L_n):=
V^{\mathrm{B}}_{g,n}\big[\sfm\big](L_1,\dots,L_n).
\label{smv_tw_bessel}
\end{align}
Some computational results are listed in \eqref{lst_smv}, 
or
found from the twisted~volume polynomials
\smash{$\VO^{\mathrm{SM}(p)}_{g,n}\big[\sfm\big]$} for the $(2,2p-2)$ minimal superstring
below in equation~\eqref{vol_minimal_superstring_twist} by a specialization
\smash{$V^{\mathrm{SMV}}_{g,n}=\VO^{\mathrm{SM}(1)}_{g,n}\big[\sfm\big]$}.

\subsubsection[Twisted volume polynomials for (2,2p-2) minimal superstring]{Twisted volume polynomials for $\boldsymbol{(2,2p-2)}$ minimal superstring}\label{subsec:gr_mg_super_twisted}

For the twisted volume polynomial
\smash{$\VO^{\mathrm{SM}(p)}_{g,n}\big[\sfm\big]$} for the $(2,2p-2)$ minimal superstring,
the twisted ABO topological recursion is solved iteratively,
and the general form of \smash{$\VO^{\mathrm{SM}(p)}_{g,n}\big[\sfm\big]$} for any odd positive integers $p$ is obtained for $g=0,1,2,3$ as follows:
\begin{align}
&
\VO^{\mathrm{SM}(p)}_{0,n}\big[\sfm\big](L_1,\dots,L_n)=0,\qquad
\VO^{\mathrm{SM}(p)}_{1,n}\big[\sfm\big](L_1,\dots,L_n)=\frac{(n-1)!}{8},\nonumber
\\
&
\VO^{\mathrm{SM}(p)}_{2,n}\big[\sfm\big](L_1,\dots,L_n)=
\frac{3(n+1)!}{128}
\Bigg[(n+2)\left(1-\frac{1}{p^2}+\frac12\right)\pi^2
+ \frac14 \sum_{i=1}^n L_i^2\Bigg],\nonumber
\\
&
\VO^{\mathrm{SM}(p)}_{3,n}\big[\sfm\big](L_1,\dots,L_n)\nonumber
\\
&\qquad=
\frac{(n+3)!}{2^{16}\cdot 5}
\bigg[\bigg\{16(n+4)\left(42n\left(1-\frac{1}{p^2}\right)+\frac{455}{2}+\frac{15}{p^2}\right)
\left(1-\frac{1}{p^2}\right)+\frac{23\cdot 40}{3}\bigg\}\pi^4\nonumber
\\
&\phantom{\qquad=}{}
+\bigg\{336(n+4)\left(1-\frac{1}{p^2}\right)+170\bigg\}\pi^2 \sum_{i=1}^nL_i^2
+84\sum_{i \ne j}^n L_i^2 L_j^2
+25\sum_{i=1}^n L_i^4\Bigg].
\label{vol_minimal_superstring_twist}
\end{align}
More computational results are listed in~\eqref{lst_super_mst_twist_s},
where we introduce a deformation parameter~$s$
by replacing $\pi$ with $\pi \sqrt{s}$ before the twist.
The twisted volume polynomial~\smash{$\VO^{\mathrm{SM}(p)}_{g,n}\big[\sfm\big]$} interpolates
the super Masur--Veech polynomial $V^{\mathrm{SMV}}_{g,n}$ in equation~\eqref{smv_tw_bessel} at~$p=1$ and the twisted super Weil--Petersson volume
\smash{$\VO^{\mathrm{SWP}}_{g,n}\big[\sfm\big]$} at $p=\infty$ summarized in~\eqref{lst_super_wp_tw_s}
such that
\begin{align*}
&\VO^{\mathrm{SM}(1)}_{g,n}\big[\sfm\big](L_1,\dots,L_n)=
\VO^{\mathrm{B}}_{g,n}\big[\sfm\big](L_1,\dots,L_n)=
\VO^{\mathrm{SMV}}_{g,n}(L_1,\dots,L_n),
\\
&\VO^{\mathrm{SM}(\infty)}_{g,n}\big[\sfm\big](L_1,\dots,L_n)=
\VO^{\mathrm{SWP}}_{g,n}\big[\sfm\big](L_1,\dots,L_n).
\end{align*}
In Proposition \ref{prop:minimal_super_combi} of Section \ref{sec:combi_super},
we find a combinatorial formula for the constant term of~\smash{$\VO^{\mathrm{SM}(p)}_{g,n}\big[\sfm\big]$}.

\section{The Masur--Veech volume and its generalizations}\label{sect:MasurVeech}

In this section, we will discuss generalizations of the Masur--Veech volume of quadratic differentials on a complex curve with marked points, and compute them for some examples on the basis of Mirzakhani's combinatorial reformulation.

\subsection{Combinatorial formula for the Masur--Veech volume}\label{sec_MV_stable}

To begin with, we summarize essential ingredients on the Masur--Veech volume of quadratic differentials discussed in \cite{DGZZ_latest} shortly.
Let $\mathcal{M}_{g,n}$ be the moduli space of complex curves of genus~$g$ with $n$ distinct labeled marked points.
On a smooth complex curve $C\in \mathcal{M}_{g,n}$,
consider a~meromorphic quadratic differential $q$ which would have
at most simple poles only at the marked points
and is not equal to the square of an Abelian differential.
The moduli space of pairs $(q,C)$ on $\mathcal{M}_{g,n}$ defines the cotangent bundle over $\mathcal{M}_{g,n}$,
and the \emph{moduli space of quadratic differentials}~$\mathcal{Q}_{g,n}$ is identified with
the total space of the cotangent bundle over $\mathcal{M}_{g,n}$ endowed with the canonical symplectic structure.
The induced volume element on $\mathcal{Q}_{g,n}$ is called the \emph{Masur--Veech volume element}.%
\footnote{
The finiteness of the Masur--Veech volume element for a subset
$\mathcal{Q}^{\mathrm{Area}(C,q)\le a}_{g,n}$ of $\mathcal{Q}_{g,n}$:
\[
\mathcal{Q}^{\mathrm{Area}(C,q)\le a}_{g,n}=\bigg\{
(C,q)\in \mathcal{Q}_{g,n}\; \bigg|\; \mathrm{Area}(C,q)=\int_C|q|\le a
\bigg\}
\subset \mathcal{Q}_{g,n},
\]
with the total area $\mathrm{Area}(C,q)$ smaller than $a>0$ is confirmed by the independent results of Masur \cite{Masur} and Veech~\cite{Veech}.
}

\begin{figure}[t]\centering
 \includegraphics[width=60mm]{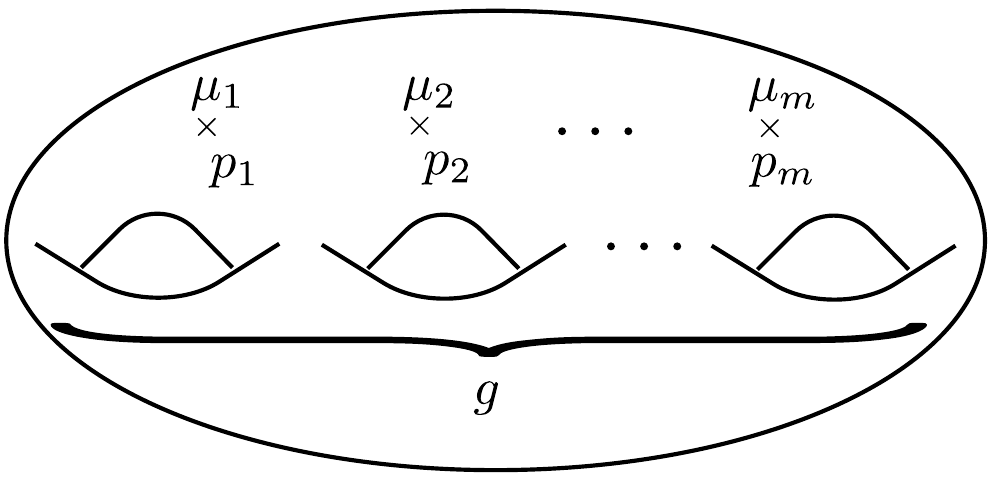}

\caption{The label set $\mu=(\mu_1,\dots,\mu_m)$ in the stratum $\mathcal{Q}(\mu)$ of
quadratic differentials on a complex curve of genus $g$ with $m$ distinct labeled marked points.}\label{fig:stratification}
\end{figure}

The moduli space $\mathcal{Q}_{g,n}$ is naturally stratified by the multiplicities of zeros and poles of quadratic differentials.

\begin{Def}[stratum of quadratic differentials]
The stratum $\mathcal{Q}(\mu)$ of quadratic differentials is the set of equivalence classes of pairs:
a smooth complex curve $C$ of genus $g$ with $m$ marked points $p_i$ $(i=1,\dots,m)$ and a quadratic differential $q$ with divisor $D=\sum_{i=1}^m\mu_ip_i$.
Here $\mu=(\mu_1,\dots,\mu_m)$ is a label set satisfying
(see Figure \ref{fig:stratification}),
\begin{align*}
\mu_i\ge -1\quad (i=1,\dots,m),
\qquad \sum_{i=1}^{m}\mu_i=4g-4,
\end{align*}
where $\mu_i=-1$ implies a simple pole.
The stratum $\mathcal{Q}(\mu)$ is a complex orbifold of dimension~${\dim_{\IC}\mathcal{Q}(\mu)=2g-2+m}$ \cite{Veech_mod}.
\end{Def}

In particular, we consider the principal stratum $\mathcal{Q}\big(1^{4g-4+n},-1^n\big)$ of meromorphic quadratic differentials on
$C \in \mathcal{M}_{g,4g-4+2n}$.
The fibers of $\mathcal{Q}\big(1^{4g-4+n},-1^n\big) \to \mathcal{Q}_{g,n}$ are discrete,
and the forgetful morphism $\mathcal{Q}_{g,4g-4+2n} \to \mathcal{Q}_{g,n}$ of
the $4g-4+n$ marked points gives a bijection between
the principal stratum $\mathcal{Q}\big(1^{4g-4+n},-1^n\big)$ modulo the choice of
$4g-4+n$ marked points and the moduli space $\mathcal{Q}_{g,n}$.
The dimension of the principal stratum coincides with that of $\mathcal{Q}_{g,n}$: $\dim_{\mathbb{C}}\mathcal{Q}\big(1^{4g-4+n},-1^n\big)=\dim_{\mathbb{C}}\mathcal{Q}_{g,n}=6g-6+2n$. We denote the Masur--Veech volume of the principal stratum $\mathcal{Q}\big(1^{4g-4+n},-1^n\big)$ as $\mathrm{Vol}\mathcal{Q}_{g,n}$.

In the work by Mirzakhani \cite{Mirz_MV}, the Masur--Veech volume $\mathrm{Vol}\mathcal{Q}_{g,n}$ is reformulated as an enumerative problem of simple closed curves in a connected bordered (hyperbolic) Riemann surface,
and a connection with the Weil--Petersson volume is found.
This reformulation is established further on the basis of the enumerative problem of square-tiled surfaces in the work by Delecroix, Goujard, Zograf and Zorich \cite{DGZZ_latest}.
By such reformulations, a relation between the Masur--Veech volume of the moduli space of quadratic differentials and the intersection numbers on the (compactified) moduli space $\overline{\mathcal{M}}_{g,n}$ is unveiled.
The stable graphs $\Gamma \in G_{g,n}$ associated with a bordered surface $S_{g,n}$ and a reduced multicurve $\gamma_{\mathrm{red}}$
in Section \ref{sec:twisting_gr} are used in the combinatorial computation of
the Masur--Veech volume $\mathrm{Vol}\mathcal{Q}_{g,n}$.

\begin{thm}[combinatorial formula for the Masur--Veech volume \cite{DGZZ_latest,Mirz_MV}]\label{thm_combi_MV}
Consider the decomposition
$S_{g,n}\setminus\gamma_{\mathrm{red}}=\cup_{a=1}^N S_{g_a,n_a}$ in~\eqref{eq:cut_Riemann_surface} represented by a stable graph $\Gamma$
associated with a bordered Riemann surface $S_{g,n}$ and
a reduced multicurve $\gamma_{\mathrm{red}}$.
Define a polynomial
\begin{align}
\label{eq:polynomial_WP}
\mathrm{Vol}^{\mathrm{WP}}_{\Gamma}(\ell_1,\dots,\ell_k):=
\prod_{a=1}^N V^{\mathrm{WP}}_{g_a,n_a}(\ell_{k_1},\dots,\ell_{k_{n_a}})
\Big|_{L_p=0 (p=1,\dots,n)},
\end{align}
where $V^{\mathrm{WP}}_{g_a,n_a}(\ell_{k_1},\dots,\ell_{k_{n_a}})$ are the Weil--Petersson volumes of moduli spaces of connected pieces $S_{g_a,n_a}$,
and all boundary lengths $L_p$ $(p=1,\dots n)$ of the original bordered Riemann surface set to zero in each connected piece.
Let $(2d_1,\dots,2d_k)^{\mathrm{WP}}_{\Gamma}$ denote the coefficient of $\ell_1^{2d_1}\cdots \ell_k^{2d_k}$ in~$\mathrm{Vol}^{\mathrm{WP}}_{\Gamma}(\ell_1,\dots,\ell_k)$.
Then, from these combinatorial data,
the Masur--Veech volume $\mathrm{Vol}\mathcal{Q}_{g,n}$ of the moduli space of quadratic differentials is given by\footnote{In \cite{DGZZ_latest,Mirz_MV}, an extra factor $2^{-M(\gamma)}$ appears in equation~\eqref{comb_formula_mv} originated from the normalization for the Weil--Petersson volume $V_{1,1}^{\mathrm{WP}}(L)$ in \cite[Table~1]{Mirz3}.
In this article, we employ another normalization $V_{1,1}^{\mathrm{WP}}(L)$ obtained from the CEO topological recursion
which is different by a factor $2$ from Mirzakhani's original computation.}
\begin{align}
\mathrm{Vol}\mathcal{Q}_{g,n}
={}&\sum_{\Gamma\in G_{g,n}}
\frac{\alpha_{g,n}}{|\mathrm{Aut}\left(\Gamma\right)|}\nonumber
 \\
&\times
\sum_{|\mathbf{d}|= 3g-3+n-k}(2d_1,\dots,2d_k)^{\mathrm{WP}}_{\Gamma}
\frac{\prod_{i=1}^k(2d_i+1)!}{(6g-6+2n)!}
\prod_{i=1}^k\zeta(2d_i+2),\label{comb_formula_mv}
\end{align}
where $|\mathbf{d}|=\sum_{i=1}^kd_i$, and the normalization constant $\alpha_{g,n}$ is
\begin{align}
\label{eq:normalization_MV}
\alpha_{g,n}:=2\cdot(6g-6+2n)\cdot(4g-4+n)!\cdot 2^{4g-3+n}.
\end{align}
\end{thm}

\begin{rem}
The normalization constant $\alpha_{g,n}$ in equation~\eqref{eq:normalization_MV} is obtained in a series of works \cite{AH1,ErSo,MoT,Wr}.
\end{rem}

\begin{rem}
Let $\mathrm{Vol}^{\mathrm{A}}_{\Gamma}(\ell_1,\dots,\ell_k)$
be a polynomial defined by replacing the Weil--Petersson volumes
in equation~\eqref{eq:polynomial_WP} with the Kontsevich--Witten symplectic volumes
of moduli spaces of stable curves in equation~\eqref{kw_volume},
and $(2d_1,\dots,2d_k)^{\mathrm{A}}_{\Gamma}$ be
the coefficient of $\ell_1^{2d_1}\cdots \ell_k^{2d_k}$ in $\mathrm{Vol}^{\mathrm{A}}_{\Gamma}(\ell_1,\dots,\ell_k)$.
The combinatorial formula \eqref{comb_formula_mv} also holds if $(2d_1,\dots,2d_k)^{\mathrm{WP}}_{\Gamma}$ is replaced with $(2d_1,\dots,2d_k)^{\mathrm{A}}_{\Gamma}$.
\end{rem}

The following proposition is shown in \cite{Andersen_MV}.

\begin{prop}[\cite{Andersen_MV}]\label{prop:mv_twisted_abo}
The Masur--Veech volumes are given by the constant terms in the Kontsevich--Witten symplectic volumes with the Masur--Veech type twist in equation~\eqref{mv_twist}:
\begin{align}
\label{eq:Masur--Veech_twisted}
\mathrm{Vol}\mathcal{Q}_{g,n}
=\beta_{g,n}
V^{\mathrm{A}}_{g,n}\big[\sfm\big](0,\dots,0)
=\beta_{g,n}
V^{\mathrm{MV}}_{g,n}(0,\dots,0),
\end{align}
where the normalization constants $\beta_{g,n}$ are
\begin{align}
&\label{eq:volume_factor}
\beta_{g,n}:=\frac{2^{4g-2+n} (4g-4+n)!}{(6g-7+2n)!}.
\end{align}
\end{prop}
\begin{proof}
The integrals of $\ell_e$ in equation~\eqref{eq:twist_volume}
give the factors involving $\zeta(2d_i+2)$
in the combinatorial formula \eqref{comb_formula_mv} for the Masur--Veech volumes.
As a result, the collection of monomials with degree $|\mathbf{d}|= 3g-3+n-k$ in $V^{\mathrm{WP}}_{g,n}$ agrees with that of $V^{\mathrm{A}}_{g,n}$.
\end{proof}

For example, the Masur--Veech volumes $\mathrm{Vol}\mathcal{Q}_{1,1}$ and $\mathrm{Vol}\mathcal{Q}_{0,4}$ are computed as follows.
(A~huge table of the Masur--Veech volumes is obtained in \cite{Goujard15}.)

\begin{ex}[$g=1$, $n=1$]
One finds a decomposition \eqref{eq:cut_Riemann_surface} which splits $S_{1,1}$ into one $S_{0,3}$, and obtains a stable graph \smash{$\Gamma_{1,1}^{(1)}$} as described in Figure \ref{fig:stable_graph_MV} (top).
\begin{figure}[t]\centering
 \includegraphics[width=110mm]{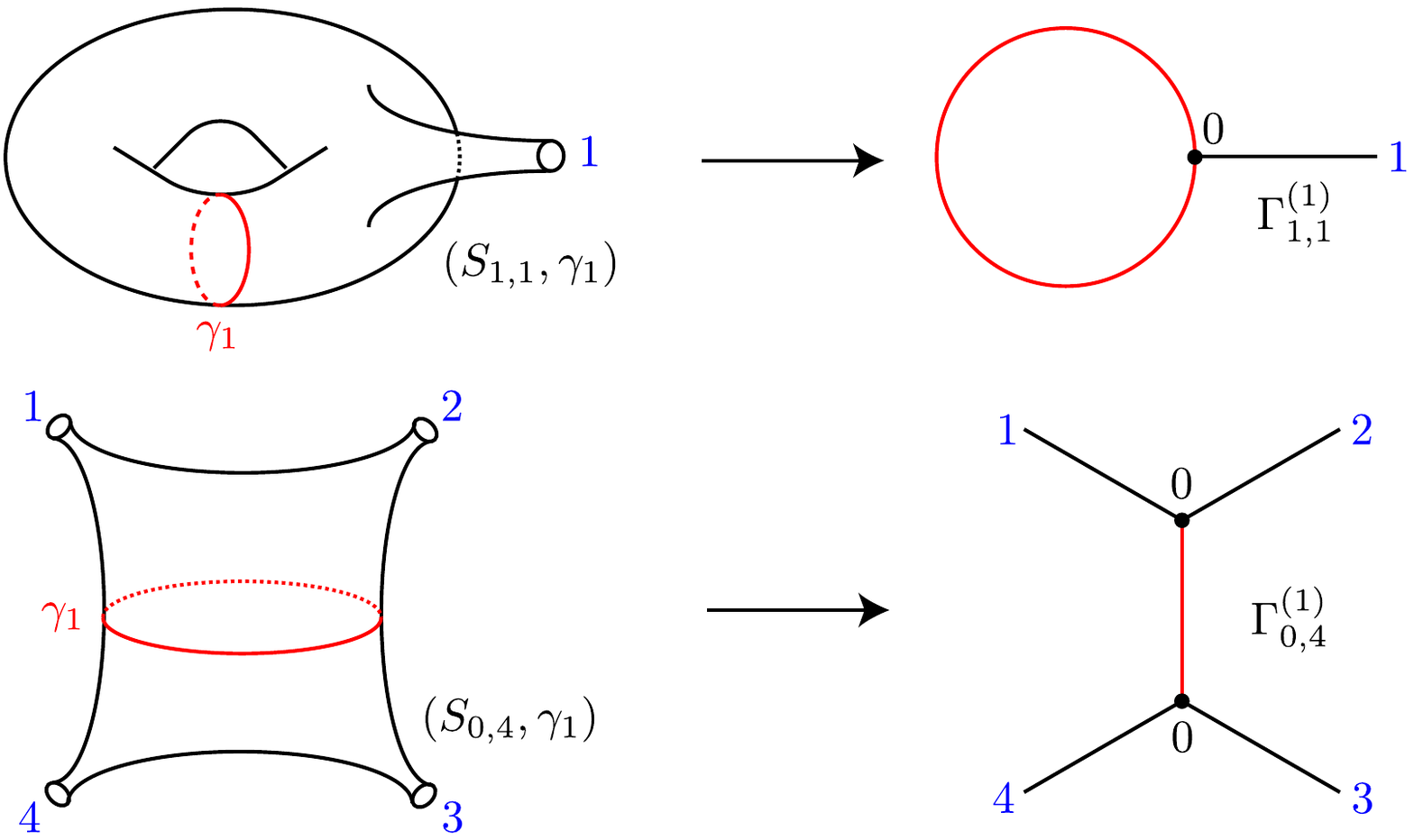}
 \vspace{0cm}
\caption{Top: A stable graph $\Gamma_{1,1}^{(1)}$ from $(S_{1,1},\gamma_1)$.
Bottom: A stable graph $\Gamma_{0,4}^{(1)}$ from $(S_{0,4},\gamma_1)$.}\label{fig:stable_graph_MV}
\end{figure}
A discrete $\mathbb{Z}/2\mathbb{Z}$ symmetry of the loop
in \smash{$\Gamma_{1,1}^{(1)}$} is found, and
\smash{$\big|\mathrm{Aut}\big(\Gamma_{1,1}^{(1)}\big)\big|=2$}.
The polynomial \eqref{eq:polynomial_WP} for this decomposition is \smash{$\mathrm{Vol}^{\mathrm{WP}}_{\Gamma_{1,1}^{(1)}}(x)
=V^{\mathrm{WP}}_{0,3}(x,x,0)=1$}.
Therefore, $k=1$, $d_1=0$ and \smash{$(2d_1)^{\mathrm{WP}}_{\Gamma_{1,1}^{(1)}}=1$} for this stable graph contributes to the sum in equation~\eqref{comb_formula_mv},
and one finds
\begin{align}
\label{MV_volume_11}
\mathrm{Vol}\mathcal{Q}_{1,1}
=\frac{\alpha_{1,1}}{\big|\mathrm{Aut}\big(\Gamma_{1,1}^{(1)}\big)\big|}
(2d_1)^{\mathrm{WP}}_{\Gamma_{1,1}^{(1)}} \frac{(2d_1+1)!}{2!}
\zeta(2d_1+2)
=\frac{2\pi^2}{3}.
\end{align}
This agrees with
$\beta_{1,1}V^{\mathrm{MV}}_{1,1}(0)=2^3\cdot \pi^2/12$ in
equation~\eqref{eq:Masur--Veech_twisted}
(see \eqref{lst_mv}). 
\end{ex}

\begin{ex}[$g=0$, $n=4$]
One finds a decomposition \eqref{eq:cut_Riemann_surface} which splits $S_{0,4}$ into two~$S_{0,3}$'s, and obtain a stable graph \smash{$\Gamma_{0,4}^{(1)}$} as described in Figure \ref{fig:stable_graph_MV} (bottom).
There are $\binom{4}{2}=6$ choices of the distribution of four labeled external legs into two $S_{0,3}$'s in this case.
A discrete $\mathbb{Z}/2\mathbb{Z}$ symmetry of the stable graph
\smash{$\Gamma_{0,4}^{(1)}$} leads to \smash{$\big|\mathrm{Aut}\big(\Gamma_{0,4}^{(1)}\big)\big|=2$}.
The polynomial \eqref{eq:polynomial_WP} for this decomposition~is
\[\mathrm{Vol}^{\mathrm{WP}}_{\Gamma_{0,4}^{(1)}}(x)=V^{\mathrm{WP}}_{0,3}(x,0,0)\cdot V^{\mathrm{WP}}_{0,3}(x,0,0)=1.\]
Therefore, $k=1$, $d_1=0$ and \smash{$(2d_1)^{\mathrm{WP}}_{\Gamma_{0,4}^{(1)}}=1$} for this stable graph contributes to the sum in equation~\eqref{comb_formula_mv}, and one finds
\begin{align}
\label{MV_volume_04}
\mathrm{Vol}\mathcal{Q}_{0,4}=
\frac{\alpha_{0,4}}{\big|\mathrm{Aut}\big(\Gamma_{0,4}^{(1)}\big)\big|}
\binom{4}{2}
(2d_1)^{\mathrm{WP}}_{\Gamma^{(1)}} \frac{(2d_1+1)!}{2!}
\zeta(2d_1+2)
=2\pi^2.
\end{align}
This agrees with
$\beta_{0,4}V^{\mathrm{MV}}_{0,4}(0,0,0,0)=2^2\cdot \pi^2/2$ in
equation~\eqref{eq:Masur--Veech_twisted}
(see \eqref{lst_mv}). 
\end{ex}

\subsection[Generalization to the (2,p) minimal string]{Generalization to the $\boldsymbol{(2,p)}$ minimal string}\label{sec:combi_minimal}

Now we consider a generalization of the combinatorial formula \eqref{comb_formula_mv} to the $(2,p)$ minimal string
whose volume polynomials interpolate the Kontsevich--Witten symplectic volumes (for $p=1$) and
the Weil--Petersson volumes (for $p=\infty$).
For this purpose, we consider a stable graph $\Gamma \in G_{g,n}$ associated to
a pair $(S_{g,n},\gamma_{\mathrm{red}})$ of the decomposition \eqref{eq:cut_Riemann_surface} in Section \ref{sec:twisting_gr},
and define%
\begin{align}
(2d_1,\dots,2d_k)_{\Gamma}={}&
\textrm{the coefficient of}\ \ell_1^{2d_1}\cdots \ell_k^{2d_k}\
\textrm{in}\nonumber\\
&
\prod_{a=1}^N V_{g_a,n_a}(\ell_{k_1},\dots,\ell_{k_{n_a}})
\Big|_{L_p=0 (p=1,\dots,n)}.
\label{vol_component_comb}
\end{align}
Here $V_{g_a,n_a}\!(\ell_{k_1},\dots,\ell_{k_{n_a}})|_{L_p=0 (p=1,\dots,n)}$ are volume polynomials associated with the stable graph~$\Gamma$
and with zero boundary lengths $L_p=0$ ($p=1,\dots n$) for the bordered boundaries in the decomposed Riemann surfaces
specified by the univalent vertices in $\Gamma$. Then, the following proposition is proved.

\begin{prop}
\label{prop:minimal_combi}
The constant term in the twisted volume polynomial
$V^{\mathrm{M}(p)}_{g,n}\big[\mathsf{f}^{\mathrm{MV}}\big]$
of the $(2,p)$ minimal string normalized by $\beta_{g,n}$ in equation~\eqref{eq:volume_factor},
\begin{align}
\label{eq:minimal_twisted}
\mathrm{Vol}\mathcal{Q}^{\mathrm{M}(p)}_{g,n}:=
\beta_{g,n}
V^{\mathrm{M}(p)}_{g,n}\big[\mathsf{f}^{\mathrm{MV}}\big](0,\dots,0),
\end{align}
is obtained, as a sum over stable graphs $\Gamma\in G_{g,n}$, by
\begin{align}
\label{eq:normalization_minimal_volume}
\mathrm{Vol}\mathcal{Q}^{\mathrm{M}(p)}_{g,n}
=&\sum_{\Gamma\in G_{g,n}}
\frac{\alpha_{g,n}}{|\mathrm{Aut}\left(\Gamma\right)|}
\nonumber \\
&
\times
\sum_{|\mathbf{d}|\le 3g-3+n-k}(2d_1,\dots,2d_k)_{\Gamma}^{\mathrm{M}(p)}
\frac{\prod_{i=1}^k(2d_i+1)!}{(6g-6+2n)!}
\prod_{i=1}^k\zeta(2d_i+2),
\end{align}
in terms of $(2d_1,\dots,2d_k)_{\Gamma}^{\mathrm{M}(p)}$ defined by equation~\eqref{vol_component_comb} for the $(2,p)$ minimal string,
where the normalization factor $\alpha_{g,n}$ is the same as equation~\eqref{eq:normalization_MV}.
\end{prop}

\begin{rem}
We refer to $\mathrm{Vol}\mathcal{Q}^{\mathrm{M}(p)}_{g,n}$ as
a twisted volume of the $(2,p)$ minimal string, since this is a natural combinatorial analogue of the Masur--Veech volume, although the geometric derivation of this volume is missing in the direct study of the quantum moduli space of the Liouville gravity.
A crucial difference between the combinatorial formulae of Theorem \ref{thm_combi_MV} and Proposition \ref{prop:minimal_combi} is the degree constraint in the sum.
In the computation of the Masur--Veech volume, the imposed degree constraint is $|\mathbf{d}|=3g-3+n-k$. On the other hand for the $(2,p)$ minimal string, the weaker degree constraint $|\mathbf{d}|\le 3g-3+n-k$ is imposed.
If the degree constraint for the Masur--Veech volume is instead imposed for the volume formula of the minimal string,
then the twisted volumes \smash{$\mathrm{Vol}\mathcal{Q}^{\mathrm{M}(p)}_{g,n}$} reduce to the Masur--Veech volumes $\mathrm{Vol}\mathcal{Q}_{g,n}$.
\end{rem}

Here we show combinatorial computations for \smash{$\mathrm{Vol}\mathcal{Q}^{\mathrm{M}(p)}_{1,1}$} and \smash{$\mathrm{Vol}\mathcal{Q}^{\mathrm{M}(p)}_{0,4}$}.

\begin{ex}[$g=1$, $n=1$]
By the weaker constraint $|\mathbf{d}|\le 3g-3+n-k$ in the sum in equation~\eqref{eq:normalization_minimal_volume},
an extra contribution to the sum in equation~\eqref{eq:normalization_minimal_volume} is found for $k=0$.
Namely for this contribution, there are no multicurves in $S_{1,1}$, and the factor $\prod_{i=1}^k(2d_i+1)!\zeta(2d_i+2)$ is absent.
The corresponding stable graph \smash{$\Gamma_{1,1}^{(0)}$} has one vertex labeled by $g=1$ and one half-edge as described in Figure \ref{fig:stable_graph_minimal} (top).
\begin{figure}[t]\centering
 \includegraphics[width=100mm]{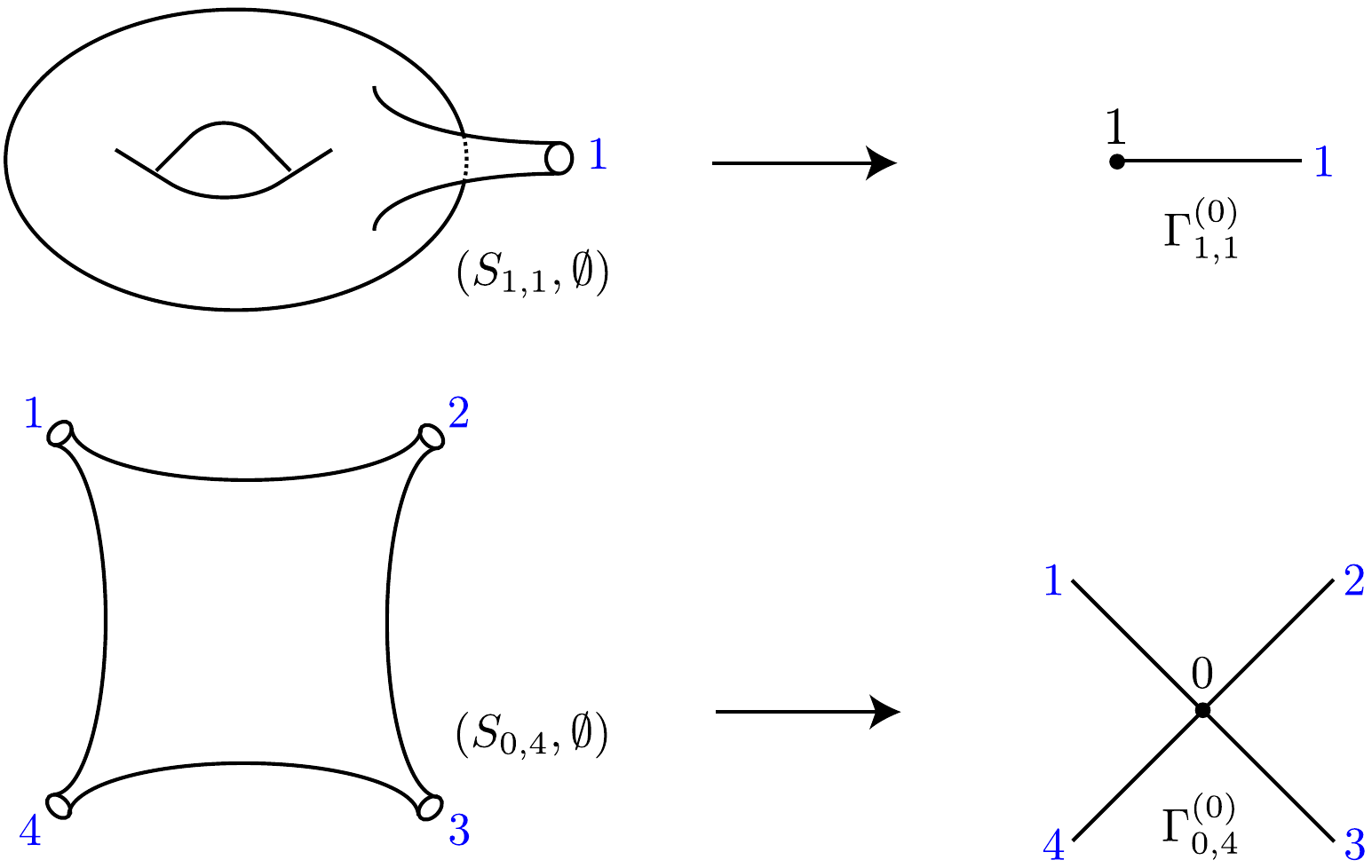}
\caption{Top: A stable graph $\Gamma_{1,1}^{(0)}$ from $(S_{1,1},\varnothing)$.
Bottom: A stable graph $\Gamma_{0,4}^{(0)}$ from $(S_{0,4},\varnothing)$.}\label{fig:stable_graph_minimal}
\end{figure}
The polynomial factor obeys
\[(\varnothing)_{\Gamma_{1,1}^{(0)}}^{\mathrm{M}(p)}
=V^{\mathrm{M}(p)}_{1,1}(0)=\pi^2s\big(1-1/p^2\big)/12\] for this stable graph,
where a deformation parameter $s$ is introduced by $\pi^2 \to \pi^2 s$.
The contributions to \smash{$\mathrm{Vol}\mathcal{Q}^{\mathrm{M}(p)}_{1,1}$} from the stable graph \smash{$\Gamma_{1,1}^{(0)}$} is
\begin{align*}
\frac{\alpha_{1,1}}{\big|\mathrm{Aut}\big(\Gamma_{1,1}^{(0)}\big)\big|}
(\varnothing)_{\Gamma_{1,1}^{(0)}}^{\mathrm{M}(p)} \frac{1}{2!}
=\frac{2\pi^2s}{3}\left(1-\frac{1}{p^2}\right).
\end{align*}
Combining the contribution \eqref{MV_volume_11} from the stable graph $\Gamma_{1,1}^{(1)}$, one finds
\begin{align*}
\mathrm{Vol}\mathcal{Q}^{\mathrm{M}(p)}_{1,1}=\frac{2\pi^2s}{3}\left(1-\frac{1}{p^2}\right)+\frac{2\pi^2}{3}.
\end{align*}
This agrees with
\[\beta_{1,1}V^{\mathrm{M}(p)}_{1,1}\big[\sfm\big](0)=2^3\cdot \pi^2\big(s-s/p^2+1\big)/12\] in
equation~\eqref{eq:minimal_twisted}
(see \eqref{lst_mst_twist_s}). 
\end{ex}

\begin{ex}[$g=0$, $n=4$]
We compute an extra contribution of $k=0$ which comes
from the weaker constraint $|\mathbf{d}|\le 3g-3+n-k$ in the sum in equation~\eqref{eq:normalization_minimal_volume}.
The corresponding stable graph \smash{$\Gamma_{0,4}^{(0)}$} has one vertex labeled by $g=0$ and four half-edges as described in Figure~\ref{fig:stable_graph_minimal} (bottom).
The polynomial factor obeys
\[(\varnothing)_{\Gamma_{0,4}^{(0)}}^{\mathrm{M}(p)}
=V^{\mathrm{M}(p)}_{0,4}(0,0,0,0)=2\pi^2s\big(1-1/p^2\big)\]
for this stable graph.
The contributions to \smash{$\mathrm{Vol}\mathcal{Q}^{\mathrm{M}(p)}_{0,4}$} from the stable graph \smash{$\Gamma_{0,4}^{(0)}$} is
\begin{align*}
\frac{\alpha_{0,4}}{\big|\mathrm{Aut}\big(\Gamma_{0,4}^{(0)}\big)\big|}
(\varnothing)_{\Gamma_{0,4}^{(0)}}^{\mathrm{M}(p)} \frac{1}{2!}
=8\pi^2s \left(1-\frac{1}{p^2}\right).
\end{align*}
Combining the contribution \eqref{MV_volume_04} from the stable graph $\Gamma_{0,4}^{(1)}$, one finds
\begin{align*}
\mathrm{Vol}\mathcal{Q}^{\mathrm{M}(p)}_{0,4}=
8\pi^2s \left(1-\frac{1}{p^2}\right)+2\pi^2.
\end{align*}
This agrees with
\[\beta_{0,4}V^{\mathrm{M}(p)}_{0,4}\big[\sfm\big](0,0,0,0)=2^2\cdot \pi^2\big(4s-4s/p^2+1\big)/2\] in
equation~\eqref{eq:minimal_twisted}
(see \eqref{lst_mst_twist_s}). 
\end{ex}

\subsection[Generalization to the (2,2p-2) minimal superstring]{Generalization to the $\boldsymbol{(2,2p-2)}$ minimal superstring}\label{sec:combi_super}

The combinatorial formula in Theorem \ref{thm_combi_MV} is also generalized to the volume polynomials for the supersymmetric models
which are associated with the BGW tau function \cite{Alexandrov:2016kjl}.
We consider the~$(2,2p-2)$ minimal superstring whose volume polynomials interpolate
the volume polynomials~$V^{\mathrm{B}}_{g,n}(L_1,\dots,L_n)$ for $p=1$ and the super Weil--Petersson volumes $V^{\mathrm{SWP}}_{g,n}(L_1,\dots,L_n)$ for~$p=\infty$,
and find the following proposition.

\begin{prop}
\label{prop:minimal_super_combi}
Let $\overline{G}_{g,n}$ be the set of stable graphs which does not contain any vertices associated to connected bordered Riemann surfaces of
genus zero.%
\footnote{
The volume polynomial $V^{\mathrm{SM}(p)}_{0,n}(0,\dots,0)$ is zero for any $n$.
}
The constant term in the twisted volume polynomial
\smash{$V^{\mathrm{SM}(p)}_{g,n}\big[\mathsf{f}^{\mathrm{MV}}\big]$}
for the $(2,2p-2)$ minimal superstring
is obtained, as a~sum over stable graphs $\Gamma\in \overline{G}_{g,n}$, by
\begin{align}
V^{\mathrm{SM}(p)}_{g,n}\big[\sfm\big](0,\dots,0)
={}&\sum_{\Gamma\in \overline{G}_{g,n}}
\frac{1}{|\mathrm{Aut}\left(\Gamma\right)|}
\sum_{|\mathbf{d}|\le g-1-k}(2d_1,\dots,2d_k)_{\Gamma}^{\mathrm{SM}(p)}\nonumber
\\
&\times
\prod_{i=1}^k(2d_i+1)! \zeta(2d_i+2),\label{eq:normalization_minimal_super_volume}
\end{align}
where $(2d_1,\dots,2d_k)_{\Gamma}^{\mathrm{SM}(p)}$ is
defined from the twisted volume polynomial
$V^{\mathrm{SM}(p)}_{g,n}\big[\mathsf{f}^{\mathrm{MV}}\big]$
by equation~\eqref{vol_component_comb}.
\end{prop}

\begin{rem}\label{rem:super_minimal_volume}
If the degree constraint $|\mathbf{d}|\le g-1-k$ in equation~\eqref{eq:normalization_minimal_super_volume} is replaced by the stronger condition $|\mathbf{d}|=g-1-k$,
then the twisted volume \smash{$V^{\mathrm{SM}(p)}_{g,n}\big[\sfm\big](0,\dots,0)$} for an odd positive integer $p$
reduces to $V^{\mathrm{B}}_{g,n}\big[\sfm\big](0,\dots,0)$.
\end{rem}

Here we show combinatorial computations for $(g,n)=(2,1)$ and $(3,1)$.

\begin{ex}[$g=2$, $n=1$]
One finds three stable graphs which correspond to decompositions \eqref{eq:cut_Riemann_surface} of $S_{2,1}$ without $g=0$ components as described in Figure~\ref{fig:stable_graph_super_21}.
\begin{figure}[t]\centering
 \includegraphics[width=100mm]{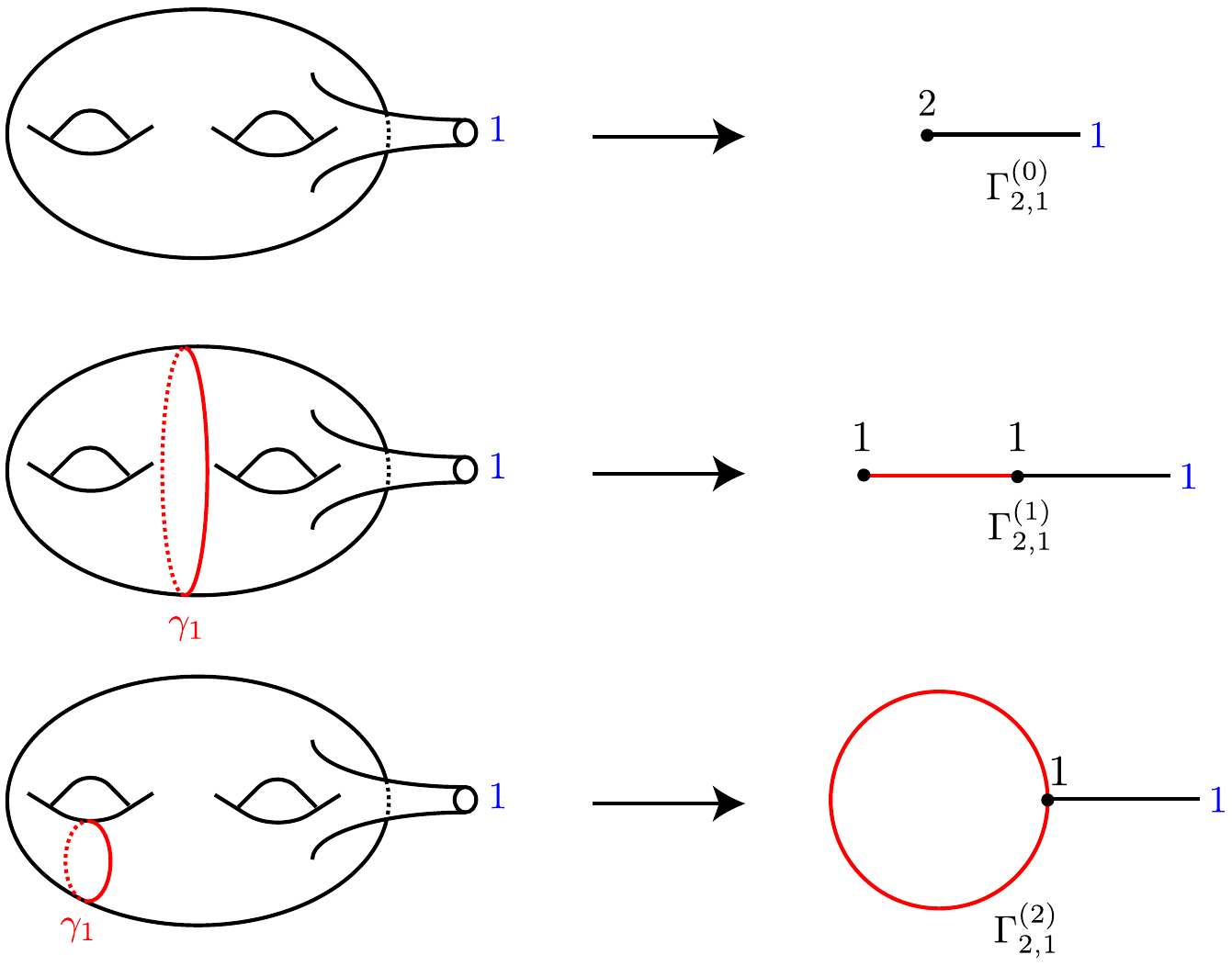}

\caption{Stable graphs $\Gamma_{2,1}^{(0)},\Gamma_{2,1}^{(1)},\Gamma_{2,1}^{(2)}\in\overline{G}_{2,1}$ obtained from the
decompositions of $S_{2,1}$ which do not involve $S_{0,n}$'s.}\label{fig:stable_graph_super_21}
\end{figure}
For the graph \smash{$\Gamma_{2,1}^{(0)}$}, a~contribution
in equation~\eqref{eq:normalization_minimal_super_volume} is
\begin{align*}
\mathrm{vol}^{(0)}_{2,1}=
\frac{1}{\big|\mathrm{Aut}\big(\Gamma_{2,1}^{(0)}\big)\big|}
(\varnothing)_{\Gamma^{(0)}_{2,1}}^{\mathrm{SM}(p)}
=\frac{9\pi^2s}{64}\left(1-\frac{1}{p^2}\right),
\qquad \text{where} \quad (\varnothing)_{\Gamma^{(0)}_{2,1}}^{\mathrm{SM}(p)}=V^{\mathrm{SM}(p)}_{2,1}(0),
\end{align*}
and a deformation parameter $s$ is introduced by $\pi^2 \to \pi^2 s$.
For~the graph \smash{$\Gamma_{2,1}^{(1)}$}, a contribution
in equation~\eqref{eq:normalization_minimal_super_volume} is
\begin{align*}
\mathrm{vol}^{(1)}_{2,1}=
\frac{1}{\big|\mathrm{Aut}\big(\Gamma_{2,1}^{(1)}\big)\big|}
(0)_{\Gamma^{(1)}_{2,1}}^{\mathrm{SM}(p)} \zeta(2)
=\frac{\pi^2}{384},
\end{align*}
where \smash{$(0)_{\Gamma^{(1)}_{2,1}}^{\mathrm{SM}(p)}$} is
the coefficient of $\ell^0$ in
\smash{$V^{\mathrm{SM}(p)}_{1,2}(\ell,0)\cdot V^{\mathrm{SM}(p)}_{1,1}(\ell)=1/64$}.
For the graph \smash{$\Gamma_{2,1}^{(2)}$}, a~contribution
in equation~\eqref{eq:normalization_minimal_super_volume} is
\begin{align*}
\mathrm{vol}^{(2)}_{2,1}=
\frac{1}{\big|\mathrm{Aut}\big(\Gamma_{2,1}^{(2)}\big)\big|} (0)_{\Gamma^{(2)}_{2,1}}^{\mathrm{SM}(p)} \zeta(2)
=\frac{\pi^2}{48},
\end{align*}
where $(0)_{\Gamma^{(2)}_{2,1}}^{\mathrm{SM}(p)}$ is
the coefficient of $\ell^0$ in
$V^{\mathrm{SM}(p)}_{1,3}(\ell,\ell,0)=1/4$, and $\mathrm{Aut}\big(\Gamma_{2,1}^{(2)}\big)=\mathbb{Z}/2\mathbb{Z}$.
Summing these three contributions, one obtains
\begin{align*}
&
V^{\mathrm{SM}(p)}_{2,1}\big[\mathsf{f}^{\mathrm{MV}}\big](0)
=\mathrm{vol}^{(0)}_{2,1}+\mathrm{vol}^{(1)}_{2,1}+\mathrm{vol}^{(2)}_{2,1}
=\frac{\pi^2}{128}
\left(
18s\left(1-\frac{1}{p^2}\right)+3
\right),
\end{align*}
which agrees with the constant term of the twisted volume polynomial $V^{\mathrm{SM}(p)}_{2,1}\big[\sfm\big]$
(see \eqref{lst_super_mst_twist_s}). 
\end{ex}

\begin{ex}[$g=3$, $n=1$]
One finds ten multicurves which decompose $S_{3,1}$ without $g=0$ components as described in Figure \ref{fig:stable_graph_super_31}.
\begin{figure}[t]\centering
 \includegraphics[width=120mm]{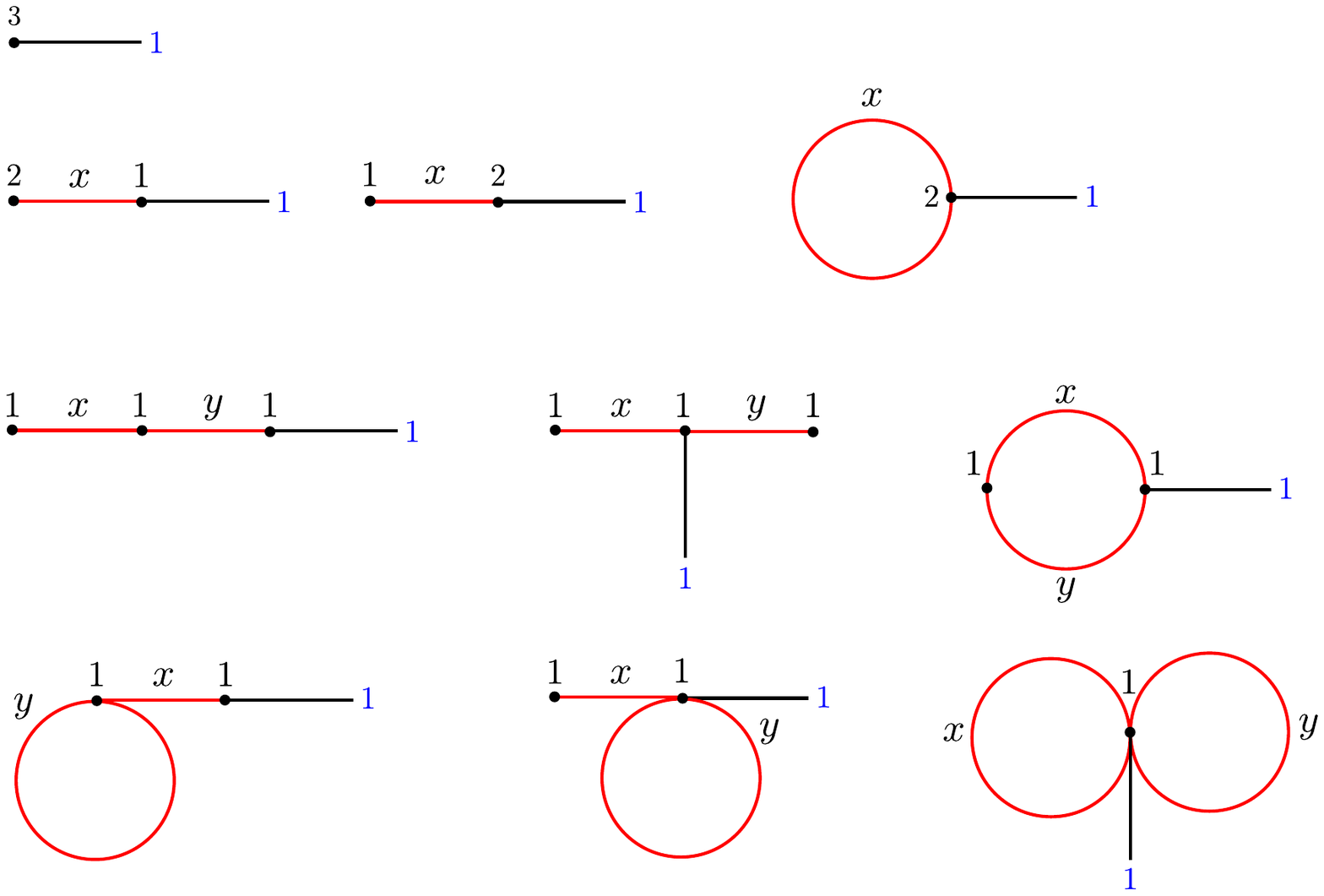}

\caption{Ten stable graphs obtained from the decompositions of $S_{3,1}$ which do not involve $S_{0,n}$'s.}\label{fig:stable_graph_super_31}
\end{figure}
Summing these ten contributions, one obtains
\begin{align*}
&V^{\mathrm{SM}(p)}_{3,1}\big[\mathsf{f}^{\mathrm{MV}}\big](0)
\\
&\qquad=V^{\mathrm{SM}(p)}_{3,1}(0)
+\left(
V^{\mathrm{SM}(p)}_{2,1}(x)\cdot V^{\mathrm{SM}(p)}_{1,2}(x,0)
+V^{\mathrm{SM}(p)}_{1,1}(x)\cdot V^{\mathrm{SM}(p)}_{2,2}(x,0)\right.\\
&\left.\qquad\quad{}
+\frac{1}{2}V^{\mathrm{SM}(p)}_{2,3}(x,x,0)
\right)\bigg|_{x^0}\zeta(2)
+\left(
V^{\mathrm{SM}(p)}_{2,1}(x)\cdot V^{\mathrm{SM}(p)}_{1,2}(x,0)
+V^{\mathrm{SM}(p)}_{1,1}(x)\cdot V^{\mathrm{SM}(p)}_{2,2}(x,0)\!
\right.\\
&\left.\qquad\quad{}+\frac{1}{2}V^{\mathrm{SM}(p)}_{2,3}(x,x,0)
\right)\bigg|_{x^2}3!\zeta(4)+\biggl(
V^{\mathrm{SM}(p)}_{1,1}(x)\cdot V^{\mathrm{SM}(p)}_{1,2}(x,y)\cdot V^{\mathrm{SM}(p)}_{1,2}(y,0)
\\
&\qquad\quad{}+\frac{1}{2}V^{\mathrm{SM}(p)}_{1,1}(x)V^{\mathrm{SM}(p)}_{1,3}(x,y,0)\cdot V^{\mathrm{SM}(p)}_{1,1}(y)
\\
&\qquad\quad{}
+\frac{1}{2}V^{\mathrm{SM}(p)}_{1,2}(x,y)V^{\mathrm{SM}(p)}_{1,3}(x,y,0)
+\frac{1}{2} V^{\mathrm{SM}(p)}_{1,3}(x,y,y)\cdot V^{\mathrm{SM}(p)}_{1,2}(x,0)
\\
&\qquad\quad{}
+\frac{1}{2}V^{\mathrm{SM}(p)}_{1,1}(x)\cdot V^{\mathrm{SM}(p)}_{1,4}(x,y,y,0)
+\frac{1}{2^3}V^{\mathrm{SM}(p)}_{1,5}(x,x,y,y,0)
\biggr)\bigg|_{x^0y^0}\zeta(2)^2
 \\
&\qquad=\frac{\pi^4}{1024}\left(
\left(1-\frac{1}{p^2}\right)
\left(6s^2\left(227-\frac{27}{p^2}\right)+255s\right)
+23\right),
\end{align*}
where $|_{x^ky^l}$ picks up coefficients of $x^ky^l$ in the polynomials.
This result agrees with the constant term of the twisted volume polynomial \smash{$V^{\mathrm{SM}(p)}_{3,1}\big[\sfm\big]$}
(see \eqref{lst_super_mst_twist_s}). 
\end{ex}

\section{CEO topological recursion}\label{sec:topological_recursion}

In this section, we apply the CEO topological recursion \cite{CE,EO2007}, which is a Laplace dual formulation of the ABO topological recursion,
to the 2D gravity models in Table \ref{tab:spectral_curve}.
In particular, in Section \ref{sec:twisting_tr} we provide a direct proof of
Theorem \ref{thm:Laplace_dual_GR} on the Laplace dual relation between
the ABO topological recursion with the Masur--Veech type twist
and the CEO topological recursion with the Masur--Veech type twist for the $(2,p)$ minimal string and the $(2,2p-2)$ minimal superstring.

\subsection{Formulation}\label{sec:def_tr}

We briefly review the formulation of the CEO topological recursion, and
describe a Laplace dual relation with the ABO topological recursion.

\begin{Def}\label{def:sp_curve}
A \emph{spectral curve} $\cC=(\Sigma;\sfx,\sfy,B)$ consists of a Riemann surface $\Sigma$, meromorphic functions $\sfx, \sfy\colon \Sigma \to {\IC}$ such that the zeros of ${\rm d}\sfx$
are different from the zeros of ${\rm d}\sfy$, and a~bidifferential $B$ on
$\Sigma^{\otimes 2}$.
\end{Def}

\begin{Def}[CEO topological recursion \cite{EO2007}]\label{def:top_rec}
For a spectral curve $\cC=(\Sigma;\sfx,\sfy,B)$ such that the zeros of ${\rm d}\sfx$ are simple,
the meromorphic multidifferentials $\omega_{g,n}(z_1, \dots, z_n)$,
$z_i \in \Sigma$, labeled by $g \ge 0$, $n \ge 1$ satisfying $2g-2+n \ge 0$,
are defined by the CEO topological recursion
\begin{align*}
\omega_{g,n}(z_1, \dots, z_n) ={}&
\sum_{\alpha \in \mathrm{Ram}}
\mathop{\mathrm{Res}}\limits_{w=\alpha}
K(z_1,w) \mathcal{R}\omega_{g,n}(w, z_K),
\end{align*}
where $K=\{2,\dots,n\}$,
$\mathrm{Ram}$ is the set of zeros of ${\rm d}\sfx$, and
$K(z,w)$ is the \emph{recursion kernel} defined~by
\begin{align}
K(z,w)=\frac{\int_{w}^{\overline{w}} B(\cdot, z)}
{2 \left(\sfy(w) {\rm d}\sfx(w)-\sfy(\overline{w}) {\rm d}\sfx(\overline{w})\right)},\label{rec_kernel}
\end{align}
and
\begin{align}
\mathcal{R}\omega_{g,n}(w, z_K)=
\omega_{g-1,n+1}(w, \overline{w}, z_K)
+ \mathop{\sum_{h+h'=g}}
\limits_{J \cup J'=K}^{\mathrm{no (0,1)}}
\omega_{h,1+|J|}(w, z_J) \omega_{h',1+|J'|}(\overline{w}, z_{J'}).
\label{rec_omega}
\end{align}
Here $\overline{w}$ is the conjugate point of $w$ near $\alpha \in \mathrm{Ram}$
such that $\overline{w}\ne w$ and $x(\overline{w})=x(w)$.
The sum in equation~\eqref{rec_omega} does not include $(g,n)=(0,1)$ part and
contains the bidifferential
\begin{align*}
\omega_{0,2}(z_1, z_2) = B(z_1, z_2),
\end{align*}
and for $J=\{i_1, \dots, i_{|J|}\} \subseteq K$,
$z_J=\{z_{i_1}, \dots, z_{i_{|J|}}\}$ and
$z_{J'}=\{z_{i_{|J|+1}}, \dots, z_{i_{n-1}}\}$.
\end{Def}

In this paper, we focus on a class of spectral curves $\cC=\big({\IP}^1;\sfx,\sfy,B\big)$ with coordinate functions
\begin{align}
\sfx=\sfx(z)=\frac12 z^2,
\qquad
\sfy=\sfy(z),
\qquad
z \in {\IP}^1,
\label{sp_curve}
\end{align}
which
has a simple ramification point only at $z=0$ (the solution to ${\rm d}\sfx(z)=0$)
and so ${\mathrm{Ram}=\{0\}}$.
$\sfy(z) \in {\IC}$ is a meromorphic function of $z$,
and $\cC$ admits the global Galois covering by the conjugation $\overline{z}=-z$.
The bidifferential $B(z_1,z_2)$ has a double pole at the diagonal locus $z_1=z_2$ such that
\begin{align}
B(z_1, z_2)=\frac{{\rm d}z_1 \otimes {\rm d}z_2}{(z_1-z_2)^2},
\qquad
z_1, z_2 \in {\IP}^1.\label{bergman}
\end{align}
For this class of spectral curves,
the recursion kernel for the CEO topological recursion is
\begin{align*}
K(z,w)=\frac{\int_{w}^{-w} B(\cdot, z)}{2w (\sfy(w)-\sfy(-w)) {\rm d}w}=\frac{(-1) {\rm d}z}{\big(z^2-w^2\big) (\sfy(w)-\sfy(-w)) {\rm d}w}.
\end{align*}
And we introduce correlation functions.
\begin{Def}\label{def:top_rec_correl}
The correlation functions $W_{g,n}(z_1, \dots, z_n)$ for $2g-2+n \ge 0$ are
\begin{align*}
W_{g,n}(z_1, \dots, z_n) \otimes_{i=1}^n {\rm d}z_i =
\omega_{g,n}(z_1, \dots, z_n).
\end{align*}
\end{Def}

From the general argument in \cite{Andersen:2017vyk,Andersen_GR},
one finds the following claim.

\begin{thm}[Laplace transform of the volume polynomial \cite{Andersen:2017vyk,Andersen_GR}]\label{thm:laplace_abo_ceo}
For the physical $2$D gravity models in Table {\rm \ref{tab:kernel}},
the correlation functions $W_{g,n}$ for $2g-2+n > 0$ are related to the volume polynomials $\VO_{g,n}$ in Definition {\rm\ref{def:geom_rec}} by Laplace transform
\begin{align}
W_{g,n}(z_1, \dots, z_n) =
\mathcal{L}\{\VO_{g,n}\}(z_1,\dots,z_n),
\label{laplace}
\end{align}
where the operator $\mathcal{L}$ acts on a function $f(z_1,\dots,z_n)$ as
\begin{align}
\mathcal{L}\{f\}(z_1,\dots,z_n):=
\int_{\IR_+^n} f(L_1,\dots,L_n)
\e^{-\sum_{i=1}^n z_iL_i} \prod_{i=1}^n L_i {\rm d}L_i.
\label{eq:Laplace}
\end{align}
In particular, for the expansion \eqref{gr_exp}, we find
\begin{align}
W_{g,n}(z_1, \dots, z_n) =
\sum_{a_1,\dots,a_n \ge 0} F^{(g)}_{a_1,\dots,a_n}
\prod_{i=1}^n \frac{1}{z_i^{2a_i+2}}.
\label{laplace_exp}
\end{align}
\end{thm}

\begin{rem}
In Appendix \ref{sec:Derivation_Mirz},
based on the Laplace dual relation \eqref{laplace}
we derive the kernel functions \eqref{eq:H_minimal} and \eqref{eq:H_super_minimal}, which are not known before in the literature,
of the $(2,p)$ minimal string and the $(2,2p-2)$ minimal superstring from
the spectral curves in equations~\eqref{sp_curve_mg}
and~\eqref{sp_curve_mg_super}.
Note that the ``local'' initial data $(\sfA^{a_1}_{a_2, a_3}, \sfB^{a_1}_{a_2, a_3}, \sfC^{a_1}_{a_2, a_3}, \sfD^{a_1})$ in equations~\eqref{geom_rec_coeff} and~\eqref{geom_rec_coeff_ini} are found in \cite{Andersen:2017vyk,Ooms:2019}.
\end{rem}

In Sections \ref{sec:examples_tr_boson} and \ref{sec:examples_tr_super},
we will discuss the correlation functions $W_{g,n}$ defined by the CEO topological recursion for the physical 2D gravity models
and show some computational results of $W_{g,n}$ explicitly for our use in other sections and our checks of consistency of computations.

\subsection{Bosonic models}\label{sec:examples_tr_boson}

Here we consider the bosonic models in Table \ref{tab:spectral_curve}.

\subsubsection{Airy and KdV}\label{subsec:tr_kw}

For the Airy spectral curve $\cC^{\mathrm{A}}=\big({\IP}^1;\sfx,\sfy^{\mathrm{A}},B\big)$ with
coordinate functions
\begin{align}
\sfx(z)=\frac12 z^2,
\qquad
\sfy^{\mathrm{A}}(z)=z,
\label{sp_curve_kw}
\end{align}
and the bidifferential $B$ in equation~\eqref{bergman},
the CEO topological recursion defines the meromorphic multidifferentials
\begin{align}
\omega_{g,n}^{\mathrm{A}}(z_1, \dots, z_n)=
\mathop{\mathrm{Res}}\limits_{w=0}
K^{\mathrm{A}}(z_1,w)
\mathcal{R}\omega_{g,n}^{\mathrm{A}}(w, z_K)
=
\sum_{a_1, \dots, a_n \ge 0} F^{\mathrm{A}(g)}_{a_1,\dots,a_n}
\otimes_{i=1}^n \frac{{\rm d}z_i}{z_i^{2a_i+2}},
\label{airy_mdiff}
\end{align}
which give the Airy volume coefficients $F^{\mathrm{A}(g)}_{a_1,\dots,a_n}$ in equation~\eqref{feg_kw}, where
\begin{align}
K^{\mathrm{A}}(z_1,w)=
\frac{(-1) {\rm d}z_1}{2\left(z_1^2-w^2\right)w{\rm d}w},
\label{rec_k_airy}
\end{align}
is the recursion kernel for $\cC^{\mathrm{A}}$.
From equation~\eqref{airy_mdiff}, some of the correlation functions are
\begin{align*}
&
W_{0,3}^{\mathrm{A}}(z_1, z_2, z_3)=
\prod_{i=1}^3 \frac{1}{z_i^2},
\qquad
W_{1,1}^{\mathrm{A}}(z_1)=
\frac{1}{8z_1^4},
\\
&
W_{0,4}^{\mathrm{A}}(z_1, \dots, z_4)=
\left(\sum_{i=1}^4 \frac{3}{z_i^2}\right)
\prod_{i=1}^4 \frac{1}{z_i^2},
\qquad
W_{1,2}^{\mathrm{A}}(z_1, z_2)=
\left(
\sum_{i=1}^2 \frac{5}{8 z_i^4}
+ \frac{3}{8 z_1^2 z_2^2}
\right)
\prod_{i=1}^2 \frac{1}{z_i^2},
\\
&
W_{0,5}^{\mathrm{A}}(z_1, \dots, z_5)=
\Bigg(
\sum_{i=1}^5 \frac{15}{z_i^4}
+ \sum_{1\le i<j \le 5} \frac{18}{z_i^2 z_j^2}
\Bigg)\prod_{i=1}^5 \frac{1}{z_i^2},
\\
&
W_{1,3}^{\mathrm{A}}(z_1, z_2, z_3)=
\Bigg(
\sum_{i=1}^3 \frac{35}{8 z_i^6}
+ \sum_{1\le i,j \le 3} \frac{15}{4 z_i^2 z_j^4}
+ \frac{9}{4 z_1^2 z_2^2 z_3^2}
\Bigg)
\prod_{i=1}^3 \frac{1}{z_i^2}, \qquad
W_{2,1}^{\mathrm{A}}(z_1)=
\frac{105}{128z_1^{10}}.
\end{align*}

We now introduce the KdV spectral curve which deforms
the Airy spectral curve.

\begin{Def}[KdV spectral curve]\label{def:kdv}
The KdV spectral curve $\cC^{\mathrm{KdV}}=\big({\IP}^1;\sfx,\sfy^{\mathrm{KdV}},B\big)$ is defined~by
\begin{align}
\sfx(z)=\frac12 z^2,
\qquad
\sfy^{\mathrm{KdV}}(z)=z + \sum_{a \ge 2}\sfu_{a} z^{a},
\label{sp_curve_kdv}
\end{align}
and the bidifferential $B$ in equation~\eqref{bergman},
where $\sfu_{a}$ are time variables.
The KdV spectral curve~$\cC^{\mathrm{KdV}}$ yields the Airy spectral curve $\cC^{\mathrm{A}}$ when $\sfu_{a}=0$.
\end{Def}

The correlation functions
\begin{align}
W_{g,n}^{\mathrm{KdV}}(z_1, \dots, z_n)=
\sum_{a_1, \dots, a_n \ge 0} F^{\mathrm{KdV}(g)}_{a_1,\dots,a_n}
\prod_{i=1}^n \frac{1}{z_i^{2a_i+2}},
\label{kdv_mdiff}
\end{align}
obtained from the CEO topological recursion for $\cC^{\mathrm{KdV}}$
obey the following proposition \cite{Dunin-Barkowski:2012kbi}.

\begin{prop}[\cite{Dunin-Barkowski:2012kbi}]\label{prop:kdv}
The coefficients \smash{$F^{\mathrm{KdV}(g)}_{a_1,\dots,a_n}$} in equation~\eqref{kdv_mdiff} are written in terms of
the~volume coefficients \smash{$F^{\mathrm{A}(g)}_{a_1,\dots,a_n}$} in equation~\eqref{airy_mdiff}
$($or equation~\eqref{feg_kw}$)$ as
\begin{align}
F^{\mathrm{KdV}(g)}_{a_1,\dots,a_n}=
\sum_{m \ge 0}\frac{(-1)^m}{m!}
\sum_{b_1,\dots,b_m \ge 2}
\Bigg(\prod_{j=1}^m \frac{\sfu_{2b_j-1}}{2b_j+1}\Bigg)
F^{\mathrm{A}(g)}_{a_1,\dots,a_n, b_1, \dots, b_m},
\label{kdv_airy}
\end{align}
where the sum over $m$ and $b_j$ satisfies
\begin{align}
\sum_{i=1}^n a_i = 3g-3+n + m - \sum_{j=1}^m b_j,
\label{hom_kdv}
\end{align}
by the homogeneity condition \eqref{hom_kw} for the volume coefficients.
\end{prop}
\begin{proof}
The CEO topological recursion for the KdV spectral curve $\cC^{\mathrm{KdV}}$ gives
\begin{align}
\omega_{g,n}^{\mathrm{KdV}}(z_1, \dots, z_n)={}&
\mathop{\mathrm{Res}}\limits_{w=0}
\frac{(-1) {\rm d}z_1}{\big(z_1^2-w^2\big)
\big(\sfy^{\mathrm{KdV}}(w)-\sfy^{\mathrm{KdV}}(-w)\big){\rm d}w}
\mathcal{R}\omega_{g,n}^{\mathrm{KdV}}(w, z_K)\nonumber
\\
={}&
\mathop{\mathrm{Res}}\limits_{w=0}
K^{\mathrm{A}}(z_1,w)
Y^{\mathrm{KdV}}(w)
\mathcal{R}\omega_{g,n}^{\mathrm{KdV}}(w, z_K),
\label{kdv_ceo}
\end{align}
where $K^{\mathrm{A}}(z_1,w)$ is the recursion kernel \eqref{rec_k_airy} for the Airy spectral curve $\cC^{\mathrm{A}}$, and
\begin{align*}
Y^{\mathrm{KdV}}(w)
=\frac{2w}{\sfy^{\mathrm{KdV}}(w)-\sfy^{\mathrm{KdV}}(-w)}
=\sum_{m \ge 0} (-1)^m
\Bigg(\sum_{b \ge 2}\sfu_{2b-1} w^{2b-2}\Bigg)^m
=1 + \mathcal{O}\big(w^2\big)
\end{align*}
is a regular even function of $w$ at $w=0$.

A key formula to prove the proposition is \cite{Dunin-Barkowski:2012kbi},
\begin{align}
\mathop{\mathrm{Res}}\limits_{w_1=0}
\frac{(-1) Y(w_1) R(w_1)}{2\big(w_0^2-w_1^2\big)w_1}
& = 
\mathop{\mathrm{Res}}\limits_{w_1=0}
\big(\mathop{\mathrm{Res}}\limits_{w_2=w_1}
+\mathop{\mathrm{Res}}\limits_{w_2=- w_1}\big)
\frac{(-1) {\rm d}w_2}{2\big(w_0^2-w_2^2\big)w_2}
\frac{w_2^2 Y(w_2) R(w_1)}{\big(w_2^2-w_1^2\big)w_1}\nonumber
\\
& = 
\mathop{\mathrm{Res}}\limits_{w_2=0}
\frac{w_2^2 Y(w_2){\rm d}w_2}{\big(w_0^2-w_2^2\big)w_2}
\mathop{\mathrm{Res}}\limits_{w_1=0}
\frac{(-1) R(w_1)}{2\big(w_2^2-w_1^2\big)w_1},
\label{key_f}
\end{align}
where $Y(w)$ is a regular function of $w$ at $w=0$ and
$R(w)$ is a meromorphic differential of $w$.
Using this formula recursively, we rewrite equation~\eqref{kdv_ceo} as
\begin{align}
\omega_{g,n}^{\mathrm{KdV}}(z_1, \dots, z_n)={}&
\sum_{m \ge 0}(-1)^m
\mathop{\mathrm{Res}}\limits_{w_m=0}
\frac{\sfu(w_m) {\rm d}z_1}{\big(z_1^2-w_m^2\big)w_m}
\mathop{\mathrm{Res}}\limits_{w_{m-1}=0}
\frac{\sfu(w_{m-1}) {\rm d}w_m}{\big(w_m^2-w_{m-1}^2\big)w_{m-1}}\cdots\nonumber
\\
&
\times
\mathop{\mathrm{Res}}\limits_{w_1=0}
\frac{\sfu(w_1) {\rm d}w_2}{\big(w_2^2-w_1^2\big)w_1}
\mathop{\mathrm{Res}}\limits_{w=0}
K^{\mathrm{A}}(w_1,w)
\mathcal{R}\omega_{g,n}^{\mathrm{KdV}}(w, z_K),\label{kdv_ceo_ind}
\end{align}
where
\begin{align*}
\sfu(w)=
\sum_{b \ge 2}\sfu_{2b-1} w^{2b},
\end{align*}
and the residue operators
\begin{align*}
\mathop{\mathrm{Res}}\limits_{w_j=0}
\frac{\sfu(w_j) {\rm d}w}{\big(w^2-w_j^2\big)w_j},
\end{align*}
acting on meromorphic differentials of $w_j$ are referred to as the \emph{Airy dilaton leaves}.
Therefore, the KdV meromorphic multidifferentials $\omega_{g,n}^{\mathrm{KdV}}$
are regarded as
the Airy meromorphic multidifferentials $\omega_{g,n}^{\mathrm{A}}$
decorated by the Airy dilaton leaves.

To find such decorations,
we compare a meromorphic even differential
\[R(w)=\sum_{a} R_a {\rm d}w/w^{2a}\] of $w$
decorated by an Airy dilaton leaf,
\begin{align}
\mathop{\mathrm{Res}}\limits_{w_j=0}
\frac{\sfu(w_j) {\rm d}w R(w_j)}{\big(w^2-w_j^2\big)w_j}
=
\sum_{d \ge 0}\frac{{\rm d}w}{w^{2d+2}}
\sum_{b \ge 2}\sfu_{2b-1} R_{d+b},
\label{kdv_decoration_1}
\end{align}
with the $(g,n)=(0,2)$ part in the CEO topological recursion for $\cC^{\mathrm{A}}$,
\begin{align}
&
\mathop{\mathrm{Res}}\limits_{w_j=0}
K^{\mathrm{A}}(w,w_j)
\left(\frac{B(w_j,v)}{dv} R(-w_j)+\frac{B(-w_j,v)}{dv} R(w_j)\right)\nonumber
\\
&\qquad=
\mathop{\mathrm{Res}}\limits_{w_j=0}
\frac{R(w_j) {\rm d}w}{2\big(w^2-w_j^2\big)w_j}
\left(\frac{1}{(v-w_j)^2}+\frac{1}{(v+w_j)^2}\right)
=
\sum_{d \ge 0}\frac{{\rm d}w}{w^{2d+2}}
\sum_{b \ge 0} \frac{2b+1}{v^{2b+2}} R_{d+b}.
\label{kdv_decoration_2}
\end{align}
By this comparison, it is found that the decoration \eqref{kdv_decoration_1} of the Airy dilaton leaf is translated into a decorated $(g,n)=(0,2)$ part in the CEO topological recursion for $\cC^{\mathrm{A}}$.
In this translation, the Airy meromorphic differential acquires an extra marked point $v=1$,
and we find the sum in equation~\eqref{kdv_decoration_1} by replacing the sum in equation~\eqref{kdv_decoration_2} by a weighted sum such that
\begin{align*}
\sum_{b \ge 0} \frac{2b+1}{v^{2b+2}} R_{d+b}
\;\to\;
\sum_{b \ge 0} \frac{2b+1}{v^{2b+2}} \mathsf{wt}_{b}(\sfu) R_{d+b}
=
\sum_{b \ge 2}\sfu_{2b-1} R_{d+b},
\end{align*}
with weight factors
\begin{align}
\mathsf{wt}_{b}(\sfu)=
\begin{cases}
0
& \textrm{for}\ b = 0,1,
\\
\dfrac{\sfu_{2b-1}}{2b+1}
& \textrm{for}\ b \ge 2.
\label{dilaton_leaf_eff}
\end{cases}
\end{align}
In the following, we refer to this translation for the Airy dilaton leaf which decorates the ${(g,n)=(0,2)}$ part in the multidifferential
as the \emph{Airy translation}. In the following discussions,
we adopt the Airy translation to prove the formula \eqref{kdv_airy}
by the mathematical induction on $2g-2+n \ge 1$.

For $(g,n)=(0,3)$, equation~\eqref{kdv_ceo_ind} gives
\begin{align*}
\omega_{0,3}^{\mathrm{KdV}}(z_1, z_2, z_3)={}&
\sum_{m \ge 0}(-1)^m
\mathop{\mathrm{Res}}\limits_{w_m=0}
\frac{\sfu(w_m) {\rm d}z_1}{\big(z_1^2-w_m^2\big)w_m} \cdots
\mathop{\mathrm{Res}}\limits_{w_1=0}
\frac{\sfu(w_1) {\rm d}w_2}{\big(w_2^2-w_1^2\big)w_1}
\omega_{0,3}^{\mathrm{A}}(w_1, z_2, z_3)
\\
={}&
\omega_{0,3}^{\mathrm{A}}(z_1, z_2, z_3),
\end{align*}
where the Airy dilaton leaves do not contribute in this case.
For $(g,n)=(1,1)$, equation~\eqref{kdv_ceo_ind} gives
\begin{align*}
\omega_{1,1}^{\mathrm{KdV}}(z_1)& = 
\sum_{m \ge 0}(-1)^m
\mathop{\mathrm{Res}}\limits_{w_m=0}
\frac{\sfu(w_m) {\rm d}z_1}{\big(z_1^2-w_m^2\big)w_m} \cdots
\mathop{\mathrm{Res}}\limits_{w_1=0}
\frac{\sfu(w_1) {\rm d}w_2}{\big(w_2^2-w_1^2\big)w_1}
\omega_{1,1}^{\mathrm{A}}(w_1)
\\
& = 
\omega_{1,1}^{\mathrm{A}}(z_1)+
\sum_{a_1 \ge 0, b_1 \ge 2}
\frac{\sfu_{2b_1-1}}{2b_1+1}
F^{\mathrm{A}(1)}_{a_1,b_1} \frac{{\rm d}z_1}{z_1^{2a_1+2}},
\end{align*}
where only one Airy dilaton leaf contributes and the Airy translation is adopted.
Thus, the equations \eqref{kdv_airy} for $(g,n)=(0,3), (1,1)$ are obtained.

Next, under the assumption that the formula \eqref{kdv_airy} is correct for any $(g,n)$ with
$2g-2+n \le k$,
we consider equation~\eqref{kdv_ceo_ind} for $(g,n)$ with $2g-2+n=k+1$.
The coefficients of \[\otimes_{i=2}^n {\rm d}z_i/z_i^{2a_i+2}\]
in the factor $\mathcal{R}\omega_{g,n}^{\mathrm{KdV}}(w, z_K)$ in
equation~\eqref{kdv_ceo_ind} are rewritten
under this assumption as
\begin{align}
&
\sum_{a, b \ge 0}
\frac{(-1) {\rm d}w \otimes {\rm d}w}{w^{2a+2b+4}}
\Bigg(F^{\mathrm{KdV}(g-1)}_{a,b,a_2,\dots,a_n}
+
\mathop{\sum_{h+h'=g}'}
\limits_{J \subseteq K}
F^{\mathrm{KdV} (h)}_{a,a_{i_1},\dots,a_{i_{|J|}}}
F^{\mathrm{KdV} (h')}_{b,a_{i_{|J|+1}},\dots,a_{i_{n-1}}}\Bigg)\nonumber
\\
&\qquad=
\sum_{a, b \ge 0}
\frac{(-1) {\rm d}w \otimes {\rm d}w}{w^{2a+2b+4}}
\sum_{m \ge 0}\frac{(-1)^m}{m!}\sum_{b_1,\dots,b_m \ge 2}
\Bigg(\prod_{j=1}^m\mathsf{wt}_{b_j}(\sfu)\Bigg)\Bigg(F^{\mathrm{A}(g-1)}_{a,b,a_2,\dots,a_n, b_1, \dots, b_m}
\nonumber
\\
&\qquad\quad
+
\mathop{\sum_{h+h'=g}'}
\limits_{J \subseteq K}
\sum_{\ell=0}^m
\binom{m}{\ell}
F^{\mathrm{A} (h)}_{a,a_{i_1},\dots,a_{i_{|J|}}, b_1, \dots, b_{\ell}}
F^{\mathrm{A} (h')}_{b,a_{i_{|J|+1}},\dots,a_{i_{n-1}}, b_{\ell+1}, \dots, b_m}\Bigg),
\label{kdv_airy_step2}
\end{align}
where the sum $\sum_{h+h'=g}'$ does not include
$(h, J)=(0,\varnothing)$, $(g,K)$.
The symmetric factor $m!$ for the insertion of $m$ Airy dilaton leaves arises, since extra marked points at $v=1$ are indistinguishable.

Plugging equation~\eqref{kdv_airy_step2} into the right-hand side of equation~\eqref{kdv_ceo_ind}, we see that the Airy dilaton leaves in equation~\eqref{kdv_ceo_ind} compensate
the $(h, h', J)=(0,g,\varnothing)$, $(g,0,K)$ parts in
equation~\eqref{kdv_airy_step2} by Airy translations.
As a result, we obtain
\begin{align*}
&
\omega_{g,n}^{\mathrm{KdV}}(z_1, \dots, z_n)
\\
&
\qquad=
\sum_{a_1, \dots, a_n \ge 0} \otimes_{i=1}^n \frac{{\rm d}z_i}{z_i^{2a_i+2}}
\sum_{m \ge 0}\frac{(-1)^m}{m!}\sum_{b_1,\dots,b_m \ge 2}
\Bigg(\prod_{j=1}^m\mathsf{wt}_{b_j}(\sfu)\Bigg)
\\
&\
\qquad\quad{}\times
\Bigg(F^{\mathrm{A}(g-1)}_{a,b,a_2,\dots,a_n, b_1, \dots, b_m}
+
\sum_{\substack{h+h'=g\\ J \subseteq K\\ 0 \le \ell \le m}}''
\binom{m}{\ell}
F^{\mathrm{A} (h)}_{a,a_{i_1},\dots,a_{i_{|J|}}, b_1, \dots, b_{\ell}}
F^{\mathrm{A} (h')}_{b,a_{i_{|J|+1}},\dots,a_{i_{n-1}}, b_{\ell+1}, \dots, b_m}\Bigg)
\\
&
\qquad=
\sum_{a_1, \dots, a_n \ge 0} \otimes_{i=1}^n \frac{{\rm d}z_i}{z_i^{2a_i+2}}
\sum_{m \ge 0}\frac{(-1)^m}{m!}\sum_{b_1,\dots,b_m \ge 2}
\Bigg(\prod_{j=1}^m \frac{\sfu_{2b_j-1}}{2b_j+1}\Bigg)
F^{\mathrm{A}(g)}_{a_1,\dots,a_n, b_1, \dots, b_m},
\end{align*}
where the sum $\sum_{h+h'=g}''$ does not include
$(h, J, \ell)=(0,\varnothing, 0)$, $(g,K, m)$.
Thus, the induction is completed and the claim is proved.
\end{proof}

\subsubsection{Weil--Petersson volumes}\label{subsec:tr_mir}

For the Weil--Petersson spectral curve $\cC^{\mathrm{WP}}=\big({\IP}^1;\sfx,\sfy^{\mathrm{WP}},B\big)$
defined by \cite{Eynard:2007fi},
\begin{align}
\sfx(z)=\frac12 z^2,
\qquad
\sfy^{\mathrm{WP}}(z)=\frac{1}{2\pi} \sin(2\pi z)
= \sum_{a \ge 1} \frac{\bigl(-2\pi^2\bigr)^{a-1}}{(2a-1)!! (a-1)!} z^{2a-1},
\label{sp_curve_mir}
\end{align}
and the bidifferential $B$ in equation~\eqref{bergman},
the CEO topological recursion computes
the Weil--Petersson volume coefficients \eqref{feg_mir},
which give the Weil--Petersson volumes $\VO_{g,n}^{\mathrm{WP}}$ in equation~\eqref{mir_volume}, by
\begin{align}
W_{g,n}^{\mathrm{WP}}(z_1, \dots, z_n)=
\sum_{a_1, \dots, a_n \ge 0} F^{\mathrm{WP}(g)}_{a_1,\dots,a_n}
\prod_{i=1}^n \frac{1}{z_i^{2a_i+2}}.
\label{wp_mdiff}
\end{align}
The coordinate function $\sfy^{\mathrm{WP}}(z)$ in equation~\eqref{sp_curve_mir} is
found from a specialization\footnote{This specialization is also found in the physics literatures \cite{Dijkgraaf:2018vnm,Okuyama:2019xbv}.}
of the coordinate function $\sfy^{\mathrm{KdV}}(z)$ in equation~\eqref{sp_curve_kdv} of the KdV spectral curve as
\begin{align}
\sfu_{2a}=0,\qquad
\sfu_{2a+1}=\frac{\bigl(-2\pi^2\bigr)^a}{(2a+1)!! a!},
\label{sp_kdv_wp}
\end{align}
and Proposition \ref{prop:kdv} implies a formula
\begin{align*}
F^{\mathrm{WP}(g)}_{a_1,\dots,a_n}=
\sum_{m \ge 0}\frac{(-1)^m}{m!}
\sum_{b_1,\dots,b_m \ge 2}
\left(\prod_{j=1}^m \frac{\bigl(-2\pi^2\bigr)^{b_j-1}}{(2b_j+1)!! (b_j-1)!}\right)
F^{\mathrm{A}(g)}_{a_1,\dots,a_n, b_1, \dots, b_m}
\end{align*}
with the condition \eqref{hom_kdv}.

\begin{rem}\label{rem:fugacity_tr_wp}
The deformation parameter $s$ in Remark \ref{rem:fugacity_kappa}
is implemented to the spectral curve~$\cC^{\mathrm{WP}}$ by
changing $\pi \to \pi \sqrt{s}$ for the coordinate function $\sfy^{\mathrm{KdV}}(z)$ in equation~\eqref{sp_curve_mir}.
\end{rem}

\subsubsection[(2,p) minimal string]{$\boldsymbol{(2,p)}$ minimal string}\label{subsec:tr_mg}

The spectral curve for the $(2,p)$ minimal string is found from the disk partition function with the FZZT boundary condition \cite{Fateev:2000ik,Teschner:2000md}.
For an odd positive integer $p$,
the $(2,p)$ minimal string spectral curve
$\cC^{\mathrm{M}(p)}=\big({\IP}^1;\sfx,\sfy^{\mathrm{M}(p)},B\big)$ is defined by
\begin{align}
&\sfx(z)=\frac12 z^2,\nonumber
\\
&\sfy^{\mathrm{M}(p)}(z)= \frac{(-1)^{\frac{p-1}{2}}}{2\pi} T_p\biggl(\frac{2\pi}{p}z\biggr)
=\frac{1}{2\pi} \sin\left(\frac{p}{2} \arccos\left(1-\frac{8\pi^2z^2}{p^2}\right)\right)\nonumber
\\
&\phantom{\sfy^{\mathrm{M}(p)}(z)}{}= \sum_{a=1}^{\frac{p+1}{2}} \frac{\bigl(-2\pi^2\bigr)^{a-1}}{(2a-1)!! (a-1)!}
\prod_{i=1}^{a-1}\left(1-\frac{(2i-1)^2}{p^2}\right) z^{2a-1},
\label{sp_curve_mg}
\end{align}
and the bidifferential $B$ in equation~\eqref{bergman},
where $T_p(z)$ denotes the Chebyshev polynomial of the first kind defined by $
T_p(\cos \theta)=\cos(p \theta)$.
The minimal string spectral curve $\cC^{\mathrm{M}(p)}$
interpolates the Airy spectral curve $\cC^{\mathrm{A}}$ and
the Weil--Petersson spectral curve $\cC^{\mathrm{WP}}$ by
\begin{align*}
\sfy^{\mathrm{M}(1)} (z) = \sfy^{\mathrm{A}}(z)\
\textrm{in equation~\eqref{sp_curve_kw}},
\qquad
\sfy^{\mathrm{M}(\infty)}(z) = \sfy^{\mathrm{WP}}(z)\
\textrm{in equation~\eqref{sp_curve_mir}}.
\end{align*}
The coordinate function $\sfy^{\mathrm{M}(p)}(z)$ in equation~\eqref{sp_curve_mg} is
found from a specialization of the coordinate function
$\sfy^{\mathrm{KdV}(p)}(z)$ in equation~\eqref{sp_curve_kdv} of the KdV spectral curve as
\begin{align}
\begin{cases}
\sfu_{2a}=0\
& \textrm{for}\ a \ge 1,
\\
\displaystyle \sfu_{2a+1}=
\frac{(-2\pi^2)^a}{(2a+1)!! a!}
\prod_{i=1}^{a}\left(1-\frac{(2i-1)^2}{p^2}\right)\
& \textrm{for}\ 1 \le a \le \dfrac{p-1}{2},
\\
\sfu_{2a+1}=0\
& \textrm{for}\ a \ge \dfrac{p+1}{2}.
\label{sp_kdv_ms}
\end{cases}
\end{align}
Proposition \ref{prop:kdv} then implies that
the $(2,p)$ minimal string volume coefficients in the correlation functions
\begin{align}
W_{g,n}^{\mathrm{M}(p)}(z_1, \dots, z_n)=
\sum_{a_1, \dots, a_n \ge 0} F^{\mathrm{M}(p) (g)}_{a_1,\dots,a_n}
\prod_{i=1}^n \frac{1}{z_i^{2a_i+2}},
\label{mst_mdiff}
\end{align}
obey a formula
\begin{align}
F^{\mathrm{M}(p) (g)}_{a_1,\dots,a_n}
={}&
\sum_{m \ge 0}\frac{(-1)^m}{m!}
\sum_{b_1,\dots,b_m \ge 2}
\Bigg(\prod_{j=1}^m
\frac{\bigl(-2\pi^2\bigr)^{b_j-1}}{(2b_j+1)!! (b_j-1)!}
\prod_{i=1}^{b_j-1}\left(1-\frac{(2i-1)^2}{p^2}\right)
\Bigg)\nonumber\\&\times
F^{\mathrm{A}(g)}_{a_1,\dots,a_n, b_1, \dots, b_m},\label{mst_airy}
\end{align}
with the condition \eqref{hom_kdv}.

\subsection{Supersymmetric models}\label{sec:examples_tr_super}

Here we consider the supersymmetric models in Table \ref{tab:spectral_curve}.

\subsubsection{Bessel and BGW}\label{subsec:tr_be}

For the Bessel spectral curve $\cC^{\mathrm{B}}=\big({\IP}^1;\sfx,\sfy^{\mathrm{B}},B\big)$ with coordinate functions \cite{Do:2016odu},
\begin{align}
\sfx(z)=\frac12 z^2,
\qquad
\sfy^{\mathrm{B}}(z)=\frac{1}{z},
\label{sp_curve_be}
\end{align}
and the bidifferential $B$ in equation~\eqref{bergman},
the CEO topological recursion defines the meromorphic multidifferentials
\begin{align}
\omega_{g,n}^{\mathrm{B}}(z_1, \dots, z_n)=
\mathop{\mathrm{Res}}\limits_{w=0}
K^{\mathrm{B}}(z_1,w)
\mathcal{R}\omega_{g,n}^{\mathrm{B}}(w, z_K)=
\sum_{a_1, \dots, a_n \ge 0} F^{\mathrm{B}(g)}_{a_1,\dots,a_n}
\otimes_{i=1}^n \frac{{\rm d}z_i}{z_i^{2a_i+2}},
\label{bessel_mdiff}
\end{align}
and the Bessel volume coefficients $F^{\mathrm{B}(g)}_{a_1,\dots,a_n}$ in
equation~\eqref{feg_be} are obtained,
where
\begin{align}
K^{\mathrm{B}}(z_1,w)=
\frac{(-1) w {\rm d}z_1}{2\left(z_1^2-w^2\right){\rm d}w},
\label{rec_k_bessel}
\end{align}
is the recursion kernel for the CEO topological recursion on the Bessel spectral curve $\cC^{\mathrm{B}}$.
Here, note the relation $\sfy^{\mathrm{B}}(z)=\partial_{\sfx} \sfy^{\mathrm{A}}(z)$
with the coordinate function $\sfy^{\mathrm{A}}(z)$ in equation~\eqref{sp_curve_kw} of
the Airy spectral curve.
From equation~\eqref{bessel_mdiff}, some of the correlation functions are
\begin{align*}
&
W_{0,n}^{\mathrm{B}}(z_1, \dots, z_n)=0,
\qquad
W_{1,n}^{\mathrm{B}}(z_1, \dots, z_n)=
\frac{(n-1)!}{8}
\prod_{i=1}^n \frac{1}{z_i^2},
\\
&
W_{2,1}^{\mathrm{B}}(z_1)=
\frac{9}{128 z_1^4},
\qquad
W_{2,2}^{\mathrm{B}}(z_1,z_2)=
\left(
\sum_{i=1}^2 \frac{27}{128 z_i^2}
\right)
\prod_{i=1}^2 \frac{1}{z_i^2},
\\
&
W_{2,3}^{\mathrm{B}}(z_1,z_2,z_3)=
\left(
\sum_{i=1}^3 \frac{27}{32 z_i^2}
\right)
\prod_{i=1}^3 \frac{1}{z_i^2},
\qquad
W_{3,1}^{\mathrm{B}}(z_1)=
\frac{225}{1024 z_1^6},
\\
&
W_{2,4}^{\mathrm{B}}(z_1,\dots,z_4)=
\left(
\sum_{i=1}^4 \frac{135}{32 z_i^2}
\right)
\prod_{i=1}^4 \frac{1}{z_i^2},
\\
&
W_{3,2}^{\mathrm{B}}(z_1,z_2)=
\left(
\sum_{i=1}^2 \frac{1125}{1024 z_i^4}
+ \frac{567}{512 z_1^2 z_2^2}
\right)
\prod_{i=1}^2 \frac{1}{z_i^2}.
\end{align*}

Let us introduce the BGW spectral curve which deforms the Bessel spectral curve.

\begin{Def}[BGW spectral curve]\label{def:bgw}
The BGW spectral curve $\cC^{\mathrm{BGW}}=\big({\IP}^1;\sfx,\sfy^{\mathrm{BGW}},B\big)$ is
defined by
\begin{align}
\sfx(z)=\frac12 z^2,
\qquad
\sfy^{\mathrm{BGW}}(z)= \frac{1}{z} + \sum_{a \ge 0}\sfv_{a} z^{a},
\label{sp_curve_bgw}
\end{align}
and the bidifferential $B$ in equation~\eqref{bergman},
where $\sfv_{a}$ are time variables.
The BGW spectral curve~$\cC^{\mathrm{BGW}}$ for $\sfv_{a}=0$ yields the Bessel spectral curve~$\cC^{\mathrm{B}}$.
\end{Def}

Similar to Proposition \ref{prop:kdv}, the correlation functions
\begin{align}
W_{g,n}^{\mathrm{BGW}}(z_1, \dots, z_n)=
\sum_{a_1, \dots, a_n \ge 0} F^{\mathrm{BGW}(g)}_{a_1,\dots,a_n}
\prod_{i=1}^n \frac{1}{z_i^{2a_i+2}},
\label{bgw_mdiff}
\end{align}
obtained from the CEO topological recursion for $\cC^{\mathrm{BGW}}$
obey the following proposition.

\begin{prop}\label{prop:bgw}
The coefficients \smash{$F^{\mathrm{BGW}(g)}_{a_1,\dots,a_n}$} in equation~\eqref{bgw_mdiff} are written in terms of
the volume coefficients \smash{$F^{\mathrm{B}(g)}_{a_1,\dots,a_n}$} in equation~\eqref{bessel_mdiff} $($or equation~\eqref{feg_be}$)$ as
\begin{align}
F^{\mathrm{BGW}(g)}_{a_1,\dots,a_n}=
\sum_{m \ge 0}\frac{(-1)^m}{m!}
\sum_{b_1,\dots,b_m \ge 1}
\Bigg(\prod_{j=1}^m \frac{\sfv_{2b_j-1}}{2b_j+1}\Bigg)
F^{\mathrm{B}(g)}_{a_1,\dots,a_n, b_1, \dots, b_m},
\label{bgw_bessel}
\end{align}
where the sum over $m$ and $b_j$ satisfies
\begin{align}
\sum_{i=1}^n a_i = g-1 - \sum_{j=1}^m b_j,
\label{hom_bgw}
\end{align}
by the homogeneity condition \eqref{hom_be}.
\end{prop}
\begin{proof}
The statement can be shown in the parallel way as the proof of Proposition \ref{prop:kdv}.
The CEO topological recursion for the BGW spectral curve $\cC^{\mathrm{BGW}}$ gives
\begin{align}
\omega_{g,n}^{\mathrm{BGW}}(z_1, \dots, z_n)={}&
\mathop{\mathrm{Res}}\limits_{w=0}
\frac{(-1) {\rm d}z_1}{\big(z_1^2-w^2\big)
\left(\sfy^{\mathrm{BGW}}(w)-\sfy^{\mathrm{BGW}}(-w)\right){\rm d}w}
\mathcal{R}\omega_{g,n}^{\mathrm{BGW}}(w, z_K)\nonumber
\\
={}&
\mathop{\mathrm{Res}}\limits_{w=0}
K^{\mathrm{B}}(z_1,w)
Y^{\mathrm{BGW}}(w)
\mathcal{R}\omega_{g,n}^{\mathrm{BGW}}(w, z_K),
\label{bgw_ceo}
\end{align}
where $K^{\mathrm{B}}(z_1,w)$ is the recursion kernel \eqref{rec_k_bessel}
for the Bessel spectral curve $\cC^{\mathrm{B}}$, and
\begin{align*}
Y^{\mathrm{BGW}}(w)
=\frac{2}{w \big(\sfy^{\mathrm{BGW}}(w)-\sfy^{\mathrm{BGW}}(-w)\big)}
=\sum_{m \ge 0} (-1)^m
\Bigg(\sum_{b \ge 1}\sfv_{2b-1} w^{2b}\Bigg)^m
=1 + \mathcal{O}\big(w^2\big),
\end{align*}
is a regular even function of $w$ around $w=0$.
Using the formula \eqref{key_f} recursively, we rewrite equation~\eqref{bgw_ceo} as
\begin{align*}
\omega_{g,n}^{\mathrm{BGW}}(z_1, \dots, z_n)={}&
\sum_{m \ge 0}(-1)^m
\mathop{\mathrm{Res}}\limits_{w_m=0}
\frac{\sfv(w_m) w_m {\rm d}z_1}{\big(z_1^2-w_m^2\big)}
\mathop{\mathrm{Res}}\limits_{w_{m-1}=0}
\frac{\sfv(w_{m-1}) w_{m-1} {\rm d}w_m}{\big(w_m^2-w_{m-1}^2\big)}\cdots
\\
&
 \times
\mathop{\mathrm{Res}}\limits_{w_1=0}
\frac{\sfv(w_1) w_1 {\rm d}w_2}{\big(w_2^2-w_1^2\big)}
\mathop{\mathrm{Res}}\limits_{w=0}
K^{\mathrm{B}}(w_1,w) \mathcal{R}\omega_{g,n}^{\mathrm{BGW}}(w, z_K),
\end{align*}
where
\begin{align*}
\sfv(w)=
\sum_{b \ge 1}\sfv_{2b-1} w^{2b},
\end{align*}
and the residue operators
\begin{align*}
\mathop{\mathrm{Res}}\limits_{w_j=0}
\frac{\sfv(w_j) w_j {\rm d}w}{\big(w^2-w_j^2\big)},
\end{align*}
acting on meromorphic differentials of $w_j$ are
referred to as the \emph{Bessel dilaton leaves}.
In such a way, the BGW meromorphic multidifferentials $\omega_{g,n}^{\mathrm{BGW}}$
are obtained as
the Bessel meromorphic multidifferentials $\omega_{g,n}^{\mathrm{B}}$
decorated by the Bessel dilaton leaves.
Just like the Airy translation for an Airy dilaton leaf which decorates $(g,n)=(0,2)$ part in
the CEO topological recursion for~$\cC^{\mathrm{A}}$ with weights \eqref{dilaton_leaf_eff},
a Bessel dilaton leaf is translated into a decoration of the $(g,n)=(0,2)$ part in
the CEO topological recursion for $\cC^{\mathrm{B}}$ with weights
\begin{align*}
\begin{cases}
0
& \textrm{for}\ b = 0,
\\
\dfrac{\sfv_{2b-1}}{2b+1}
& \textrm{for}\ b \ge 1.
\end{cases}
\end{align*}
Repeating the same analysis as the proof of Proposition \ref{prop:kdv},
we find that equation~\eqref{bgw_bessel} holds.\looseness=1
\end{proof}

\subsubsection{Super Weil--Petersson volumes}\label{subsec:tr_sw}

For the super Weil--Petersson spectral curve $\cC^{\mathrm{SWP}}=\big({\IP}^1;\sfx,\sfy^{\mathrm{SWP}},B\big)$
defined by \cite{Norbury:2020vyi,Stanford:2019vob},
\begin{align}
\sfx(z)=\frac12 z^2,
\qquad
\sfy^{\mathrm{SWP}}(z)=\frac{1}{z} \cos(2\pi z)
= \sum_{a \ge 0} \frac{\bigl(-2\pi^2\bigr)^{a}}{(2a-1)!! a!} z^{2a-1},
\label{sp_curve_sw}
\end{align}
and the bidifferential $B$ in equation~\eqref{bergman},
the CEO topological recursion computes the super Weil--Petersson volume coefficients \eqref{feg_sw}, which give the super Weil--Petersson volumes, by
\begin{align}
W_{g,n}^{\mathrm{SWP}}(z_1, \dots, z_n)=
\sum_{a_1, \dots, a_n \ge 0} F^{\mathrm{SWP}(g)}_{a_1,\dots,a_n}
\prod_{i=1}^n \frac{1}{z_i^{2a_i+2}}.
\label{super_wp_mdiff}
\end{align}
Here, note the relation $\sfy^{\mathrm{SWP}}(z)=\partial_{\sfx} \sfy^{\mathrm{WP}}(z)$
with the coordinate function $\sfy^{\mathrm{WP}}(z)$ in equation~\eqref{sp_curve_mir} of the Weil--Petersson spectral curve.
The coordinate function $\sfy^{\mathrm{SWP}}(z)$ in equation~\eqref{sp_curve_sw} is found from a specialization of
the coordinate function $\sfy^{\mathrm{BGW}}(z)$ in equation~\eqref{sp_curve_bgw} of the BGW spectral curve as
\begin{align}
\sfv_{2a}=0,\qquad
\sfv_{2a+1}=
\frac{\bigl(-2\pi^2\bigr)^{a+1}}{(2a+1)!! (a+1)!},
\label{sp_bgw_swp}
\end{align}
and Proposition \ref{prop:bgw} implies a formula
\begin{align*}
F^{\mathrm{SWP}(g)}_{a_1,\dots,a_n}=
\sum_{m \ge 0}\frac{(-1)^m}{m!}
\sum_{b_1,\dots,b_m \ge 1}
\left(\prod_{j=1}^m
\frac{\bigl(-2\pi^2\bigr)^{b_j}}{(2b_j+1)!! b_j!}\right)
F^{\mathrm{B}(g)}_{a_1,\dots,a_n, b_1, \dots, b_m}
\end{align*}
with the condition \eqref{hom_bgw}.

\subsubsection[(2,2p-2) minimal superstring]{$\boldsymbol{(2,2p-2)}$ minimal superstring}\label{subsec:tr_mg_super}

A spectral curve
$\cC^{\mathrm{SM}(p)}=\big({\IP}^1;\sfx,\sfy^{\mathrm{SM}(p)},B\big)$
for the $(2,2p-2)$ minimal superstring
with an odd positive integer $p$, which is heuristically introduced by
\smash{$\sfy^{\mathrm{SM}(p)}(z)=\partial_{\sfx} \sfy^{\mathrm{M}(p)}(z)$}
from the coordinate function \smash{$\sfy^{\mathrm{M}(p)}(z)$}
in equation~\eqref{sp_curve_mg} of the $(2,p)$ minimal string spectral curve
(see Appendix~\ref{sec:type0A}), consists of
\begin{align}
&\sfx(z)=\frac12 z^2,\nonumber
\\
&\sfy^{\mathrm{SM}(p)}(z)= \frac{(-1)^{\frac{p-1}{2}}}{z} U_{p-1}\biggl(\frac{2\pi}{p}z\biggr)
=\frac{4\pi}{p \sqrt{1-\left(1-\frac{8\pi^2z^2}{p^2}\right)^2}}
\cos\left(\frac{p}{2} \arccos\left(1-\frac{8\pi^2z^2}{p^2}\right)\right)\nonumber
\\
&\phantom{\sfy^{\mathrm{SM}(p)}(z)}{}= \sum_{a=0}^{\frac{p-1}{2}} \frac{\bigl(-2\pi^2\bigr)^{a}}{(2a-1)!! a!}
\prod_{i=1}^{a}\left(1-\frac{(2i-1)^2}{p^2}\right) z^{2a-1},
\label{sp_curve_mg_super}
\end{align}
and the bidifferential $B$ in equation~\eqref{bergman},
where $U_p(z)$ is the Chebyshev polynomial of the second kind defined by
\begin{align*}
U_{p-1}(\cos \theta)=\frac{\sin (p \theta)}{\sin \theta}.
\end{align*}
Notice that the spectral curve $\cC^{\mathrm{SM}(p)}$
interpolates the Bessel spectral curve $\cC^{\mathrm{B}}$ and
the super Weil--Petersson spectral curve $\cC^{\mathrm{SWP}}$ by
\begin{align*}
\sfy^{\mathrm{SM}(1)}(z) = \sfy^{\mathrm{B}}(z)\
\textrm{in equation~\eqref{sp_curve_be}},
\qquad
\sfy^{\mathrm{SM}(\infty)}(z) = \sfy^{\mathrm{SWP}}(z)\
\textrm{in equation~\eqref{sp_curve_sw}}.
\end{align*}
The coordinate function $\sfy^{\mathrm{SM}(p)}(z)$ in equation~\eqref{sp_curve_mg_super} is found from
a specialization of the coordinate function $\sfy^{\mathrm{BGW}(p)}(z)$ in
equation~\eqref{sp_curve_bgw} of the BGW spectral curve as
\begin{align}
\begin{cases}
\sfv_{2a}=0\
& \textrm{for}\ a \ge 0,
\\
\displaystyle \sfv_{2a+1}=
\frac{(-2\pi^2)^{a+1}}{(2a+1)!! (a+1)!}
\prod_{i=1}^{a+1}\left(1-\frac{(2i-1)^2}{p^2}\right)\
& \textrm{for}\ 0 \le a \le \dfrac{p-3}{2},
\\
\sfv_{2a+1}=0\
& \textrm{for}\ a \ge \dfrac{p-1}{2}.
\label{sp_bgw_ms_super}
\end{cases}
\end{align}
Adopting this specialization to Proposition \ref{prop:bgw}, we find that
the $(2,2p-2)$ minimal superstring volume coefficients
\smash{$F^{\mathrm{SM}(p) (g)}_{a_1,\dots,a_n}$} in the correlation functions
\begin{align}
W_{g,n}^{\mathrm{SM}(p)}(z_1, \dots, z_n)=
\sum_{a_1, \dots, a_n \ge 0} F^{\mathrm{SM}(p) (g)}_{a_1,\dots,a_n}
\prod_{i=1}^n \frac{1}{z_i^{2a_i+2}},
\label{super_mst_mdiff}
\end{align}
obey a formula:
\begin{align}
F^{\mathrm{SM}(p) (g)}_{a_1,\dots,a_n}={}&
\sum_{m \ge 0}\frac{(-1)^m}{m!}\nonumber\\&\times
\sum_{b_1,\dots,b_m \ge 1}
\Bigg(
\prod_{j=1}^m
\frac{\bigl(-2\pi^2\bigr)^{b_j}}{(2b_j+1)!! b_j!}
\prod_{i=1}^{b_j}\left(1-\frac{(2i-1)^2}{p^2}\right)
\Bigg)
F^{\mathrm{B}(g)}_{a_1,\dots,a_n, b_1, \dots, b_m},
\label{msst_bessel}
\end{align}
with the condition \eqref{hom_bgw}.

\subsection{Twisting}\label{sec:twisting_tr}

\begin{Def}\label{def:sp_curve_tw}
We refer to a \emph{twisted spectral curve} $\cC[\sff]=(\Sigma;\sfx,\sfy,B[\sff])$
as a spectral curve with a twisted bidifferential $B[\sff]$ on $C^{\otimes 2}$ with the admissible test function $\sff\colon\mathbb{R}_+\to\mathbb{C}$.
\end{Def}

The twisted volume polynomial $V_{g,n}[\mathsf{f}](L_1,\dots,L_n)$ with
an admissible test function $\mathsf{f}$ is equivalent to
the multidifferential $\omega_{g,n}[\mathsf{f}](z_1,\dots,z_n)$
which satisfies the CEO topological recursion
for a twisted spectral curve \cite{Andersen_MV,Andersen_GR}.
These are related by the Laplace transform involving an action of
\emph{twist-elimination} explained below.

Here we consider a spectral curve $\cC=(\Sigma;\sfx,\sfy,B)$ such that
the zeros of ${\rm d}\sfx$ are simple, and a~local coordinate $p$ near
a branch point $\alpha$ obeying
\[\sfx=p^2/2+\sfx(\alpha).\]
For this set-up, we introduce a globally defined 1-form $\xi_{\alpha,a}(z)$ on the spectral curve $\cC$ with $a \in \mathbb{Z}_{\ge 0}$ and a branch point $\alpha$~by~\cite{Andersen:2017vyk},
\[
\xi_{\alpha,a}(z)=\underset{w=\alpha}{\mathrm{Res}} \frac{{\rm d}p(w)}{p(w)^{2a+2}}\left(\int_{\alpha}^wB(\cdot,z)\right).
\]
In the following discussion, we will focus on the spectral curve $\mathcal{C}$ only with a single branch point~$\alpha_0$, and
define $\xi_{a}(z)$ by\footnote{For a spectral curve with the coordinate function $\sfx=z^2/2$, $z \in \IP^1$ in equation~\eqref{sp_curve} and
the bidifferential~${B(z, w)={\rm d}z \otimes {\rm d}w/(z-w)^2}$, $z, w \in \IP^1$ in equation~\eqref{bergman}, $\xi_a(z)={\rm d}z/z^{2a+2}$.} $\xi_{a}(z):=\xi_{\alpha_0,a}(z)$.

The solutions $\omega_{g,n}(z_1,\dots,z_n)$ $(2g-2+n > 0)$ of the CEO topological recursion are represented by the 1-forms $\xi_{a_i}(z_i)$
($i=1,\dots,n$) as the basis of the multidifferentials:
\begin{align*}
\omega_{g,n}(z_1,\dots,z_n)=
\sum_{a_1,\dots,a_n\ge 0}
F^{(g)}_{a_1,\dots,a_n} \otimes_{i=1}^n \xi_{a_i}(z_i).
\end{align*}
For a twisted spectral curve $\cC[\sff]=(\Sigma;\sfx,\sfy,B[\sff])$,
the solutions $\omega_{g,n}[\mathsf{f}](z_1,\dots,z_n)$ of
the CEO topological recursion are also given on basis of
\begin{align}
\xi_a[\mathsf{f}](z)=
\underset{w=\alpha_0}{\mathrm{Res}}\;\frac{{\rm d}p(w)}{p(w)^{2a+2}}
\left(
\int_{\alpha_0}^w B[\sff](\cdot,z)
\right),
\label{eq:twisted_1-form}
\end{align}
by
\begin{align}
\omega_{g,n}[\mathsf{f}](z_1,\dots,z_n)=
\sum_{a_1,\dots,a_n\ge 0}
F^{(g)}[\mathsf{f}]_{a_1,\dots,a_n} \otimes_{i=1}^n\xi_{a_i}[\mathsf{f}](z_i).
\label{eq:twisted_multidifferential}
\end{align}

Now we introduce twisted correlation functions and a twist-elimination map below.

\begin{Def}\label{def:top_rec_correl_tw}
The twisted correlation functions $W_{g,n}[\sff](z_1, \dots, z_n)$ for $2g-2+n \ge 0$ are defined by
\begin{align}
W_{g,n}[\mathsf{f}](z_1,\dots,z_n) \otimes_{i=1}^n {\rm d}z_i =
\omega_{g,n}[\mathsf{f}](z_1,\dots,z_n),
\label{eq:TR_correlator}
\end{align}
and the twisted recursion kernel $K[\mathsf{f}]$ is defined by
\begin{align*}
K[\mathsf{f}](z,w)=\frac{\int_{w}^{\overline{w}} B[\sff](\cdot, z)}
{2 \left(\sfy(w) {\rm d}\sfx(w)-\sfy(\overline{w}) {\rm d}\sfx(\overline{w})\right)}.
\end{align*}
\end{Def}

\begin{Def}[twist-elimination map]\label{def:twist_elimination}
Let $\Omega(\mathcal{C})$ be the space of meromorphic multidifferentials on a twisted spectral curve $\cC[\sff]=(\Sigma;\sfx,\sfy,B[\sff])$, spanned by the symmetric tensors of the basis $\xi_{a_i}[\mathsf{f}](z_i)$ ($a_i\in\mathbb{Z}_{\ge 0}$, $i=1,\dots,n$),
the twisted bidifferential $B[\sff]$ and
the twisted recursion kernel $K[\mathsf{f}]$.%
\footnote{
In this definition, the twisted recursion kernel $K[\mathsf{f}](z,w)$ is regarded as
a meromorphic differential of the variable $z$.
}
The \emph{twist-elimination map} \smash{$\widehat{\mathcal{E}}[z_I]$} with an index set $I=\{i_1,\dots,i_{\ell}\}$ is a map acting on $\Omega(\mathcal{C})$
and prescribed by the following four properties:
\begin{itemize}\itemsep=0pt
\item[(1)]
for a 1-form $c_{a_i} \xi_{a_i}[\sff](z_i) \in \Omega(\cC)$ with
coefficient $c_{a_i}$,
\begin{align*}
\widehat{\cE}[z_I]\left(c_{a_i} \xi_{a_i}[\mathsf{f}](z_i)\right)=
\begin{cases}
c_{a_i} \xi_{a_i}(z_i)
&\textrm{if}\ i\in I\cap\{1,\dots,n\},
\\
c_{a_i} \xi_{a_i}[\mathsf{f}](z_i)
&\textrm{if}\ i\notin I\cap\{1,\dots,n\},
\end{cases}
\end{align*}
\item[(2)]
for the twisted bidifferential $B[\sff]$ and the twisted recursion kernel $K[\mathsf{f}]$,
\begin{align}
&
\widehat{\cE}[z]\left(B[\sff](z,w)\right)=B(z,w),
\qquad
\widehat{\cE}[z]\left(K[\mathsf{f}](z,w)\right)=K(z,w),
\label{eq:elim_kernel}
\end{align}
\item[(3)]
for a sum of multidifferentials $\omega^{(1)}[\mathsf{f}]$, $\omega^{(2)}[\mathsf{f}]\in \Omega(\mathcal{C})$,
\begin{align}\label{eq:property1_elim}
&\widehat{\mathcal{E}}[z_I]\big(\omega^{(1)}[\mathsf{f}]+\omega^{(2)}[\mathsf{f}]\big)
=\widehat{\mathcal{E}}[z_I]\big(\omega^{(1)}[\mathsf{f}]\big)+
\widehat{\mathcal{E}}[z_{I}]\big(\omega^{(2)}[\mathsf{f}]\big),
\end{align}
\item[(4)]
for a tensor product of multidifferentials $\omega^{(1)}[\mathsf{f}]$, $\omega^{(2)}[\mathsf{f}]\in \Omega(\mathcal{C})$,
\begin{align}
\label{eq:property2_elim}
&\widehat{\mathcal{E}}[z_I]\big(\omega^{(1)}[\mathsf{f}]\otimes\omega^{(2)}[\mathsf{f}]\big)
=\widehat{\mathcal{E}}[z_I]\big(\omega^{(1)}[\mathsf{f}]\big)\otimes
\widehat{\mathcal{E}}[z_{I}]\big(\omega^{(2)}[\mathsf{f}]\big).
\end{align}
\end{itemize}
The twist-elimination map $\mathcal{E}[z_I]$ acting
on the twisted correlation functions \eqref{eq:TR_correlator} is induced by
\begin{align*}
\mathcal{E}[z_{I}]
\left\{W_{g,n}[\mathsf{f}]\right\}(z_1,\dots,z_n) \otimes_{i=1}^n {\rm d}z_i=
\widehat{\mathcal{E}}[z_I]\left(\omega_{g,n}[\mathsf{f}](z_1,\dots,z_n)\right).
\end{align*}
\end{Def}

For the case $I=\{i_1,\dots,i_{\ell}\}\subseteq\{1,\dots,n\}$,
we obtain partially untwisted correlation functions for $2g-2+n>0$,
\begin{align}
&
\mathcal{E}[z_{i_1},\dots,z_{i_\ell}]
\left\{W_{g,n}[\mathsf{f}]\right\}(z_1,\dots,z_n) \otimes_{i=1}^n {\rm d}z_i\nonumber
\\
&\qquad=
\sum_{a_1,\dots,a_n\ge 0}
F^{(g)}[\mathsf{f}]_{a_1,\dots,a_n}
\otimes_{m=1}^{\ell} \xi_{a_{i_m}}(z_{i_m})
\otimes_{k \in\{1,\dots,n\}\setminus \{i_1,\dots,i_{\ell}\}}
\xi_{a_k}[\sff](z_k).
\label{eq:partial_untwisted}
\end{align}
More generally, if $\{i_1,\dots,i_\ell\}\nsubseteq \{1,\dots,n\}$,
the partially untwisted correlation functions are
given by
the maximal subset $\{j_1,\dots,j_{\ell'}\}\subseteq\{i_1,\dots,i_\ell\}$
which obeys
$\{j_1,\dots,j_{\ell'}\}\subseteq\{1,\dots,n\}$,
\begin{align*}
&
\mathcal{E}[z_{i_1},\dots,z_{i_\ell}]
\left\{W_{g,n}[\mathsf{f}]\right\}(z_1,\dots,z_n) \otimes_{i=1}^n {\rm d}z_i
\\
&\qquad=
\sum_{a_1,\dots,a_n\ge 0}
F^{(g)}[\mathsf{f}]_{a_1,\dots,a_n}
\otimes_{m=1}^{\ell'} \xi_{a_{j_m}}(z_{j_m})
\otimes_{k \in\{1,\dots,n\}\setminus \{j_1,\dots,j_{\ell'}\}}
\xi_{a_k}[\sff](z_k).
\end{align*}

\begin{figure}[t]\centering
 \includegraphics[width=70mm]{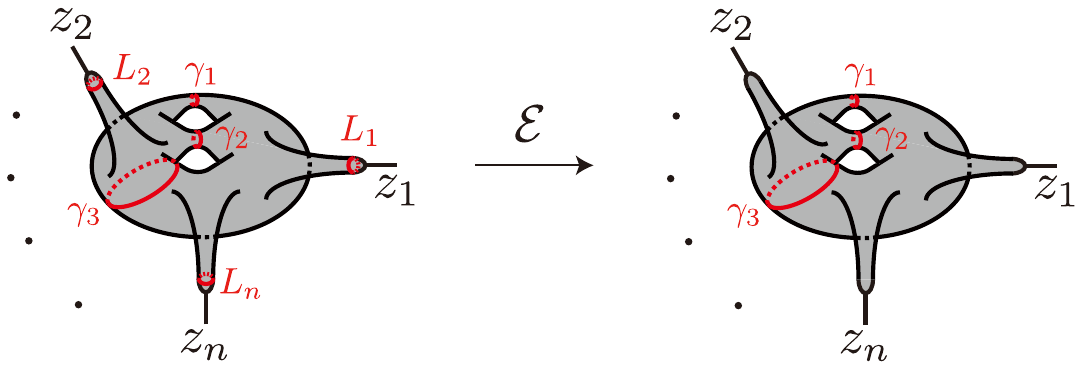}
\caption{In the geometrical interpretation,
the action of the twist-elimination map $\mathcal{E}[z_1,\dots,z_n]$ on the correlation function $W_{g,n}[\sff^{\mathrm{MV}}](z_1,\dots,z_n)$
is depicted as eliminations of the closed geodesics on $n$ boundaries of the bordered Riemann surface,
and recovers the combinatorial aspects of twisted volumes in Section~\ref{sect:MasurVeech}.}\label{fig:multi_boundaries_elimination}
\end{figure}

\begin{rem}\label{rel:Picture_elimination}
In Figure \ref{fig:multi_boundaries_elimination}, a geometrical interpretation of the action of the twist-elimination map $\mathcal{E}[z_1,\dots,z_n]$ on the correlation function $W_{g,n}[\sff^{\mathrm{MV}}](z_1,\dots,z_n)$
with the Masur--Veech type twist is depicted.
In this picture, the correlation function is described by a capped Riemann surface where caps are glued along all boundaries in a bordered Riemann surface.
In this interpretation, the closed geodesics on the bordered boundaries and bulk in the Riemann surface
represent the Masur--Veech type twists on the bases $\xi_a[\sff^{\mathrm{MV}}]$ and
coefficients \smash{$F^{(g)}\big[\mathsf{f}^{\mathrm{MV}}\big]_{a_1,\dots,a_n}$} in the multidifferential $\omega_{g,n}[\sff^{\mathrm{MV}}]$ of equation~\eqref{eq:partial_untwisted}, respectively.
In the computation of twisted volumes by the combinatorial method discussed in Section \ref{sect:MasurVeech},
we only enumerate the multicurves wrapping around the closed geodesics in the bulk%
\footnote{The ``closed geodesics in the bulk'' means the ``non-boundary closed geodesics''.} of bordered Riemann surfaces.
To get the twisted volume polynomial $V_{g,n}[\sff^{\mathrm{MV}}](L_1,\dots,L_n)$ from the correlation function $W_{g,n}[\sff^{\mathrm{MV}}](z_1,\dots,z_n)$,
we need to eliminate the effects of twists coming from the boundaries, while keep those from the closed geodesics in the bulk.
This geometrical interpretation of the twist is used in the physical interpretation discussed in Appendix \ref{sec:Wilson_loop}.
\end{rem}

In the following, we consider a class of twisted spectral curves with
the coordinate functions in equation~\eqref{sp_curve} and
twisted bidifferential
\cite{Andersen_MV},
\begin{align}
B\big[\sfm\big](z_1, z_2)=\frac{{\rm d}z_1 \otimes {\rm d}z_2}{(z_1-z_2)^2}
+ \frac12 \sum_{\mfm \in \IZ^*} \frac{{\rm d}z_1 \otimes {\rm d}z_2}{(z_1-z_2+\mfm)^2}
=\zeta_{\mathrm{H}}(2;z_1-z_2) {\rm d}z_1 \otimes {\rm d}z_2,
\label{bergman_twist}
\end{align}
with the Masur--Veech type twist function $\sfm$
in equation~\eqref{mv_twist}, where
\begin{align}
\zeta_{\mathrm{H}}(2d;z)={}&
\frac{1}{z^{2d}} + \frac12 \sum_{\mfm \in \IZ^*} \frac{1}{(z+\mfm)^{2d}}=
\frac{1}{z^{2d}}+
\sum_{k \ge 0}\binom{2k+2d-1}{2k} \zeta(2k+2d) z^{2k},
\label{hurwitz_def}
\end{align}
is the Hurwitz zeta function.
For this twisted bidifferential $B\big[\sfm\big]$,
the recursion kernel \eqref{rec_kernel} of the CEO topological recursion yields
\begin{align*}
K[\sff^{\mathrm{MV}}](z,w)={}&
\frac{(-1) {\rm d}z}{\left(\sfy(w)-\sfy(-w)\right){\rm d}w}
\left(\frac{1}{z^2-w^2} + \frac12 \sum_{\mfm \in \IZ^*} \frac{1}{(z+\mfm)^2-w^2}
\right)
\\
={}&
\frac{(-1) {\rm d}z}{\left(\sfy(w)-\sfy(-w)\right){\rm d}w}
\sum_{d\ge 0}\zeta_{\mathrm{H}}(2d+2;z) w^{2d},
\end{align*}
and the twisted 1-form \eqref{eq:twisted_1-form} is
\begin{align}
\xi_a\big[\sfm\big](z)=\zeta_{\mathrm{H}}(2a+2;z) {\rm d}z.
\label{twisted_1-form_sp}
\end{align}
By acting the twist-elimination map $\cE$ involved with the properties of equations~\eqref{eq:elim_kernel}, \eqref{eq:property1_elim} and \eqref{eq:property2_elim} on the CEO topological recursion for the twisted spectral curve, we obtain
\begin{align}
&
\cE[z_1,z_K]\left\{W_{g,n}\big[\mathsf{f}^\mathrm{MV}\big]\right\}(z_1,z_K)\nonumber
\\
&\qquad=\big[w^0\big]
\sum_{d\ge 0} \frac{w^{2d+1}}{z_1^{2d+2}\left(\sfy(w)-\sfy(-w)\right)}
\Biggl(
\sum_{m=2}^n \cW_{0,2}(w,z_m)
\cE[z_{K\setminus \{m\}}]\nonumber
\\
&
\qquad\quad {}\times\left\{W_{g,n-1}\big[\sfm\big]\right\}(w,z_{K\setminus \{m\}})+\cE[z_K]\left\{Q_{g,n}\big[\sfm\big]\right\}(w,-w,z_K)
\biggr),
\label{eq:top_rec_one_branch_untwisted}
\end{align}
where $K=\{2,\dots,n\}$, and
$\big[w^0\big]$ implies to pick up all zeroth order terms in the expansion around $w=0$. Here
\begin{align}
\begin{split}
\cW_{0,2}(w,z_m)
={}&W_{0,2}(w,z_m)+W_{0,2}(-w,z_m)
=2\sum_{d\ge 0}(2d+1)\frac{w^{2d}}{z_m^{2d+2}},
\label{w02_sym}
\end{split}
\end{align}
and
\begin{align*}
Q_{g,n}\big[\sfm\big](w,-w,z_K)={}&
W_{g-1,n+1}\big[\sfm\big](w,-w,z_K)
\\
&
+\sum_{\substack{h+h'=g\\ J \sqcup J'=K}}^{\mathrm{stable}}
W_{h,1+|J|}\big[\sfm\big](w,z_{J}) W_{h',1+|J'|}\big[\sfm\big](-w,z_{J'}).
\end{align*}
We refer to the recursion \eqref{eq:top_rec_one_branch_untwisted}
as the \emph{partially twist-eliminated CEO topological recursion}.

In the rest of this section, we provide a direct proof of the following claim.

\begin{thm}[Laplace transform of the twisted volume polynomial \cite{Andersen_MV,Andersen_GR}]\label{thm:Laplace_dual_GR}
For the $(2,p)$ minimal string and the $(2,2p-2)$ minimal superstring,
the twisted correlation function $W_{g,n}\big[\sfm\big](z_1,\allowbreak\dots,z_n)$ with the action of the twist-elimination map $\mathcal{E}$
agrees with the Laplace transform $\mathcal{L}$
of the twisted volume polynomial $V_{g,n}\big[\sfm\big](L_1,\dots,L_n)$:
\begin{align}
\mathcal{L}\left\{V_{g,n}\big[\sfm\big]\right\}(z_1,\dots,z_n)=
\cE[z_1,\dots,z_n]\left\{W_{g,n}\big[\sfm\big]\right\}(z_1,\dots,z_n),
\label{eq:Laplace_dual_GR}
\end{align}
where the operator $\mathcal{L}$ is defined by equation~\eqref{eq:Laplace}.%
\footnote{This relation is derived formally in the general set-up in \cite{Andersen_MV}.
In this article, we give a direct proof specialized for
the physical 2D gravity models.
}
\end{thm}

\subsubsection[Laplace dual relation for the (2,p) minimal string]{Laplace dual relation for the $\boldsymbol{(2,p)}$ minimal string}\label{sec:Laplace_bosonic}

We will show that the Laplace dual relation \eqref{eq:Laplace_dual_GR} holds for the $(2,p)$ minimal string manifestly.
Accordingly, by the specializations $p=\infty$ and $1$,
we also find that the Laplace dual relation holds for the Weil--Petersson volumes and the Kontsevich--Witten symplectic volumes, respectively.

\begin{prop}\label{prop:Laplace_dual_untwisted}
For the $(2,p)$ minimal string, the Laplace transform of
the Mirzakhani type ABO topological recursion \eqref{eq:Mirzakhani's}
with the Masur--Veech type twist,
\begin{align}
&
\int_{\IR_+^n}
\Biggl[
\sum_{m=2}^n \int_{\IR_+} \left(R(L_1,L_m,x) + L_1 \sfm(x)\right)
xV_{g,n-1}\big[\sfm\big](x,L_{K\setminus \{m\}}) {\rm d}x\nonumber
\\
&
\qquad
{}+\frac{1}{2}\int_{\IR_+^2}
\left(D(L_1,x,y)+R(L_1,x,y) \sfm(y) + R(L_1,y,x) \sfm(x)
+
L_1 \sfm(x) \sfm(y)\right)\nonumber
\\
&\qquad\phantom{+}{}
\times
 xyP_{g,n}\big[\sfm\big](x,y,L_K) {\rm d}x{\rm d}y
\Biggr]
\e^{-z_1L_1} {\rm d}L_1
\prod_{i=2}^n \e^{-z_iL_i} L_i{\rm d}L_i,\label{eq:Laplace_dual_untwisted}
\end{align}
where $P_{g,n}\big[\sfm\big](x,y,L_K)$ is given by equation~\eqref{geo_rec_p}
for the twisted volume polynomials,
agrees with the partially twist-eliminated CEO topological recursion \eqref{eq:top_rec_one_branch_untwisted} for $2g-2+n > 1$.
\end{prop}

To prove this proposition,
we prepare some key integration formulae involving the kernel function $H^{\mathrm{M}(p)}(x,y)$ in equation~\eqref{eq:H_minimal}
for the $(2,p)$ minimal string. (See \cite[equations~(2.2) and (2.3)]{MulSaf} for analogous formulae of Mirzakhani's recursion for the Weil--Petersson volumes.)
For $k\in\mathbb{Z}_{\ge 0}$, one finds
\begin{align}
h_{2k+1}^{\mathrm{M}(p)}(t):={}&
\int_{\IR_+} \frac{x^{2k+1}}{(2k+1)!} H^{\mathrm{M}(p)}(x,t) {\rm d}x\nonumber
\\
={}&
\frac{t^{2k+2}}{(2k+2)!}
-2\sum_{j=1}^{(p-1)/2}
(-1)^{j}\cos\left(\frac{\pi}{p}j\right)
\sum_{\ell=0}^{k}u_j^{2\ell-2k-2}\frac{t^{2\ell}}{(2\ell)!}\nonumber
\\
={}&
\sum_{\ell=0}^{k+1}s_{\ell} \frac{t^{2k+2-2\ell}}{(2k+2-2\ell)!},
\label{eq:formula1}
\end{align}
where $u_j=(p/2\pi)\sin(j\pi/p)$ in equation~\eqref{eq:uj_bosonic}, and
\begin{align}
\int_{\IR_+^2}\frac{x^{2a+1}y^{2b+1}}{(2a+1)!(2b+1)!} H^{\mathrm{M}(p)}(x+y,t) {\rm d}x{\rm d}y=h_{2a+2b+3}^{\mathrm{M}(p)}(t).
\label{eq:formula2}
\end{align}
The coefficients $s_{\ell}$'s in equation~\eqref{eq:formula1} agree with
those in the following expansion for $1/\sfy^{\mathrm{M}(p)}(z)$
of the $(2,p)$ minimal string spectral curve:
\begin{align}
\frac{1}{\sfy^{\mathrm{M}(p)}(z)}={}&
(-1)^{\frac{p-1}{2}}\frac{2\pi}{T_p\big(\frac{2\pi}{p}z\big)}
=\frac{1}{z}+\sum_{j=1}^{(p-1)/2}
(-1)^j\cos\left(\frac{\pi}{p}j\right)\left(
\frac{1}{z-u_j}+\frac{1}{z+u_j}
\right)\nonumber
\\
={}&\sum_{\ell \ge 0}s_\ell z^{2\ell-1}.
\label{eq:1/y}
\end{align}
In addition, we will use an integration formula involving the Masur--Veech type twist function,
\begin{align}
\int_{\IR_+}\mathsf{f}^{\mathrm{MV}}(x) x^{2k+1} {\rm d}x
=(2k+1)! \zeta(2k+2).
\label{eq:twist_zeta}
\end{align}

\begin{lem}\label{lem:Laplace_B-term}
In equation~\eqref{eq:Laplace_dual_untwisted},
the $x^{2k+1}$ term in $xV_{g,n-1}\big[\mathsf{f}^{\mathrm{MV}}\big](x,L_{K\setminus\{m\}})$ obeys
\begin{align}
&
\int_{\IR_+^3}
\big(R(L_1,L_m,x)+L_1\mathsf{f}^{\mathrm{MV}}(x)\big)x^{2k+1}
\mathrm{e}^{-z_1L_1-z_mL_m} {\rm d}x {\rm d}L_1 L_m{\rm d}L_m\nonumber
\\
&
\qquad=\frac{1}{2}\big[w^0\big]\sum_{d\ge 0} \frac{1}{z_1^{2d+2}}\frac{w^{2d+1}}{\sfy^{\mathrm{M}(p)}(w)}
\cW_{0,2}(w,z_m)
\zeta_{\mathrm{H}}(2k+2;w) (2k+1)!.
\label{eq:twist_B-term_bosonic}
\end{align}
\end{lem}
\begin{proof}
We perform the integrations on the left-hand side of equation~\eqref{eq:twist_B-term_bosonic} using~\eqref{eq:DR_H_bosonic}, \eqref{eq:formula1} and \eqref{eq:twist_zeta}:
\begin{align}
&
\int_{\IR_+^3}
\left[
\frac12 \int_{0}^{L_1}
\big(H^{\mathrm{M}(p)}(x,t+L_m)+H^{\mathrm{M}(p)}(x,t-L_m)\big){\rm d}t
+ L_1 \sfm(x)\right]
x^{2k+1}\nonumber
\\
&\qquad{}\times
\mathrm{e}^{-z_1L_1-z_mL_m} {\rm d}x {\rm d}L_1 L_m{\rm d}L_m
\nonumber\\
&
\qquad\phantom{\times}{}=(2k+1)!\left(\sum_{\ell=0}^{k+1}\sum_{d=0}^{k+1-\ell}s_{\ell}
\cdot(2k+3-2\ell-2d) \frac{1}{z_1^{2d+2}}\frac{1}{z_m^{2k+4-2\ell-2d}}\right.\nonumber\\
&\left.\qquad\phantom{\times=}{}+\frac{1}{z_1^2}\frac{1}{z_m^2} \zeta(2k+2)\right),\label{eq:twist_B-term_bosonic_LHS}
\end{align}
where we also used
a formula of the Laplace transform:
\begin{align}
\int_{\IR_+} L^{n} \mathrm{e}^{-zL} {\rm d}L=\frac{n!}{z^{n+1}}.
\label{eq:Laplace_basic}
\end{align}
On the other hand, by equation~\eqref{w02_sym} the right-hand side of equation~\eqref{eq:twist_B-term_bosonic} is
\begin{align}
\frac{1}{2}\big[w^0\big]\sum_{d\ge 0}\frac{1}{z_1^{2d+2}}\frac{w^{2d+1}}{\sfy^{\mathrm{M}(p)}(w)}
\left(2\sum_{d'\ge 0}(2d'+1)\frac{w^{2d'}}{z_m^{2d'+2}}\right)
\zeta_{\mathrm{H}}(2k+2;w) (2k+1)!.
\label{eq:twist_B-term_bosonic_RHS}
\end{align}
Adopting the formula \eqref{hurwitz_def} for $\zeta_{\mathrm{H}}(2k+2;w)$
and the expansion \eqref{eq:1/y} for $1/\sfy^{\mathrm{M}(p)}(w)$ to this expression,
we find the agreement between equations~\eqref{eq:twist_B-term_bosonic_LHS} and \eqref{eq:twist_B-term_bosonic_RHS}.
\end{proof}

\begin{lem}\label{lem:Laplace_C-term}
In equation~\eqref{eq:Laplace_dual_untwisted},
the $x^{2a+1}y^{2b+1}$ term in $xyP_{g,n}\big[\sfm\big](x,y,L_K)$ obeys
\begin{align}
&
\frac{1}{2}\int_{\IR_+^3}
\big(D(L_1,x,y)+R(L_1,x,y) \mathsf{f}^{\mathrm{MV}}(y)
+R(L_1,y,x) \mathsf{f}^{\mathrm{MV}}(x)
+L_1\mathsf{f}^{\mathrm{MV}}(x) \mathsf{f}^{\mathrm{MV}}(y)\big)\nonumber
\\
&\qquad{}\times
x^{2a+1}y^{2b+1} \e^{-z_1L_1} {\rm d}x{\rm d}y{\rm d}L_1\nonumber
\\
&\qquad\phantom{\times}{}
=\frac{1}{2}\big[w^0\big]
\sum_{d\ge 0}\frac{1}{z_1^{2d+2}}
\frac{w^{2d+1}}{\sfy^{\mathrm{M}(p)}(w)}
\left(\zeta_{\mathrm{H}}(2a+2;w) (2a+1)!\right)
\left(\zeta_{\mathrm{H}}(2b+2;w) (2b+1)!\right).\!\!\!
\label{eq:twist_C-term_bosonic}
\end{align}
\end{lem}
\begin{proof}
We perform the integrations on the left-hand side of equation~\eqref{eq:twist_C-term_bosonic} using~\eqref{eq:DR_H_bosonic},
\eqref{eq:formula1}, \eqref{eq:formula2} and \eqref{eq:twist_zeta}.
The term involving $D(L_1,x,y)$ yields
\begin{align}
&
\frac{1}{2}\int_{\IR_+^3}
\left(\int_0^{L_1} H^{\mathrm{M}(p)}(x+y,t) {\rm d}t\right)
x^{2a+1}y^{2b+1} \e^{-z_1L_1} {\rm d}x{\rm d}y{\rm d}L_1\nonumber
\\
&\qquad=\frac{1}{2}(2a+1)!(2b+1)!
\sum_{\ell=0}^{a+b+2} \frac{s_{\ell}}{z_1^{2a+2b-2\ell+4}}.\label{eq:Laplace_C-term_1}
\end{align}
The term involving $R(L_1,x,y)\mathsf{f}^{\mathrm{MV}}(y)$ yields:
\begin{align}
&
\frac{1}{4}\int_{\IR_+^3}
\left(\int_0^{L_1}\big(H^{\mathrm{M}(p)}(x,t+y)+H^{\mathrm{M}(p)}(x,t-y)\big){\rm d}t\right)
x^{2a+1}y^{2b+1} \mathsf{f}^{\mathrm{MV}}(y) \e^{-z_1L_1} {\rm d}x{\rm d}y{\rm d}L_1\nonumber
\\
&\quad
=\frac{1}{2}(2a+1)!
\sum_{\ell=0}^{a+1}\sum_{d=0}^{a+1-\ell}s_{\ell} \frac{(2a+2b+3-2\ell-2d)!}{(2a+2-2\ell-2d)!}
\zeta(2a+2b+4-2\ell-2d) \frac{1}{z_1^{2d+2}},\label{eq:Laplace_C-term_2}
\end{align}
and we find the term involving $R(L_1,y,x)\mathsf{f}^{\mathrm{MV}}(x)$ by replacing the role of parameters $a$ and $b$ in equation~\eqref{eq:Laplace_C-term_2}:
\begin{align*}
\frac{1}{2}(2b+1)!
\sum_{\ell=0}^{b+1}\sum_{d=0}^{b+1-\ell}s_{\ell} \frac{(2a+2b+3-2\ell-2d)!}{(2b+2-2\ell-2d)!}
\zeta(2a+2b+4-2\ell-2d) \frac{1}{z_1^{2d+2}}.
\end{align*}
Finally the term involving $\mathsf{f}^{\mathrm{MV}}(x)\mathsf{f}^{\mathrm{MV}}(y)$ yields
\begin{align}
\label{eq:Laplace_C-term_4}
\frac{1}{2z_1^2} (2a+1)! (2b+1)! \zeta(2a+2) \zeta(2b+2).
\end{align}
On the other hand,
adopting equations~\eqref{hurwitz_def} and \eqref{eq:1/y} on the right-hand side of equation~\eqref{eq:twist_C-term_bosonic},
we correctly recover the sum of four terms \eqref{eq:Laplace_C-term_1} -- \eqref{eq:Laplace_C-term_4}.
\end{proof}

Combining the claims in Lemmas \ref{lem:Laplace_B-term} and \ref{lem:Laplace_C-term} as well as
the expansions \eqref{eq:expansion_volume} and \eqref{eq:twisted_multidifferential} with \eqref{twisted_1-form_sp},
we find the claim of Proposition \ref{prop:Laplace_dual_untwisted}.

\subsubsection[Laplace dual relation for the (2,2p-2) minimal superstring]{Laplace dual relation for the $\boldsymbol{(2,2p-2)}$ minimal superstring}\label{sec:Laplace_super}

We will show that the Laplace dual relation \eqref{eq:Laplace_dual_GR} holds for the $(2,2p-2)$ minimal superstring.
Accordingly, by the specializations $p=\infty$ and $1$,
we also find that the Laplace dual relation holds for the super Weil--Petersson volumes and the supersymmetric analogue of the Kontsevich--Witten symplectic volumes, respectively.

\begin{prop}\label{prop:Laplace_dual_untwisted_super}
For the $(2,2p-2)$ minimal superstring, the Laplace transform of
the Mirzakhani type ABO topological recursion \eqref{eq:Mirzakhani's}
with the Masur--Veech type twist,
\begin{align}
&
\int_{\IR_+^n}
\Biggl[
\sum_{m=2}^n \int_{\IR_+} R(L_1,L_m,x)
xV_{g,n-1}\big[\sfm\big](x,L_{K\setminus \{m\}}) {\rm d}x\nonumber
\\
&\qquad{}
+\frac{1}{2}\int_{\IR_+^2}
\left(D(L_1,x,y)+R(L_1,x,y) \sfm(y) + R(L_1,y,x) \sfm(x)\right)\nonumber
\\
&\qquad\phantom{+}{}
\times xyP_{g,n}\big[\sfm\big](x,y,L_K) {\rm d}x{\rm d}y
\biggr]
\e^{-z_1L_1} {\rm d}L_1
\prod_{i=2}^n \e^{-z_iL_i} L_i{\rm d}L_i,\label{eq:Laplace_dual_untwisted_super}
\end{align}
agrees with the partially twist-eliminated CEO topological recursion \eqref{eq:top_rec_one_branch_untwisted} for $2g-2+n > 1$.
\end{prop}

To prove Proposition \ref{prop:Laplace_dual_untwisted_super}, we use
an integration formula
\begin{align}
h_{2k+1}^{\mathrm{SM}(p)}(t):={}&
\int_{\IR_+} \frac{x^{2k+1}}{(2k+1)!} H^{\mathrm{SM}(p)}(x,t) {\rm d}x\nonumber
\\
={}&
\frac{t^{2k+1} \delta_{p,1}}{(2k+1)!}
-\frac{1}{\pi}\sum_{j=1}^{(p-1)/2}
(-1)^{j}\cos^2\left(\frac{\pi}{p}\left(j-\frac{1}{2}\right)\right)
\sum_{\ell=0}^{k}(u'_j)^{2\ell-2k-1}\frac{t^{2\ell+1}}{(2\ell+1)!}\nonumber
\\
={}&
\sum_{\ell=0}^{k} s'_{\ell} \frac{t^{2k+1-2\ell}}{(2k+1-2\ell)!},
\label{eq:formula1_super}
\end{align}
involving the kernel function $H^{\mathrm{SM}(p)}(x,y)$ in equation~\eqref{eq:H_super_minimal} for the $(2,2p-2)$ minimal superstring,
where $u_j'=(p/2\pi)\sin((j-1/2)\pi/p)$ in equation~\eqref{eq:uj_super}.
(See \cite[Section 5.4]{Norbury:2020vyi} for an analogous formula of Stanford--Witten's recursion for the super Weil--Petersson volumes.)
The following expansion for the $(2,2p-2)$ minimal superstring spectral curve
also gives the coefficients $s'_{\ell}$'s in equation~\eqref{eq:formula1_super}:
\begin{align}
\frac{1}{\sfy^{\mathrm{SM}(p)}(z)}={}&
(-1)^{\frac{p-1}{2}}\frac{z}{U_{p-1}\left(\frac{2\pi}{p}z\right)}\nonumber
 \\
={}&z \delta_{p,1}+\frac{z}{2\pi}\sum_{j=1}^{(p-1)/2}
(-1)^j \cos^2\left(\frac{\pi}{p}\left(j-\frac{1}{2}\right)\right)\left(
\frac{1}{z-u'_j}-\frac{1}{z+u'_j}
\right)\nonumber
\\
={}&\sum_{\ell \ge 0}s'_{\ell} z^{2\ell+1}.
\label{eq:1/y_super}
\end{align}

\begin{lem}\label{lem:Laplace_B-term_super}
In equation~\eqref{eq:Laplace_dual_untwisted_super},
the $x^{2k+1}$ term in $xV_{g,n-1}\big[\mathsf{f}^{\mathrm{MV}}\big](x,L_{K \setminus\{m\}})$ obeys
\begin{align}
&
\int_{\IR_+^3} R(L_1,L_m,x) x^{2k+1} \e^{-z_1L_1-z_mL_m}
{\rm d}x {\rm d}L_1 L_m{\rm d}L_m\nonumber
\\
&
\qquad=\frac{1}{2}\big[w^0\big]\sum_{d\ge 0}\frac{1}{z_1^{2d+2}}
\frac{w^{2d+1}}{\sfy^{\mathrm{SM}(p)}(w)}
\cW_{0,2}(w,z_m)
\zeta_{\mathrm{H}}(2k+2;w) (2k+1)!.
\label{eq:twist_B-term_super}
\end{align}
\end{lem}
\begin{proof}
The equation \eqref{eq:twist_B-term_super} is verified in the parallel way as equation~\eqref{eq:twist_B-term_bosonic}.
We rewrite the left-hand side of equation~\eqref{eq:twist_B-term_super}
by equations~\eqref{eq:DR_H_super}, \eqref{eq:Laplace_basic}, and \eqref{eq:formula1_super}:
\begin{align*}
&
\frac12\int_{\IR_+^3}
\big(H^{\mathrm{SM}(p)}(x,L_1+L_m)+H^{\mathrm{SM}(p)}(x,L_1-L_m)\big)
x^{2k+1} \e^{-z_1L_1-z_mL_m}
{\rm d}x {\rm d}L_1 L_m{\rm d}L_m
\\
&
\qquad=(2k+1)!\sum_{\ell=0}^{k}\sum_{d=0}^{k-\ell}s'_{\ell}\cdot(2k-2\ell-2d+1)
\frac{1}{z_1^{2d+2}}\frac{1}{z_m^{2k-2\ell-2d+2}}.
\end{align*}
The right-hand side of equation~\eqref{eq:twist_B-term_super} is in the same form \eqref{eq:twist_B-term_bosonic_RHS} as the $(2,p)$ minimal string.
By the formula \eqref{hurwitz_def} for $\zeta_{\mathrm{H}}(2k+2;w)$
and the expansion \eqref{eq:1/y_super} for $1/\sfy^{\mathrm{SM}(p)}(z)$,
the claim follows.
\end{proof}

\begin{lem}\label{lem:Laplace_C-term_super}
In equation~\eqref{eq:Laplace_dual_untwisted_super},
the $x^{2a+1}y^{2b+1}$ term in $xyP_{g,n}\big[\sfm\big](x,y,L_K)$ obeys
\begin{align}
&
\frac{1}{2}\int_{\IR_+^3}
\big(D(L_1,x,y)+R(L_1,x,y) \mathsf{f}^{\mathrm{MV}}(y)
+R(L_1,y,x) \mathsf{f}^{\mathrm{MV}}(x)\big)
x^{2a+1}y^{2b+1} \e^{-z_1L_1} {\rm d}x{\rm d}y{\rm d}L_1\nonumber
\\
&
\qquad=\frac{1}{2}\big[w^0\big]
\sum_{d\ge 0}\frac{1}{z_1^{2d+2}}
\frac{w^{2d+1}}{\sfy^{\mathrm{SM}(p)}(w)}
(\zeta_{\mathrm{H}}(2a+2;w) (2a+1)!)
(\zeta_{\mathrm{H}}(2b+2;w) (2b+1)!).
\label{eq:twist_C-term_super}
\end{align}
\end{lem}
\begin{proof}
The equation \eqref{eq:twist_C-term_super} is verified in the parallel way as equation~\eqref{eq:twist_C-term_bosonic}.
\end{proof}

Combining the claims in Lemmas \ref{lem:Laplace_B-term_super} and \ref{lem:Laplace_C-term_super}, we find the claim of Proposition \ref{prop:Laplace_dual_untwisted_super}.

\section{Virasoro constraints}\label{sec:virasoro_constraint}

In this section, we first overview an algebraic formulation, called the quantum Airy structures~\cite{Andersen:2017vyk,Kontsevich:2017vdc},
of the ABO topological recursion and the CEO topological recursion.
In particular, we will see the equivalence between
the quantum Airy structures and the Virasoro constraints
for the physical 2D gravity models in Table \ref{tab:spectral_curve}.
We then discuss explicit computation of the volume coefficients
$F^{(g)}_{a_1,\dots,a_n}$ as well as the twisted volume coefficients
\smash{$F^{(g)}\big[\sfm\big]_{a_1,\dots,a_n}$} using
the cut-and-join equations in \cite{Alexandrov:2010bn,Alexandrov:2016kjl}
derived from the Virasoro constraints and homogeneity conditions,
and a group action in \cite{Andersen:2017vyk} which is associated with the twist action.

\subsection{Formulation}\label{sec:def_vir}

Consider the generating function of the volume coefficients
$F^{(g)}_{a_1,\dots,a_n}$ in equation~\eqref{gr_exp} of the ABO topological recursion
or equation~\eqref{laplace_exp} of the CEO topological recursion for $2g-2+n > 0$:
\begin{align}
& Z(\hbar; \textbf{t})= \e^{F(\hbar; \textbf{t})},\nonumber
\\
& F(\hbar; \textbf{t})
= 
\sum_{g \ge 0} \hbar^{2g-2} F_{g}(\textbf{t})
=
\sum_{g \ge 0, n\ge 1} \hbar^{2g-2}
\sum_{a_1,\dots,a_n \ge 0} F^{(g)}_{a_1,\dots,a_n}
\frac{t_{a_1}\cdots t_{a_n}}{n!}.
\label{gen_fn}
\end{align}
Here $\textbf{t}=\{t_0, t_1, t_2, \dots \}$ is the set of variables $t_a$
which are related to the length variables $L_i$ in equation~\eqref{gr_exp}
and the spectral curve variables $z_i$ in equation~\eqref{laplace_exp} by
\begin{align*}
\frac{t_{a_1}\cdots t_{a_n}}{n!}
\ \ \longleftrightarrow\ \
\frac{L_1^{2a_1}}{(2a_1+1)!} \cdots \frac{L_n^{2a_n}}{(2a_n+1)!}
\ \ \longleftrightarrow\ \
\frac{1}{z_1^{2a_1+2}} \cdots \frac{1}{z_n^{2a_n+2}}.
\end{align*}
In this set of variables, the recursion \eqref{geom_rec_coeff} leads to
\begin{align*}
\partial_{k} F(\hbar; \textbf{t})={}&
\sum_{a, b \ge 0} \sfB^{k}_{a, b} t_a \partial_b F(\hbar; \textbf{t})
+ \frac{\hbar^2}{2} \sum_{a, b \ge 0} \sfC^{k}_{a, b}
(\partial_a \partial_b F(\hbar; \textbf{t})
+ \partial_a F(\hbar; \textbf{t}) \partial_b F(\hbar; \textbf{t}) )
\\
&
+ \frac{1}{2\hbar^2} \sum_{a, b \ge 0} \sfA^{k}_{a, b} t_a t_b
+ \sfD^{k},
\end{align*}
where $\partial_a = \partial/\partial t_a$,
and the following proposition is obtained.

\begin{prop}[\cite{Andersen:2017vyk}]\label{prop:viraroso_const}
The generating function \eqref{gen_fn} satisfies constraint equations,
\begin{align}
\widehat{L}_k Z(\hbar; \textbf{t})=0,\qquad
k \ge -1,
\label{airy_constraint}
\end{align}
where
\begin{align}
\widehat{L}_k=
-\frac12 \partial_{k+1}
+ \frac{1}{4\hbar^2} \sum_{a, b \ge 0} \sfA^{k+1}_{a, b} t_a t_b
+ \frac12 \sum_{a, b \ge 0} \sfB^{k+1}_{a, b} t_a \partial_b
+ \frac{\hbar^2}{4} \sum_{a, b \ge 0} \sfC^{k+1}_{a, b} \partial_a \partial_b
+ \frac12 \sfD^{k+1}.
\label{airy_gen}
\end{align}
\end{prop}

When the differential operators $\widehat{L}_k$ satisfy
\begin{align}
\big[\widehat{L}_k, \widehat{L}_{\ell} \big]
= \sum_{a \ge -1} f^{a}_{k,\ell} \widehat{L}_{a},
\qquad
k, \ell \ge -1,
\label{q_airy_rel}
\end{align}
where $f^{a}_{k,\ell}$ are scalars, the operators $\widehat{L}_k$ define
a so called \emph{quantum Airy structure} on the space of the variables $t_a$ \cite{Andersen:2017vyk,Kontsevich:2017vdc}.
The quantum Airy structure is shown to be a sufficient condition for
the existence of the solution to the constraint equations \eqref{airy_constraint} \cite{Andersen:2017vyk}.
In particular, when the differential operators $\widehat{L}_k$ satisfy
the \emph{Virasoro relations}
\begin{align}
\big[\widehat{L}_k, \widehat{L}_{\ell}\big]
= \left(k-\ell \right) \widehat{L}_{k+\ell},
\qquad
k, \ell \ge -1,
\label{virasoro_rel}
\end{align}
the constraint equations \eqref{airy_constraint} are referred to as
the \emph{Virasoro constraints}.
As we will show below, the generating functions
of the volume coefficients for the physical 2D gravity models in Table~\ref{tab:spectral_curve}, with or without twist,
satisfy the Virasoro constraints.

\subsection{Bosonic models}\label{sec:examples_vir_boson}

Here we discuss the bosonic models in Table \ref{tab:spectral_curve}.

\subsubsection{Airy and KdV}\label{subsec:vir_kw}

From equation~\eqref{gr_kw_abcd},
the Airy initial data for the Kontsevich--Witten symplectic volumes of moduli spaces of stable curves are
\begin{alignat}{3}
&\sfA^{a_1}_{a_2,a_3}= \delta_{a_1, a_2, a_3, 0},
\qquad&&
\sfB^{a_1}_{a_2,a_3}=(2a_2+1) \delta_{a_1+a_2, a_3+1},&\nonumber
\\
&\sfC^{a_1}_{a_2,a_3}= \delta_{a_1, a_2+a_3+2},
\qquad&&
\sfD^{a_1}=\frac{1}{8} \delta_{a_1, 1},&\label{vir_kw_abcd}
\end{alignat}
and the differential operators for $k \ge -1$ in equation~\eqref{airy_gen},
\begin{align}
\widehat{L}_k^{\mathrm{A}}=
-\frac12 \partial_{k+1}
+ \sum_{a \ge 0} \left( a+\frac12 \right) t_a \partial_{a+k}
+ \frac{\hbar^2}{4}
\mathop{\sum_{a, b \ge 0}}\limits_{a+b=k-1} \partial_a \partial_b
+ \frac{1}{4\hbar^2} t_0^2 \delta_{k, -1}
+ \frac{1}{16} \delta_{k, 0},
\label{vir_kw}
\end{align}
satisfy the Virasoro relations
\begin{align*}
\big[\widehat{L}_k^{\mathrm{A}}, \widehat{L}_{\ell}^{\mathrm{A}}\big]
= (k-\ell) \widehat{L}_{k+\ell}^{\mathrm{A}},
\qquad
k, \ell \ge -1.
\end{align*}
Then, the constraint equations \eqref{airy_constraint} provide
the Virasoro constraints \cite{Dijkgraaf:1990rs,Fukuma:1990jw},
\begin{align}
\widehat{L}_k^{\mathrm{A}} Z^{\mathrm{A}}(\hbar; \textbf{t})=0,\qquad
k \ge -1,
\label{vir_const_kw}
\end{align}
for the generating function of the Airy volume coefficients \eqref{feg_kw}:
\begin{align}
\log Z^{\mathrm{A}}(\hbar; \textbf{t})=
\sum_{g \ge 0} \hbar^{2g-2} F^{\mathrm{A}}_{g}(\textbf{t})
=
\sum_{g \ge 0, n\ge 1} \hbar^{2g-2}
\mathop{\sum_{a_1,\dots,a_n \ge 0}}\limits_{|\mathbf{a}|=3g-3+n}
F^{\mathrm{A}(g)}_{a_1,\dots,a_n}
\frac{t_{a_1}\cdots t_{a_n}}{n!},
\label{kw_pf}
\end{align}
which satisfies the homogeneity condition \eqref{hom_kw}.

\begin{rem}
The Virasoro operators $\widehat{L}_k^{\mathrm{A}}$ in equation~\eqref{vir_kw} admits
the free field realization \cite{Dijkgraaf:1990rs,Fukuma:1990jw}:
\begin{align*}
T^{\mathrm{A}}(x)=: 
  \partial\phi^{\mathrm{A}}(x)\partial\phi^{\mathrm{A}}(x) \!:
+ \frac{1}{16x^2}
= \sum_{n \in {\IZ}} \widehat{L}_n^{\mathrm{A}} x^{-n-2},
\end{align*}
by a chiral bosonic field with the anti-periodic boundary condition:
\begin{align*}
\partial \phi^{\mathrm{A}}(x)=
\frac{1}{\hbar} \sum_{n \ge 0}
\left(n+\frac12\right)\left(t_n-\frac13 \delta_{n,1}\right) x^{n-\frac12}
+ \frac{\hbar}{2} \sum_{n \ge 0} \partial_n x^{-n-\frac32}
= \sum_{n \in {\IZ}} \alpha_{n+\frac12}^{\mathrm{A}} x^{-n-\frac32}.
\end{align*}
\end{rem}

From the Virasoro constraints \eqref{vir_const_kw} with
the homogeneity condition \eqref{hom_kw},
a cut-and-join representation of the Airy generating function \eqref{kw_pf}
is derived in \cite{Alexandrov:2010bn} following \cite{Morozov:2009xk}.

\begin{prop}[\cite{Alexandrov:2010bn}]\label{prop:cut_join_kw}
The Airy generating function \eqref{kw_pf}, which is a solution of the Virasoro constraints \eqref{vir_const_kw},
is given by
\begin{align}
Z^{\mathrm{A}}(x\hbar; x\textbf{t})
= \sum_{k \ge 0} x^k Z^{\mathrm{A}}_k(\hbar; \textbf{t})
= \e^{x \widehat{W}^{\mathrm{A}}}\cdot 1,
\label{cj_eq_kw}
\end{align}
where the cut-and-join operator $\widehat{W}^{\mathrm{A}}$ is
\begin{align*}
\widehat{W}^{\mathrm{A}}= &
\frac13 \sum_{a, b \ge 0}\left(2a+1\right)\left(2b+1\right)
t_a t_b \partial_{a+b-1}
+ \frac{\hbar^2}{6} \sum_{a, b \ge 0}\left(2a+2b+5\right)
t_{a+b+2}\partial_a \partial_b
+ \frac{t_0^3}{6\hbar^2} + \frac{t_1}{8},
\end{align*}
where $\partial_{-1}=0$.
\end{prop}
\begin{proof}
From the Virasoro constraints \eqref{vir_const_kw}, we find
\begin{align}
0=\frac23 \sum_{k \ge 0} (2k+1) t_k
\widehat{L}_{k-1}^{\mathrm{A}} Z^{\mathrm{A}}(\hbar; \textbf{t})
=
-\widehat{D}^{\mathrm{A}} Z^{\mathrm{A}}(\hbar; \textbf{t})
+ \widehat{W}^{\mathrm{A}} Z^{\mathrm{A}}(\hbar; \textbf{t}),
\label{cj_kw_proof}
\end{align}
where $\widehat{D}^{\mathrm{A}}$ denotes the Euler operator
\begin{align*}
\widehat{D}^{\mathrm{A}}=\frac13 \sum_{k \ge 0}
(2k+1) t_k \partial_k.
\end{align*}
The homogeneity condition \eqref{hom_kw} leads to
the action of the Euler operator on $Z^{\mathrm{A}}$ such that
\begin{align*}
\widehat{D}^{\mathrm{A}} Z^{\mathrm{A}}_k(\hbar; \textbf{t})=
k Z^{\mathrm{A}}_k(\hbar; \textbf{t}),
\end{align*}
and equation~\eqref{cj_kw_proof} gives
\begin{align}
\sum_{k \ge 0} x^{k+1}
\widehat{W}^{\mathrm{A}} Z^{\mathrm{A}}_k(\hbar; \textbf{t})
= \sum_{k \ge 1} x^k k Z^{\mathrm{A}}_k(\hbar; \textbf{t}).
\label{eq:cut_join_airy0}
\end{align}
From $x^k$ terms in equation~\eqref{eq:cut_join_airy0},
a recursion relation
\smash{$\widehat{W}^{\mathrm{A}} Z^{\mathrm{A}}_{k-1}(\hbar; \textbf{t})
=kZ^{\mathrm{A}}_k(\hbar; \textbf{t})$}
is found, and we obtain a solution $Z^{\mathrm{A}}_k(\hbar; \textbf{t})$ by adopting
the recursion relation iteratively,
\begin{align}
\label{eq:solution_Zk_A}
Z^{\mathrm{A}}_k(\hbar; \textbf{t})
= \frac{1}{k} \widehat{W}^{\mathrm{A}} Z^{\mathrm{A}}_{k-1}(\hbar; \textbf{t})
= \cdots
= \frac{1}{k!} \big(\widehat{W}^{\mathrm{A}}\big)^k \cdot 1,
\end{align}
where $Z^{\mathrm{A}}_{0}(\hbar; \textbf{t})=1$.
And finally, we find the cut-and-join representation \eqref{cj_eq_kw} by taking a~generating series of $Z^{\mathrm{A}}_k(\hbar; \textbf{t})$ in equation~\eqref{eq:solution_Zk_A}.
\end{proof}

By the iterative use of the cut-and-join equation \eqref{cj_eq_kw},
we obtain
the first few volume coefficients in equation~\eqref{kw_pf} as:
\begin{align}
& F^{\mathrm{A}}_{0}(\textbf{t})= 
\frac{t_{0}^{3}}{6} + \frac{t_{0}^{3} t_{1}}{2} +
\left(\frac{5}{8} t_{0}^{4} t_{2} + \frac{3}{2} t_{0}^{3} t_{1}^{2}\right)
+
\left(\frac{7}{8} t_{0}^{5} t_{3}+\frac{45}{8} t_{0}^{4} t_{1} t_{2}+\frac{9}{2} t_{0}^{3} t_{1}^{3}\right) + \cdots,\nonumber
\\
& F^{\mathrm{A}}_{1}(\textbf{t})= 
\frac{t_{1}}{8} + \left(\frac{5 t_{0} t_{2}}{8}+\frac{3 t_{1}^{2}}{16}\right)
+ \left(\frac{35}{16} t_{0}^{2} t_{3}+\frac{15}{4} t_{0} t_{1} t_{2}+\frac{3}{8} t_{1}^{3}\right) + \cdots,\nonumber
\\
& F^{\mathrm{A}}_{2}(\textbf{t})= 
\frac{105 t_{4}}{128}
+ \left(\frac{1155 t_{0} t_{5}}{128}+\frac{945 t_{1} t_{4}}{128}+\frac{1015 t_{2} t_{3}}{128}\right) + \cdots,\nonumber
\\
& F^{\mathrm{A}}_{3}(\textbf{t})= 
\frac{25025 t_{7}}{1024}
+ \left(\frac{425425 t_{8} t_{0}}{1024}+\frac{375375 t_{1} t_{7}}{1024}+\frac{385385 t_{2} t_{6}}{1024}+\frac{193655 t_{3} t_{5}}{512}+\frac{191205 t_{4}^{2}}{1024}\right)\nonumber
\\
& \hphantom{F^{\mathrm{A}}_{3}(\textbf{t})= }{}
+ \cdots.
\label{feg_comp_kw}
\end{align}

The Airy generating function \eqref{kw_pf} generates the coefficients \smash{$F^{\mathrm{KdV}(g)}_{a_1,\dots,a_n}$} in equation~\eqref{kdv_mdiff} by the following proposition.

\begin{prop}\label{prop:shift_kdv}
The generating function of the coefficients \smash{$F^{\mathrm{KdV}(g)}_{a_1,\dots,a_n}$} in equation~\eqref{kdv_mdiff},
\begin{align}
\log Z^{\mathrm{KdV}}(\hbar; \textbf{t})=
\sum_{g \ge 0, n\ge 1} \hbar^{2g-2}
\mathop{\sum_{a_1,\dots,a_n \ge 0}}\limits_{|\mathbf{a}| \le 3g-3+n}
F^{\mathrm{KdV}(g)}_{a_1,\dots,a_n}
\frac{t_{a_1}\cdots t_{a_n}}{n!},
\label{kdv_pf}
\end{align}
is obtained by shifting the variables in the Airy generating function \eqref{kw_pf} as
\begin{align}
t_a
 \to
t_a - \frac{\sfu_{2a-1}}{2a+1}
\qquad
\textrm{for}\quad a \ge 2.
\label{shift_kdv}
\end{align}
\end{prop}
\begin{proof}
The shift \eqref{shift_kdv} for the Airy generating function \eqref{kw_pf} gives
\begin{align*}
&
\mathop{\sum_{a_1,\dots,a_{n+m} \ge 0}}\limits_{|\mathbf{a}|=3g-3+n+m}
\frac{1}{(n+m)!}
F^{\mathrm{A}(g)}_{a_1,\dots,a_{n+m}}
\prod_{i=1}^{n+m}\left(t_{a_i} - \frac{\sfu_{2a_i-1}}{2a_i+1}\right)
\\
&\qquad=
\mathop{\sum_{a_1,\dots,a_{n+m} \ge 0}}\limits_{|\mathbf{a}|=3g-3+n+m}
\frac{1}{(n+m)!}
\binom{n+m}{m}
F^{\mathrm{A}(g)}_{a_1,\dots,a_{n+m}}
\Bigg(\prod_{i=1}^{n}t_{a_i}\Bigg)
\Bigg(\prod_{j=1}^{m} (-1) \frac{\sfu_{2a_{n+j}-1}}{2a_{n+j}+1}\Bigg)
+
\cdots,
\end{align*}
where $\sfu_{-1}=\sfu_{1}=0$.
This leads to the right-hand side of the formula \eqref{kdv_airy} in
Proposition \ref{prop:kdv}, and the claim is proved.
\end{proof}

Proposition \ref{prop:shift_kdv} implies that
the KdV generating function \eqref{kdv_pf} also satisfies the Virasoro constraints and can be computed by the cut-and-join equation \eqref{cj_eq_kw}
with the shift \eqref{shift_kdv}.

\subsubsection{Weil--Petersson volumes}\label{subsec:vir_mir}

From Proposition \ref{prop:shift_kdv},
the generating function of the Weil--Petersson volume coefficients
$F^{\mathrm{WP}(g)}_{a_1,\dots,a_n}$ in equations~\eqref{feg_mir} and \eqref{wp_mdiff},
\begin{align}
\log Z^{\mathrm{WP}}(\hbar; \textbf{t})=
\sum_{g \ge 0} \hbar^{2g-2} F^{\mathrm{WP}}_{g}(\textbf{t})
=
\sum_{g \ge 0, n\ge 1} \hbar^{2g-2}
\mathop{\sum_{a_1,\dots,a_n \ge 0}}\limits_{|\mathbf{a}| \le 3g-3+n}
F^{\mathrm{WP}(g)}_{a_1,\dots,a_n}
\frac{t_{a_1}\cdots t_{a_n}}{n!},
\label{mir_pf}
\end{align}
is given by the shift of variables in the Airy generating function \eqref{kw_pf} \cite{ManZog2000}
\begin{align}
t_a
 \to
t_a - \frac{\bigl(-2\pi^2\bigr)^{a-1}}{(2a+1)!! (a-1)!}
\qquad
\textrm{for}\quad a \ge 2.
\label{shift_mir}
\end{align}
This shift is found from the specialization \eqref{sp_kdv_wp} of the time variables $\sfu_{a}$ of the KdV spectral curve to obtain the Weil--Petersson spectral curve $\cC^{\mathrm{WP}}$.
The cut-and-join description of
the generating function \eqref{mir_pf} is obtained in \cite{Alexandrov:2021mjv}.

Adopting the shift \eqref{shift_mir} to equation~\eqref{feg_comp_kw},
we find the first few volume coefficients in equation~\eqref{mir_pf} such that
\begin{align*}
F^{\mathrm{WP}}_{0}(\textbf{t})={}&
\frac{t_{0}^{3}}{6}
+ \left(\frac{\pi^{2}}{12} t_{0}^{4} +\frac{1}{2} t_{1} t_{0}^{3}\right)
+ \left(\frac{\pi^{4}}{12} t_{0}^{5} +\frac{3}{4} \pi^{2} t_{0}^{4} t_{1}+\frac{5}{8} t_{0}^{4} t_{2}+\frac{3}{2} t_{0}^{3} t_{1}^{2}\right)
+ \cdots,
\\
F^{\mathrm{WP}}_{1}(\textbf{t})={}&
\left(\frac{t_{0} \pi^{2}}{12}+\frac{t_{1}}{8}\right)
+ \left(\frac{\pi^{4}}{8} t_{0}^{2} +\frac{\pi^{2}}{2} t_{0} t_{1} +\frac{5}{8} t_{2} t_{0}+\frac{3}{16} t_{1}^{2}\right)+ \left(\frac{7}{27}\pi^{6} t_{0}^{3} +\frac{13}{8}\pi^{4} t_{0}^{2} t_{1} \right.
\\
&\left.
{}+\frac{5}{2}\pi^{2} t_{0}^{2} t_{2}+\frac{9}{4}\pi^{2} t_{0} t_{1}^{2} +\frac{35}{16} t_{3} t_{0}^{2}+\frac{15}{4} t_{2} t_{1} t_{0}+\frac{3}{8} t_{1}^{3}\right) + \cdots,
\\
F^{\mathrm{WP}}_{2}(\textbf{t})={}&
\left(\frac{29}{192}\pi^{8} t_{0} +\frac{169}{480}\pi^{6} t_{1} +\frac{139}{192}\pi^{4} t_{2} +\frac{203}{192}\pi^{2} t_{3}+\frac{105}{128} t_{4}
\right) + \cdots.
\end{align*}

\subsubsection[(2,p) minimal string]{$\boldsymbol{(2,p)}$ minimal string}\label{subsec:vir_mg}

From Proposition \ref{prop:shift_kdv},
the generating function of the $(2,p)$ minimal string volume coefficients \smash{$F^{\mathrm{M}(p) (g)}_{a_1,\dots,a_n}$} in equation~\eqref{mst_mdiff},
\begin{align}
\log Z^{\mathrm{M}(p)}(\hbar; \textbf{t})=
\sum_{g \ge 0} \hbar^{2g-2} F^{\mathrm{M}(p)}_{g}(\textbf{t})
=
\sum_{g \ge 0, n\ge 1} \hbar^{2g-2}
\mathop{\sum_{a_1,\dots,a_n \ge 0}}\limits_{|\mathbf{a}| \le 3g-3+n}
F^{\mathrm{M}(p) (g)}_{a_1,\dots,a_n}
\frac{t_{a_1}\cdots t_{a_n}}{n!},
\label{mg_pf}
\end{align}
is given by the shift of variables in the Airy generating function \eqref{kw_pf},
\begin{align}
t_a
 \to
t_a - \frac{\bigl(-2\pi^2\bigr)^{a-1}}{(2a+1)!! (a-1)!}
\prod_{i=1}^{a-1}\left(1-\frac{(2i-1)^2}{p^2}\right)
\qquad
\textrm{for}\quad2 \le a \le \frac{p+1}{2}.
\label{shift_minimal_grav}
\end{align}
This shift is found from the specialization \eqref{sp_kdv_ms} of the time variables $\sfu_{a}$
of the KdV spectral curve to find the $(2,p)$ minimal string spectral curve $\cC^{\mathrm{M}(p)}$.

By the iterative use of the cut-and-join equation \eqref{cj_eq_kw} and the shift of variables \eqref{shift_minimal_grav}, we find
the first few volume coefficients in equation~\eqref{mg_pf} such that
\begin{align*}
&F^{\mathrm{M}(p)}_{0}(\textbf{t})= 
\frac{t_{0}^{3}}{6}
+ \left(\left(\frac{t_{0}^{4}}{12}-\frac{t_{0}^{4}}{12 p^{2}}\right) \pi^{2}+\frac{t_{1} t_{0}^{3}}{2}\right)
\\
& \hphantom{F^{\mathrm{M}(p)}_{0}(\textbf{t})=}{}
+ \left(\left(\frac{t_{0}^{5}}{12}-\frac{t_{0}^{5}}{30 p^{2}}-\frac{t_{0}^{5}}{20 p^{4}}\right) \pi^{4}+\left(\frac{3 t_{0}^{4} t_{1}}{4}-\frac{3 t_{0}^{4} t_{1}}{4 p^{2}}\right) \pi^{2}+\frac{5 t_{0}^{4} t_{2}}{8}+\frac{3 t_{0}^{3} t_{1}^{2}}{2}\right) + \cdots,
\\
&F^{\mathrm{M}(p)}_{1}(\textbf{t})= 
\left(\left(\frac{t_{0}}{12}-\frac{t_{0}}{12 p^{2}}\right) \pi^{2}+\frac{t_{1}}{8}\right)
\\
&\hphantom{F^{\mathrm{M}(p)}_{1}(\textbf{t})=}{}
+ \left(\left(\frac{t_{0}^{2}}{8}+\frac{t_{0}^{2}}{12 p^{2}}-\frac{5 t_{0}^{2}}{24 p^{4}}\right) \pi^{4}+\left(\frac{t_{0} t_{1}}{2}-\frac{t_{0} t_{1}}{2 p^{2}}\right) \pi^{2}+\frac{5 t_{2} t_{0}}{8}+\frac{3 t_{1}^{2}}{16}\right)
\\
&\hphantom{F^{\mathrm{M}(p)}_{1}(\textbf{t})=}{}
+ \left(\left(\frac{7 t_{0}^{3}}{27}+\frac{13 t_{0}^{3}}{27 p^{2}}-\frac{7 t_{0}^{3}}{9 p^{6}}+\frac{t_{0}^{3}}{27 p^{4}}\right) \pi^{6}+\left(-\frac{t_{0}^{2} t_{1}}{4 p^{2}}-\frac{11 t_{0}^{2} t_{1}}{8 p^{4}}+\frac{13 t_{0}^{2} t_{1}}{8}\right) \pi^{4}\right.
\\
&
\left.\hphantom{F^{\mathrm{M}(p)}_{1}(\textbf{t})=}{}\!
+\left(-\frac{5 t_{0}^{2} t_{2}}{2 p^{2}}-\frac{9 t_{0} t_{1}^{2}}{4 p^{2}}+\frac{5 t_{0}^{2} t_{2}}{2}+\frac{9 t_{0} t_{1}^{2}}{4}\right) \pi^{2}+\frac{35 t_{3} t_{0}^{2}}{16}+\frac{15 t_{2} t_{1} t_{0}}{4}+\frac{3 t_{1}^{3}}{8}\right)
+ \cdots,
\\
&F^{\mathrm{M}(p)}_{2}(\textbf{t})= 
\left(\left(\frac{29 t_{0}}{192}-\frac{6557 t_{0}}{2880 p^{8}}+\frac{497 t_{0}}{720 p^{2}}+\frac{587 t_{0}}{480 p^{4}}+\frac{17 t_{0}}{80 p^{6}}\right) \pi^{8}\right.\\
&\left.\hphantom{F^{\mathrm{M}(p)}_{2}(\textbf{t})= }{}
+\left(\frac{169 t_{1}}{480}-\frac{361 t_{1}}{480 p^{6}}+\frac{87 t_{1}}{160 p^{2}}-\frac{23 t_{1}}{160 p^{4}}\right) \pi^{6}\right.
\\
&
\left.\hphantom{F^{\mathrm{M}(p)}_{2}(\textbf{t})= }{}
+\left(\frac{139 t_{2}}{192}-\frac{23 t_{2}}{96 p^{2}}-\frac{31 t_{2}}{64 p^{4}}\right) \pi^{4}+\left(\frac{203 t_{3}}{192}-\frac{203 t_{3}}{192 p^{2}}\right) \pi^{2}+\frac{105 t_{4}}{128}\right)
+ \cdots.
\end{align*}
Note that the minimal string volume coefficients interpolate
the Airy volume coefficients at $p=1$ and
the Weil--Petersson volume coefficients at $p=\infty$:
\begin{align*}
F^{\mathrm{M}(1)}_{g}(\textbf{t})=F^{\mathrm{A}}_{g}(\textbf{t}),
\qquad
F^{\mathrm{M}(\infty)}_{g}(\textbf{t})=F^{\mathrm{WP}}_{g}(\textbf{t}).
\end{align*}

\subsection{Supersymmetric models}\label{sec:examples_vir_super}

Here we discuss the supersymmetric models in Table \ref{tab:spectral_curve}.

\subsubsection{Bessel and BGW}\label{subsec:vir_be}

From equation~\eqref{gr_be_abcd}, the Bessel initial data for the supersymmetric analogue of the symplectic volumes of moduli spaces of stable curves are
\begin{align}
&\sfA^{a_1}_{a_2,a_3}=0,
\qquad
\sfB^{a_1}_{a_2,a_3}=(2a_2+1) \delta_{a_1+a_2, a_3},\nonumber
\\
&\sfC^{a_1}_{a_2,a_3}=\delta_{a_1, a_2+a_3+1},
\qquad
\sfD^{a_1}=\frac{1}{8} \delta_{a_1, 0},
\label{vir_be_abcd}
\end{align}
and the differential operators in equation~\eqref{airy_gen},%
\footnote{Note that
$\widehat{L}_k^{\mathrm{B}}+\frac12 \partial_{k}=
\widehat{L}_k^{\mathrm{A}}+\frac12 \partial_{k+1}$ for $k \ge 0$.}
\begin{align}
\widehat{L}_k^{\mathrm{B}}=
-\frac12 \partial_{k}
+ \sum_{a \ge 0} \left( a+\frac12 \right) t_a \partial_{a+k}
+ \frac{\hbar^2}{4}
\mathop{\sum_{a, b \ge 0}}\limits_{a+b=k-1} \partial_a \partial_b
+ \frac{1}{16} \delta_{k, 0},
\qquad
k \ge 0,
\label{vir_be}
\end{align}
satisfy the Virasoro relations
\begin{align*}
\big[\widehat{L}_k^{\mathrm{B}}, \widehat{L}_{\ell}^{\mathrm{B}}\big]
= \left(k-\ell \right) \widehat{L}_{k+\ell}^{\mathrm{B}},
\qquad
k, \ell \ge 0.
\end{align*}
Then, the constraint equations \eqref{airy_constraint} provide
the Virasoro constraints \cite{Gross:1991ji,Mironov:1994mv},
\begin{align}
\widehat{L}_k^{\mathrm{B}} Z^{\mathrm{B}}(\hbar; \textbf{t})=0,\qquad
k \ge 0,
\label{vir_const_be}
\end{align}
for the generating function
of the Bessel volume coefficients \eqref{feg_be}
\begin{align}
\log Z^{\mathrm{B}}(\hbar; \textbf{t})=
\sum_{g \ge 0} \hbar^{2g-2} F^{\mathrm{B}}_{g}(\textbf{t})
=
\sum_{g \ge 0, n\ge 1} \hbar^{2g-2}
\mathop{\sum_{a_1,\dots,a_n \ge 0}}\limits_{|\mathbf{a}|=g-1}
F^{\mathrm{B}(g)}_{a_1,\dots,a_n}
\frac{t_{a_1}\cdots t_{a_n}}{n!},
\label{be_pf}
\end{align}
which satisfies the homogeneity condition \eqref{hom_be}.

From the Virasoro constraints \eqref{vir_const_be} with
the homogeneity condition \eqref{hom_be},
the following claim is proved in \cite{Alexandrov:2016kjl}.

\begin{prop}[\cite{Alexandrov:2016kjl}]\label{prop:cut_join_be}
The Bessel generating function \eqref{be_pf},
which is a solution of the Virasoro constraints \eqref{vir_const_be},
is given by
\begin{align}
Z^{\mathrm{B}}(x\hbar; x\textbf{t})
=\sum_{k \ge 0} x^k Z^{\mathrm{B}}_k(\hbar; \textbf{t})
=\e^{x \widehat{W}^{\mathrm{B}}}\cdot 1,
\label{cj_eq_be}
\end{align}
where the cut-and-join operator $\widehat{W}^{\mathrm{B}}$ is
\begin{align*}
\widehat{W}^{\mathrm{B}}=
\sum_{a, b \ge 0}(2a+1)(2b+1)
t_a t_b \partial_{a+b}
+ \frac{\hbar^2}{2} \sum_{a, b \ge 0}(2a+2b+3)
t_{a+b+1}\partial_a \partial_b
+ \frac{t_0}{8}.
\end{align*}
\end{prop}
\begin{proof}
From the Virasoro constraints \eqref{vir_const_be}, we find
\begin{align}
0=2 \sum_{k \ge 0} (2k+1) t_k
\widehat{L}_{k}^{\mathrm{B}} Z^{\mathrm{B}}(\hbar; \textbf{t})
=
-\widehat{D}^{\mathrm{B}} Z^{\mathrm{B}}(\hbar; \textbf{t})
+ \widehat{W}^{\mathrm{B}} Z^{\mathrm{B}}(\hbar; \textbf{t}),
\label{cj_be_proof}
\end{align}
where $\widehat{D}^{\mathrm{B}}$ denotes the Euler operator
\begin{align*}
\widehat{D}^{\mathrm{B}}=
\sum_{k \ge 0} (2k+1) t_k \partial_k.
\end{align*}
The homogeneity condition \eqref{hom_be} lead to the action of
the Euler operator acting on $Z^{\mathrm{B}}$ such that
\begin{align*}
\widehat{D}^{\mathrm{B}} Z^{\mathrm{B}}_k(\hbar; \textbf{t})=
k Z^{\mathrm{B}}_k(\hbar; \textbf{t}).
\end{align*}
Repeating the same analysis as the proof of Proposition \ref{prop:cut_join_kw},
we see that equation~\eqref{cj_be_proof} leads to equation~\eqref{cj_eq_be}.
\end{proof}

Due to the homogeneity condition \eqref{hom_be},
the $t_0$-dependence of \smash{$F^{\mathrm{B}}_{g}$} in
equation~\eqref{be_pf} are irrelevant to the genus growth,
and in fact, \smash{$F^{\mathrm{B}}_{g}$} for $g \ge 2$
are found from the following simple replacements of variables for
\smash{$F^{\mathrm{B}}_{g}$} with $t_0=0$
\cite{Alexandrov:2016kjl,Alexandrov:2009gn,Gross:1991aj,Okuyama:2020qpm}:
\begin{align*}
\hbar
\to
\frac{\hbar}{1-t_0},
\qquad
t_a
\to
\frac{t_a}{1-t_0}.
\end{align*}
The first few volume coefficients are found by the iterative use of the cut-and-join equation~\eqref{cj_eq_be} such that
\begin{align}
F^{\mathrm{B}}_{0}(\textbf{t})={}&0,
\qquad
F^{\mathrm{B}}_{1}(\textbf{t})
= -\frac18 \log (1-t_0)
= \frac{t_{0}}{8} + \frac{t_{0}^{2}}{16} + \frac{t_{0}^{3}}{24} +
\frac{t_{0}^{4}}{32} + \frac{t_{0}^{5}}{40} + \frac{t_{0}^{6}}{48} + \cdots,\nonumber
\\
F^{\mathrm{B}}_{2}(\textbf{t})
={}& \frac{9t_1}{128(1-t_0)^3}
= \frac{9 t_{1}}{128} + \frac{27 t_{0} t_{1}}{128}
+ \frac{27 t_{0}^{2} t_{1}}{64} + \frac{45 t_{0}^{3} t_{1}}{64}
+ \frac{135 t_{0}^{4} t_{1}}{128} + \cdots,\nonumber
\\
F^{\mathrm{B}}_{3}(\textbf{t})
={}& \frac{225t_2}{1024(1-t_0)^5} + \frac{567t_1^2}{1024(1-t_0)^6}
= \frac{225 t_{2}}{1024} +
\left(\frac{1125 t_{0} t_{2}}{1024}+\frac{567 t_{1}^{2}}{1024}\right)
+ \cdots,\nonumber
\\
F^{\mathrm{B}}_{4}(\textbf{t})
={}& \frac{55125 t_{3}}{32768 (1-t_0)^7} + \frac{388125 t_{1} t_{2}}{32768 (1-t_0)^8}
+ \frac{64989 t_1^3}{4096 (1-t_0)^9}\nonumber
\\
={}&\frac{55125 t_{3}}{32768}
+ \left(\frac{385875 t_{3} t_{0}}{32768}+\frac{388125 t_{2} t_{1}}{32768}\right)
\nonumber
\\
&+ \left(\frac{64989}{4096} t_{1}^{3}+\frac{385875}{8192} t_{3} t_{0}^{2}+\frac{388125}{4096} t_{2} t_{1} t_{0}\right)
+ \cdots.
\label{feg_comp_be}
\end{align}

Similar to Proposition \ref{prop:shift_kdv}, the Bessel generating function \eqref{be_pf} generates the coefficients \smash{$F^{\mathrm{BGW}(g)}_{a_1,\dots,a_n}$} in equation~\eqref{bgw_mdiff} by the following proposition which also implies the Virasoro constraints for the BGW generating function.

\begin{prop}\label{prop:shift_bgw}
The generating function of the coefficients
$F^{\mathrm{BGW} (g)}_{a_1,\dots,a_n}$ in equation~\eqref{bgw_mdiff},
\begin{align*}
\log Z^{\mathrm{BGW}}(\hbar; \textbf{t})=
\sum_{g \ge 0, n\ge 1} \hbar^{2g-2}
\mathop{\sum_{a_1,\dots,a_n \ge 0}}\limits_{|\mathbf{a}| \le g-1}
F^{\mathrm{BGW}(g)}_{a_1,\dots,a_n}
\frac{t_{a_1}\cdots t_{a_n}}{n!},
\end{align*}
is obtained by shifting the variables in the Bessel generating function \eqref{be_pf} as
\begin{align}
t_a
 \to
t_a - \frac{\sfv_{2a-1}}{2a+1}
\qquad
\textrm{for}\quad
 a \ge 1.
\label{shift_bgw}
\end{align}
\end{prop}

\subsubsection{Super Weil--Petersson volumes}\label{subsec:vir_sw}

From Proposition \ref{prop:shift_bgw},
the generating function of the super Weil--Petersson volume coefficients \smash{$F^{\mathrm{SWP}(g)}_{a_1,\dots,a_n}$} in equations~\eqref{feg_sw} and \eqref{super_wp_mdiff},
\begin{align}
\log Z^{\mathrm{SWP}}(\hbar; \textbf{t})=
\sum_{g \ge 0} \hbar^{2g-2} F^{\mathrm{SWP}}_{g}(\textbf{t})
=
\sum_{g \ge 0, n\ge 1} \hbar^{2g-2}
\mathop{\sum_{a_1,\dots,a_n \ge 0}}\limits_{|\mathbf{a}| \le g-1}
F^{\mathrm{SWP}(g)}_{a_1,\dots,a_n}
\frac{t_{a_1}\cdots t_{a_n}}{n!},
\label{sw_pf}
\end{align}
is given by the shift of variables in the Bessel generating function \eqref{be_pf} \cite{Norbury:2020vyi}:
\begin{align}
t_a
 \to
t_a - \frac{\bigl(-2\pi^2\bigr)^a}{(2a+1)!! a!}
\qquad
\textrm{for}\quad
 a \ge 1.
\label{shift_sw}
\end{align}
This shift is found from the specialization \eqref{sp_bgw_swp} of
the time variables $\sfv_{a}$ of the BGW spectral curve
to obtain the super Weil--Petersson spectral curve $\cC^{\mathrm{SWP}}$.
The cut-and-join description of
the generating function \eqref{sw_pf} is also obtained in \cite{Alexandrov:2021mjv}.

Adopting the shift \eqref{shift_sw} to equation~\eqref{feg_comp_be}, we find the first few volume coefficients in equation~\eqref{sw_pf} such that
\begin{align*}
F^{\mathrm{SWP}}_{0}(\textbf{t})={}&F^{\mathrm{B}}_{0}(\textbf{t})=0,
\qquad
F^{\mathrm{SWP}}_{1}(\textbf{t})=F^{\mathrm{B}}_{1}(\textbf{t})
=-\frac18 \log (1-t_0),
\\
F^{\mathrm{SWP}}_{2}(\textbf{t})
={}& \frac{9t_1}{128(1-t_0)^3}
+\frac{3\pi^2}{64}\left(\frac{1}{(1-t_0)^3}-1\right)
\\
={}&\left(\frac{9 \pi^{2} t_{0}}{64}+\frac{9 t_{1}}{128}\right)
+ \left(\frac{9}{32} \pi^{2} t_{0}^{2}+\frac{27}{128} t_{1} t_{0}\right)
+ \left(\frac{15}{32} \pi^{2} t_{0}^{3}+\frac{27}{64} t_{1} t_{0}^{2}\right)
+ \cdots,
\\
F^{\mathrm{SWP}}_{3}(\textbf{t})
={}& \frac{225t_2}{1024(1-t_0)^5} + \frac{567t_1^2}{1024(1-t_0)^6}
+ \frac{189 \pi^2 t_{1}}{256(1-t_{0})^{6}}
+ \frac{3\pi^4 (37+5 t_{0})}{512(1-t_{0})^{6}}
- \frac{111 \pi^{4}}{512}
\\
={}& \left(\frac{681}{512} \pi^{4} t_{0}+\frac{189}{256} \pi^{2} t_{1}+\frac{225}{1024} t_{2}\right)
\\
&+ \left(\frac{2421}{512} \pi^{4} t_{0}^{2}+\frac{567}{128} \pi^{2} t_{1} t_{0}+\frac{1125}{1024} t_{2} t_{0}+\frac{567}{1024} t_{1}^{2}\right)
+ \cdots.
\end{align*}

\subsubsection[(2,2p-2) minimal superstring]{$\boldsymbol{(2,2p-2)}$ minimal superstring}\label{subsec:vir_mg_super}

From Proposition \ref{prop:shift_bgw},
the generating function of the $(2,2p-2)$ minimal superstring free energies \smash{$F^{\mathrm{SM}(p) (g)}_{a_1,\dots,a_n}$} in equation~\eqref{super_mst_mdiff}
labeled by an odd positive integer $p$,
\begin{align}
\log Z^{\mathrm{SM}(p)}(\hbar; \textbf{t})=
\sum_{g \ge 0} \hbar^{2g-2} F^{\mathrm{SM}(p)}_{g}(\textbf{t})
=
\sum_{g \ge 0, n\ge 1} \hbar^{2g-2}
\mathop{\sum_{a_1,\dots,a_n \ge 0}}\limits_{|\mathbf{a}| \le g-1}
F^{\mathrm{SM}(p) (g)}_{a_1,\dots,a_n}
\frac{t_{a_1}\cdots t_{a_n}}{n!},
\label{super_mg_pf}
\end{align}
is given by the shift of variables in the Bessel generating function \eqref{be_pf},
\begin{align}
t_a
 \to
t_a - \frac{\bigl(-2\pi^2\bigr)^a}{(2a+1)!! a!}
\prod_{i=1}^{a}\left(1-\frac{(2i-1)^2}{p^2}\right)
\qquad
\textrm{for}\quad 1 \le a \le \frac{p-1}{2}.
\label{shift_minimal_grav_super}
\end{align}
This shift is found from the specialization
\eqref{sp_bgw_ms_super} of the time variables $\sfv_{a}$ of the BGW spectral curve
to obtain the $(2,2p-2)$ minimal superstring spectral curve $\cC^{\mathrm{SM}(p)}$.

Adopting the shift \eqref{shift_minimal_grav_super} to equation~\eqref{feg_comp_be},
we find the first few volume coefficients in equation~\eqref{super_mg_pf} such that
\begin{align*}
F^{\mathrm{SM}(p)}_{0}(\textbf{t})={}&F^{\mathrm{B}}_{0}(\textbf{t})=0,
\qquad
F^{\mathrm{SM}(p)}_{1}(\textbf{t})=F^{\mathrm{B}}_{1}(\textbf{t})
=-\frac18 \log (1-t_0),
\\
F^{\mathrm{SM}(p)}_{2}(\textbf{t})
={}& \frac{9t_1}{128(1-t_0)^3}
+ \left(\frac{3\pi^2}{64(1-t_0)^3}-\frac{3\pi^2}{64}\right)P_1
\\
={}& \left(\left(\frac{9}{64}-\frac{9}{64 p^{2}}\right)t_0 \pi^{2}
+\frac{9 t_{1}}{128}\right)
+ \left(\left(\frac{9}{32}-\frac{9}{32 p^{2}}\right)t_{0}^{2} \pi^{2}+\frac{27 t_{1} t_{0}}{128}\right)
+\cdots,
\\
F^{\mathrm{SM}(p)}_{3}(\textbf{t})
={}& \frac{225t_2}{1024(1-t_0)^5} + \frac{567t_1^2}{1024(1-t_0)^6}
+ \frac{189 \pi^2 t_{1} P_1}{256(1-t_{0})^{6}}
\\
&
+ \frac{3\pi^4 P_1}{512}
\left(
\frac{42P_1-5P_2+5 t_{0} P_2}{(1-t_{0})^{6}}
- 42P_1+5P_2 \right)
\\
={}& \left(\left(\frac{681 t_{0}}{512}-\frac{381 t_{0}}{256 p^{2}}+\frac{81 t_{0}}{512 p^{4}}\right) \pi^{4}+\left(\frac{189 t_{1}}{256}-\frac{189 t_{1}}{256 p^{2}}\right) \pi^{2}+\frac{225 t_{2}}{1024}\right)
\\
&
+ \left(\left(\frac{2421}{512}-\frac{1521}{256 p^{2}}+\frac{621}{512 p^{4}}\right) t_0^2 \pi^{4}
+\left(\frac{567}{128}-\frac{567}{128 p^{2}}\right) t_{0} t_{1} \pi^{2}\right.
\\
&\left.+\frac{1125 t_{2} t_{0}}{1024}+\frac{567 t_{1}^{2}}{1024}\right)
+ \cdots,
\end{align*}
where $P_1= 1-1/p^2$ and $P_2= 1-9/p^2$.
Note that the minimal superstring volume coefficients interpolate
the Bessel volume coefficients at $p=1$ and
the super Weil--Petersson volume coefficients at $p=\infty$,
\begin{align*}
F^{\mathrm{SM}(1)}_{g}(\textbf{t})=F^{\mathrm{B}}_{g}(\textbf{t}),
\qquad
F^{\mathrm{SM}(\infty)}_{g}(\textbf{t})=F^{\mathrm{SWP}}_{g}(\textbf{t}).
\end{align*}

\subsection{Twisting}\label{sec:twisting_vir}

The twisted initial data \eqref{vir_abcd_twist} of the ABO topological recursion
defines the generating function~$Z[\sff]$ of
the twisted volume coefficients $F^{(g)}[\sff]_{a_1,\dots,a_n}$,
\begin{align*}
Z[\sff](\hbar; \textbf{t})=
\exp\bigg(
\sum_{g \ge 0, n\ge 1} \hbar^{2g-2}
\sum_{a_1,\dots,a_n \ge 0} F^{(g)}[\sff]_{a_1,\dots,a_n}
\frac{t_{a_1}\cdots t_{a_n}}{n!}\bigg),
\end{align*}
which satisfies constraint equations
\begin{align*}
\widehat{L}[\sff]_k Z[\sff](\hbar; \textbf{t})=0,\qquad
k \ge -1,
\end{align*}
where $\widehat{L}[\sff]_k$ are twisted differential operators
\begin{align*}
\widehat{L}[\sff]_k={}&
-\frac12 \partial_{k+1}
+ \frac{1}{4\hbar^2} \sum_{a, b \ge 0} \sfA[\sff]^{k+1}_{a, b} t_a t_b
+ \frac12 \sum_{a, b \ge 0} \sfB[\sff]^{k+1}_{a, b} t_a \partial_b
\\
&
+ \frac{\hbar^2}{4} \sum_{a, b \ge 0} \sfC[\sff]^{k+1}_{a, b} \partial_a \partial_b
+ \frac12 \sfD[\sff]^{k+1}.
\end{align*}
The following proposition is then established.

\begin{prop}[\cite{Andersen:2017vyk}]\label{prop:givental_act}
The generating functions $Z$ and $Z[\sff]$, and
the differential operators $\widehat{L}_k$ and $\widehat{L}[\sff]_k$
are related by the group action of
\begin{align}
\widehat{U}[\sff]=
\exp\bigg(
\frac{\hbar^2}{2}\sum_{a, b \ge 0} u[\sff]_{a,b} \partial_a\partial_b
\bigg),
\label{giv_op}
\end{align}
defined by the twist function $u[\sff]_{a,b}$ in equation~\eqref{twist_u}, as
\begin{align}
Z[\sff](\hbar; \textbf{t})=\widehat{U}[\sff] Z(\hbar; \textbf{t}),
\qquad
\widehat{L}[\sff]_k = \widehat{U}[\sff] \widehat{L}_k \widehat{U}[\sff]^{-1}.
\label{givental_act}
\end{align}
\end{prop}

When the differential operators $\widehat{L}_k$ satisfy the constraint equations \eqref{q_airy_rel}
of the quantum Airy structure,
it is found from the group actions \eqref{givental_act} that the twisted differential operators $\widehat{L}[\sff]_k$ also
satisfy
\begin{align*}
\big[\widehat{L}[\sff]_k, \widehat{L}[\sff]_{\ell} \big]
= \sum_{a \ge -1} f^{a}_{k,\ell} \widehat{L}[\sff]_{a},
\qquad
k, \ell \ge -1.
\end{align*}
In particular, when the differential operators $\widehat{L}_k$ satisfy
the Virasoro relations \eqref{virasoro_rel}, the twisted operators $\widehat{L}[\sff]_k$ also satisfy the Virasoro relations
\begin{align*}
\big[\widehat{L}[\sff]_k, \widehat{L}[\sff]_{\ell}\big]
= (k-\ell ) \widehat{L}[\sff]_{k+\ell},
\qquad
k, \ell \ge -1.
\end{align*}

In the case of the Masur--Veech type twist,
$u_{a, b}^{\mathrm{MV}}$ in equation~\eqref{mv_twist_u},
the operator \eqref{giv_op} of the group action yields
\begin{align}
\widehat{U}^{\mathrm{MV}}=
\widehat{U}\big[\sfm\big]=
\exp\bigg(
\frac{\hbar^2}{2}\sum_{a, b \ge 0}
\frac{(2a+2b+1)!}{(2a+1)!(2b+1)!}
\zeta(2a+2b+2) \partial_a\partial_b
\bigg).
\label{mv_twist_vir_u}
\end{align}

\begin{rem}The twist action $\widehat{U}[\sff]$ is regarded as an exponentiated quadratic operator of Bogoliubov's transformation type.
It preserves Virasoro algebra of differential operators acting on a partition function, but can be harmful for integrable hierarchy equations.\footnote{This aspect of the twist action is pointed out by an anonymous referee.}
\end{rem}

\subsection{Masur--Veech type twist}\label{sec:examples_vir_tw}

In the following, we discuss the Virasoro constraints with the Masur--Veech type twist for the physical 2D gravity models in Table \ref{tab:spectral_curve}.

\subsubsection{Masur--Veech polynomials}\label{subsec:vir_mvp}

The initial data for the Masur--Veech polynomials are found from
the Airy initial data \eqref{vir_kw_abcd} twisted by
equations~\eqref{vir_abcd_twist} and \eqref{mv_twist_u}:
\begin{align*}
&\sfA\big[\sfm\big]^{a_1}_{a_2,a_3}=\delta_{a_1, a_2, a_3, 0},
\\
&\sfB\big[\sfm\big]^{a_1}_{a_2,a_3}=(2a_2+1) \delta_{a_1+a_2, a_3+1}
+ \zeta(2a_3+2) \delta_{a_1, a_2, 0},
\\
&\sfC\big[\sfm\big]^{a_1}_{a_2,a_3}=\delta_{a_1, a_2+a_3+2}
+ \zeta(2a_2+2) \zeta(2a_3+2) \delta_{a_1, 0}
\\
&\phantom{\sfC\big[\sfm\big]^{a_1}_{a_2,a_3}=}{}
+ \left(\binom{2a_2+2a_3-2a_1+3}{2a_2+1}
+ \binom{2a_2+2a_3-2a_1+3}{2a_3+1}\right) \\
&\phantom{\sfC\big[\sfm\big]^{a_1}_{a_2,a_3}=+}{}\times
\zeta(2a_2+2a_3-2a_1+4),
\\
&\sfD\big[\sfm\big]^{a_1}=\frac{1}{8} \delta_{a_1, 1}
+ \frac{\zeta(2)}{2} \delta_{a_1, 1}.
\end{align*}
By Proposition \ref{prop:givental_act},
the twisted differential operators
\begin{align*}
\widehat{L}_k^{\mathrm{MV}}=
\widehat{L}^{\mathrm{A}}\big[\sfm\big]_k=
\widehat{U}^{\mathrm{MV}}
\widehat{L}_k^{\mathrm{A}}
\big(\widehat{U}^{\mathrm{MV}}\big)^{-1},
\end{align*}
of the Airy Virasoro operators $\widehat{L}_k^{\mathrm{A}}$ in equation~\eqref{vir_kw}
by the twist \eqref{mv_twist_vir_u}
also satisfy the Virasoro relations\footnote{These Virasoro relations can also be checked directly by
verifying the commutation relations or
using the free field realization
in Remark~\ref{rem:free_mv} below.}
\begin{align}
\big[\widehat{L}_k^{\mathrm{MV}}, \widehat{L}_{\ell}^{\mathrm{MV}}\big]
= \left(k-\ell \right) \widehat{L}_{k+\ell}^{\mathrm{MV}},
\qquad
k, \ell \ge -1,
\label{vir_rel_mv}
\end{align}
and provide the Virasoro constraints
\begin{align*}
\widehat{L}_k^{\mathrm{MV}} Z^{\mathrm{MV}}(\hbar; \textbf{t})=0,\qquad
k \ge -1,
\end{align*}
for the generating function of the Masur--Veech polynomials
\begin{align}
\log Z^{\mathrm{MV}}(\hbar; \textbf{t})={}&
\log Z^{\mathrm{A}}\big[\sfm\big](\hbar; \textbf{t})=
\log \big(\widehat{U}^{\mathrm{MV}} Z^{\mathrm{A}}(\hbar; \textbf{t})\big)\nonumber
\\
={}&
\sum_{g \ge 0} \hbar^{2g-2} F^{\mathrm{MV}}_{g}(\textbf{t})
=
\sum_{g \ge 0, n\ge 1} \hbar^{2g-2}
\mathop{\sum_{a_1,\dots,a_n \ge 0}}\limits_{|\mathbf{a}| \le 3g-3+n}
F^{\mathrm{MV}(g)}_{a_1,\dots,a_n}
\frac{t_{a_1}\cdots t_{a_n}}{n!}.
\label{mv_pf}
\end{align}

\begin{rem}\label{rem:free_mv}
The twisted Virasoro operators $\widehat{L}_k^{\mathrm{MV}}$ in equation~\eqref{vir_rel_mv} admit
the free field realization
\begin{align*}
T^{\mathrm{MV}}(x)={}
{:} \partial\phi^{\mathrm{MV}}(x)\partial\phi^{\mathrm{MV}}(x) {:}
+ \frac{1}{16x^2} + \frac{\zeta(2)}{4x}
= \sum_{n \in {\IZ}} \widehat{L}_n^{\mathrm{MV}} x^{-n-2},
\end{align*}
by a ``twisted'' chiral bosonic field with the anti-periodic boundary condition
\begin{align*}
\partial \phi^{\mathrm{MV}}(x)={}&
\frac{1}{\hbar} \sum_{n \ge 0}
\left(n+\frac12\right)
\Bigg(\left(t_n-\frac13 \delta_{n,1}\right)
+ \hbar^2 \sum_{a \ge 0} u_{a, n}^{\mathrm{MV}} \partial_a \Bigg)
x^{n-\frac12}
+ \frac{\hbar}{2} \sum_{n \ge 0} \partial_n x^{-n-\frac32}
\\
={}& \sum_{n \in {\IZ}} \alpha_{n+\frac12}^{\mathrm{MV}} x^{-n-\frac32}.
\end{align*}
Note that the Virasoro relation \eqref{vir_rel_mv} follows
from this free field realization immediately.
\end{rem}

By Proposition \ref{prop:cut_join_kw},
the Masur--Veech generating function \eqref{mv_pf} is computed by
\begin{align*}
Z^{\mathrm{MV}}(x \hbar; x \textbf{t})={}&
\widehat{U}^{\mathrm{MV}} Z^{\mathrm{A}}(x \hbar; x \textbf{t})=
\widehat{U}^{\mathrm{MV}} \e^{x \widehat{W}^{\mathrm{A}}}\cdot 1,
\end{align*}
and we obtain the first few volume coefficients such that
\begin{align}
F^{\mathrm{MV}}_{0}(\textbf{t})={}&
\frac{t_{0}^{3}}{6} + \left(\frac{\pi^{2}}{48} t_{0}^{4}+\frac{1}{2} t_{1} t_{0}^{3}\right) + \left(\frac{\pi^{4}}{160} t_{0}^{5}+\frac{\pi^{2}}{8} t_{0}^{4} t_{1}+\frac{5}{8} t_{0}^{4} t_{2}+\frac{3}{2} t_{1}^{2} t_{0}^{3}\right) + \cdots,\nonumber
\\
F^{\mathrm{MV}}_{1}(\textbf{t})={}&
\left(\frac{\pi^{2} t_{0}}{12}+\frac{t_{1}}{8}\right)
+
\left(\frac{\pi^{4}}{32} t_{0}^{2}+\frac{\pi^{2}}{4} t_{0} t_{1}+\frac{5}{8} t_{0} t_{2}+\frac{3}{16} t_{1}^{2}\right)+
\left(\frac{11}{576}\pi^{6} t_{0}^{3} +\frac{3}{16} \pi^{4} t_{0}^{2} t_{1}\right.\nonumber
\\
&\left.
+\frac{65}{96} \pi^{2} t_{0}^{2} t_{2}+\frac{3}{4} \pi^{2} t_{0} t_{1}^{2}+\frac{35}{16} t_{0}^{2} t_{3}+\frac{15}{4} t_{0} t_{1} t_{2}+\frac{3}{8} t_{1}^{3}\right) + \cdots,\nonumber
\\
F^{\mathrm{MV}}_{2}(\textbf{t})={}&
\left(\frac{29}{2560} \pi^{8} t_{0}+\frac{1}{32} \pi^{6} t_{1}+\frac{119}{1152} \pi^{4} t_{2}+\frac{35}{96} \pi^{2} t_{3}+\frac{105}{128} t_{4}\right) + \cdots.
\label{feg_comp_mv}
\end{align}

Proposition \ref{prop:shift_kdv} leads to the following proposition.

\begin{prop}\label{prop:shift_kdv_tw}
The generating function of the twisted KdV volume coefficients,
\begin{align*}
\log Z^{\mathrm{KdV}}\big[\sfm\big](\hbar; \textbf{t})=
\sum_{g \ge 0, n\ge 1} \hbar^{2g-2}
\mathop{\sum_{a_1,\dots,a_n \ge 0}}\limits_{|\mathbf{a}| \le 3g-3+n}
F^{\mathrm{KdV}(g)}\big[\sfm\big]_{a_1,\dots,a_n}
\frac{t_{a_1}\cdots t_{a_n}}{n!},
\end{align*}
is obtained by shifting the variables $t_a$ as equation~\eqref{shift_kdv}
in the Masur--Veech $($i.e., twisted Airy$)$ generating function \eqref{mv_pf}.
\end{prop}

\subsubsection[Twisted (2,p) minimal string volume polynomials]{Twisted $\boldsymbol{(2,p)}$ minimal string volume polynomials}\label{subsec:vir_mg_twisted}

The generating function of the twisted volume coefficients for
the $(2,p)$ minimal string:
\begin{align*}
\log Z^{\mathrm{M}(p)}\big[\sfm\big](\hbar; \textbf{t})={}&
\sum_{g \ge 0} \hbar^{2g-2} F^{\mathrm{M}(p)}_{g}\big[\sfm\big](\textbf{t})
\\
={}&
\sum_{g \ge 0, n\ge 1} \hbar^{2g-2}
\mathop{\sum_{a_1,\dots,a_n \ge 0}}\limits_{|\mathbf{a}| \le 3g-3+n}
F^{\mathrm{M}(p) (g)}\big[\sfm\big]_{a_1,\dots,a_n}
\frac{t_{a_1}\cdots t_{a_n}}{n!},
\end{align*}
is obtained by shifting the variables $t_a$ as equation~\eqref{shift_minimal_grav} in the Masur--Veech generating function~\eqref{mv_pf}.
Here we introduce a deformation parameter $s$ as
\begin{align}
t_a
 \to
t_a - \frac{\bigl(-2\pi^2 s\bigr)^{a-1}}{(2a+1)!! (a-1)!}
\prod_{i=1}^{a-1}\left(1-\frac{(2i-1)^2}{p^2}\right)
\qquad
\textrm{for}\quad 2 \le a \le \frac{p+1}{2},
\label{shift_minimal_grav_s}
\end{align}
which yields the shift \eqref{shift_minimal_grav} at $s=1$.

Adopting the shift \eqref{shift_minimal_grav_s} to the Masur--Veech generating functions \eqref{feg_comp_mv},
we find the first few volume coefficients such that
\begin{align*}
F^{\mathrm{M}(p)}_{0}\big[\sfm\big](\textbf{t})={}&
\frac{t_{0}^{3}}{6}
+ \left(\left(\frac{t_{0}^{4} s}{12}-\frac{t_{0}^{4} s}{12 p^{2}}+\frac{t_{0}^{4}}{48}\right) \pi^{2}+\frac{t_{1} t_{0}^{3}}{2}\right)
\\
&
+ \left(\left(\frac{t_{0}^{5} s^{2}}{12}-\frac{t_{0}^{5} s^{2}}{30 p^{2}}-\frac{t_{0}^{5} s^{2}}{20 p^{4}}+\frac{t_{0}^{5} s}{36}-\frac{t_{0}^{5} s}{36 p^{2}}+\frac{t_{0}^{5}}{160}\right) \pi^{4}\right.
\\
&
\left. + \left(\frac{3 t_{0}^{4} t_{1} s}{4}-\frac{3 t_{0}^{4} t_{1} s}{4 p^{2}}+\frac{t_{0}^{4} t_{1}}{8}\right) \pi^{2}+\frac{5 t_{0}^{4} t_{2}}{8}+\frac{3 t_{0}^{3} t_{1}^{2}}{2}\right)
+ \cdots,
\\
F^{\mathrm{M}(p)}_{1}\big[\sfm\big](\textbf{t})={}&
\left(\left(\frac{t_{0} s}{12}-\frac{t_{0} s}{12 p^{2}}+\frac{t_{0}}{12}\right) \pi^{2}+\frac{t_{1}}{8}\right)
+ \left(\left(\frac{t_{0}^{2} s^{2}}{8}+\frac{t_{0}^{2} s^{2}}{12 p^{2}}-\frac{5 t_{0}^{2} s^{2}}{24 p^{4}}+\frac{13 t_{0}^{2} s}{144} \right. \right.
\\
&
\left.\left. -\frac{13 t_{0}^{2} s}{144 p^{2}}+\frac{t_{0}^{2}}{32}\right) \pi^{4} +\left(\frac{t_{0} t_{1} s}{2}-\frac{t_{0} t_{1} s}{2 p^{2}}+\frac{t_{0} t_{1}}{4}\right) \pi^{2}+\frac{5 t_{2} t_{0}}{8}+\frac{3 t_{1}^{2}}{16}\right)
\\
&
+ \left(\left(\frac{7 t_{0}^{3} s^{3}}{27}+\frac{71 t_{0}^{3} s^{2}}{432}+\frac{61 t_{0}^{3} s}{864}+\frac{11 t_{0}^{3}}{576}
+ \frac{13 t_{0}^{3} s^{3}}{27 p^{2}}+\frac{t_{0}^{3} s^{3}}{27 p^{4}}-\frac{7 t_{0}^{3} s^{3}}{9 p^{6}}
\right.\right.
\\
&
\left.\left.-\frac{61 t_{0}^{3} s}{864 p^{2}}-\frac{11 t_{0}^{3} s^{2}}{216 p^{2}}-\frac{49 t_{0}^{3} s^{2}}{432 p^{4}}\right) \pi^{6}+\left(\frac{3 t_{0}^{2} t_{1}}{16}-\frac{t_{0}^{2} t_{1} s^{2}}{4 p^{2}}
-\frac{11 t_{0}^{2} t_{1} s^{2}}{8 p^{4}}
\right.\right.
\\
&
\left.\left.- \frac{13 t_{0}^{2} t_{1} s}{16 p^{2}}+\frac{13 t_{0}^{2} t_{1} s^{2}}{8}+\frac{13 t_{0}^{2} t_{1} s}{16}\right) \pi^{4}
+\left(\frac{3 t_{0} t_{1}^{2}}{4}+\frac{65 t_{0}^{2} t_{2}}{96}-\frac{9 t_{0} t_{1}^{2} s}{4 p^{2}}\right.\right.
\\
&
\left.\left.-\frac{5 t_{0}^{2} s t_{2}}{2 p^{2}}
+\frac{9 t_{0} t_{1}^{2} s}{4}+\frac{5 t_{0}^{2} s t_{2}}{2}\right) \pi^{2}+\frac{15 t_{2} t_{1} t_{0}}{4}+\frac{3 t_{1}^{3}}{8}+\frac{35 t_{3} t_{0}^{2}}{16}\right)
+ \cdots,
\\
F^{\mathrm{M}(p)}_{2}\big[\sfm\big](\textbf{t})={}&
\left(
\left(\frac{497 t_{0} s^{4}}{720 p^{2}}+\frac{587 t_{0} s^{4}}{480 p^{4}}+\frac{17 t_{0} s^{4}}{80 p^{6}}-\frac{6557 t_{0} s^{4}}{2880 p^{8}}-\frac{5 t_{0} s}{128 p^{2}}-\frac{527 t_{0} s^{2}}{25920 p^{2}}
\right.\right.
\\
&
\left.\left.-\frac{629 t_{0} s^{2}}{10368 p^{4}}+\frac{23 t_{0} s^{3}}{96 p^{2}}
+\frac{t_{0} s^{3}}{96 p^{4}}-\frac{331 t_{0} s^{3}}{864 p^{6}}+\frac{115 t_{0} s^{3}}{864}+\frac{5 t_{0} s}{128}+\frac{29 t_{0} s^{4}}{192}
\right.\right.
\\
&
\left.\left.+\frac{4199 t_{0} s^{2}}{51840}+\frac{29 t_{0}}{2560}\right) \pi^{8}
+\left(\frac{87 t_{1} s^{3}}{160 p^{2}}
-\frac{23 t_{1} s^{3}}{160 p^{4}}-\frac{361 t_{1} s^{3}}{480 p^{6}}-\frac{119 t_{1} s}{960 p^{2}}
\right.\right.
\\
&
\left.\left.-\frac{t_{1} s^{2}}{18 p^{2}}-\frac{2 t_{1} s^{2}}{9 p^{4}}+\frac{169 t_{1} s^{3}}{480}+\frac{119 t_{1} s}{960}+\frac{5 t_{1} s^{2}}{18}+\frac{t_{1}}{32}\right) \pi^{6}\right.
\\
&
\left.
+\left(-\frac{245 s t_{2}}{576 p^{2}}-\frac{23 s^{2} t_{2}}{96 p^{2}}-\frac{31 s^{2} t_{2}}{64 p^{4}}+\frac{139 s^{2} t_{2}}{192}+\frac{245 s t_{2}}{576}+\frac{119 t_{2}}{1152}\right) \pi^{4}
\right.
\\
&
\left.
+\left(-\frac{203 s t_{3}}{192 p^{2}}+\frac{203 s t_{3}}{192}+\frac{35 t_{3}}{96}\right) \pi^{2}+\frac{105 t_{4}}{128}\right)
+ \cdots.
\end{align*}
These volume coefficients reduce to the Masur--Veech volume coefficients \eqref{feg_comp_mv}
at $p=1$ and the twisted Weil--Petersson volume coefficients at $p=\infty$.

\subsubsection{Super Masur--Veech polynomials}\label{subsec:vir_supermvp}

The initial data for the super Masur--Veech polynomials are found from
the Bessel initial data~\eqref{vir_be_abcd} twisted by
equations~\eqref{vir_abcd_twist} and \eqref{mv_twist_u}:
\begin{align*}
&\sfA\big[\sfm\big]^{a_1}_{a_2,a_3}=0,
\qquad
\sfB\big[\sfm\big]^{a_1}_{a_2,a_3}=(2a_2+1) \delta_{a_1+a_2, a_3},
\\
&\sfC\big[\sfm\big]^{a_1}_{a_2,a_3}=\delta_{a_1, a_2+a_3+1}
+ \left(\binom{2a_2+2a_3-2a_1+1}{2a_2+1}
+ \binom{2a_2+2a_3-2a_1+1}{2a_3+1}\right)
\\
&\phantom{\sfC\big[\sfm\big]^{a_1}_{a_2,a_3}=}{}\times \zeta(2a_2+2a_3-2a_1+2),
\\
&\sfD\big[\sfm\big]^{a_1}=\frac{1}{8} \delta_{a_1, 0}.
\end{align*}
By Proposition \ref{prop:givental_act}, the twisted differential operators
\begin{align*}
\widehat{L}_k^{\mathrm{SMV}}=
\widehat{L}^{\mathrm{B}}\big[\sfm\big]_k=
\widehat{U}^{\mathrm{MV}}
\widehat{L}_k^{\mathrm{B}}
\big(\widehat{U}^{\mathrm{MV}}\big)^{-1},
\end{align*}
of the Bessel Virasoro operators $\widehat{L}_k^{\mathrm{B}}$ in equation~\eqref{vir_be}
by the twist \eqref{mv_twist_vir_u} also satisfy the Virasoro relations
\begin{align*}
\big[\widehat{L}_k^{\mathrm{SMV}}, \widehat{L}_{\ell}^{\mathrm{SMV}}\big]
= (k-\ell ) \widehat{L}_{k+\ell}^{\mathrm{SMV}},
\qquad
k, \ell \ge 0,
\end{align*}
and provide the Virasoro constraints
\begin{align*}
\widehat{L}_k^{\mathrm{SMV}} Z^{\mathrm{SMV}}(\hbar; \textbf{t})=0,\qquad
k \ge 0,
\end{align*}
for the generating function of the super Masur--Veech polynomials
\begin{align}
\log Z^{\mathrm{SMV}}(\hbar; \textbf{t})={}&
\log Z^{\mathrm{B}}\big[\sfm\big](\hbar; \textbf{t})=
\log \big(\widehat{U}^{\mathrm{MV}} Z^{\mathrm{B}}(\hbar; \textbf{t})\big)\nonumber
\\
={}&
\sum_{g \ge 0} \hbar^{2g-2} F^{\mathrm{SMV}}_{g}(\textbf{t})
=
\sum_{g \ge 0, n\ge 1} \hbar^{2g-2}
\mathop{\sum_{a_1,\dots,a_n \ge 0}}\limits_{|\mathbf{a}| \le g-1}
F^{\mathrm{SMV}(g)}_{a_1,\dots,a_n}
\frac{t_{a_1}\cdots t_{a_n}}{n!}.
\label{supermv_pf}
\end{align}
By Proposition \ref{prop:cut_join_be}, the super Masur--Veech generating function
is computed by
\begin{align*}
Z^{\mathrm{SMV}}(x \hbar; x \textbf{t})={}&
\widehat{U}^{\mathrm{MV}} Z^{\mathrm{B}}(x \hbar; x \textbf{t})=
\widehat{U}^{\mathrm{MV}} \e^{x \widehat{W}^{\mathrm{B}}}\cdot 1,
\end{align*}
and we obtain the first few volume coefficients such that
\begin{align}
F^{\mathrm{SMV}}_{0}(\textbf{t})={}&F^{\mathrm{B}}_{0}(\textbf{t})=0,
\qquad
F^{\mathrm{SMV}}_{1}(\textbf{t})=F^{\mathrm{B}}_{1}(\textbf{t})
=-\frac18 \log (1-t_0),\nonumber
\\
F^{\mathrm{SMV}}_{2}(\textbf{t})
={}& \frac{9t_1}{128(1-t_0)^3}
+ \frac{3\pi^2}{256}\left(\frac{1}{(1-t_0)^2}-1\right)\nonumber
\\
={}&\left(\frac{3 \pi^{2} t_{0}}{128} + \frac{9 t_{1}}{128}\right)
+ \left(\frac{9}{256} \pi^{2} t_{0}^{2} + \frac{27}{128} t_{1} t_{0}\right)
+ \left(\frac{3}{64} \pi^{2} t_{0}^{3}+\frac{27}{64} t_{1} t_{0}^{2}\right)
+ \cdots,\nonumber
\\
F^{\mathrm{SMV}}_{3}(\textbf{t})
={}& \frac{225t_2}{1024(1-t_0)^5} + \frac{567t_1^2}{1024(1-t_0)^6}
+ \frac{153 \pi^2 t_1}{2048(1-t_0)^5}
+ \frac{23 \pi^4}{4096}\left(\frac{1}{(1-t_0)^4}-1\right)\nonumber
\\
={}& \left(\frac{23}{1024} \pi^{4} t_{0}+\frac{153}{2048} \pi^{2} t_{1}+ \frac{225}{1024} t_{2}\right)\nonumber
\\
&
+ \left(\frac{115}{2048} \pi^{4} t_{0}^{2}+\frac{765}{2048} \pi^{2} t_{1} t_{0}+
\frac{1125}{1024} t_{2} t_{0}+\frac{567}{1024} t_{1}^{2}\right)\nonumber
\\
&
+ \left(\frac{115}{1024}\pi^{4} t_{0}^{3} +\frac{2295}{2048} \pi^{2} t_{1} t_{0}^{2}+
\frac{3375}{1024} t_{2} t_{0}^{2}+\frac{1701}{512} t_{1}^{2} t_{0}\right)
+ \cdots.
\label{feg_comp_mv_super}
\end{align}

Proposition \ref{prop:shift_bgw} leads to the following proposition.

\begin{prop}\label{prop:shift_bgw_tw}
The generating function of the twisted BGW volume coefficients,
\begin{align*}
\log Z^{\mathrm{BGW}}\big[\sfm\big](\hbar; \textbf{t})=
\sum_{g \ge 0, n\ge 1} \hbar^{2g-2}
\mathop{\sum_{a_1,\dots,a_n \ge 0}}\limits_{|\mathbf{a}| \le g-1}
F^{\mathrm{BGW}(g)}\big[\sfm\big]_{a_1,\dots,a_n}
\frac{t_{a_1}\cdots t_{a_n}}{n!},
\end{align*}
is obtained by shifting the variables $t_a$ as equation~\eqref{shift_bgw} in the super Masur--Veech $($i.e., twisted Bessel$)$ generating function \eqref{supermv_pf}.
\end{prop}

\subsubsection[Twisted (2,2p-2) minimal superstring volume polynomials]{Twisted $\boldsymbol{(2,2p-2)}$ minimal superstring volume polynomials}\label{subsec:vir_mg_super_twisted}

The generating function of twisted volume coefficients for the $(2,2p-2)$ minimal superstring,
\begin{align*}
\log Z^{\mathrm{SM}(p)}\big[\sfm\big](\hbar; \textbf{t})={}&
\sum_{g \ge 0} \hbar^{2g-2} F^{\mathrm{SM}(p)}_{g}\big[\sfm\big](\textbf{t}),
\\
={}&
\sum_{g \ge 0, n\ge 1} \hbar^{2g-2}
\mathop{\sum_{a_1,\dots,a_n \ge 0}}\limits_{|\mathbf{a}| \le g-1}
F^{\mathrm{SM}(p) (g)}\big[\sfm\big]_{a_1,\dots,a_n}
\frac{t_{a_1}\cdots t_{a_n}}{n!},
\end{align*}
is obtained by shifting the variables $t_a$ in the super Masur--Veech
generating function \eqref{supermv_pf} as
\begin{align}
t_a
 \to
t_a - \frac{\bigl(-2\pi^2 s\bigr)^a}{(2a+1)!! a!}
\prod_{i=1}^{a}\left(1-\frac{(2i-1)^2}{p^2}\right)
\qquad
\textrm{for}\quad 1 \le a \le \frac{p-1}{2},
\label{shift_minimal_grav_super_s}
\end{align}
where a deformation parameter $s$ is introduced.

Adopting the shift \eqref{shift_minimal_grav_super_s} to the super Masur--Veech generating functions \eqref{feg_comp_mv_super},
we find the first few volume coefficients such that
\begin{align*}
F^{\mathrm{SM}(p)}_{0}\big[\sfm\big](\textbf{t})={}&F^{\mathrm{B}}_{0}(\textbf{t})=0,
\qquad
F^{\mathrm{SM}(p)}_{1}\big[\sfm\big](\textbf{t})=F^{\mathrm{B}}_{1}(\textbf{t})
=-\frac18 \log (1-t_0),
\\
F^{\mathrm{SM}(p)}_{2}\big[\sfm\big](\textbf{t})
={}& \frac{9t_1}{128(1-t_0)^3}
+ \left(\frac{3\pi^2}{64(1-t_0)^3}-\frac{3\pi^2}{64}\right)s P_1
+ \frac{3\pi^2}{256(1-t_0)^2}-\frac{3\pi^2}{256}
\\
={}& \left(\left(\frac{9 s t_{0}}{64}-\frac{9 s t_{0}}{64 p^{2}}+\frac{3 t_{0}}{128}\right) \pi^{2}+\frac{9 t_{1}}{128}\right)
\\
&+ \left(\left(\frac{9 s t_{0}^{2}}{32}-\frac{9 s t_{0}^{2}}{32 p^{2}}+\frac{9 t_{0}^{2}}{256}\right) \pi^{2}+\frac{27 t_{1} t_{0}}{128}\right)
+ \cdots,
\\
F^{\mathrm{SM}(p)}_{3}\big[\sfm\big](\textbf{t})
={}& \frac{225t_2}{1024(1-t_0)^5} + \frac{567t_1^2}{1024(1-t_0)^6}
+ \frac{189 \pi^2 t_{1} s P_1}{256(1-t_{0})^{6}}
+ \frac{153 \pi^2 t_1}{2048(1-t_0)^5}
\\
&+ \frac{23 \pi^4}{4096}\left(\frac{1}{(1-t_0)^4}-1\right)
+ \frac{3\pi^4 s^2 P_1}{512}
\left(
\frac{42P_1-5P_2+5 t_{0} P_2}{(1-t_{0})^{6}}
\right.
\\
&\left.- 42P_1+5P_2 \right)+ \frac{51 \pi^4 s P_1}{1024}\left(\frac{1}{(1-t_0)^5}-1\right)
\\
={}&
\left(\left(\frac{681 s^{2} t_{0}}{512}-\frac{381 s^{2} t_{0}}{256 p^{2}}+\frac{81 s^{2} t_{0}}{512 p^{4}}+\frac{255 s t_{0}}{1024}-\frac{255 s t_{0}}{1024 p^{2}}+\frac{23 t_{0}}{1024}\right) \pi^{4}
\right.
\\
&
\left.+\left(\frac{189 s t_{1}}{256}-\frac{189 s t_{1}}{256 p^{2}}
+\frac{153 t_{1}}{2048}\right) \pi^{2}+\frac{225 t_{2}}{1024}\right)
\\
&
+
\left(\left(\frac{2421 s^{2} t_{0}^{2}}{512}-\frac{1521 s^{2} t_{0}^{2}}{256 p^{2}}+\frac{621 s^{2} t_{0}^{2}}{512 p^{4}}+\frac{765 s t_{0}^{2}}{1024}-\frac{765 s t_{0}^{2}}{1024 p^{2}}
+\frac{115 t_{0}^{2}}{2048}\right) \pi^{4}\right.\\
&\left.+\left(\frac{567 s t_{1} t_{0}}{128}-\frac{567 s t_{1} t_{0}}{128 p^{2}}+\frac{765 t_{1} t_{0}}{2048}\right) \pi^{2}+\frac{1125 t_{2} t_{0}}{1024}+\frac{567 t_{1}^{2}}{1024}\right)
+ \cdots,
\end{align*}
where $P_1= 1-1/p^2$ and $P_2= 1-9/p^2$.
These volume coefficients reduce to the super Masur--Veech volume coefficients \eqref{feg_comp_mv_super} at $p=1$
and the twisted super Weil--Petersson volume coefficients at $p=\infty$.

\appendix

\section[Physical derivations of spectral curves and the Masur--Veech twist]{Physical derivations of spectral curves\\ and the Masur--Veech twist}\label{sec:physical}

In this appendix, we discuss a physical derivation of spectral curves for the JT gravity and the~$(2,p)$ minimal string
as well as their supersymmetric generalizations.
And we also discuss a~physical interpretation of
the Masur--Veech type twist in terms of the JT gravity.

\subsection{JT gravity and Weil--Petersson volume}\label{sec:JT_WP}

Here we summarize basic results of the JT gravity in \cite{Saad:2019lba} which is necessary for our physical derivation of the Masur--Veech type twist.

The JT gravity is the two-dimensional dilaton gravity which appears in a model of $\mathrm{AdS}_2/\mathrm{CFT}_1\!$ correspondence.
In this gravity theory, the dilaton function plays a role of the Lagrange multiplier setting a hyperbolic constraint $R=-2$ for the Ricci curvature of a two-dimensional surface.

From the physical duality conjecture of the $\mathrm{AdS_2/CFT_1}$ correspondence, it is found that the partition function $Z_{0,1}^{\mathrm{JT}}(\beta)$ of the JT gravity on the Euclidean $\mathrm{AdS}_2$ homeomorphic to a hyperbolic disk with a wiggly boundary is dual to the thermal partition function \smash{$\langle Z(\beta)\rangle=\langle\mathrm{e}^{-\beta H_{\mathrm{SYK}}}\rangle$} of the Sachdev--Ye--Kitaev (SYK) model on the boundary circle in the low energy limit which is described by the Schwarzian theory.
On the boundary of the disk parametrized by a proper length coordinate $u$ in Figure \ref{fig:disk_wiggly} (left),
\begin{figure}[t]\centering
 \includegraphics[width=90mm]{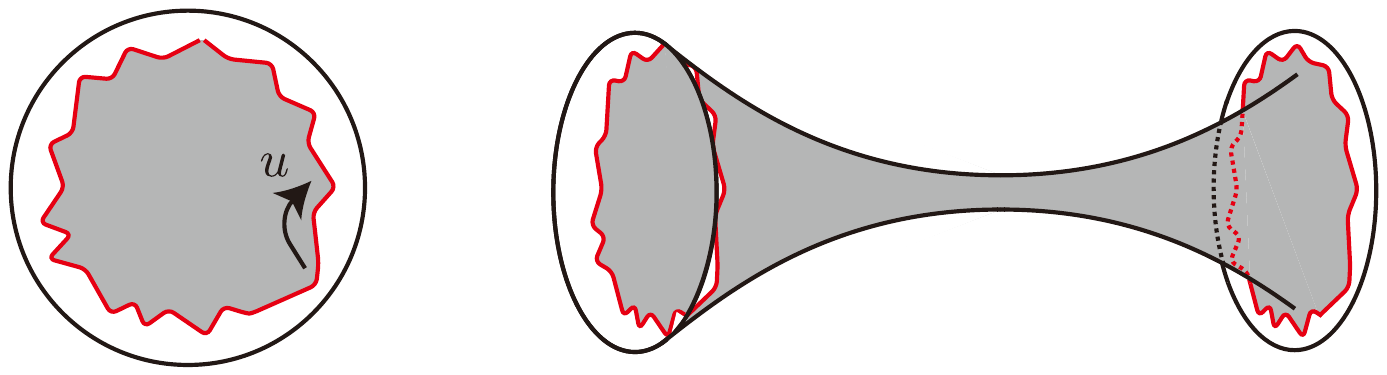}
\caption{Left: A disk with a wiggly boundary parametrized by a proper length coordinate $u$.
Right: A double trumpet with two wiggly boundaries.}\label{fig:disk_wiggly}
\end{figure}
the metric $g_{uu}$ and the dilaton field $\phi$ of the JT gravity obey the \emph{wiggly} boundary conditions below with a parameter $\gamma$:
\begin{align}
\label{eq:wiggly}
g_{uu}\big|_{\mathrm{bdy}}=\frac{1}{\epsilon^2},\qquad \phi\big|_{\mathrm{bdy}}=\frac{\gamma}{\epsilon},\qquad \epsilon\to 0,
\end{align}
and it is necessary to perform the path integral over the boundary graviton modes to get the partition function of the JT gravity.
In \cite{Stanford:2017thb}, a direct computation of the disk partition function~$Z_{0,1}^{\mathrm{JT}}(\beta)$ is performed, and the following striking formula is obtained
\begin{align*}
Z_{0,1}^{\mathrm{JT}}(\beta)=
\frac{\gamma^{\frac{3}{2}}}{(2\pi)^{\frac{1}{2}}\beta^{\frac{3}{2}}}
\mathrm{e}^{\frac{2\pi^2\gamma}{\beta}}.
\end{align*}

Further evidence of the duality conjecture is observed for the partition function of the JT gravity on the hyperbolic \emph{double trumpet} with wiggly boundaries which is homeomorphic to the cylinder in Figure \ref{fig:disk_wiggly} (right).
The double trumpet partition function of the JT gravity is shown to be dual to the spectral form factor $\langle Z(\beta+\mathrm{i}T)Z(\beta-\mathrm{i}T)\rangle$ of the SYK model in the ramp region.
\begin{figure}[t]\centering
 \includegraphics[width=95mm]{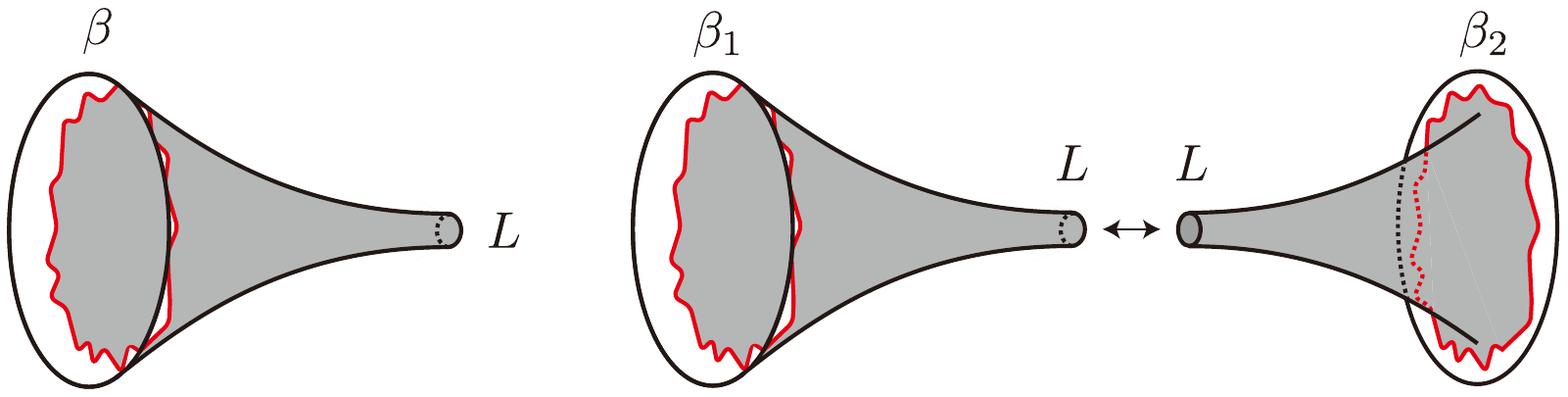}
\caption{Left: A hyperbolic trumpet with a wiggly boundary and a geodesic boundary
of length $L$.
Right: A double trumpet made of two hyperbolic trumpets.}\label{fig:trumpet}
\end{figure}

The double trumpet is divided into two hyperbolic trumpets by cutting along the waist curve (see Figure \ref{fig:trumpet}).
Each hyperbolic trumpet is also homeomorphic to a cylinder, and
one of two boundaries is the wiggly boundary which obeys the boundary condition \eqref{eq:wiggly} for the JT gravity fields.
Another boundary is the geodesic boundary, and we choose its length to be $L$.
In \cite{Saad:2019lba,Stanford:2019vob}, a striking formula of the trumpet partition function $Z_{\mathrm{trumpet}}^{\mathrm{JT}}(\beta,L)$ is obtained
\begin{align*}
Z_{\mathrm{trumpet}}^{\mathrm{JT}}(\beta,L)
=\frac{\gamma^{\frac{1}{2}}}{(2\pi)^{\frac{1}{2}}\beta^{\frac{1}{2}}}
\mathrm{e}^{-\frac{\gamma}{2}\frac{L^2}{\beta}}.
\end{align*}
Gluing two hyperbolic trumpets along the geodesic boundaries, one obtains the double trumpet, and
the gluing formula for the double trumpet partition function $Z_{0,2}^{\mathrm{JT}}(\beta,L)$ is
\begin{align}
\label{eq:double_trumpet}
Z_{0,2}^{\mathrm{JT}}(\beta_1,\beta_2)=
\int_{\IR_+} Z^{\mathrm{JT}}_{\mathrm{trumpet}}(\beta_1,L)
Z^{\mathrm{JT}}_{\mathrm{trumpet}}(\beta_2,L) L{\rm d}L
=\frac{\sqrt{\beta_1\beta_2}}{2\pi(\beta_1+\beta_2)}.
\end{align}

The gluing formula is generalized to the genus $g$ partition function with $n$ boundaries in Figure \ref{fig:multi_boundaries}.
\begin{figure}[t]\centering
 \includegraphics[width=50mm]{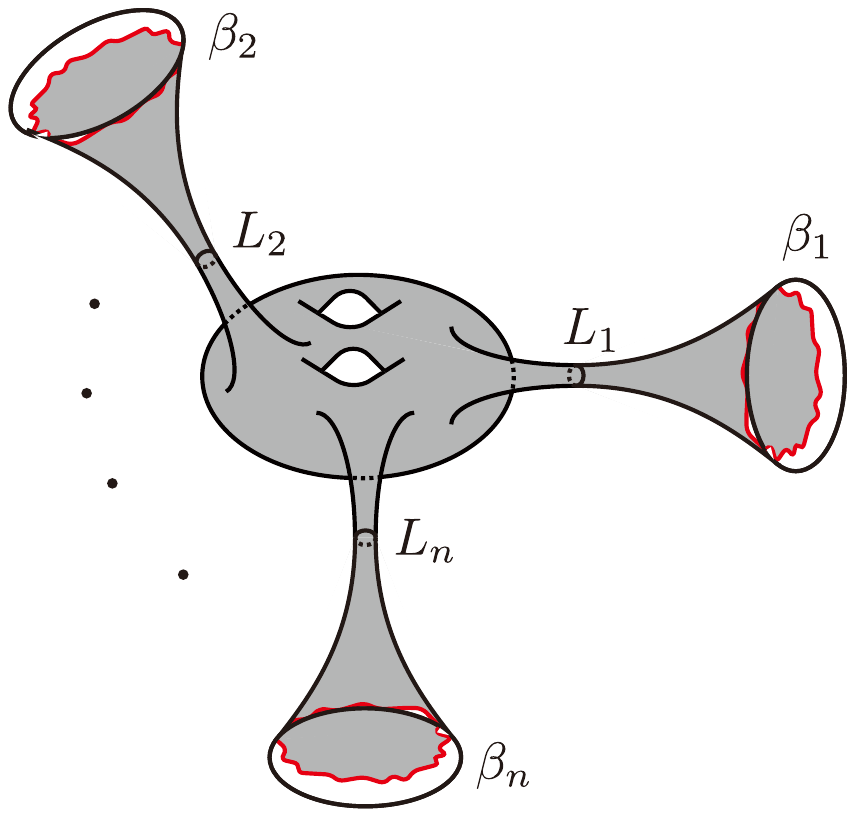}
\caption{A hyperbolic Riemann surface with genus $g$ and $n$ wiggly boundaries, which is constructed
by gluing $n$ hyperbolic trumpets to the bordered hyperbolic Riemann surface.}\label{fig:multi_boundaries}
\end{figure}
The JT gravity partition function $Z_{g,n}^{\mathrm{JT}}(\beta_1,\dots,\beta_n)$ of a genus $g$ hyperbolic bordered Riemann surface with $n$ wiggly boundaries obeys
\begin{align}
\label{eq:JT_partition_function}
Z_{g,n}^{\mathrm{JT}}(\beta_1,\dots,\beta_n)
=\int_{\IR_+^n}
\left(\prod_{i=1}^nZ^{\mathrm{JT}}_{\mathrm{trumpet}}(\beta_i,L_i)\right)
V_{g,n}^{\mathrm{WP}}(L_1,\dots,L_n)
\prod_{i=1}^n L_i{\rm d}L_i,
\end{align}
where $V_{g,n}^{\mathrm{WP}}(L_1,\dots,L_n)$ denotes the Weil--Petersson volume of genus $g$ bordered Riemann surface with boundary lengths $L_1,\dots,L_n$.
In the path integral of the JT gravity, the Weil--Petersson volume arises from the path integral with respect to the metric $g_{\mu\nu}$ and the dilaton field $\phi$ on the bulk of bordered Riemann surfaces.

{\bf Weil--Petersson spectral curve from JT gravity.} The JT gravity partition function~\eqref{eq:JT_partition_function} is directly related to the connected correlation function $W_{g,n}^{\mathrm{WP}}(z_1,\dots,z_n)$
of the CEO topological recursion for the Weil--Petersson volume.
Here we will focus on the derivation of the basic data, the $\sfy$-coordinate function in equation~\eqref{sp_curve_mir} and
the bidifferential $B$ in equation~\eqref{bergman} of
the Weil--Petersson spectral curve
$\cC^{\mathrm{WP}}=\big({\IP}^1;\sfx,\sfy^{\mathrm{WP}},B\big)$
from the disk and double trumpet partition functions.

The $\sfy$-coordinate function is found from the disk partition function $Z_{0,1}^{\mathrm{JT}}(\beta)$ rewritten in the form
\begin{align*}
Z_{0,1}^{\mathrm{JT}}(\beta)
=\int_{\IR_+}
\rho_0^{\mathrm{JT}}(E) \mathrm{e}^{-\beta E} {\rm d}E,
\end{align*}
where $\rho_0^{\mathrm{JT}}(E)$ denotes the genus zero density of states
\begin{align*}
\rho_0^{\mathrm{JT}}(E)=\frac{\gamma}{2\pi^2}\sinh\big(2\pi\sqrt{2\gamma E}\big).
\end{align*}
By a change of variable $E=-z^2$ and the analytic continuation, one finds
the $\sfy$-coordinate function \cite{Saad:2019lba}
\begin{align*}
\sfy^{\mathrm{JT}}(z)=
-\pi \mathrm{i} \rho_0^{\mathrm{JT}}(E)=
\frac{\gamma}{2\pi}\sin\big(2\pi\sqrt{2\gamma}z\big).
\end{align*}
Putting $\gamma=1/2$, we find the $\sfy$-coordinate function $\sfy^{\mathrm{WP}}(z)$ in equation~\eqref{sp_curve_mir} up to an overall constant factor.

We now use a formula which relate the correlation function $W_{g,n}^{\mathrm{WP}}(z_1,\dots,z_n)$ for $2g+2-n>0$ of the CEO topological recursion for the Weil--Petersson volume
and the JT gravity partition function $Z_{g,n}^{\mathrm{JT}}(\beta_1,\dots,\beta_n)$ with $\gamma=1/2$:
\begin{align}
\label{eq:Zgn_Wgn}
W_{g,n}^{\mathrm{WP}}(z_1,\dots,z_n)=
2^n z_1\cdots z_n \int_{\IR_+^n}
Z_{g,n}^{\mathrm{JT}}(\beta_1,\dots,\beta_n)
\mathrm{e}^{-\sum_{i=1}^n \beta_i z_i^2}
\prod_{i=1}^n {\rm d}\beta_i.
\end{align}
This formula is found from the Laplace dual relation between the Weil--Petersson volume $V^{\mathrm{WP}}(L_1,\dots,L_n)$ and the correlation function $W_{g,n}^{\mathrm{WP}}(z_1,\dots,z_n)$. Here the integral formula involving $Z_{\mathrm{trumpet}}^{\mathrm{JT}}(\beta,L)$,
\begin{align}
\label{eq:dual_trumpet}
2z \int_{\IR_+}
Z_{\mathrm{trumpet}}^{\mathrm{JT}}(\beta,L)
\mathrm{e}^{-\beta z^2} {\rm d}\beta
=2z \int_{\IR_+}
\frac{\gamma^{\frac{1}{2}}}{(2\pi)^{\frac{1}{2}}\beta^{\frac{1}{2}}}
\mathrm{e}^{-\beta z^2-\frac{\gamma}{2}\frac{L^2}{\beta}} {\rm d}\beta
=(2\gamma)^{\frac{1}{2}} \mathrm{e}^{-\sqrt{2\gamma}Lz},
\end{align}
is applied to equation~\eqref{eq:JT_partition_function}.
The bidifferential $B$ in equation~\eqref{bergman} is found by applying the above formula to equation~\eqref{eq:double_trumpet},
\begin{align}
\label{eq:Z02_W02_1}
2^2 z_1z_2
\int_{\IR_+^2}
Z_{0,2}^{\mathrm{JT}}(\beta_1,\beta_2)
\mathrm{e}^{-\beta_1z_1^2-\beta_2z_2^2} {\rm d}\beta_1{\rm d}\beta_2
=\frac{1}{(z_1+z_2)^2},
\end{align}
and this agrees with the regularized $(g,n)=(0,2)$ correlation function \cite{Eynard:2004mh} given by
\begin{align*}
B(z_1,z_2)-\frac{{\rm d}\sfx(z_1) \otimes {\rm d}\sfx(z_2)}{\left(\sfx(z_1)-\sfx(z_2)\right)^2}
=
\frac{{\rm d}z_1\otimes {\rm d}z_2}{(z_1-z_2)^2}-\frac{4z_1z_2 {\rm d}z_1\otimes {\rm d}z_2}{(z_1^2-z_2^2)^2}
=\frac{{\rm d}z_1\otimes {\rm d}z_2}{(z_1+z_2)^2}.
\end{align*}
Thus the basic data of the Weil--Petersson spectral curve $\cC^{\mathrm{WP}}=\big({\IP}^1;\sfx,\sfy^{\mathrm{WP}},B\big)$ are found from the JT gravity partition functions.

\subsection{Including a scalar field and the Masur--Veech type twist}\label{sec:Wilson_loop}

Here we introduce an extra scalar field with mass $m$ coupled to the JT gravity fields.
The partition function of the scalar field $Z_{\mathrm{scalar}}(L;\Delta)$ on the hyperbolic trumpet with a geodesic boundary of length $L$ is found in \cite{Jafferis:2022wez} via the heat kernel method,
\begin{align}
\label{eq:scalar_partition_function}
Z_{\mathrm{scalar}}(L;\Delta)
={}&\prod_{p \ge 0}\frac{1}{1-\mathrm{e}^{-L(\Delta+p)}}
=\exp\Bigg(
\sum_{w \ge 1}\frac{\mathrm{e}^{-wL\Delta}}{w\big(1-\mathrm{e}^{-wL}\big)}
\Bigg),
\end{align}
where\footnote{When the bulk scalar field is associated to an operator defined in the boundary field theory on the Riemann surface,
the quantity $\Delta$ is identified with its scaling dimension (see, e.g., \cite{Aharony:1999ti}).}
\begin{align*}
\Delta=\frac12 + \sqrt{\frac14 + m^2}.
\end{align*}
The partition function of the scalar coupled JT gravity is computed by introducing the scalar partition factor
$Z_{\mathrm{scalar}}(L;\Delta)$ for each of the closed geodesics of the hyperbolic bordered Riemann surface with wiggly boundaries in the gluing formula.
Namely, $Z_{\mathrm{scalar}}(L)$ is regarded as a twist function of the ABO topological recursion.
In particular, we focus on a sector of the scalar field path integral
in the exponential formula of equation~\eqref{eq:scalar_partition_function},
\begin{align*}
\mathsf{f}_{w}(L;\Delta):=\frac{\mathrm{e}^{-\Delta wL}}{w(1-\mathrm{e}^{-wL})}
=\frac{1}{w}\sum_{k \ge 0}\mathrm{e}^{-(k+\Delta)wL}.
\end{align*}
This sector in the scalar partition function comes from
path integral contributions of the closed world lines of the scalar field wrapping $w$ times around a closed geodesic of length $L$.
In particular, by considering the $w=1$ sector for
a massless scalar field with $\Delta=1$ ($m=0$),
we find the Masur--Veech type twist function in equation~\eqref{mv_twist},
\begin{align*}
\mathsf{f}_{w=1}(L;1)=\mathsf{f}^{\mathrm{MV}}(L)
=\frac{1}{\e^L -1}.
\end{align*}
Thus we get a physical interpretation of the Masur--Veech type twist from the physical arguments based on the JT gravity.

\begin{rem}\label{rem:ge_mv_tw_phys}
The generalized Masur--Veech type twist function $\mathsf{f}^{\mathrm{g\mathchar`-MV}}$ in equation~\eqref{eq:gen_MV_twist},
normalized by the winding number $w$, is
also found as a sector of equation~\eqref{eq:scalar_partition_function}
for the massless scalar field,
\begin{align*}
\mathsf{f}_{w}(L;1)=
\frac{1}{w} \mathsf{f}^{\mathrm{g\mathchar`-MV}}(L;w)
=\frac{1}{w(\mathrm{e}^{wL}-1)}.
\end{align*}
\end{rem}

{\bf Twisted Weil--Petersson spectral curve from scalar coupled JT gravity.} Based on the above physical interpretation in the JT gravity,
we will find the bidifferential
$B[\sff^{\mathrm{MV}}](z_1, z_2)$ in equation~\eqref{bergman_twist}
with the Masur--Veech type twist, and recover
the twisted Weil--Petersson spectral curve $\cC^{\mathrm{WP}}\big[\sfm\big]=\big({\IP}^1;\sfx,\sfy^{\mathrm{WP}},B\big[\sfm\big]\big)$.
We start from the scalar coupled JT gravity partition function $Z_{0,2}^{\mathrm{JT\mathchar`-scalar}}(\beta_1,\beta_2)$ on the double trumpet
that is found by gluing two partition functions for the hyperbolic trumpets and introducing the scalar partition function along a~simple closed geodesic with geodesic length $L$ as depicted in Figure \ref{fig:double_trumpet} \cite{Jafferis:2022wez},
\begin{figure}[t]\centering
 \includegraphics[width=65mm]{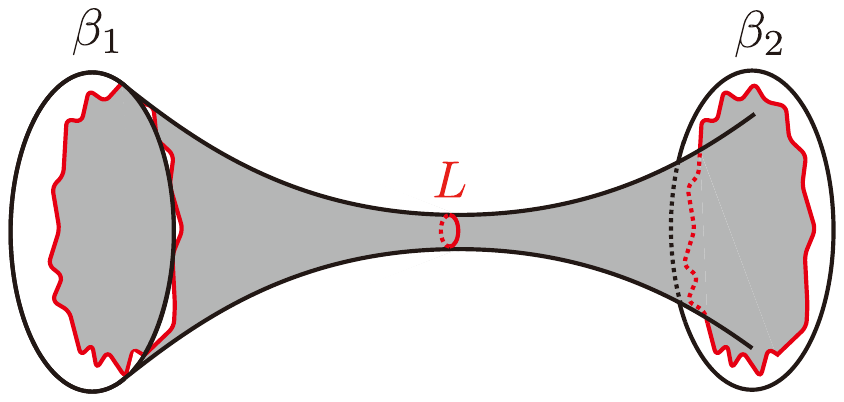}
\caption{A simple closed geodesic with geodesic length $L$ in a hyperbolic double trumpet.}\label{fig:double_trumpet}
\end{figure}
\begin{align}
\label{eq:Z02_scalar}
Z_{0,2}^{\mathrm{JT\mathchar`-scalar}}(\beta_1,\beta_2)
=\int_{\IR_+}
Z_{\mathrm{trumpet}}^{\mathrm{JT}}(\beta_1,L)
Z_{\mathrm{scalar}}(L;\Delta)
Z_{\mathrm{trumpet}}^{\mathrm{JT}}(\beta_2,L) L{\rm d}L.
\end{align}
We now replace $Z_{\mathrm{scalar}}(L;\Delta)$ by the Masur--Veech type twist function $\mathsf{f}^{\mathrm{MV}}(L)$ for the $L$-integral in the right-hand side of equation~\eqref{eq:Z02_scalar}.
The sum of the twisted and untwisted partition functions denoted by $Z_{0,2}^{\mathrm{JT}}\big[\mathsf{f}^{\mathrm{MV}}\big](\beta_1,\beta_2)$ is
\begin{align}
\label{eq:Z02_twisted}
Z_{0,2}^{\mathrm{JT}}\big[\mathsf{f}^{\mathrm{MV}}\big](\beta_1,\beta_2)
=\int_{\IR_+}
Z_{\mathrm{trumpet}}^{\mathrm{JT}}(\beta_1,L)
\left(1+\mathsf{f}^{\mathrm{MV}}(L)\right)
Z_{\mathrm{trumpet}}^{\mathrm{JT}}(\beta_2,L) L{\rm d}L.
\end{align}
Adopting equation~\eqref{eq:dual_trumpet} with $\gamma=1/2$,
for $\mathrm{Re} (z_1+z_2)>0$, we find
\begin{align*}
&
2^2 z_1z_2\int_{\IR_+^2}
Z_{0,2}^{\mathrm{JT}}\big[\mathsf{f}^{\mathrm{MV}}\big](\beta_1,\beta_2)
\mathrm{e}^{-\beta_1z_1^2-\beta_2z_2^2}
{\rm d}\beta_1 {\rm d}\beta_2
\\
&\qquad=\int_{\IR_+}
\mathrm{e}^{-L(z_1+z_2)}\left(1+\frac{1}{\mathrm{e}^L-1}\right)L{\rm d}L
=\frac{1}{(z_1+z_2)^2}
+\sum_{\mfm \ge 1} \frac{1}{(z_1+z_2+\mfm)^2}\\
&\qquad=\zeta(2;z_1+z_2),
\end{align*}
where $\zeta(2;z)$ is known as the generalized zeta function. By an analytic continuation to $z\in\mathbb{C}\setminus \IZ$, $\zeta(2;z)$ is replaced by the Hurwitz zeta function $\zeta_{\mathrm{H}}(2;z)$ in equation~\eqref{hurwitz_def}.
By changing the signature $z_2\to -z_2$ to obtain the correlation function between two points in the same branch of the double cover of the spectral curve,
we find the twisted bidifferential
$B[\sff^{\mathrm{MV}}](z_1, z_2)$ in equation~\eqref{bergman_twist}.
Thus, the basic data of the twisted Weil--Petersson spectral curve $\cC^{\mathrm{WP}}\big[\sfm\big]=\big({\IP}^1;\sfx,\sfy^{\mathrm{WP}},B\big[\sfm\big]\big)$ are also found from the physical arguments of the scalar coupled JT gravity.\looseness=1

As a generalization of the double trumpet partition function \eqref{eq:Z02_twisted}, we consider the twisted JT gravity partition function of a genus $g$ hyperbolic bordered Riemann surface with $n$ wiggly boundaries (see Figure \ref{fig:multi_boundaries_MV}) twisted by the Masur--Veech type twist function:
\begin{figure}[t]\centering
 \includegraphics[width=50mm]{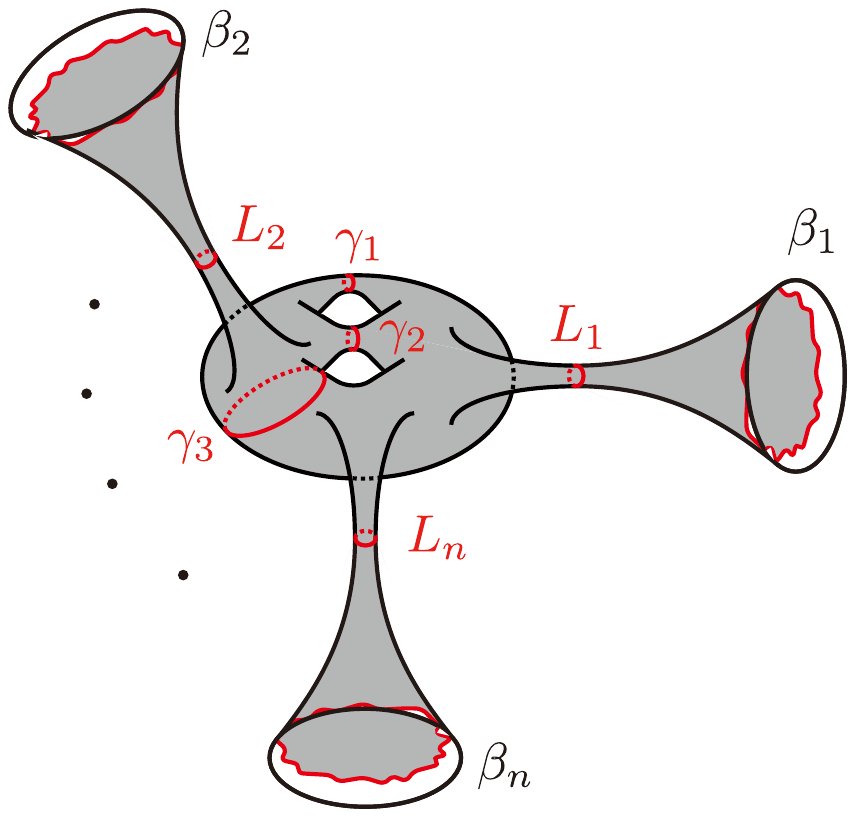}
\caption{Closed geodesics on a hyperbolic Riemann surface with genus $g$ and $n$ wiggly boundaries.}\label{fig:multi_boundaries_MV}
\end{figure}
\begin{align}
Z_{g,n}^{\mathrm{JT}}\big[\mathsf{f}^{\mathrm{MV}}\big](\beta_1,\dots,\beta_n)={}&
\int_{\IR_+^n}
\left(\prod_{i=1}^n\mathsf{f}^{\mathrm{MV}}(L_i) Z^{\mathrm{JT}}_{\mathrm{trumpet}}(\beta_i,L_i)\right)\nonumber
\\
&\times
V^{\mathrm{WP}}_{g,n}\big[\mathsf{f}^{\mathrm{MV}}\big](L_1,\dots,L_n)
\prod_{i=1}^n L_i{\rm d}L_i,\label{eq:twisted_JT}
\end{align}
where $V^{\mathrm{WP}}_{g,n}\big[\mathsf{f}^{\mathrm{MV}}\big](L_1,\dots,L_n)$ is the solution of the ABO topological recursion of the Weil--Petersson volume with the Masur--Veech type twist, i.e., the twisted Weil--Petersson volume.
The ABO topological recursion twisted by $\mathsf{f}^{\mathrm{MV}}$ is equivalent to the Laplace dual of the CEO topological recursion for the twisted Weil--Petersson spectral curve $\cC^{\mathrm{WP}}\big[\sfm\big]$.

\begin{rem}
Here we comment on a physical observation of the twist-elimination map introduced in Section \ref{sec:twisting_tr} in terms of the twisted JT gravity partition functions.
To extract the twisted Weil--Petersson volume from the twisted JT gravity partition function,
we should eliminate the twist factors $\mathsf{f}^{\mathrm{MV}}(L_i)$ and the trumpet partition functions $Z^{\mathrm{JT}}_{\mathrm{trumpet}}(\beta_i,L_i)$ in equation~\eqref{eq:twisted_JT}.
Such manipulation is geometrically interpreted as follows.
By the physical argument, similar to the relation \eqref{eq:Zgn_Wgn},
the twisted JT gravity partition function $Z_{g,n}^{\mathrm{JT}}\big[\mathsf{f}^{\mathrm{MV}}\big]$ and the twisted correlation function \smash{$W_{g,n}^{\mathrm{WP}}\big[\mathsf{f}^{\mathrm{MV}}\big]$} of the CEO topological recursion are
related by an integral transform $\mathcal{I}\left(Z_{g,n}^{\mathrm{WP}}\big[\mathsf{f}^{\mathrm{MV}}\big]\right)=W_{g,n}^{\mathrm{WP}}\big[\mathsf{f}^{\mathrm{MV}}\big]$ such that
\begin{align*}
W_{g,n}^{\mathrm{WP}}\big[\mathsf{f}^{\mathrm{MV}}\big](z_1,\dots,z_n)=
2^n z_1\cdots z_n \int_{\IR_+^n}
Z_{g,n}^{\mathrm{JT}}\big[\mathsf{f}^{\mathrm{MV}}\big](\beta_1,\dots,\beta_n)
\mathrm{e}^{-\sum_{i=1}^n \beta_i z_i^2}
\prod_{i=1}^n {\rm d}\beta_i.
\end{align*}
In the geometric picture, this integral transform $\mathcal{I}$ replaces the hyperbolic trumpet in the Riemann surface with wiggly boundaries
by the marked points
depicted as the left arrow in Figure~\ref{fig:multi_boundaries_MV3} (see \cite{Post:2022dfi} for physical discussions on this geometrical picture for this integral transform).
Furthermore, the twist-elimination map $\mathcal{E}$ eliminates the twist factors associated to the boundaries of the bordered Riemann surface
in the twisted correlation function $W_{g,n}^{\mathrm{WP}}\big[\mathsf{f}^{\mathrm{MV}}\big]$.
Finally, adopting the inverse Laplace transform $\mathcal{L}^{-1}$, we obtain the twisted Weil--Petersson volume $V_{g,n}^{\mathrm{WP}}\big[\mathsf{f}^{\mathrm{MV}}\big]$ with boundary length variables $L_i$.

\begin{figure}[t]\centering
 \includegraphics[width=115mm]{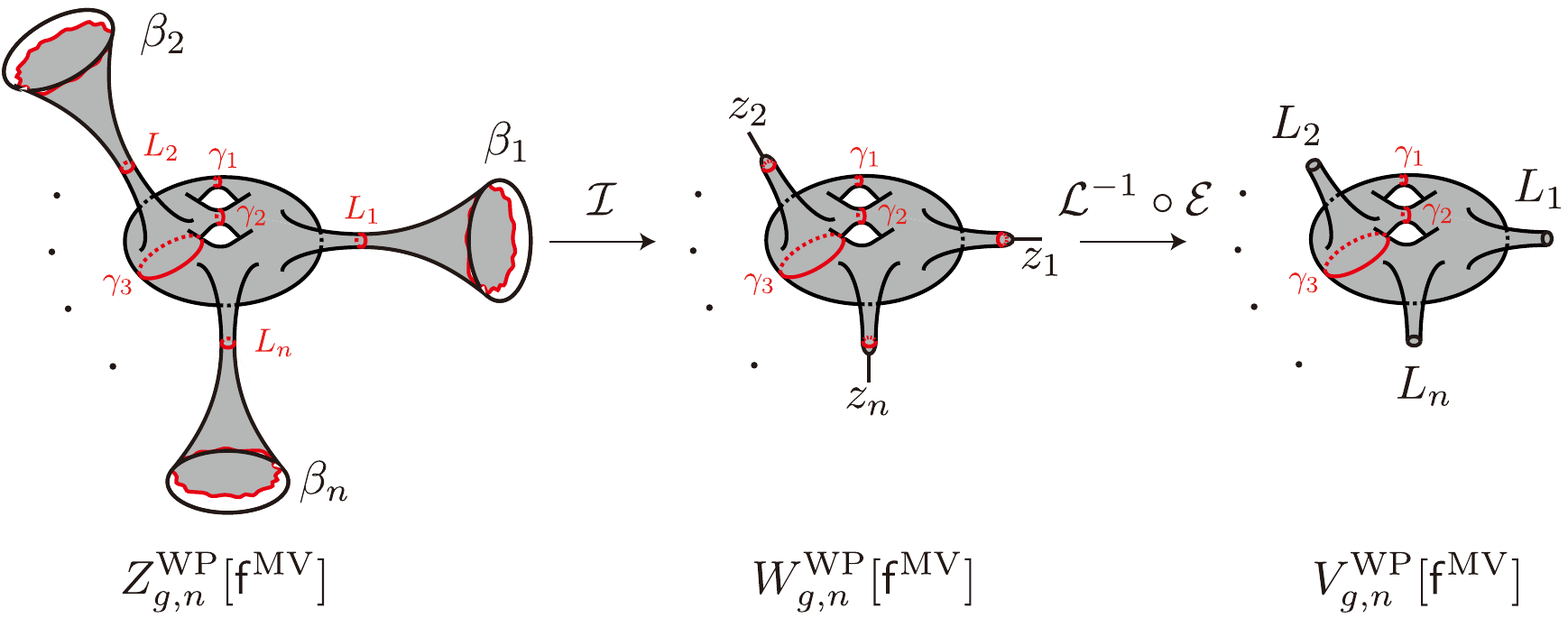}
\caption{A geometrical interpretation of the integral transform $\mathcal{I}$ and the twist-elimination map $\mathcal{E}$ involved with the inverse Laplace transform $\mathcal{L}^{-1}$.}\label{fig:multi_boundaries_MV3}
\end{figure}
\end{rem}

\subsection[Liouville gravity and (2,p) minimal string]{Liouville gravity and $\boldsymbol{(2,p)}$ minimal string}\label{sec:minimal}

Consider a matter CFT coupled to the ghost sector and Liouville CFT
which is referred to as the \emph{Liouville gravity}.
When the matter CFT is the minimal model CFT labeled by a pair of relatively prime integers $(p,p')$, the Liouville gravity yields
the $(p,p')$ \emph{minimal string}
which is a~class of non-critical string theory.
In particular, the $(2,p)$ minimal string with any odd positive integers $p$ is
found by specializing the Liouville parameter $b$ in the Liouville gravity as \smash{$b=\sqrt{\frac{2}{p}}$}.
It is pointed out in \cite{Saad:2019lba} that the partition function of the $(2,p)$ minimal string leads to
the partition function of the JT gravity in the $p\to\infty$ limit.
Detailed studies on the relation between the Liouville gravity and the JT gravity in terms of the topological expansion of the partition functions
are found in \cite{Mertens:2020hbs}.

The Liouville gravity partition function on the two-dimensional surface with a simple boundary condition \cite{Fateev:2000ik,Teschner:2000md}
which fixes the boundary cosmological constant $\mu_{\mathrm{B}}$ is studied in \cite{Edwards:1991jx,Moore:1991ag,Seiberg:2003nm}.
Such a boundary condition is known as the FZZT boundary condition.
The inverse Laplace transform of the Liouville gravity partition function with the FZZT boundary conditions gives
the correlation function of macroscopic loop operators of the Liouville gravity, which is the Liouville gravity partition function of the two-dimensional surface with fixed boundary lengths.
In \cite{Saad:2019lba}, it is pointed out that the JT gravity partition function on the hyperbolic surface with wiggly boundaries
coincides with the Liouville gravity partition function with fixed boundary lengths in the $b\to 0$ limit. (Detailed derivations are found in~\cite{Mertens:2020hbs}.)

The disk partition function \smash{$Z^{\mathrm{L}(b)}_{0,1}(\ell)$} of the Liouville gravity with a fixed boundary length $\ell$ is given by the integral transform of the density of states \smash{$\rho_0^{\mathrm{L}(b)}(E)$} \cite{Moore:1991ir} as%
\footnote{The disk partition function of the Liouville gravity in the JT gravity notation is found such as \cite[equation~(150)]{Saad:2019lba} and \cite[equation~(3.11)]{Mertens:2020hbs}.}
\begin{align*}
&Z^{\mathrm{L}(b)}_{0,1}(\ell)\sim
\int_{\kappa}^{\infty} \mathrm{e}^{-\ell E} \rho_0^{\mathrm{L}(b)}(E) {\rm d}E,
\qquad
\rho_0^{\mathrm{L}(b)}(E)=\frac{1}{4\pi^2}\sinh\left(\frac{1}{b^2}\mathrm{arccosh}\left(\frac{E}{\kappa}\right)\right),
\end{align*}
where $\sim$ implies a multiplication of a constant factor $N(b)$ and a bulk cosmological factor $\mu^{1/(2b^2)}$ (cf.\ \cite[equation (3.9)]{Mertens:2020hbs}),
and the parameter $\kappa$ is given by $\kappa^2=\mu/\sin(\pi b^2)$.
By rescaling the energy $E$ and the boundary length $\ell$ as\footnote{The boundary length $\ell$ of the Liouville gravity is identified as the Euclidean time of the JT gravity in this coincidence. Due to this identification, we use the inverse temperature $\beta$ for the boundary length $\ell$ in the Liouville gravity partition function.}
\begin{align*}
E=\kappa\big(1+2\pi^2b^4E_{\mathrm{JT}}\big),
\qquad \ell=\frac{\beta}{2\pi^2\kappa b^4},
\end{align*}
we obtain the partition function of the JT gravity in the $b\to 0$ limit.
By an analytic continuation of the genus zero density of states \smash{$\rho_0^{\mathrm{L}(b)}(E)$} with $E=\kappa\big(1-2\pi^2b^4z^2\big)$, we find the $\sfy$-coordinate function
\begin{align*}
\sfy^{\mathrm{L}(b)}(z)=
-\pi \mathrm{i} \rho_0^{\mathrm{L}(b)}(E)=
\frac{1}{2\pi}\sin\left(\frac{1}{b^2}\mathrm{arccos}\big(1-2\pi^2b^4z^2\big)\right),
\end{align*}
which defines the Liouville gravity spectral curve $\cC^{\mathrm{L}(b)}$.
In particular for $b=\sqrt{2/p}$, the coordinate function $\sfy^{\mathrm{L}(b)}(z)$ yields
the coordinate function $\sfy^{\mathrm{M}(p)}(z)$ in equation~\eqref{sp_curve_mg} of the~${(2,p)}$ minimal string spectral curve
$\cC^{\mathrm{M}(p)}=\big({\IP}^1;\sfx,\sfy^{\mathrm{M}(p)},B\big)$.\footnote{The coordinate function is also found in \cite{Seiberg:2003nm} from the ground ring relation \cite{Witten:1991zd} of the tachyon module.}
Now we consider two specializations of the parameter~$p$.
One specialization is the limit $p\to\infty$. In this limit, the minimal string reduces to the JT gravity,
and the coordinate functions~\eqref{sp_curve_mir}
of the Weil--Petersson spectral curve~${\cC^{\mathrm{WP}}=\big({\IP}^1;\sfx,\sfy^{\mathrm{WP}},B\big)}$ are recovered.
Another specialization is $p=1$.
In this case, the minimal string is identified with the topological gravity
whose matter CFT has central charge $c=-2$~\cite{DiFrancesco:1993cyw,Ginsparg:1991bi,Ginsparg:1993is},
and the coordinate functions \eqref{sp_curve_kw}
of the Airy spectral curve~${\cC^{\mathrm{A}}=\big({\IP}^1;\sfx,\sfy^{\mathrm{A}},B\big)}$
for the Kontsevich--Witten theory are recovered.
In \cite{Gregori:2021tvs}, the CEO topological recursion for the spectral curve $\mathcal{C}^{\mathrm{M}(p)}$ is
studied , and the non-perturbative behavior of the minimal string theory is discussed in great detail.

Next, we consider the cylinder partition function of the Liouville gravity with fixed boundary lengths $\ell_1$, $\ell_2$ to find the bidifferential of the spectral curve.
From the detailed analysis
of the two point function of the Liouville gravity, a gluing formula
for the cylinder partition function~\smash{$Z^{\mathrm{L}(b)}_{0,2}(\ell_1,\ell_2)$} is obtained in \cite{Martinec:2003ka,Mertens:2020hbs} such that
\begin{align}
\label{eq:cylinder_liouville}
Z^{\mathrm{L}(b)}_{0,2}(\ell_1,\ell_2)
=\frac{2}{\pi}\int_{\IR_+}
\tanh(\pi \lambda) K_{\mathrm{i}\lambda}(\kappa\ell_1) K_{\mathrm{i}\lambda}(\kappa\ell_2)
\lambda {\rm d}\lambda
=\frac{\sqrt{\ell_1\ell_2}}{\ell_1+\ell_2}
\mathrm{e}^{-\kappa(\ell_1+\ell_2)},
\end{align}
where $K_{s}(\ell)$ denotes the modified Bessel function of the second kind, and a formula of the Kontorovich--Lebedev transform is adopted to obtain this result (e.g., see \cite[Section 12.1]{Erdelyi}).
In the Liouville gravity,
the trumpet partition function \smash{$Z^{\mathrm{L}(b)}_{\mathrm{trumpet}}(\ell,\lambda)$}
is computed as the bulk one point function of the disk with a fixed boundary length $\ell$, and given by (see \cite[equations~(4.8) and~(7.57)]{Mertens:2020hbs}),
\begin{align}
\label{eq:trumpet_Liouville}
Z^{\mathrm{L}(b)}_{\mathrm{trumpet}}(\ell,\lambda)
=K_{\mathrm{i}\lambda}(\kappa\ell).
\end{align}
Then, the formula \eqref{eq:cylinder_liouville} is considered as an analogue of the gluing formula \eqref{eq:double_trumpet} for the double trumpet in the JT gravity.
Performing the following integral transform for the $(g,n)=(0,2)$ correlation function \smash{$Z^{\mathrm{L}(b)}_{0,2}(\ell_1,\ell_2)$} with
$E_i=\kappa\big(1-2\pi^2b^4z_i^2\big)$ $(i=1,2)$:
\begin{align}
\label{eq:W02_Liouville}
\frac{2^2 \big(2\pi^2b^4\kappa\big)^2 z_1z_2}{2\pi}
\int_{\IR_+^2}Z^{\mathrm{L}(b)}_{0,2}(\ell_1,\ell_2)
\mathrm{e}^{\ell_1E_1+\ell_2E_2} {\rm d}\ell_1{\rm d}\ell_2
=\frac{1}{(z_1+z_2)^2},
\end{align}
we obtain the same answer as the $(g,n)=(0,2)$ correlation function \eqref{eq:Z02_W02_1} of the JT gravity.
Consequently, by a change of signature $z_2\to -z_2$ of equation~\eqref{eq:W02_Liouville},
we find the bidifferential~$B$ in equation~\eqref{eq:spectral_curve_data_1} for the Liouville gravity spectral curve.
Thus, the basic data of the Liouville gravity spectral curve $\cC^{\mathrm{L}(b)}=\big({\IP}^1;\sfx,\sfy^{\mathrm{L}(b)},B\big)$
are derived from the Liouville gravity partition function, and the spectral curve \smash{$\cC^{\mathrm{L}(b)}$} reduces to the $(2,p)$ minimal string spectral curve \smash{$\cC^{\mathrm{M}(p)}=\big({\IP}^1;\sfx,\sfy^{\mathrm{M}(p)},B\big)$} in Table \ref{tab:spectral_curve} by the specialization $b=\sqrt{2/p}$.

Rescaling the parameters in the Liouville gravity such that
\begin{align*}
\lambda=\frac{L}{2\pi b^2},
\qquad \ell=\frac{\beta}{2\kappa \pi^2 b^4},
\end{align*}
we recover the JT gravity results in the $b\to 0$ limit.
Indeed the trumpet partition function~\eqref{eq:trumpet_Liouville}
and the measure factor in equation~\eqref{eq:cylinder_liouville} reduces to those of the JT gravity in this scaling limit~\cite{Mertens:2020hbs}:
\begin{align*}
\mathrm{e}^{\kappa\ell} Z^{\mathrm{L}(b)}_{\mathrm{trumpet}}(\beta,\lambda)
&\to\pi b^2\sqrt{\frac{\pi}{\beta}} \mathrm{e}^{-\frac{L^2}{4\beta}},
\qquad
\tanh(\pi\lambda) \lambda {\rm d}\lambda \to\frac{1}{4\pi^2b^4} L{\rm d}L,
\end{align*}
and the Liouville gravity partition function reduces to the JT gravity partition function.

Having obtained a set of basic data of the CEO topological recursion,
we can consider an analogue of the Weil--Petersson volume for the $(2,p)$ minimal string.
Let \smash{$V^{\mathrm{L}(b)}_{g,n}(L_1,\dots,L_n)$} be the volume polynomial
which is the Laplace dual of the solution of the CEO topological recursion for the Liouville gravity spectral curve $\cC^{\mathrm{L}(b)}$.
By the specialization $b=\sqrt{2/p}$, the volume polynomial~\smash{$V^{\mathrm{L}(b)}_{g,n}(L_1,\dots,L_n)$} reduces to the volume polynomial
\smash{$V^{\mathrm{M}(p)}_{g,n}(L_1,\dots,L_n)$} for the $(2,p)$ minimal string.\footnote{The $p$-deformed Weil--Petersson volume is computed independently on the basis of
the $(2,p)$ minimal string in \cite[Section~7]{Mertens:2020hbs}, and
we find agreements between the volume polynomials $V^{\mathrm{M}(p)}_{g,n}(L_1,\dots,L_n)$ in \eqref{lst_mst} 
and the $p$-deformed Weil--Petersson volumes
for $(g,n)=(0,4)$ and $(1,1)$ under a change of parameters as $b_i^2=L_i^2-4\pi^2/p^2$
(see version 5 of \cite{Mertens:2020hbs} in arXiv).}
In particular, for $p=\infty$ (resp.\ $p=1$),
the volume polynomial \smash{$V^{\mathrm{M}(p)}_{g,n}(L_1,\dots,L_n)$} reduces to
the Weil--Petersson volume \smash{$V^{\mathrm{WP}}_{g,n}(L_1,\dots,L_n)$}
(resp.\ the Kontsevich--Witten symplectic volume \smash{$V^{\mathrm{A}}_{g,n}(L_1,\dots,L_n)$})
which appears in the path integral for the JT gravity (resp.\ topological gravity) partition functions.
To establish the $p\to\infty$ limit, the non-perturbative study of
the~$(2,p)$ minimal string is also necessary,
and such aspect is addressed in \cite{Gregori:2021tvs}.

\subsection{JT supergravity}\label{sec:super_models}

Supersymmetric generalizations\footnote{A supersymmetric extension of the McShane identity is studied independently in \cite{Penner_super} on the basis of Bowditch's approach \cite{Bowditch} by Markoff triples.}
of Mirzakhani's recursion relation \cite{Mirz3,Mirz1} are given by Stan\-ford--Witten's work \cite{Stanford:2019vob} from the study of the JT supergravity defined on a bordered super Riemann surface.
Mathematical aspects of the volume of the moduli space of super hyperbolic surfaces and the CEO topological recursion formalism are further studied in \cite{Norbury:2020vyi}.
The super Riemann surface is a Riemann surface equipped with a spin structure, and
the partition function of the~$\mathcal{N}=1$ supersymmetric JT gravity (i.e., JT supergravity abbreviated by SJT) is given by the integral over the moduli space of Riemann surfaces involving the sum over spin structures~\cite{Stanford:2019vob}.
The topological expansion of the JT supergravity partition function with respect to the underlying Riemann surfaces
is found in the parallel way as the (bosonic) JT gravity, and
a~classification of fermionic and non-orientable extensions of the JT gravity is discussed in \cite{Stanford:2019vob}.

The basic data for the topological expansion of a fermionic/(non-)orientable JT gravity partition function are also found essentially from the disk, trumpet and double trumpet partition functions.
In the strategy of Stanford--Witten's work, the disk and trumpet partition functions of the fermionic/(non-)orientable JT gravity are studied on the basis of the dual SYK-like models.
Furthermore, the fermionic/non-orientable JT gravity partition function has a description by a~random matrix integral, and
a correspondence between the classification of the fermionic/(non-)orientable JT gravity and
the Dyson $\beta$-ensembles~\cite{Dyson:1962es}
or the Altland--Zirnbauer $(\alpha,\beta)$-ensembles~\cite{Altland:1997zz}
is established as Tables~1--4 of~\cite{Stanford:2019vob}.

The disk and trumpet partition functions for the JT supergravity are\footnote{For the disk and trumpet partition functions \smash{$Z^{\mathrm{D}}_{\mathrm{SJT}}(\beta)$}, \smash{$Z^{\mathrm{T}}_{\mathrm{SJT}}(\beta,L)$} computed in \cite[Appendix C]{Stanford:2019vob} from the boundary super Schwarzian path integrals,
a factor $1/2$ is multiplied to employ the matrix model normalization~\cite{Okuyama:2020qpm}: \smash{$
Z_{0,1}^{\mathrm{SJT}}(\beta)=\frac{1}{2}Z^{\mathrm{D}}_{\mathrm{SJT}}(\beta)$},
\smash{$Z_{\mathrm{trumpet}}^{\mathrm{SJT}}(\beta,L)=\frac{1}{2}Z^{\mathrm{T}}_{\mathrm{SJT}}(\beta,L)$}.
}
\begin{align*}
&
Z_{0,1}^{\mathrm{SJT}}(\beta)
=\sqrt{\frac{1}{2\pi\beta}} \mathrm{e}^{\frac{\pi^2}{\beta}},
\qquad
Z_{\mathrm{trumpet}}^{\mathrm{SJT}}(\beta,L)
=\frac{1}{2\sqrt{2\pi\beta}} \mathrm{e}^{-\frac{L^2}{4\beta}},
\end{align*}
where the parameters $\gamma$ and $\mathrm{e}^{S_0}$ of the super Schwarzian path integrals in \cite[Appendix C]{Stanford:2019vob} are fixed to be $1/2$
and $1$, respectively.
From the disk partition function \smash{$Z_{0,1}^{\mathrm{SJT}}(\beta)$}, as was done in the JT gravity, we obtain the coordinate functions
$\sfx(z)$ and $\sfy^{\mathrm{SWP}}(z)$ in equation~\eqref{sp_curve_sw} of
the super Weil--Petersson spectral curve $\cC^{\mathrm{SWP}}$,
where we employed a normalization given in \cite{Norbury:2020vyi}.
Gluing two trumpet partition functions along the geodesic circle of length $L$, we obtain the double trumpet function
\begin{align*}
Z_{0,2}^{\mathrm{SJT}}(\beta_1,\beta_2)
=2\int_{\IR_+} Z_{\mathrm{trumpet}}^{\mathrm{SJT}}(\beta_1,L)
Z_{\mathrm{trumpet}}^{\mathrm{SJT}}(\beta_2,L) L{\rm d}L
=\frac{\sqrt{\beta_1\beta_2}}{2\pi(\beta_1+\beta_2)},
\end{align*}
where we assume the JT supergravity on orientable surfaces without time-reversal symmetry,
and the factor $2$ in front of the middle integral arises from the sum over the spin structures.
The JT supergravity partition function on an orientable surface with genus $g$ and $n$ wiggly boundaries without time-reversal symmetry is
\begin{align}
\label{eq:partition_function_SJT}
Z^{\mathrm{SJT}}_{g,n}(\beta_1,\dots,\beta_n)
=
\int_{\IR_+^n}
\left(\prod_{i=1}^n2 Z_{\mathrm{trumpet}}^{\mathrm{SJT}}(\beta_i,L_i)\right)
V_{g,n}^{\mathrm{SW}}(L_1,\dots,L_n)
\prod_{i=1}^n L_i{\rm d}L_i,
\end{align}
where $V_{g,n}^{\mathrm{SW}}(L_1,\dots,L_n)$ denotes the Weil--Petersson volume for the moduli space of super Riemann surfaces
computed from the supersymmetric generalization of Mirzakhani's recursion in~\cite[Appendix D]{Stanford:2019vob}.

From this double trumpet partition function, we also obtain the bidifferential $B(z_1,z_2)$ of the super Weil--Petersson spectral curve $\cC^{\mathrm{SWP}}$ in the same way as the JT gravity,
and the super Weil--Petersson spectral curve $\cC^{\mathrm{SWP}}=\big({\IP}^1;\sfx,\sfy^{\mathrm{SWP}},B\big)$ is obtained.
Then, the correlation functions~\smash{$W^{\mathrm{SWP}}_{g,n}(z_1,\dots,z_n)$} are
obtained from the CEO topological recursion, and
their inverse Laplace transforms give
the volume polynomials $V_{g,n}^{\mathrm{SWP}}(L_1,\dots,L_n)$ of the moduli space of super Riemann surfaces which are the supersymmetric analogue of the Weil--Petersson volumes~\smash{$V_{g,n}^{\mathrm{WP}}(L_1,\dots,L_n)$}
and obey the ABO topological recursion.
The Weil--Petersson volume for the moduli space of super Riemann surfaces in the JT supergravity partition function \eqref{eq:partition_function_SJT} and the super Weil--Petersson volume defined in equation~\eqref{sw_volume} are related by \cite[Section 5.3]{Norbury:2020vyi}
\begin{align*}
V_{g,n}^{\mathrm{SW}}(L_1,\dots,L_n)
=(-1)^{n}2^{1-g} V_{g,n}^{\mathrm{SWP}}(L_1,\dots,L_n).
\end{align*}

\subsection{Type 0A minimal superstring}\label{sec:type0A}

It is conjectured in \cite{Okuyama:2020qpm} that the Bessel generating function $Z^{\mathrm{B}}(\hbar; \textbf{t})$ in equation~\eqref{be_pf} is obtained from the string equation \cite{Dalley:1991qg,Dalley:1992br,Morris:1990bw} for the complex matrix model which gives a non-perturbative definition of the free energy of type 0A minimal superstring \cite{Klebanov:2003wg}.%
\footnote{
Proofs of the correspondence between the (generalized) BGW free energy and Kontsevich--Witten free energy
is given in \cite{YZ21,YZ23}.}
In particular, when a finite number of variables $t_a$ is turned on such that
$t_a\ne 0$ ($1\le a\le (p-1)/2$),
the string equation gives the free energy of type 0A $(2,2p-2)$ minimal superstring.%
\footnote{
For the JT supergravity, the free energy is found from the string equation by tuning an infinite number of~$t_a$'s.
The role of the string equation in the JT supergravity is discussed in \cite{Johnson:2019eik,Johnson:2020exp,Johnson:2020heh}.
}

Heuristically we consider a spectral curve
which interpolates the super Weil--Petersson spectral curve $\cC^{\mathrm{SWP}}$ and the Bessel spectral curve
$\cC^{\mathrm{B}}$ with coordinate functions in equation~\eqref{sp_curve_be}.
We obtain the $z\sfy^{\mathrm{SWP}}(z)$ and $z\sfy^{\mathrm{B}}(z)$ by differentiating $\sfy^{\mathrm{WP}}(z)$ and $\sfy^{\mathrm{A}}(z)$
of the spectral curves for the Weil--Petersson volumes and Kontsevich--Witten symplectic volumes in the bosonic models\footnote{This property will be originated from the fact that the supersymmetric formulae for the Mirzakhani--McShane identity found in \cite{Stanford:2019vob}
allows a superfield representation with fermionic coordinates.}
\begin{align*}
\sfy^{\mathrm{SWP}}(z)=\frac{1}{z}\frac{{\rm d}\sfy^{\mathrm{WP}}(z)}{{\rm d}z},
\qquad \sfy^{\mathrm{B}}(z)=\frac{1}{z}\frac{{\rm d}\sfy^{\mathrm{A}}(z)}{{\rm d}z}.
\end{align*}
Adopting this heuristic relation to the $\sfy$-coordinate in equation~\eqref{sp_curve_mg}
of the $(2,p)$ minimal string spectral curve $\cC^{\mathrm{M}(p)}$,
we find a supersymmetric analogue of the $(2,p)$ minimal string
spectral curve $\cC^{\mathrm{SM}(p)}$ with
the coordinate functions in equation~\eqref{sp_curve_mg_super}.
The cut-and-join equation obtained from the spectral curve $\cC^{\mathrm{SM}(p)}$
is studied in Section \ref{subsec:vir_mg_super},
and we find that the generating function
$Z^{\mathrm{SM}(p)}(\hbar; \textbf{t})$ in equation~\eqref{super_mg_pf}
is obtained by the shift \eqref{shift_minimal_grav_super} for
the finite number of variables $t_a$ in
the Bessel generating function $Z^{\mathrm{B}}(\hbar; \textbf{t})$:
\begin{align*}
t_a
 \to
t_a + \gamma_a
\qquad \textrm{with} \quad
\begin{cases}
\displaystyle \gamma_a=-\frac{(-2\pi^2)^a}{(2a+1)!!a!}\prod_{i=1}^a\left(
1-\frac{(2i-1)^2}{p^2}
\right)
& \textrm{for}\ 1 \le a \le \frac{p-1}{2},
\\
\gamma_a=0
& \textrm{for others}.
\end{cases}
\end{align*}
Accordingly, we can regard
\begin{align*}
\widehat{\mathcal{B}}(\beta_1) \cdots \widehat{\mathcal{B}}(\beta_n)
Z^{\mathrm{SM}(p)}(\hbar; \textbf{t})\Big|_{t_a=0}
=
\widehat{\mathcal{B}}(\beta_1) \cdots \widehat{\mathcal{B}}(\beta_n)
Z^{\mathrm{B}}(\hbar; \textbf{t})
\Big|_{t_a=\gamma_a}
\end{align*}
with an operator
\begin{align*}
\widehat{\mathcal{B}}(\beta)=\sum_{a \ge 0} \beta^a \frac{\partial}{\partial t_a},
\end{align*}
as a correlation function of type 0A $(2,2p-2)$ minimal superstring,
where the operator $\widehat{\mathcal{B}}(\beta)$ is considered as a boundary creation operator \cite{Okuyama:2019xbv,Okuyama:2020qpm,Okuyama:2020ncd} with a variable $\beta$ associated to a~boundary.
Although we do not know the appropriate boundary condition for the disk partition function of type 0A $(2,2p-2)$ minimal superstring
which leads to the spectral curve $\cC^{\mathrm{SM}(p)}$,
from these physical observations, we expect that
the spectral curve $\cC^{\mathrm{SM}(p)}$
is obtained from the brane partition functions in type 0A $(2,2p-2)$ minimal superstring.\footnote{The boundary condition of the branes in type 0A minimal superstring should be clarified from the matrix models and super Liouville field theories (e.g., \cite{Seiberg:2003nm}).}

\section{Derivation of the ABO topological recursion data}\label{sec:Derivation_Mirz}

Originally, a derivation of the CEO topological recursion for $\cC^{\mathrm{WP}}$ is given explicitly in \cite{Eynard:2007fi} by
Eynard and Orantin
from Mirzakhani's topological recursion for the Weil--Petersson volume.
We apply their computation to the $(2,p)$ minimal string and the $(2,2p-2)$ minimal superstring in a reverse way (i.e., inverse Laplace transform of the CEO topological recursion), and derive the ABO topological recursion data given by
the kernel functions $H^{\mathrm{M}(p)}(x,y)$ in equation~\eqref{eq:H_minimal}
and $H^{\mathrm{SM}(p)}(x,y)$ in equation~\eqref{eq:H_super_minimal} explicitly.

We focus on the spectral curves $\cC^{\mathrm{M}(p)}$ and $\cC^{\mathrm{SM}(p)}$,
and start from a partially (inverse) Laplace transformed CEO topological recursion for $2g-2+n > 1$:
\begin{align}
2 \widetilde{W}_{g,n}(z,L_K)={}&
\underset{w=0}{\mathrm{Res}}\Biggl[\frac{{\rm d}w}{\sfy(w)}\frac{1}{z^2-w^2}
\Biggl(
\widetilde{Q}_{g,n}(w,w,L_K)\nonumber
\\
&
+
\sum_{m=2}^n 2 \widetilde{W}_{0,2}(w,L_m)\widetilde{W}_{g,n-1}(w,L_{K\setminus \{m\}})
\Biggr)\Biggr],\label{eq:partial_CEO}
\end{align}
where $K=\{2,\dots,n\}$.
In this recursion, $\widetilde{W}_{g,n}$ are partially inverse Laplace transformed correlation functions for $(g,n) \ne (0,2)$,
\begin{align}
\label{eq:tilde_W}
\int_{\IR_+^{n-1}} \widetilde{W}_{g,n}(z,L_K)
\mathrm{e}^{-\sum_{i=2}^n z_iL_i}
\prod_{i=2}^n L_i{\rm d}L_i
=W_{g,n}(z,z_K),
\end{align}
and for $(g,n) = (0,2)$,
\begin{align*}
\int_{\IR_+} \widetilde{W}_{0,2}(z,L_m) \e^{-z_mL_m} L_m{\rm d}L_m
=W_{0,2}(z,z_m)=\frac{1}{(z-z_m)^2},
\end{align*}
which gives
\begin{align}\label{eq:W02}
\widetilde{W}_{0,2}(z,L_m)=\mathrm{e}^{zL_m}.
\end{align}
In terms of the fully inverse Laplace transformed function $V_{g,n}$ for the correlation function $W_{g,n}$,
\begin{align*}
W_{g,n}(z,z_K)
=\int_{\IR_+^n} V_{g,n}(L,L_K) \mathrm{e}^{-zL-\sum_{i=2}^n z_iL_i}
L{\rm d}L \prod_{i=2}^n L_i{\rm d}L_i,
\end{align*}
the relation \eqref{eq:tilde_W} leads to
\begin{align*}
\widetilde{W}_{g,n}(z,L_K)=
\int_{\IR_+} V_{g,n}(L,L_K) \mathrm{e}^{-zL} L{\rm d}L.
\end{align*}
In the recursion \eqref{eq:partial_CEO},
$\widetilde{Q}_{g,n}$ is a partial Laplace transform of
$P_{g,n}$:
\begin{align*}
\widetilde{Q}_{g,n}(z,z,L_K)=
\int_{\IR_+^2} xyP_{g,n}(x,y,L_K) \mathrm{e}^{-z(x+y)} {\rm d}x{\rm d}y,
\end{align*}
where $P_{g,n}$ is the Laplace dual of $Q_{g,n}$ given by
\begin{align*}
Q_{g,n}(z,z,z_K):={}&
W_{g-1,n+1}(z,z,z_K)+\sum_{\substack{h+h'=g\\ J \sqcup J'=K}}^{\mathrm{stable}}
W_{h,1+|J|}(z,z_{J}) W_{h',1+|J'|}(z,z_{J'})
\\
={}&
\int_{\IR_+^{n+1}} xyP_{g,n}(x,y,L_K)
\mathrm{e}^{-z(x+y)-\sum_{i=2}^{n}z_iL_i}
{\rm d}x{\rm d}y \prod_{i=2}^n L_i{\rm d}L_i.
\end{align*}
In the following computations, we will rewrite equation~\eqref{eq:partial_CEO} into the recursion relation for $V_{g,n}$ to find the basic data
of the ABO topological recursion.

\subsection[Derivation for the (2,p) minimal string]{Derivation for the $\boldsymbol{(2,p)}$ minimal string}\label{sec:derivation_minimal}

We will rewrite the partially Laplace transformed CEO topological recursion \eqref{eq:partial_CEO} with the coordinate function $\sfy=\sfy^{\mathrm{M}(p)}$
in equation~\eqref{sp_curve_mg} into the form of
the Mirzakhani type ABO topological recursion \eqref{eq:Mirzakhani's} written in terms of the kernel function in equation~\eqref{eq:DR_H_bosonic} with
$H^{\mathrm{M}(p)}(x,y)$:
\begin{gather}
2\int_{\IR_+} V_{g,n}(L,L_K) \mathrm{e}^{-zL} L{\rm d}L\nonumber
\\
\qquad=
\int_{\IR_+^3}
\left(\int_0^L H^{\mathrm{M}(p)}(x+y,t) {\rm d}t\right)
xyP_{g,n}(x,y,L_K)
\mathrm{e}^{-zL} {\rm d}L{\rm d}x{\rm d}y
\nonumber \\
\qquad\phantom{}{}
+\sum_{m=2}^n \int_{\IR_+^2}
\left(\int_0^L
\left(H^{\mathrm{M}(p)}(x,t+L_m)+H^{\mathrm{M}(p)}(x,t-L_m)\right){\rm d}t \right)
\nonumber \\
\qquad \phantom{=+}{}\times
xV_{g,n-1}(x,L_{K\setminus \{m\}})
\mathrm{e}^{-zL} {\rm d}L{\rm d}x
\nonumber\\
\qquad=
\frac{1}{z} \int_{\IR_+^3}
xy H^{\mathrm{M}(p)}(x+y,t) P_{g,n}(x,y,L_K) \e^{-zt} {\rm d}t{\rm d}x{\rm d}y
\nonumber\\
\qquad \phantom{=}{}
+\sum_{m=2}^n \frac{1}{z}
\int_{\IR_+^2}
\left(H^{\mathrm{M}(p)}(x,t+L_m)+H^{\mathrm{M}(p)}(x,t-L_m)\right)
\nonumber \\
\qquad \phantom{=+}{}\times
xV_{g,n-1}(x,L_{K\setminus \{m\}}) \e^{-zt} {\rm d}t{\rm d}x.
\label{eq:Mirz_LD} \end{gather}

Firstly, we focus on the first term on the right-hand side of equation~\eqref{eq:partial_CEO},
\begin{align}
\underset{w=0}{\mathrm{Res}}
\frac{{\rm d}w}{2w\sfy^{\mathrm{M}(p)}(w)}\left(\frac{1}{z-w}-\frac{1}{z+w}\right)
\widetilde{Q}_{g,n}(w,w,L_K).
\label{fst_term_p_ceo}
\end{align}
Based on the following properties:
\begin{itemize}\itemsep=0pt
\item
$\widetilde{Q}_{g,n}(w,w,L_K)$ and $w \sfy^{\mathrm{M}(p)}(w)$ are even functions of $w$,
\item
${\rm d}w/\sfy^{\mathrm{M}(p)}(w)$ has poles at $w=u_j=(p/2\pi)\sin(j\pi/p)$
($j=0,\pm 1,\dots,\pm (p-1)/2$)
with the residue $(-1)^j\cos(\pi j/p)$ by equation~\eqref{eq:1/y},
\end{itemize}
equation~\eqref{fst_term_p_ceo} is rewritten as
\begin{gather*}
-\left[
\sum_{j=1}^{\frac{p-1}{2}}\underset{w=\pm u_j}{\mathrm{Res}}
+\underset{w=z}{\mathrm{Res}}
\right]
\frac{{\rm d}w}{w\sfy^{\mathrm{M}(p)}(w)}\frac{1}{z-w}
\widetilde{Q}_{g,n}(w,w,L_K)
\\
\qquad=
-\sum_{j=1}^{\frac{p-1}{2}}
(-1)^j\cos\left(\frac{\pi j}{p}\right)
\frac{2}{z^2-u_j^2}
\widetilde{Q}_{g,n}\left(u_j,u_j,L_K\right)
\\
\qquad \phantom{=}{}
+\frac{1}{z}\Bigg(\frac{1}{z}
+\sum_{j=1}^{\frac{p-1}{2}}
(-1)^j\cos\left(\frac{\pi j}{p}\right)
\left(\frac{1}{z-u_j}+\frac{1}{z+u_j}\right)\Bigg)
\widetilde{Q}_{g,n}(z,z,L_K)
\\
\qquad{}=
\frac{1}{z}\int_{\IR_+^2}
\Biggl(
-\sum_{j=1}^{\frac{p-1}{2}}
(-1)^j\cos\left(\frac{\pi j}{p}\right)
\frac{\mathrm{e}^{-u_j(x+y)}}{z+u_j}
+
\sum_{j=0}^{\frac{p-1}{2}}
(-1)^j\cos\left(\frac{\pi j}{p}\right)
\frac{\mathrm{e}^{-z(x+y)}}{z+u_j}
\\
\qquad \phantom{=}{}-\sum_{j=1}^{\frac{p-1}{2}}
(-1)^j\cos\left(\frac{\pi j}{p}\right)
\frac{\mathrm{e}^{-u_j(x+y)}-\mathrm{e}^{-z(x+y)}}{z-u_j}
\Biggr)
xyP_{g,n}(x,y,L_K) {\rm d}x{\rm d}y
\\
\qquad {}=
\frac{1}{z}\int_{\IR_+^2}
\Biggl(
-\sum_{j=1}^{\frac{p-1}{2}}
(-1)^j\cos\left(\frac{\pi j}{p}\right)
\left(\int_0^{\infty} \mathrm{e}^{-u_j(x+y+t)} \e^{-zt} {\rm d}t
+\int_0^{x+y} \mathrm{e}^{-u_j(x+y-t)} \e^{-zt} {\rm d}t\right)
\\
\qquad \phantom{=}{}+\sum_{j=0}^{\frac{p-1}{2}}
(-1)^j\cos\left(\frac{\pi j}{p}\right)
\left(\int_{x+y}^{\infty} \mathrm{e}^{+u_j(x+y-t)} \e^{-zt} {\rm d}t\right)
\Biggr)
xyP_{g,n}(x,y,L_K) {\rm d}x{\rm d}y.
\end{gather*}
Comparing the final expression with the first term in equation~\eqref{eq:Mirz_LD}, we find the kernel function~${H^{\mathrm{M}(p)}(x,y)}$ given in equation~\eqref{eq:H_minimal}.\footnote{The kernel function $H^{\mathrm{M}(p)}(x,y)$ in equation~\eqref{eq:H_minimal} is symmetrized under the action $y\to -y$.
The extra terms which appear in the anti-symmetrization do not contribute to the integrals in the first term in equation~\eqref{eq:Mirz_LD},
because the Heaviside functions in the extra terms vanish in these integrals.}

{\bf Consistency check.} To check the consistency of the kernel function $H^{\mathrm{M}(p)}(x,y)$
derived above, we will focus on the second term on the right-hand side of equation~\eqref{eq:partial_CEO}, and show the following relation:
\begin{align}
&
\int_{\IR_+}
\underset{w=0}{\mathrm{Res}}
\Biggl[\frac{{\rm d}w}{\sfy^{\mathrm{M}(p)}(w)}\frac{2}{z^2-w^2}
\widetilde{W}_{0,2}(w,L_m)\widetilde{W}_{g,n-1}(w,L_{K\setminus \{m\}})
\Biggr]
\mathrm{e}^{-z_mL_m} {\rm d}L_m\nonumber
\\
&\qquad=
\frac{1}{z}\int_{\IR_+^3}
\big(H^{\mathrm{M}(p)}(x,t+L_m)+H^{\mathrm{M}(p)}(x,t-L_m)\big)
xV_{g,n-1}(x,L_{K\setminus \{m\}}\nonumber
\\
&\qquad \phantom{=}{}
\times \e^{-zt-z_mL_m} {\rm d}t{\rm d}x{\rm d}L_m.\label{eq:Laplace_dual_minimal_zm}
\end{align}
Using the expression \eqref{eq:W02} of $\widetilde{W}_{0,2}(w,L_m)$,
we rewrite the left-hand side of equation~\eqref{eq:Laplace_dual_minimal_zm} as
\begin{gather}
\underset{w=0}{\mathrm{Res}}
\frac{{\rm d}w}{\sfy^{\mathrm{M}(p)}(w)}
\frac{2}{z^2-w^2}
\frac{1}{z_m-w} \widetilde{W}_{g,n-1}(w,L_{K\setminus \{m\}})\nonumber
\\
\quad=\underset{w=0}{\mathrm{Res}}
\frac{{\rm d}w}{\sfy^{\mathrm{M}(p)}(w)}
\frac{1}{z^2-w^2}
\left(\frac{1}{z_m-w}+\frac{1}{z_m+w}\right)\widetilde{W}_{g,n-1}(w,L_{K\setminus \{m\}})\nonumber
\\
\quad
=
\underset{w=0}{\mathrm{Res}}
\frac{2{\rm d}w}{\sfy^{\mathrm{M}(p)}(w)}
\frac{z_m}{(z^2-w^2)(z_m^2-w^2)} \widetilde{W}_{g,n-1}(w,L_{K\setminus \{m\}})\nonumber
\\
\quad
=
\underset{w=0}{\mathrm{Res}}
\frac{2{\rm d}w}{\sfy^{\mathrm{M}(p)}(w)}
\frac{1}{z (z_m^2-z^2)}
\left(\frac{z_m}{z-w}-\frac{z}{z_m-w}\right)
\widetilde{W}_{g,n-1}(w,L_{K\setminus \{m\}})\nonumber
\\
\quad
=
-\left[\sum_{j=1}^{\frac{p-1}{2}}\underset{w=\pm u_j}{\mathrm{Res}}
+\underset{w=z,z_m}{\mathrm{Res}}\right]
\frac{2{\rm d}w}{\sfy^{\mathrm{M}(p)}(w)}
\frac{1}{z (z_m^2-z^2)}
\left(\frac{z_m}{z-w}-\frac{z}{z_m-w}\right)
\widetilde{W}_{g,n-1}(w,L_{K\setminus \{m\}})\nonumber\\
\quad
=
\frac{2}{z (z_m^2-z^2)}\Bigg(
-\sum_{j=1}^{\frac{p-1}{2}}
(-1)^j\cos\left(\frac{j\pi}{p}\right)
\left(\frac{z_m}{z+u_j}+\frac{z_m}{z-u_j}
-\frac{z}{z_m+u_j}-\frac{z}{z_m-u_j}\right)
\nonumber \\
\quad\phantom{=}{}
\times
\widetilde{W}_{g,n-1}(u_j,L_{K\setminus \{m\}})
+\sum_{-\frac{p-1}{2} \le j \le \frac{p-1}{2}}
(-1)^j\cos\left(\frac{j\pi}{p}\right) \frac{z_m}{z-u_j}
\widetilde{W}_{g,n-1}(z,L_{K\setminus \{m\}})\nonumber
\\
\quad\phantom{=\times}{}
-\sum_{-\frac{p-1}{2} \le j \le \frac{p-1}{2}}
(-1)^j\cos\left(\frac{j\pi}{p}\right) \frac{z}{z_m-u_j}
\widetilde{W}_{g,n-1}(z_m,L_{K\setminus \{m\}})
\Bigg)\nonumber
\\
\quad=
\frac{2}{z (z_m^2-z^2)}
\int_{\IR_+}\Bigg[
-\sum_{j=1}^{\frac{p-1}{2}}(-1)^j\cos\left(\frac{j\pi}{p}\right)
\left(\frac{z_m \e^{-u_j x}}{z+u_j}-\frac{z \e^{-u_j x}}{z_m+u_j}\right)\nonumber
\\
\quad\phantom{=}{}
-\sum_{j=1}^{\frac{p-1}{2}}(-1)^j\cos\left(\frac{j\pi}{p}\right)
\left(\frac{z_m (\e^{-u_j x}-\e^{-z x})}{z-u_j}
-\frac{z (\e^{-u_j x}-\e^{-z_m x})}{z_m-u_j}\right)\nonumber
\\
\quad\phantom{=}{}
+\sum_{j=0}^{\frac{p-1}{2}}(-1)^j\cos\left(\frac{j\pi}{p}\right)
\left(\frac{z_m \e^{-z x}}{z+u_j}-\frac{z \e^{-z_m x}}{z_m+u_j}\right)
\Bigg]
xV_{g,n-1}(x,L_{K\setminus \{m\}}) {\rm d}x
\nonumber
\\
\quad=
\frac{2}{z}\int_{\IR_+}\Bigg[
-\sum_{j=1}^{\frac{p-1}{2}}(-1)^j\cos\left(\frac{j\pi}{p}\right)
\left(\int_0^{\infty} \mathrm{e}^{-u_j(x+t)}
\frac{z_m \e^{-zt} - z \e^{-z_mt}}{z_m^2-z^2} {\rm d}t\right.
\nonumber\\
\quad\phantom{=}{}
+\left.
\int_0^x \mathrm{e}^{-u_j(x-t)}
\frac{z_m \e^{-zt} - z \e^{-z_mt}}{z_m^2-z^2} {\rm d}t \right)
\nonumber\\
\quad\phantom{=}{}
+\sum_{j=0}^{\frac{p-1}{2}}(-1)^j\cos\left(\frac{j\pi}{p}\right)
\left(\int_x^{\infty} \mathrm{e}^{+u_j(x-t)}
\frac{z_m \e^{-zt} - z \e^{-z_mt}}{z_m^2-z^2} {\rm d}t\right)
\Bigg]
\nonumber \\
\quad\phantom{=+}{}
\times
xV_{g,n-1}(x,L_{K\setminus \{m\}}) {\rm d}x.
\label{eq:inverse_laplace_minimal_2}
\end{gather}
On the other hand, the right-hand side of equation~\eqref{eq:Laplace_dual_minimal_zm} is rewritten as follows \cite{Eynard:2007fi}:
\begin{gather}
\frac{1}{z}\int_{\IR_+^2}
\left(\int_{L_m}^{\infty} \e^{-z(t-L_m)} H^{\mathrm{M}(p)}(x,t) {\rm d}t
+\int_{-L_m}^{\infty} \e^{-z(t+L_m)} H^{\mathrm{M}(p)}(x,t) {\rm d}t \right)\nonumber
 \\
 \phantom{\quad=}{}
\times
xV_{g,n-1}(x,L_{K\setminus \{m\}})
 \e^{-z_mL_m} {\rm d}x{\rm d}L_m
\nonumber \\
\quad=
\frac{1}{z}\int_{\IR_+}
\Bigg(
\int_0^{\infty}
\left(\int_0^t \e^{-(z_m-z)L_m} {\rm d}L_m
+\int_0^{\infty} \e^{-(z_m+z)L_m} {\rm d}L_m \right)
\e^{-zt} H^{\mathrm{M}(p)}(x,t) {\rm d}t
\nonumber\\
 \phantom{\quad=}{}
+\int_{-\infty}^0
\left(\int_{-t}^{\infty} \e^{-(z_m+z)L_m} {\rm d}L_m\right)
\e^{-zt} H^{\mathrm{M}(p)}(x,t) {\rm d}t \Bigg)
xV_{g,n-1}(x,L_{K\setminus \{m\}}) {\rm d}x
\nonumber\\
 \phantom{\quad}{}
=
\frac{1}{z}\int_{\IR_+^2}
\left(
\frac{\mathrm{e}^{-zt}-\mathrm{e}^{-z_mt}}{z_m-z}
+\frac{\mathrm{e}^{-zt}+\mathrm{e}^{-z_mt}}{z_m+z}
\right)
H^{\mathrm{M}(p)}(x,t) xV_{g,n-1}(x,L_{K\setminus \{m\}}) {\rm d}t{\rm d}x,
\label{eq:inverse_laplace_minimal_2_2}
\end{gather}
where $H^{\mathrm{M}(p)}(x,-t)=H^{\mathrm{M}(p)}(x,t)$ is used in the last equality.
Applying the kernel function~${H^{\mathrm{M}(p)}(x,y)}$ in equation~\eqref{eq:H_minimal} to the final answer of equation~\eqref{eq:inverse_laplace_minimal_2_2},
we find the final expression of equation~\eqref{eq:inverse_laplace_minimal_2}.
Thus we derived the kernel function $H^{\mathrm{M}(p)}(x,y)$ in equation~\eqref{eq:H_minimal} of the Mirzakhani type ABO topological recursion
for the $(2,p)$ minimal string.

\subsection[Derivation for the (2,2p-2) minimal superstring]{Derivation for the $\boldsymbol{(2,2p-2)}$ minimal superstring}\label{sec:derivation_super_minimal}

In the same way as the $(2,p)$ minimal string, we will rewrite the partially Laplace transformed CEO topological recursion \eqref{eq:partial_CEO} with the coordinate function $\sfy=\sfy^{\mathrm{SM}(p)}$
in equation~\eqref{sp_curve_mg_super} into the form of the Mirzakhani type ABO topological recursion \eqref{eq:Mirzakhani's} written in terms of the kernel function in equation~\eqref{eq:DR_H_super} with $H^{\mathrm{SM}(p)}(x,y)$:
\begin{gather}
2\int_{\IR_+} V_{g,n}(L,L_K) \mathrm{e}^{-zL} L{\rm d}L\nonumber
\\
\qquad=\int_{\IR_+^3}
xy H^{\mathrm{SM}(p)}(x+y,L) P_{g,n}(x,y,L_K) \mathrm{e}^{-zL}
{\rm d}L{\rm d}x{\rm d}y\nonumber
\\
\qquad \phantom{=}{}
+\sum_{m=2}^n \int_{\IR_+^2}
\big(H^{\mathrm{SM}(p)}(x,L+L_m)+H^{\mathrm{SM}(p)}(x,L-L_m)\big)\nonumber\\
\qquad \phantom{=+}{}\times xV_{g,n-1}(x,L_{K\setminus \{m\}}) \mathrm{e}^{-zL} {\rm d}L{\rm d}x.\label{eq:Mirz_LD2}
\end{gather}

Adopting the following properties:
\begin{itemize}\itemsep=0pt
\item
$\widetilde{Q}_{g,n}(w,w,L_K)$ and $w \sfy^{\mathrm{SM}(p)}(w)$ are even functions of $w$;
\item
${\rm d}w/\sfy^{\mathrm{SM}(p)}(w)$ has poles at $w=\pm u_j'=\pm(p/2\pi)\sin((j-1/2)\pi/p)$, $(j=1,\dots, (p-1)/2)$
with the residue
$+u_j'(1/2\pi)(-1)^j\cos^2(\pi (j-1/2)/p)$ by equation~\eqref{eq:1/y_super};
\end{itemize}
the first term on the right-hand side of the partially Laplace transformed CEO topological recursion \eqref{eq:partial_CEO} is rewritten as
\begin{gather*}
-\left[
\sum_{j=1}^{\frac{p-1}{2}}\underset{w=\pm u'_j}{\mathrm{Res}}
+\underset{w=z}{\mathrm{Res}}
\right]
\frac{{\rm d}w}{z-w}\frac{1}{w\sfy^{\mathrm{SM}(p)}(w)}
\widetilde{Q}_{g,n}(w,w,L_K)
\\
\qquad=
-\sum_{j=1}^{\frac{p-1}{2}}
\frac{(-1)^j}{2\pi}\cos^2\left(\frac{\pi}{p}\left(j-\frac12\right)\right)
\left(
\frac{1}{z-u_j'}-\frac{1}{z+u'_j}
\right)
\\
\qquad \phantom{=-}{}
\times
\big(\widetilde{Q}_{g,n}(u_j',u_j',L_K)
-\widetilde{Q}_{g,n}(z,z,L_K)\big)
+ \widetilde{Q}_{g,n}(z,z,L_K) \delta_{p,1}
\\
\qquad{}
=
\int_{\IR_+^2}
\left(
\sum_{j=1}^{\frac{p-1}{2}}
\frac{(-1)^j}{2\pi}\cos^2\left(\frac{\pi}{p}\left(j-\frac12\right)\right)
\left(
\frac{1}{z+u_j'}-\frac{1}{z-u'_j}
\right)
\big(\e^{-u_j'(x+y)}-\e^{-z(x+y)}\big)\right.
\\
\left.\qquad \phantom{=}{}
+ \e^{-z(x+y)} \delta_{p,1}
\right)
xyP_{g,n}(x,y,L_K) {\rm d}x{\rm d}y
\\
\qquad {}
=
\int_{\IR_+^2}
\left(
\sum_{j=1}^{\frac{p-1}{2}}
\frac{(-1)^j}{2\pi}\cos^2\left(\frac{\pi}{p}\left(j-\frac12\right)\right)
\left(
\int_0^{\infty} \e^{-u_j'(x+y+L)} \e^{-zL} {\rm d}L
\right.\right.
\\
\left.\left.\qquad \phantom{=}{}
-\int_0^{x+y} \e^{-u_j'(x+y-L)} \e^{-zL} {\rm d}L
-\int_{x+y}^{\infty} \e^{+u_j'(x+y-L)} \e^{-zL} {\rm d}L
\right)\right.
\\
\left.\qquad \phantom{=}{}
+\delta_{p,1} \int_0^{\infty} \delta(L-x-y) \e^{-zL} {\rm d}L \right)
xyP_{g,n}(x,y,L_K) {\rm d}x{\rm d}y.
\end{gather*}
Comparing the final expression with the first term in equation~\eqref{eq:Mirz_LD2}, we find the kernel function~$H^{\mathrm{SM}(p)}(x,y)$ given in equation~\eqref{eq:H_super_minimal}.%
\footnote{The kernel function $H^{\mathrm{SM}(p)}(x,y)$ in equation~\eqref{eq:H_super_minimal} is anti-symmetrized under the action
$y\to -y$.
The extra terms which appear in the anti-symmetrization do not contribute to the integrals in the first term in equation~\eqref{eq:Mirz_LD2},
because the Heaviside functions and delta functions in the extra terms vanish in these integrals.
}

{\bf Consistency check.} As the (bosonic) minimal string, we will show the following relation
to check the consistency of the kernel function $H^{\mathrm{SM}(p)}(x,y)$:
\begin{align}
&
\int_{\IR_+}\underset{w=0}{\mathrm{Res}}
\left[\frac{{\rm d}w}{\sfy^{\mathrm{SM}(p)}(w)}\frac{2}{z^2-w^2}
\widetilde{W}_{0,2}(w,L_m)\widetilde{W}_{g,n-1}(w,L_{K\setminus \{m\}})
\right]
\mathrm{e}^{-z_mL_m} {\rm d}L_m\nonumber
\\
&\qquad=
\int_{\IR_+^3}
\big(H^{\mathrm{SM}(p)}(x,L+L_m)+H^{\mathrm{SM}(p)}(x,L-L_m)\big)
xV_{g,n-1}(x,L_{K\setminus \{m\}})\nonumber
\\
&\qquad\phantom{=}{}
\times \mathrm{e}^{-zL-z_mL_m} {\rm d}L{\rm d}x{\rm d}L_m.\label{eq:Laplace_dual_super_minimal_zm}
\end{align}
The left-hand side of equation~\eqref{eq:Laplace_dual_super_minimal_zm} is rewritten
in the similar way as equation~\eqref{eq:inverse_laplace_minimal_2}:
\begin{align}
&
-\left[\sum_{j=1}^{\frac{p-1}{2}}\underset{w=\pm u_j'}{\mathrm{Res}}
+\underset{w=z,z_m}{\mathrm{Res}}\right]
\frac{2{\rm d}w}{\sfy^{\mathrm{SM}(p)}(w)}
\frac{1}{z \big(z_m^2-z^2\big)}
\left(\frac{z_m}{z-w}-\frac{z}{z_m-w}\right)
\widetilde{W}_{g,n-1}(w,L_{K\setminus \{m\}})
\nonumber \\
&
\qquad=
\frac{2}{z_m^2-z^2}\left[
-\sum_{j=1}^{\frac{p-1}{2}}
\frac{(-1)^j}{2\pi} \cos^2\left(\frac{\pi}{p}\left(j-\frac12\right)\right)\right.
\nonumber \\
&\left. \qquad \phantom{=}{}\times
\left(-\frac{z_m}{z+u_j'}+\frac{z_m}{z-u_j'}
+\frac{z_m}{z_m+u_j'}-\frac{z_m}{z_m-u_j'}\right)
\widetilde{W}_{g,n-1}(u_j',L_{K\setminus \{m\}})\right.
\nonumber \\
&\left. \qquad \phantom{=\times}{}
+\left(z_m \delta_{p,1}+\sum_{j=1}^{\frac{p-1}{2}}
\frac{(-1)^j}{2\pi} \cos^2\left(\frac{\pi}{p}\left(j-\frac12\right)\right)
\left(\frac{z_m}{z-u_j'}-\frac{z_m}{z+u_j'}\right)\right)\right.
\nonumber \\
&\left.\qquad \phantom{=}{}
\times
\widetilde{W}_{g,n-1}(z,L_{K\setminus \{m\}})\right.
\nonumber \\
&\left.\left.\qquad \phantom{=\times}{}
-\left(z_m \delta_{p,1}+\sum_{j=1}^{\frac{p-1}{2}}
\frac{(-1)^j}{2\pi} \cos^2\left(\frac{\pi}{p}\left(j-\frac12\right)\right)
\left(\frac{z_m}{z_m-u_j'}-\frac{z_m}{z_m+u_j'}\right)\right)
\right.\right.
\nonumber \\
&\left.\qquad \phantom{ =}{}
\times\widetilde{W}_{g,n-1}(z_m,L_{K\setminus \{m\}})
\right]
\nonumber\\
&\qquad
=
\frac{2z_m}{z_m^2-z^2}\int_{\IR_+}
\left[
\sum_{j=1}^{\frac{p-1}{2}}
\frac{(-1)^j}{2\pi} \cos^2\left(\frac{\pi}{p}\left(j-\frac12\right)\right)\right.
\nonumber \\
&\left.\qquad \phantom{=}{}
\times
\left(\frac{\e^{-u_j' x} - \e^{-z x}}{z+u_j'}
-\frac{\e^{-u_j' x} - \e^{-z x}}{z-u_j'}
-\frac{\e^{-u_j' x} - \e^{-z_m x}}{z_m+u_j'}
+\frac{\e^{-u_j' x} - \e^{-z_m x}}{z_m-u_j'}
\right)\right.
\nonumber \\
&\left. \qquad \phantom{=\times}{}
+ \left(\e^{-z x}-\e^{-z_m x}\right)\delta_{p,1}
\right]
xV_{g,n-1}(x,L_{K\setminus \{m\}}) {\rm d}x
\nonumber
\\
&\qquad
=
2z_m \int_{\IR_+}
\left[
\sum_{j=1}^{\frac{p-1}{2}}
\frac{(-1)^j}{2\pi} \cos^2\left(\frac{\pi}{p}\left(j-\frac12\right)\right)
\left(\int_0^{\infty} \e^{-u_j' (x+L)}
\frac{\e^{-zL}-\e^{-z_mL}}{z_m^2 - z^2} {\rm d}L
\right.\right.
\nonumber \\
&\left.\left.\qquad \phantom{=}{}
- \int_0^{x} \e^{-u_j' (x-L)}
\frac{\e^{-zL}-\e^{-z_mL}}{z_m^2 - z^2} {\rm d}L
- \int_x^{\infty} \e^{+u_j' (x-L)}
\frac{\e^{-zL}-\e^{-z_mL}}{z_m^2 - z^2} {\rm d}L
\right)\right.
\nonumber\\
&\left.\qquad \phantom{=}{}
+ \delta_{p,1} \int_{0}^{\infty} \delta(L-x) \frac{\e^{-z L}-\e^{-z_m L}}{z_m^2 - z^2} {\rm d}L
\right]
xV_{g,n-1}(x,L_{K\setminus \{m\}}) {\rm d}x.
\label{eq:inverse_laplace_super_minimal_2}
\end{align}
On the other hand, similar to equation~\eqref{eq:inverse_laplace_minimal_2_2},
the right-hand side of equation~\eqref{eq:Laplace_dual_super_minimal_zm} is rewritten~as
\begin{align}
\label{eq:inverse_laplace_super_minimal_2_2}
2z_m \int_{\IR_+^2}
\left(
\frac{\mathrm{e}^{-zt}-\mathrm{e}^{-z_mt}}{z_m-z}
+\frac{\mathrm{e}^{-zt}-\mathrm{e}^{-z_mt}}{z_m+z}
\right)
H^{\mathrm{SM}(p)}(x,L) xV_{g,n-1}(x,L_{K\setminus \{m\}}) {\rm d}L{\rm d}x,
\end{align}
where $H^{\mathrm{SM}(p)}(x,-L)=-H^{\mathrm{SM}(p)}(x,L)$ is used.
Applying the kernel function $H^{\mathrm{SM}(p)}(x,y)$ in equation~\eqref{eq:H_super_minimal} to equation~\eqref{eq:inverse_laplace_super_minimal_2_2}, the final expression of equation~\eqref{eq:inverse_laplace_super_minimal_2} is found.
Thus,
the consistency equation~\eqref{eq:Laplace_dual_super_minimal_zm}
for the kernel function $H^{\mathrm{SM}(p)}(x,y)$ is verified, and
the CEO topological recursion for the spectral curve $\cC^{\mathrm{SM}(p)}$ is shown to be equivalent to
the Mirzakhani type ABO topological recursion with the kernel function $H^{\mathrm{SM}(p)}(x,y)$.

\section{Volume polynomials}\label{sec:table_vol}

In this appendix, we give some computational results of volume polynomials
for the 2D gravity models and their Masur--Veech type twist.

\subsection{Volume polynomials}\label{sec:table_vol_no_twist}
Weil--Petersson volumes:
\begin{gather}
\VO^{\mathrm{WP}}_{0,3}= 1,\nonumber\\
\VO^{\mathrm{WP}}_{1,1} =\frac{\pi^2}{12} + \frac{1}{48}L_1^2,\nonumber\\
\VO^{\mathrm{WP}}_{0,4}=2\pi^2 + \frac12\sum_{i=1}^4 L_i^2,\nonumber\\
\VO^{\mathrm{WP}}_{1,2} =\frac{\pi^4}{4} + \frac{\pi^2}{12}\sum_{i=1}^2 L_i^2+\frac{1}{192}\sum_{i=1}^2 L_i^4 + \frac{1}{96} L_1^2L_2^2,\nonumber\\
\VO^{\mathrm{WP}}_{0,5} = 10\pi^4 + 3\pi^2\sum_{i=1}^5 L_i^2+\frac{1}{8} \sum_{i=1}^5 L_i^4 + \frac12 \sum_{1\le i<j\le 5} L_i^2L_j^2,\nonumber\\
\VO^{\mathrm{WP}}_{1,3} =
\frac{14\pi^6}{9} + \frac{13\pi^4}{24} \sum_{i=1}^3 L_i^2
+ \frac{\pi^2}{24} \sum_{i=1}^3 L_i^4
+ \frac{\pi^2}{8} \sum_{1\le i<j\le 3} L_i^2L_j^2
+ \frac{1}{1152} \sum_{i=1}^3 L_i^6\nonumber\\
\hphantom{\VO^{\mathrm{WP}}_{1,3} =}{}
+ \frac{1}{192}
\sum_{\begin{subarray}{c}
1\le i, j\le 3\\ (i\ne j)
\end{subarray}}
L_i^2 L_j^4
+ \frac{1}{96} L_1^2L_2^2L_3^2,\nonumber\\
\VO^{\mathrm{WP}}_{2,1} =
\frac{29\pi^8}{192} + \frac{169\pi^6}{2880}L_1^2
+ \frac{139\pi^4}{23040}L_1^4 + \frac{29\pi^2}{138240}L_1^6
+ \frac{1}{442368}L_1^8.\label{lst_wp}
\end{gather}
Volume polynomials of $(2,p)$ minimal string:
\begin{gather}
 \VO^{\mathrm{M}(p)}_{0,3} =1,\nonumber\\
 \VO^{\mathrm{M}(p)}_{1,1} = \frac{\pi^2}{12}\left(1-\frac{1}{p^2}\right)+ \frac{1}{48}L_1^2,\nonumber\\
 \VO^{\mathrm{M}(p)}_{0,4} = 2\pi^2\left(1-\frac{1}{p^2}\right)+ \frac12\sum_{i=1}^4 L_i^2,\nonumber\\
 \VO^{\mathrm{M}(p)}_{1,2} =
 \frac{\pi^4}{4}\left(1-\frac{1}{p^2}\right)\left(1+\frac{5}{3p^2}\right)
+ \frac{\pi^2}{12}\left(1-\frac{1}{p^2}\right)\sum_{i=1}^2 L_i^2
+ \frac{1}{192}\sum_{i=1}^2 L_i^4 + \frac{1}{96} L_1^2L_2^2,\nonumber\\
 \VO^{\mathrm{M}(p)}_{0,5} =
 10\pi^4\left(1-\frac{1}{p^2}\right)\left(1+\frac{3}{5p^2}\right)
+ 3\pi^2\left(1-\frac{1}{p^2}\right)\sum_{i=1}^5 L_i^2
+ \frac{1}{8} \sum_{i=1}^5 L_i^4
 + \frac12 \sum_{1\le i<j\le 5} L_i^2L_j^2,\nonumber\\
 \VO^{\mathrm{M}(p)}_{1,3} =
 \frac{14\pi^6}{9}\left(1-\frac{1}{p^2}\right)
\left(1+\frac{20}{7p^2}+\frac{3}{p^4}\right)
+ \frac{13\pi^4}{24}\left(1-\frac{1}{p^2}\right)\left(1+\frac{11}{13p^2}\right)
\sum_{i=1}^3 L_i^2
\nonumber\\
 \hphantom{\VO^{\mathrm{M}(p)}_{1,3} =}{}
 + \frac{\pi^2}{24}\left(1-\frac{1}{p^2}\right) \sum_{i=1}^3 L_i^4
+ \frac{\pi^2}{8}\left(1-\frac{1}{p^2}\right) \sum_{1\le i<j\le 3} L_i^2L_j^2
+ \frac{1}{1152} \sum_{i=1}^3 L_i^6
\nonumber\\
\hphantom{\VO^{\mathrm{M}(p)}_{1,3} =}{} + \frac{1}{192}
\sum_{\begin{subarray}{c}
1\le i, j\le 3\\ (i\ne j)
\end{subarray}}
L_i^2 L_j^4
+ \frac{1}{96} L_1^2L_2^2L_3^2,\nonumber\\
 \VO^{\mathrm{M}(p)}_{2,1} =
 \frac{29\pi^8}{192}\left(1-\frac{1}{p^2}\right)
\left(1+\frac{2423}{435p^2}+\frac{41}{3p^4}+\frac{6557}{435p^6}\right)
\nonumber\\
\hphantom{\VO^{\mathrm{M}(p)}_{2,1} =}{} + \frac{169\pi^6}{2880}\!\left(1-\frac{1}{p^2}\right)\!
\left(1+\frac{430}{169p^2}+\frac{361}{169p^4}\right) \! L_1^2
+ \frac{139\pi^4}{23040}\left(1-\frac{1}{p^2}\right)\!
\left(1+\frac{93}{139p^2}\right)\!L_1^4
\nonumber\\
\hphantom{\VO^{\mathrm{M}(p)}_{2,1} =}{} + \frac{29\pi^2}{138240}\left(1-\frac{1}{p^2}\right)L_1^6
+ \frac{1}{442368}L_1^8.\label{lst_mst}
\end{gather}
Super Weil--Petersson volumes:
\begin{gather}
 \VO^{\mathrm{SWP}}_{0,n} = 0,\nonumber\\[-0.25mm]
 \VO^{\mathrm{SWP}}_{1,n} = \frac{(n-1)!}{8},\nonumber\\
 \VO^{\mathrm{SWP}}_{2,1} = \frac{9\pi^2}{64}+\frac{3}{256}L_1^2,\nonumber\\[-0.25mm]
 \VO^{\mathrm{SWP}}_{2,2} =
 \frac{9\pi^2}{16}+\frac{9}{256}\sum_{i=1}^2 L_i^2,\nonumber\\[-0.25mm]
 \VO^{\mathrm{SWP}}_{2,3} =
 \frac{45\pi^2}{16}+\frac{9}{64}\sum_{i=1}^3 L_i^2,\nonumber\\[-0.25mm]
 \VO^{\mathrm{SWP}}_{2,4} =
 \frac{135\pi^2}{8}+\frac{45}{64}\sum_{i=1}^4 L_i^2,\nonumber\\[-0.25mm]
 \VO^{\mathrm{SWP}}_{2,5} =
 \frac{945\pi^2}{8}+\frac{135}{32}\sum_{i=1}^5 L_i^2,\nonumber\\[-0.25mm]
 \VO^{\mathrm{SWP}}_{3,1} =
 \frac{681\pi^4}{512}+\frac{63\pi^2}{512}L_1^2+\frac{15}{8192}L_1^4,\nonumber\\[-0.25mm]
 \VO^{\mathrm{SWP}}_{3,2} =
 \frac{2421\pi^4}{256}+
\frac{189\pi^2}{256}\sum_{i=1}^2 L_i^2 +
\frac{75}{8192}\sum_{i=1}^2 L_i^4 + \frac{63}{2048} L_1^2L_2^2,\nonumber\\[-0.25mm]
 \VO^{\mathrm{SWP}}_{3,3} =
 \frac{19593\pi^4}{256}+
\frac{1323\pi^2}{256}\sum_{i=1}^3 L_i^2 +
\frac{225}{4096}\sum_{i=1}^3 L_i^4 + \frac{189}{1024}
\sum_{1\le i<j\le 3} L_i^2L_j^2,\nonumber\\[-0.25mm]
 \VO^{\mathrm{SWP}}_{4,1} =
 \frac{278833\pi^6}{8192} + \frac{106911\pi^4}{32768} L_1^2+
\frac{8625\pi^2}{131072} L_1^4 +\frac{175}{524288} L_1^6 .\label{lst_swp}
\end{gather}
Volume polynomials of the $(2,2p-2)$ minimal superstring:
\begin{gather}
 \VO^{\mathrm{SM}(p)}_{0,n} = 0,\nonumber\\[-0.25mm]
 \VO^{\mathrm{SM}(p)}_{1,n} = \frac{(n-1)!}{8},\nonumber\\[-0.25mm]
 \VO^{\mathrm{SM}(p)}_{2,1} =
 \frac{9\pi^2}{64}\left(1-\frac{1}{p^2}\right)
+\frac{3}{256}L_1^2,\nonumber\\[-0.25mm]
 \VO^{\mathrm{SM}(p)}_{2,2} =
 \frac{9\pi^2}{16}\left(1-\frac{1}{p^2}\right)
+\frac{9}{256}\sum_{i=1}^2 L_i^2,\nonumber\\[-0.25mm]
 \VO^{\mathrm{SM}(p)}_{2,3} =
 \frac{45\pi^2}{16}\left(1-\frac{1}{p^2}\right)
+\frac{9}{64}\sum_{i=1}^3 L_i^2,\nonumber\\[-0.25mm]
 \VO^{\mathrm{SM}(p)}_{2,4} =
 \frac{135\pi^2}{8}\left(1-\frac{1}{p^2}\right)
+\frac{45}{64}\sum_{i=1}^4 L_i^2,\nonumber\\[-0.25mm]
 \VO^{\mathrm{SM}(p)}_{2,5} =
 \frac{945\pi^2}{8}\left(1-\frac{1}{p^2}\right)
+\frac{135}{32}\sum_{i=1}^5 L_i^2,\nonumber\\[-0.25mm]
 \VO^{\mathrm{SM}(p)}_{3,1} =
 \frac{681\pi^4}{512}
\left(1-\frac{1}{p^2}\right)\left(1-\frac{27}{227p^2}\right)
+\frac{63\pi^2}{512}L_1^2\left(1-\frac{1}{p^2}\right)
+\frac{15}{8192}L_1^4,\nonumber\\[-0.25mm]
 \VO^{\mathrm{SM}(p)}_{3,2} =
 \frac{2421\pi^4}{256}
\left(1-\frac{1}{p^2}\right)\left(1-\frac{69}{269p^2}\right)
+ \frac{189\pi^2}{256}\left(1-\frac{1}{p^2}\right)\sum_{i=1}^2 L_i^2
\nonumber\\[-0.25mm]
 \hphantom{\VO^{\mathrm{SM}(p)}_{3,2} =}{}
 + \frac{75}{8192}\sum_{i=1}^2 L_i^4
+ \frac{63}{2048} L_1^2L_2^2,\nonumber\\[-0.25mm]
 \VO^{\mathrm{SM}(p)}_{3,3} =
 \frac{19593\pi^4}{256}
\left(1-\frac{1}{p^2}\right)\left(1-\frac{111}{311p^2}\right)
+ \frac{1323\pi^2}{256}\left(1-\frac{1}{p^2}\right)\sum_{i=1}^3 L_i^2
\nonumber\\[-0.25mm]
\hphantom{\VO^{\mathrm{SM}(p)}_{3,3} =}
 + \frac{225}{4096}\sum_{i=1}^3 L_i^4
+ \frac{189}{1024} \sum_{1\le i<j\le 3} L_i^2L_j^2,\nonumber\\[-0.25mm]
 \VO^{\mathrm{SM}(p)}_{4,1} =
 \frac{278833\pi^6}{8192}
\left(1-\frac{1}{p^2}\right)\left(1-\frac{44866}{278833p^2}
+\frac{1233}{278833p^4}\right)\label{lst_super_mst}
\\[-0.25mm]
\hphantom{\VO^{\mathrm{SM}(p)}_{4,1} =}{}
 + \frac{106911\pi^4}{32768}
\left(1-\frac{1}{p^2}\right)\left(1-\frac{12637}{35637p^2}\right)L_1^2
+ \frac{8625\pi^2}{131072}\left(1-\frac{1}{p^2}\right)L_1^4
+ \frac{175}{524288} L_1^6 .\nonumber
\end{gather}

\subsection{Twisted volume polynomials}\label{sec:table_vol_with_twist}

Masur--Veech polynomials:
\begin{gather}
 \VO^{\mathrm{MV}}_{0,3} = 1, \nonumber\\
 \VO^{\mathrm{MV}}_{1,1} =
 \frac{\pi^2}{12} + \frac{1}{48}L_1^2, \nonumber\\
 \VO^{\mathrm{MV}}_{0,4} =
 \frac{\pi^2}{2} + \frac12\sum_{i=1}^4 L_i^2, \nonumber\\
 \VO^{\mathrm{MV}}_{1,2} =
 \frac{\pi^4}{16} + \frac{\pi^2}{24}\sum_{i=1}^2 L_i^2
+\frac{1}{192}\sum_{i=1}^2 L_i^4 + \frac{1}{96} L_1^2L_2^2, \nonumber\\
 \VO^{\mathrm{MV}}_{0,5} =
 \frac{3\pi^4}{4} + \frac{\pi^2}{2}\sum_{i=1}^5 L_i^2
+\frac{1}{8} \sum_{i=1}^5 L_i^4 + \frac12 \sum_{1\le i<j\le 5} L_i^2L_j^2, \nonumber\\
 \VO^{\mathrm{MV}}_{1,3} =
 \frac{11\pi^6}{96} + \frac{\pi^4}{16} \sum_{i=1}^3 L_i^2
+ \frac{13\pi^2}{1152} \sum_{i=1}^3 L_i^4
+ \frac{\pi^2}{24} \sum_{1\le i<j\le 3} L_i^2L_j^2
\nonumber\\
\hphantom{\VO^{\mathrm{MV}}_{1,3} =}{}
 + \frac{1}{1152} \sum_{i=1}^3 L_i^6
+ \frac{1}{192}
\sum_{\begin{subarray}{c}
1\le i, j\le 3\\ (i\ne j)
\end{subarray}}
L_i^2 L_j^4
+ \frac{1}{96} L_1^2L_2^2L_3^2, \nonumber\\
 \VO^{\mathrm{MV}}_{2,1} =
 \frac{29\pi^8}{2560} + \frac{\pi^6}{192}L_1^2
+ \frac{119\pi^4}{138240}L_1^4 + \frac{\pi^2}{13824}L_1^6
+ \frac{1}{442368}L_1^8 .\label{lst_mv}
\end{gather}
Twisted Weil--Petersson volumes:
\begin{gather}
 \VO^{\mathrm{WP}}_{0,3}\big[\sfm\big] = 1,\nonumber\\
 \VO^{\mathrm{WP}}_{1,1}\big[\sfm\big] =
 \frac{(s+1)\pi^2}{12} + \frac{1}{48}L_1^2,\nonumber\\
 \VO^{\mathrm{WP}}_{0,4}\big[\sfm\big] =
 \frac{(4s+1)\pi^2}{2} + \frac12\sum_{i=1}^4 L_i^2,\nonumber\\
 \VO^{\mathrm{WP}}_{1,2}\big[\sfm\big] =
 \frac{\bigl(36 s^{2}+26 s+9\bigr)\pi^4}{144}
+ \frac{(2s+1)\pi^2}{24}\sum_{i=1}^2 L_i^2
+ \frac{1}{192}\sum_{i=1}^2 L_i^4 + \frac{1}{96} L_1^2L_2^2,\nonumber\\
 \VO^{\mathrm{WP}}_{0,5}\big[\sfm\big] =
 \frac{\bigl(120 s^{2}+40 s+9\bigr)\pi^4}{12}
+ \frac{(6s + 1)\pi^2}{2}\sum_{i=1}^5 L_i^2
+ \frac{1}{8} \sum_{i=1}^5 L_i^4 + \frac12 \sum_{1\le i<j\le 5} L_i^2L_j^2,\nonumber\\
 \VO^{\mathrm{WP}}_{1,3}\big[\sfm\big] =
 \frac{\bigl(448 s^{3}+284 s^{2}+122 s+33\bigr)\pi^6}{288}
+ \frac{\bigl(26 s^{2}+13 s+3\bigr)\pi^4}{48}\sum_{i=1}^3 L_i^2
\nonumber\\
\hphantom{\VO^{\mathrm{WP}}_{1,3}\big[\sfm\big] =}{}
 + \frac{(48 s+13)\pi^2}{1152} \sum_{i=1}^3 L_i^4
+ \frac{(3 s+1)\pi^2}{24} \sum_{1\le i<j\le 3} L_i^2L_j^2
\nonumber\\
\hphantom{\VO^{\mathrm{WP}}_{1,3}\big[\sfm\big] =}{} + \frac{1}{1152} \sum_{i=1}^3 L_i^6
+ \frac{1}{192}
\sum_{\begin{subarray}{c}
1\le i, j\le 3\\ (i\ne j)
\end{subarray}}
L_i^2 L_j^4
+ \frac{1}{96} L_1^2L_2^2L_3^2,\nonumber\\
 \VO^{\mathrm{WP}}_{2,1}\big[\sfm\big] =
 \frac{\bigl(31320 s^{4}+27600 s^{3}+16796 s^{2}+8100 s+2349\bigr)\pi^8}{207360}\nonumber\\
\hphantom{\VO^{\mathrm{WP}}_{2,1}\big[\sfm\big] =}{}
+ \frac{\bigl(1014 s^{3}+800 s^{2}+357 s+90\bigr)\pi^6}{17280} L_1^2 + \frac{(834 s^{2}+490 s+119)\pi^4}{138240} L_1^4
\nonumber\\
\hphantom{\VO^{\mathrm{WP}}_{2,1}\big[\sfm\big] =}{}
+ \frac{(29 s+10)\pi^2}{138240} L_1^6 + \frac{1}{442368}L_1^8 .\label{lst_wp_tw_s}
\end{gather}
Twisted volume polynomials of the $(2,p)$ minimal string:
\begin{gather}
 \VO^{\mathrm{M}(p)}_{0,3}\big[\sfm\big] = 1,\nonumber\\
 \VO^{\mathrm{M}(p)}_{1,1}\big[\sfm\big] =
 \left[s\left(1-\frac{1}{p^2}\right)+1\right]\frac{\pi^2}{12}
+ \frac{1}{48}L_1^2,\nonumber\\
 \VO^{\mathrm{M}(p)}_{0,4}\big[\sfm\big] =
 \left[4s\left(1-\frac{1}{p^2}\right)+1\right]\frac{\pi^2}{2}
+ \frac12\sum_{i=1}^4 L_i^2,\nonumber\\
 \VO^{\mathrm{M}(p)}_{1,2}\big[\sfm\big] =
 \left[\left(1-\frac{1}{p^2}\right)
\left(12s^2\left(3+\frac{5}{p^2}\right)+26s\right)+9\right]\frac{\pi^4}{144}\nonumber\\
\hphantom{\VO^{\mathrm{M}(p)}_{1,2}\big[\sfm\big] =}{}
+ \left[2s\left(1-\frac{1}{p^2}\right)+1\right]
\frac{\pi^2}{24}\sum_{i=1}^2 L_i^2
 + \frac{1}{192}\sum_{i=1}^2 L_i^4 + \frac{1}{96} L_1^2L_2^2,\nonumber\\
 \VO^{\mathrm{M}(p)}_{0,5}\big[\sfm\big] =
 \left[\left(1-\frac{1}{p^2}\right)
\left(24s^2\left(5+\frac{3}{p^2}\right)+40s\right)+9\right]\frac{\pi^4}{12}\nonumber\\
\hphantom{\VO^{\mathrm{M}(p)}_{0,5}\big[\sfm\big] =}{}
+ \left[6s\left(1-\frac{1}{p^2}\right)+1\right]\frac{\pi^2}{2}\sum_{i=1}^5 L_i^2
 + \frac{1}{8} \sum_{i=1}^5 L_i^4 + \frac12 \sum_{1\le i<j\le 5} L_i^2L_j^2,\nonumber\\
 \VO^{\mathrm{M}(p)}_{1,3}\big[\sfm\big] =
 \left[\left(1-\frac{1}{p^2}\right)
\left(64s^3\left(7+\frac{20}{p^2}+\frac{21}{p^4}\right)
+4s^2\left(71+\frac{49}{p^2}\right)+122s \right) + 33 \right]
\frac{\pi^6}{288}
\nonumber\\
\hphantom{\VO^{\mathrm{M}(p)}_{1,3}\big[\sfm\big] =}{}
 + \left[\left(1-\frac{1}{p^2}\right)
\left(2s^2\left(13+\frac{11}{p^2}\right) + 13s \right) + 3 \right]
\frac{\pi^4}{48}\sum_{i=1}^3 L_i^2
\nonumber\\
\hphantom{\VO^{\mathrm{M}(p)}_{1,3}\big[\sfm\big] =}{} + \left[48s\left(1-\frac{1}{p^2}\right)+13\right]
\frac{\pi^2}{1152} \sum_{i=1}^3 L_i^4
+ \left[3s\left(1-\frac{1}{p^2}\right)+1\right]
\frac{\pi^2}{24} \sum_{1\le i<j\le 3} L_i^2L_j^2
\nonumber\\
\hphantom{\VO^{\mathrm{M}(p)}_{1,3}\big[\sfm\big] =}{} + \frac{1}{1152} \sum_{i=1}^3 L_i^6
+ \frac{1}{192}
\sum_{\begin{subarray}{c}
1\le i, j\le 3\\ (i\ne j)
\end{subarray}}
L_i^2 L_j^4
+ \frac{1}{96} L_1^2L_2^2L_3^2,\nonumber\\
 \VO^{\mathrm{M}(p)}_{2,1}\big[\sfm\big] =
 \left[\left(1-\frac{1}{p^2}\right)
\left(72s^4\left(435+\frac{2423}{p^2}+\frac{5945}{p^4}+\frac{6557}{p^6}\right)\right.\right.\nonumber\\
 \left.\left.
 \hphantom{\VO^{\mathrm{M}(p)}_{2,1}\big[\sfm\big] =}{}
+ 240s^3\!\left(115 + \frac{322}{p^2} + \frac{331}{p^4}\right)\!
+ 68s^2\!\left(247 + \frac{185}{p^2}\right)\! +8100s \right)\! +2349\right]\!
\frac{\pi^8}{207360}
\nonumber\\
\hphantom{\VO^{\mathrm{M}(p)}_{2,1}\big[\sfm\big] =}{} + \left[\left(1-\frac{1}{p^2}\right)
\left(6s^3\left(169+\frac{430}{p^2}+\frac{361}{p^4}\right)\right.\right.\nonumber\\
 \left.\left.
 \hphantom{\VO^{\mathrm{M}(p)}_{2,1}\big[\sfm\big] =}{}
+ 160s^2\left(5+\frac{4}{p^2}\right) + 357s \right) +90\right]
\frac{\pi^6}{17280} L_1^2
\nonumber\\
\hphantom{\VO^{\mathrm{M}(p)}_{2,1}\big[\sfm\big] =}{}
 + \left[\left(1-\frac{1}{p^2}\right)
\left(6s^2\left(139 + \frac{93}{p^2}\right) +490s \right) + 119\right]
\frac{\pi^4}{138240} L_1^4
\nonumber\\
\hphantom{\VO^{\mathrm{M}(p)}_{2,1}\big[\sfm\big] =}{}
 + \left[29s\left(1-\frac{1}{p^2}\right)+10\right]\frac{\pi^2}{138240} L_1^6
+ \frac{1}{442368}L_1^8 .\label{lst_mst_twist_s}
\end{gather}
Super Masur--Veech polynomials:
\begin{gather}
 \VO^{\mathrm{SMV}}_{0,n} = 0,\nonumber\\
 \VO^{\mathrm{SMV}}_{1,n} = \frac{(n-1)!}{8},\nonumber\\
 \VO^{\mathrm{SMV}}_{2,1} =
 \frac{3\pi^2}{128}+\frac{3}{256}L_1^2,\nonumber\\
 \VO^{\mathrm{SMV}}_{2,2} =
 \frac{9\pi^2}{128}+\frac{9}{256}\sum_{i=1}^2 L_i^2,\nonumber\\
 \VO^{\mathrm{SMV}}_{2,3} =
 \frac{9\pi^2}{32}+\frac{9}{64}\sum_{i=1}^3 L_i^2,\nonumber\\
 \VO^{\mathrm{SMV}}_{2,4} =
 \frac{45\pi^2}{32}+\frac{45}{64}\sum_{i=1}^4 L_i^2,\nonumber\\
 \VO^{\mathrm{SMV}}_{2,5} =
 \frac{135\pi^2}{16}+\frac{135}{32}\sum_{i=1}^5 L_i^2,\nonumber\\
 \VO^{\mathrm{SMV}}_{3,1} =
 \frac{23\pi^4}{1024}+\frac{51\pi^2}{4096}L_1^2+\frac{15}{8192}L_1^4,\nonumber\\
 \VO^{\mathrm{SMV}}_{3,2} =
 \frac{115\pi^4}{1024}+
\frac{225\pi^2}{4096}\sum_{i=1}^2 L_i^2 +
\frac{75}{8192}\sum_{i=1}^2 L_i^4 + \frac{63}{2048} L_1^2L_2^2,\nonumber\\
 \VO^{\mathrm{SMV}}_{3,3} =
 \frac{345\pi^4}{512}+
\frac{765\pi^2}{2048}\sum_{i=1}^3 L_i^2 +
\frac{225}{4096}\sum_{i=1}^3 L_i^4 + \frac{189}{1024}
\sum_{1\le i<j\le 3} L_i^2L_j^2,\nonumber\\
 \VO^{\mathrm{SMV}}_{4,1} =
 \frac{1827\pi^6}{32768} + \frac{473\pi^4}{16384} L_1^2+
\frac{625\pi^2}{131072} L_1^4 +\frac{175}{524288} L_1^6 .\label{lst_smv}
\end{gather}
Twisted super Weil--Petersson volumes:
\begin{gather}
 \VO^{\mathrm{SWP}}_{0,n}\big[\sfm\big] = 0,\nonumber\\
 \VO^{\mathrm{SWP}}_{1,n}\big[\sfm\big] =
 \frac{(n-1)!}{8},\nonumber\\
 \VO^{\mathrm{SWP}}_{2,1}\big[\sfm\big] =
 \frac{3(6s+1)\pi^2}{128}+\frac{3}{256}L_1^2,\nonumber\\
 \VO^{\mathrm{SWP}}_{2,2}\big[\sfm\big] =
 \frac{9(8s+1)\pi^2}{128}+\frac{9}{256}\sum_{i=1}^2 L_i^2,\nonumber\\
 \VO^{\mathrm{SWP}}_{2,3}\big[\sfm\big] =
 \frac{9(10s+1)\pi^2}{32}+\frac{9}{64}\sum_{i=1}^3 L_i^2,\nonumber\\
 \VO^{\mathrm{SWP}}_{2,4}\big[\sfm\big] =
 \frac{45(12s+1)\pi^2}{32}+\frac{45}{64}\sum_{i=1}^4 L_i^2,\nonumber\\
 \VO^{\mathrm{SWP}}_{2,5}\big[\sfm\big] =
 \frac{135(14s+1)\pi^2}{16}+\frac{135}{32}\sum_{i=1}^5 L_i^2,\nonumber\\
 \VO^{\mathrm{SWP}}_{3,1}\big[\sfm\big] =
 \frac{\bigl(1362 s^{2}+255 s+23\bigr)\pi^4}{1024}
+\frac{3(168 s+17)\pi^2}{4096}L_1^2+\frac{15}{8192}L_1^4,\nonumber\\
 \VO^{\mathrm{SWP}}_{3,2}\big[\sfm\big] =
 \frac{\bigl(9684 s^{2}+1530 s+115\bigr)\pi^4}{1024}+
\frac{3(1008s+85)\pi^2}{4096}\sum_{i=1}^2 L_i^2
\nonumber\\
\hphantom{\VO^{\mathrm{SWP}}_{3,2}\big[\sfm\big] =}{}
 + \frac{75}{8192}\sum_{i=1}^2 L_i^4 + \frac{63}{2048} L_1^2L_2^2,\nonumber\\
 \VO^{\mathrm{SWP}}_{3,3}\big[\sfm\big] =
 \frac{3\bigl(13062 s^{2}+1785 s+115\bigr)\pi^4}{512}+
\frac{9(1176 s+85)\pi^2}{2048}\sum_{i=1}^3 L_i^2
\nonumber\\
\hphantom{\VO^{\mathrm{SWP}}_{3,3}\big[\sfm\big] =}{}
+ \frac{225}{4096}\sum_{i=1}^3 L_i^4 + \frac{189}{1024}
\sum_{1\le i<j\le 3} L_i^2L_j^2,\nonumber\\
 \VO^{\mathrm{SWP}}_{4,1}\big[\sfm\big] =
 \frac{\bigl(1115332 s^{3}+216788 s^{2}+26488 s+1827\bigr)\pi^6}{32768}\nonumber\\
 \hphantom{\VO^{\mathrm{SWP}}_{4,1}\big[\sfm\big] =}{}
 +
\frac{\bigl(106911 s^{2}+14643 s+946\bigr)\pi^4}{32768} L_1^2
 + \frac{125(69 s+5)\pi^2}{131072} L_1^4 +\frac{175}{524288} L_1^6 .\!\!\!\label{lst_super_wp_tw_s}
\end{gather}
Twisted volume polynomials of the $(2,2p-2)$ minimal superstring:
\begin{gather}
 \VO^{\mathrm{SM}(p)}_{0,n}\big[\sfm\big] = 0,\nonumber\\
 \VO^{\mathrm{SM}(p)}_{1,n}\big[\sfm\big] =
 \frac{(n-1)!}{8},\nonumber\\
 \VO^{\mathrm{SM}(p)}_{2,1}\big[\sfm\big] =
 \left[18s\left(1-\frac{1}{p^2}\right)+3\right]\frac{\pi^2}{128}
+\frac{3}{256}L_1^2,\nonumber\\
 \VO^{\mathrm{SM}(p)}_{2,2}\big[\sfm\big] =
 \left[72s\left(1-\frac{1}{p^2}\right)+9\right]\frac{\pi^2}{128}
+\frac{9}{256}\sum_{i=1}^2 L_i^2,\nonumber\\
 \VO^{\mathrm{SM}(p)}_{2,3}\big[\sfm\big] =
 \left[90s\left(1-\frac{1}{p^2}\right)+9\right]\frac{\pi^2}{32}
+\frac{9}{64}\sum_{i=1}^3 L_i^2,\nonumber\\
 \VO^{\mathrm{SM}(p)}_{2,4}\big[\sfm\big] =
 \left[540s\left(1-\frac{1}{p^2}\right)+45\right]\frac{\pi^2}{32}
+ \frac{45}{64}\sum_{i=1}^4 L_i^2,\nonumber\\
 \VO^{\mathrm{SM}(p)}_{2,5}\big[\sfm\big] =
 \left[1890s\left(1-\frac{1}{p^2}\right)+135 \right]\frac{\pi^2}{16}
+ \frac{135}{32}\sum_{i=1}^5 L_i^2,\nonumber\\
 \VO^{\mathrm{SM}(p)}_{3,1}\big[\sfm\big] =
 \left[\left(1-\frac{1}{p^2}\right)
\left(6s^2\left(227-\frac{27}{p^2}\right)+255s \right)+23 \right]
\frac{\pi^4}{1024}
\nonumber\\
\hphantom{\VO^{\mathrm{SM}(p)}_{3,1}\big[\sfm\big] =}{}
 + \left[504s\left(1-\frac{1}{p^2}\right)+51\right]\frac{\pi^2}{4096}L_1^2
+ \frac{15}{8192}L_1^4,\nonumber\\
 \VO^{\mathrm{SM}(p)}_{3,2}\big[\sfm\big] =
 \left[\left(1-\frac{1}{p^2}\right)
\left(36s^2\left(269-\frac{69}{p^2}\right)+1530s \right)+115 \right]
\frac{\pi^4}{1024}
\nonumber\\
\hphantom{\VO^{\mathrm{SM}(p)}_{3,2}\big[\sfm\big] =}{}
 + \left[3024s\left(1-\frac{1}{p^2}\right)+255 \right]
\frac{\pi^2}{4096}\sum_{i=1}^2 L_i^2
+ \frac{75}{8192}\sum_{i=1}^2 L_i^4 + \frac{63}{2048} L_1^2L_2^2,\nonumber\\
 \VO^{\mathrm{SM}(p)}_{3,3}\big[\sfm\big] =
 \left[\left(1-\frac{1}{p^2}\right)
\left(126s^2\left(311 - \frac{111}{p^2}\right)+5355s \right)+345 \right]
\frac{\pi^4}{512}
\nonumber\\
 \hphantom{\VO^{\mathrm{SM}(p)}_{3,3}\big[\sfm\big] =}{}
 + \left[10584s\left(1-\frac{1}{p^2}\right)+765 \right]
\frac{\pi^2}{2048}\sum_{i=1}^3 L_i^2
+ \frac{225}{4096}\sum_{i=1}^3 L_i^4
\nonumber\\
\hphantom{\VO^{\mathrm{SM}(p)}_{3,3}\big[\sfm\big] =}{}
 + \frac{189}{1024}
\sum_{1\le i<j\le 3} L_i^2L_j^2,\nonumber\\
 \VO^{\mathrm{SM}(p)}_{4,1}\big[\sfm\big] =
 \left[\left(1-\frac{1}{p^2}\right)
\left(4s^3\left(278833 - \frac{44866}{p^2} + \frac{1233}{p^4}\right)
\right.\right.
\nonumber\\
 \hphantom{\VO^{\mathrm{SM}(p)}_{4,1}\big[\sfm\big] =}{}
+ 4s^2\left.\left.\left(54197 - \frac{19197}{p^2}\right)+26488s \right)+1827 \right]
\frac{\pi^6}{32768}
\nonumber\\
\hphantom{\VO^{\mathrm{SM}(p)}_{4,1}\big[\sfm\big] =}{}
 + \left[\left(1-\frac{1}{p^2}\right)
\left(3s^2\left(35637 - \frac{12637}{p^2}\right)+14643s \right)+946 \right]
\frac{\pi^4}{32768} L_1^2
\nonumber\\
\hphantom{\VO^{\mathrm{SM}(p)}_{4,1}\big[\sfm\big] =}{}
 + \left[8625s\left(1-\frac{1}{p^2}\right)+625 \right]\frac{\pi^2}{131072} L_1^4
+ \frac{175}{524288} L_1^6 .\label{lst_super_mst_twist_s}
\end{gather}

\subsection*{Acknowledgements}

Authors thank Kohei Iwaki, Kazumi Okuyama, Kento Osuga, and Yuji Terashima for discussions and useful comments.
One of the authors (H.F.) is grateful to J{\o}rgen Ellegaard Andersen for instructive discussions on the geometric recursions.
The authors also thank the anonymous referees for helpful and constructive comments which improved the quality of this manuscript.
This work was supported by JSPS KAKENHI Grant Numbers JP18K03281, JP20K03601, JP20K03931, JP21H04994, JP22H01117.

\pdfbookmark[1]{References}{ref}
\LastPageEnding

\end{document}